\documentclass[conference]{IEEEtran}
\IEEEoverridecommandlockouts
\usepackage{cite}
\usepackage{amsmath,amssymb,amsfonts}
\usepackage{algorithmic}
\usepackage{graphicx}
\usepackage{textcomp}
\usepackage{xcolor}
\def\BibTeX{{\rm B\kern-.05em{\sc i\kern-.025em b}\kern-.08em
    T\kern-.1667em\lower.7ex\hbox{E}\kern-.125emX}}
\usepackage[hyphens]{url}
\usepackage{hyperref}
\usepackage{cleveref}
\usepackage{xargs}
\usepackage{listingsutf8}
\usepackage{mathtools}

\usepackage{xspace}
\usepackage{amsthm}
\usepackage{booktabs}
\usepackage{wrapfig}
\usepackage{paralist}
\usepackage{multirow}

\newif\ifcomments
\newif\ifshowremoved

\usepackage{xcolor}
\usepackage[absolute]{textpos}
\usepackage{everypage}

\definecolor{lightblue}{RGB}{210,210,225}
\definecolor{lightred}{RGB}{225,210,210}
\definecolor{lightgreen}{RGB}{210,225,210}
\definecolor{lightyellow}{RGB}{225,222,200}
\definecolor{lightpurple}{RGB}{225,210,225}
\definecolor{warningyellow}{RGB}{247, 245, 187}
\definecolor{darkergreen}{RGB}{0,64,0}
\definecolor{darkred}{RGB}{128,0,0}
\definecolor{darkblue}{RGB}{0,0,139}
\definecolor{darkgreen}{RGB}{0,128,0}
\definecolor{darkpurple}{RGB}{128,0,128}
\definecolor{warningorange}{RGB}{124, 81, 0}
\definecolor{eyecancerpink}{rgb}{1.0, 0.0, 1.0}
\definecolor{radiationyellow}{rgb}{0.8, 1.0, 0.0}

\def\enablecomments{
  \commentstrue
  \enableheader
}

\def\enableheader{
\AddEverypageHook{{
\begin{textblock*}{\paperwidth}(0mm,2mm) {
\centering
\textcolor{gray}{\textbf{
\Huge! D R A F T !
}\\\vspace*{1mm}
\footnotesize
}
}
\end{textblock*}
}}\par
\AddEverypageHook{{
\begin{textblock*}{\paperwidth}(0mm,9mm) {
\centering
\textcolor{gray}{
\footnotesize
Disable author comments to remove this watermark.
}\\
}
\end{textblock*}
}}\par
}

\makeatletter
\def\THICKhrulefill{\leavevmode \leaders \hrule height 5pt\hfill \kern \z@}
\makeatother
\newcommand{\highlightedremark}[4]{%
  \begin{center}\fcolorbox{#1}{#2}{%
  \begin{minipage}{.98\linewidth}\color{#1}%
  \textbf{\THICKhrulefill[ #3 ]\THICKhrulefill}%
  \par\noindent#4\end{minipage}}\end{center}%
}

\newcommand{\removed}[1]{\ifcomments\highlightedremark{warningorange}{warningyellow}{REMOVED CONTENT}{\ifshowremoved #1 \else Add ``\textbackslash showremovedtrue'' below ``\textbackslash enablecomments'' to display removed content. \fi}\fi}
\renewcommand{\removed}[1]{}

\newif\ifdraft
\newif\ifextended
\extendedtrue
\draftfalse

\definecolor{lightgray}{rgb}{.9,.9,.9}
\definecolor{darkgray}{rgb}{.4,.4,.4}
\definecolor{purple}{rgb}{0.65, 0.12, 0.82}
\definecolor{darkgreen}{rgb}{0, 0.5, 0}
\definecolor{turquoise}{rgb}{0, 0.5, 0.5}
\definecolor{plum}{rgb}{.4, .14, .37}
\definecolor{mediumgreen}{HTML}{009900}
\definecolor{mediumred}{HTML}{CC0000}

\ifdraft
\usepackage[users=SBHM]{uchanges}
\enablecomments
\else
\usepackage[users=SBHM, hide] {uchanges} 
\fi

\ChangesSetUserColor{S}{blue}
\ChangesSetUserColor{B}{red}
\ChangesSetUserColor{H}{darkgreen}
\ChangesSetUserColor{M}{turquoise}

\newcommand{\define}{:=}
\newcommand{\nil}{\epsilon}
\newcommand{\cons}[2]{{#1}::{#2}}
\newcommand{\BB}{\mathbb{B}}
\newcommand{\NN}{\mathbb{N}}
\newcommand{\arrayof}[1]{[{#1}]}
\newcommand{\sequenceof}[1]{\mathcal{L}(#1)}
\newcommand{\setof}[1]{\mathcal{P}(#1)}
\newcommand{\lam}[2]{\lambda {#1}. \, {#2}}
\newcommand{\none}{\bot}
\newcommand{\some}[1]{\left \lfloor {#1} \right \rfloor}
\newcommand{\access}[2]{{#1}. \mathbf{#2}} 
\newcommand{\size}[1]{|{#1}|}

\newcommand{\domain}[1]{\mathcal{D}({#1})}

\newcommand{\bytearray}{\arrayof{\BB^8}}

\newcommand{\integer}[1]{\NN_{#1}}
\newcommand{\logevents}{\textit{Ev}_{\textit{log}}}

\newcommand{\accounts}{\mathbb{A}}
\newcommand{\blockheaders}{\mathcal{H}}

\newcommand{\teffects}{N}
\newcommand{\transenvs}{\mathcal{T}_{\textit{env}}}

\newcommand{\exenvs}{I}
\newcommand{\mstates}{M}

\newcommand{\stackof}[1]{\mathcal{L}(#1)}
\newcommand{\callstacks}{\mathbb{S}}

\newcommand{\gstates}{\Sigma}

\newcommand{\project}[2]{{#1}\downarrow_{#2}}
\newcommand{\filtercallscreates}[1]{\textsf{calls}_{#1}}

\newcommand{\sstep}[3]{{#1} \vDash {#2} \, \rightarrow {#3}} 
\newcommand{\ssteps}[3]{{#1} \vDash {#2} \, \rightarrow^* {#3}} 
\newcommand{\ssteptrace}[4]{{#1} \vDash {#2} \, \xrightarrow[]{#4} {#3}} 
\newcommand{\sstepstrace}[4]{{#1} \vDash {#2} \, \xrightarrow[]{#4}^* {#3}}

\newcommand{\callstack}{S} 		
\newcommand{\transenv}{\Gamma} 	 
\newcommand{\originator}{\textsf{o}} 
\newcommand{\gasprize}{\textit{prize}}
\newcommand{\blockheader}{H}
\newcommand{\parent}{\textit{parent}}
\newcommand{\beneficiary}{\textit{beneficiary}}
\newcommand{\difficulty}{\textit{difficulty}}
\newcommand{\blocknumber}{\textit{number}}
\newcommand{\gaslimit}{\textit{gaslimit}}
\newcommand{\timestamp}{\textit{timestamp}}
\newcommand{\mstate}{\mu}
\newcommand{\exenv}{\iota}
\newcommand{\gstate}{\sigma}
\newcommand{\regstate}[3]{(#1, #2, #3)}
\newcommand{\regstatefull}[4]{(#1, #2, #3)}
\newcommand{\excstate}{\textit{EXC}}
\newcommand{\haltstate}[4]{\textit{HALT}(#1, #2, #3)}

\newcommand{\lgas}{\textit{g}}
\newcommand{\callstackplain}{S_{\textit{P}}}
\newcommand{\smstate}[5]{(#1, #2, #3, #4, #5)}
\newcommand{\sexenv}[5]{(#1, #2, \allowbreak #3, #4, #5)}
\newcommand{\accountstate}[4]{(#1, #2, #3, #4)}

\newcommand{\gas}{\textsf{g}}
\newcommand{\data}{\textsf{d}}
\newcommand{\pc}{\textsf{pc}}
\newcommand{\stack}{\textsf{s}}

\newcommand{\memo}{\textsf{m}}
\newcommand{\activeaccount}{\textsf{actor}}
\newcommand{\inputdata}{\textsf{input}}
\newcommand{\sender}{\textsf{sender}}
\newcommand{\tvalue}{\textsf{value}}
\newcommand{\activecode}{\textsf{code}}
\newcommand{\nonce}{\textsf{n}}
\newcommand{\balance}{\textsf{b}}

\newcommand{\accountcode}{\textsf{code}}

\newcommand{\transeffects}{\eta}
\newcommand{\exconf}[2]{(#1, #2)} 

\newcommand{\concatstack}[2]{{#1}++{#2} }

\newcommand{\recipient}{\textit{to}}

\newcommand{\contractaddress}{a} 	
\newcommand{\contractcode}{\textit{code}}

\newcommand{\annotate}[2]{{#1}_{#2}}
\renewcommand{\annotate}[2]{#1_#2}
\newcommand{\exstate}{s}

\newcommand{\finalstate}[1]{\textit{final} \, ({#1})}

\newcommand{\equalupto}[1]{=_{/{#1}}}

\newcommand{\nif}{n_{\textit{if}}}

\newcommand{\ISZERO}{\textsf{ISZERO}}
\newcommand{\SUICIDE}{\textsf{SELFDESTRUCT}}

\newcommand{\ADD}{\textsf{ADD}}
\newcommand{\CALLCODE}{\textsf{CALLCODE}}
\newcommand{\DELEGATECALL}{\textsf{DELEGATECALL}}

\newcommand{\pluseq}{\mathrel{+}=}
\newcommand{\minuseq}{\mathrel{-}=}
\newcommand{\cond}[3]{{#1} \, ? \,  {#2}\, : \,  {#3}}
\newcommand{\mini}[2]{\textit{min} \, (#1, #2)}
\newcommand{\update}[3]{{#1}[{#2} \rightarrow {#3}]}

\newcommand{\arraypos}[2]{{#1} \, [#2]}

\newcommand{\updategstate}[3]{{#1} \big \langle {#2} \rightarrow {#3} \big \rangle}

\newcommand{\concat}[2]{{#1} \cdot {#2}}
\newcommand{\extract}[3]{{#1}[#2, #3]}
\newcommand{\inc}[3]{{#1}[{#2}\pluseq{#3}]}
\newcommand{\dec}[3]{{#1}[{#2}\minuseq{#3}]}

\newcommand{\account}[4]{(#1, #2, #3, #4)}

\newcommand{\suicideset}{\textit{S}_\dagger}

\renewcommand{\exconf}[2]{#1}

\newcommand{\CALL}{\textsf{CALL}}
\newcommand{\STOP}{\textsf{STOP}}
\newcommand{\RETURN}{\textsf{RETURN}}
\newcommand{\CREATE}{\textsf{CREATE}}
\newcommand{\SHA}{\textsf{SHA3}}
\newcommand{\BLOCKHASH}{\textsf{BLOCKHASH}}
\newcommand{\JUMP}{\textsf{JUMP}}
\newcommand{\JUMPI}{\textsf{JUMPI}}
\newcommand{\EXTCODESIZE}{\textsf{EXTCODESIZE}}
\newcommand{\EXTCODECOPY}{\textsf{EXTCODECOPY}}

\newcommand{\curropcode}[2]{\omega_{#1, #2}}

\newcommand{\valu}{\textit{va}}
\newcommand{\io}{\textit{io}}
\newcommand{\is}{\textit{is}}
\newcommand{\oo}{\textit{oo}}
\newcommand{\os}{\textit{os}}
\newcommand{\aw}{\textit{aw}}
\newcommand{\addr}{\textit{addr}}

\newcommand{\emptyarray}{\epsilon}

\newcommand{\memext}[3]{M\,(#1, #2, #3)} 

\newcommand{\costmem}[2]{C_{\textit{mem}} \, (#1, #2)}

\newcommand{\costs}{\textit{c}}
\newcommand{\cd}{\textit{cd}}

\newcommandx{\pcheckpremscallhelp}[6]{\textsf{checkPrems}^{\textsf{\CALL}}_{\tempa} \, ({\tempb}, {\tempc}, {\tempd}, {\tempe}, {\tempf}, {\tempg}, {\temph}, {\tempi}, {#1}, {#2}, {#3}, {#4}, {#5}, {#6})}

\newcommand{\pos}{\textit{pos}}

\newcommand{\lpc}{\textit{pc}}
\newcommand{\code}{\textit{code}}
\newcommand{\stor}{\textit{stor}}
\newcommand{\siz}{\textit{size}}

\newcommand{\arrayinterval}[3]{{#1}[#2, #3]}

\newcommand{\accountnoncev}{\textit{nonce}}
\newcommand{\accountbalancev}{\textit{balance}}
\newcommand{\accountstorv}{\textit{stor}}
\newcommand{\accountcodev}{\textit{code}}

\newcommand{\datav}{d}

\newcommand{\downv}{\textit{down}}
\newcommand{\upv}{\textit{up}}

\newcommand{\bitstringtobytearray}[1]{{#1}_{\bytearray}}
\newcommand{\bytearraytobitstring}[1]{{#1}_\BB}

\newcommand{\sinteger}[1]{\textit{Int}_{#1}}
\newcommand{\callstacksplain}{\callstacks_{\textit{plain}}}
\newcommand{\accountv}{\textit{acc}}

\newcommand{\binopv}{\textit{i}_\textit{bin}}
\newcommand{\binops}{\textit{Inst}_{\textit{bin}}}
\newcommand{\binopcost}[1]{\text{cost}_\textit{bin}(#1)}
\newcommand{\binopfun}[1]{\text{fun}_\textit{bin}(#1)}
\newcommand{\signed}[1]{{#1}^-}
\newcommand{\unsigned}[1]{{#1}^+}
\newcommand{\signof}[1]{\textit{sign}{(#1)}}
\newcommand{\bitand}{\&} 
\newcommand{\bitor}{\|}
\newcommand{\bitxor}{\oplus}
\newcommand{\bitneg}{\neg}
\newcommand{\keccak}[1]{\textit{Keccak}(#1)}

\newcommand{\funP}[3]{P\,(#1, #2, #3)}

\newcommand{\parentc}{\textsf{parent}}
\newcommand{\beneficiaryc}{\textsf{beneficiary}}
\newcommand{\difficultyc}{\textsf{difficulty}}  
\newcommand{\blocknumberc}{\textsf{number}}
\newcommand{\gaslimitc}{\textsf{gaslimit}}
\newcommand{\timestampc}{\textsf{timestamp}}

\newcommand{\SUB}{\textsf{SUB}}
\newcommand{\LT}{\textsf{LT}}
\newcommand{\GT}{\textsf{GT}}
\newcommand{\EQ}{\textsf{EQ}}
\newcommand{\AND}{\textsf{AND}}
\newcommand{\OR}{\textsf{OR}}
\newcommand{\XOR}{\textsf{XOR}}
\newcommand{\SLT}{\textsf{SLT}}
\newcommand{\SGT}{\textsf{SGT}}
\newcommand{\MUL}{\textsf{MUL}}
\newcommand{\DIV}{\textsf{DIV}}
\newcommand{\SDIV}{\textsf{SDIV}}
\newcommand{\MOD}{\textsf{MOD}}
\newcommand{\SMOD}{\textsf{SMOD}}
\newcommand{\SIGNEXTEND}{\textsf{SIGNEXTEND}}
\newcommand{\BYTE}{\textsf{BYTE}}
\newcommand{\ADDMOD}{\textsf{ADDMOD}}
\newcommand{\MULMOD}{\textsf{MULMOD}}
\newcommand{\NOT}{\textsf{NOT}}
\newcommand{\EXP}{\textsf{EXP}}
\newcommand{\ADDRESS}{\textsf{ADDRESS}}
\newcommand{\CALLER}{\textsf{CALLER}}
\newcommand{\CALLVALUE}{\textsf{CALLVALUE}}
\newcommand{\CODESIZE}{\textsf{CODESIZE}}
\newcommand{\BALANCE}{\textsf{BALANCE}}
\newcommand{\ORIGIN}{\textsf{ORIGIN}}
\newcommand{\CALLDATASIZE}{\textsf{CALLDATASIZE}}
\newcommand{\CALLDATALOAD}{\textsf{CALLDATALOAD}}
\newcommand{\CODECOPY}{\textsf{CODECOPY}}
\newcommand{\CALLDATACOPY}{\textsf{CALLDATACOPY}}
\newcommand{\GASPRICE}{\textsf{GASPRICE}}
\newcommand{\COINBASE}{\textsf{COINBASE}}
\newcommand{\TIMESTAMP}{\textsf{TIMESTAMP}}
\newcommand{\NUMBER}{\textsf{NUMBER}}
\newcommand{\DIFFICULTY}{\textsf{DIFFICULTY}}
\newcommand{\GASLIMIT}{\textsf{GASLIMIT}}

\newcommand{\MLOAD}{\textsf{MLOAD}}
\newcommand{\MSTORE}{\textsf{MSTORE}}

\newcommand{\SLOAD}{\textsf{SLOAD}}
\newcommand{\SSTORE}{\textsf{SSTORE}}

\newcommand{\PC}{\textsf{PC}}
\newcommand{\MSIZE}{\textsf{MSIZE}}
\newcommand{\GAS}{\textsf{GAS}}

\newcommand{\LOG}[1]{\textsf{LOG}{#1}}
\newcommand{\INVALID}{\textsf{INVALID}}

\newcommand{\origin}{origin}
\newcommand{\edgestep}[3]{\mathcal{C}, \cd \vDash \,{#2} \, \fedge {#3}} 
\newcommand{\predstep}[3]{\mathcal{C}, \cd \vDash \, {#2} \, \qedge {#3}} 
\newcommand{\stateupdate}[1]{f~=(\lambda \cfgstate. #1)} 
\newcommand{\predupdate}[1]{Q~=(\lambda \cfgstate. #1)} 
\newcommand{\cfgnode}[2]{(#1, #2)} 
\newcommand{\None}{\bot}    

\newcommand{\Def}[1]{\textit{Def}~=~#1}    
\newcommand{\Use}[1]{\textit{Use}~=~#1}    
\newcommand{\Defs}[1]{\textit{Def}~=~\{#1\}}    
\newcommand{\Uses}[1]{\textit{Use}~=~\{#1\} 
}    

\newcommand{\optional}[1]{\lfloor #1 \rfloor}

\newcommand{\memvar}[2]{{#1}^{\cfgmem}}
\newcommand{\stackvar}[1]{{#1}^{\cfgstack}}
\newcommand{\storvar}[2]{{#1}^{\cfgstor}}
\newcommand{\envvar}[1]{{#1}^{e}}
\newcommand{\locenvvar}[1]{{#1}^{\cfglocenv}}
\newcommand{\globenvvar}[1]{{#1}^{\cfgglobenv}}
\newcommand{\extenv}{\textsf{external}}
\newcommand{\precontract}{C}
\newcommand{\instruction}{\textit{op}}
\newcommand{\nextpc}{\lpc_\textit{next}}

\newcommand{\pre}{\textit{pre}}
\newcommand{\toevmstate}{\textit{toEVM}}
\newcommand{\tocfgstate}{\textit{toCFG}}
\newcommand{\cfgstate}{\theta}
\newcommand{\cfgstatefull}{\Theta}
\newcommand{\cfgstatecopy}{\overline{\cfgstate}}
\newcommand{\cfgstatecopybot}{\cfgstatecopy_\bot}
\newcommand{\restrict}[2]{\project{#1}{\domain{#2}}}
\newcommand{\tempvar}[1]{\overline{#1}}
\newcommand{\transactionstep}[3]{#1 \vDash #2 \xrightarrow{T} #3}
\newcommand{\combistep}[3]{#1 \vDash #2 \hookrightarrow #3}
\newcommand{\cfgstep}[4]{#1 \vDash \, {#2} \, -\! {#3} \! \xrightarrow{} {#4}} 
\newcommand{\acfgstep}[3]{{#1} \, -\! {#2} \! \xrightarrow{} {#3}} 
\newcommand{\fact}{\Uparrow\!\! f}
\newcommand{\qact}{(Q)_{\surd}}
\newcommand{\cfgact}{a}
\newcommand{\cfgnodindex}{i}
\newcommand{\cfgmedstep}[3]{#1 \vDash \, {#2} \, \xRightarrow{} {#3}} 
 
\newcommand{\cfgactedge}[1]{-\! {#1}\! \xrightarrow{}}
\newcommand{\cfgconfig}[2]{\langle #1, #2 \rangle}

\newcommand{\acfgconfigstep}[3]{{#1}\, -\! {#2}\! \xrightarrow{}{#3}}
\newcommand{\evmcontract}{C}
\newcommand{\node}{n}
\newcommand{\memvarconc}[2]{{#1}^{m}.\textsf{S}}
\newcommand{\memvarconcname}{\text{\textsf{S}-variable}\xspace}
\newcommand{\memvarconcnames}{\text{\textsf{S}-variables}\xspace}
\newcommand{\memvarabs}[2]{{#1}^{m}.\textsf{D}}
\newcommand{\memvarabsname}{\text{\textsf{D}-variable}\xspace}
\newcommand{\memvarabsnames}{\text{\textsf{D}-variables}\xspace}
\newcommand{\storvarconc}[2]{{#1}^{g}.\textsf{S}}
\newcommand{\storvarabs}[2]{{#1}^{g}.\textsf{D}}
\newcommand{\varconc}[1]{{#1}.\textsf{S}}
\newcommand{\varabs}[1]{{#1}.\textsf{D}}
\newcommand{\storeinstructions}{\textit{Inst}_{\textit{\SSTORE}}}
\newcommand{\bss}[1]{\textit{BS}(#1)}
\newcommand{\lfp}[1]{\textit{lfp}(#1)}
\newcommand{\rules}[1]{\mathcal{R}(#1)}
\newcommand{\pat}{\ensuremath{\textsf{P}}}

\newcommand{\astates}{\Theta}
\newcommand{\edgelist}{\textit{as}}
\newcommand{\defn}{\textit{Def}}
\newcommand{\usen}{\textit{Use}}
\newcommand{\fun}[2]{\lambda{#1}.{#2}}

\newcommand{\updatevarset}[5]{#1 \leftarrow (#2 \define #3)_{#4 \in #5}} 
\newcommand{\updatestate}[3]{#1 \leftarrow #2 \define #3} 
\newcommand{\accesscfgstate}[2]{#1 [{#2}]}
\newcommand{\cfgstack}{\textit{ls}}
\newcommand{\cfgmem}{m}
\newcommand{\cfgstor}{g}
\newcommand{\cfglocenv}{\textit{el}}
\newcommand{\cfgglobenv}{\textit{eg}}
\newcommand{\cfgexternal}{\textit{external}}
\newcommand{\cfgmemconc}{{m.\textsf{S}}}
\newcommand{\cfgmemabs}{{m.\textsf{D}}}
\newcommand{\cfgstorconc}{{g.\textsf{S}}}
\newcommand{\cfgstorabs}{{g.\textsf{D}}}
\newcommand{\load}{\textit{load}}
\newcommand{\exitnode}{\textit{exit}}
\newcommand{\exceptionnode}{\textit{exception}}
\newcommand{\haltnode}{\textit{halt}}
\newcommand{\cfgstatequiv}[1]{\approx_{\cfgstate}^{#1}}
\newcommand{\msize}{\textsf{i}}

\newcommand{\applycall}{\textit{applyCall}}

\newcommand{\gaspc}[1]{\lgas}
\newcommand{\msizepc}[1]{\msize}
\newcommand{\cfgexternalpc}[1]{\cfgexternal}

\newcommand{\maxint}[1]{\textsf{MAX}^{\textit{Int}_{#1}}}
\newcommand{\startvalue}{\textit{start}}
\newcommand{\restartcounter}{\textit{reset}}
\newcommand{\getcallnodeoffset}{\textit{getNodeOffset}}
\newcommand{\loctype}{\mathcal{L}}
\newcommand{\tagtype}{\mathcal{T}}
\newcommand{\instructions}{\mathcal{I}}
\newcommand{\ASSIGN}{\textsf{ASSIGN}}
\newcommand{\STATICCALL}{\textsf{STATICCALL}}
\newcommand{\CREATETWO}{\textsf{CREATE2}}
\newcommand{\applycreate}{\textit{applyCreate}}
\newcommand{\startmemoffset}{o_\textit{mem}}

\newcommand{\isregular}[1]{\textit{isRegular}(#1)}
\newcommand{\ishalt}[1]{\textit{isHalt}(#1)}

\newcommand*\textnode[1]{\tikz[baseline=(char.base)]{
            \node[shape=circle, fill= black, text= white, inner sep=0.8pt] (char) {#1};}}

\newcommand{\tracenoninterference}{\text{trace noninterference}}
\newcommand{\Tracenoninterference}{\text{Trace noninterference}}
\newcommand{\tracenoninter}[3]{\textit{TNI}(#1, #2, #3)}
\newcommand{\inputvars}{Z}
\newcommand{\inputvar}{z}
\newcommand{\inputcomponents}{\textsf{Z}}
\newcommand{\inputcomponent}{\textsf{z}}
\newcommand{\componentToVar}{\textit{toVar}}
\newcommand{\outputproj}{f}

\newcommand{\patnoninter}[3]{\pat_{#1,#2}^{#3}}

\newcommand{\istep}[6]{\langle #1, #2 \rangle \xrightarrow[]{#3}^{#4} \langle #5, #6\rangle }
\newcommand{\nodesucc}[1]{{#1}^+}
\newcommand{\defnode}[1]{n_{#1}}
\newcommand{\defnodes}[1]{N_{#1}}
\newcommand{\independent}[2]{#1 ~\bot~ #2}

\newcommand{\strongconsistent}[2]{\textit{s-consistent}(#1, #2)}
\newcommand{\salt}{\textit{salt}}

\usepackage{infer}
\usepackage{tikz}
\usepackage{ wasysym }
\usetikzlibrary{automata,positioning,arrows.meta,calc,decorations.markings,chains}

\usepackage{csquotes}

\lstdefinelanguage{Horst}{
    keywords=[1]{:=,
        ;,
        clause,
        expect,
        for,
        in,
        init,
        query,
        rule,
        test}
    keywordstyle=[1]\color{blue}\bfseries,
    keywords=[2]{
        const,
        datatype,
        eqtype,
        op,
        pred,
        sel},	
    keywordstyle=[2]\color{teal}\bfseries,
    keywords=[3]{bool,
        int},	
    keywordstyle=[3]\color{violet}\bfseries,
    identifierstyle=\color{black},
    sensitive=false,
    otherkeywords={!,[,],?,:,@,->},
    comment=[l]{//},
    morecomment=[s]{/*}{*/},
    commentstyle=\color{orange}\ttfamily,
    stringstyle=\color{red}\ttfamily,
    morestring=[b]',
    morestring=[b]", 
    escapechar=~,
}

\definecolor{verylightgray}{rgb}{.97,.97,.97}
\lstdefinelanguage{Solidity}{
    keywords=[1]{anonymous, assembly, assert, balance, break, call, callcode, case, catch, class, constant, continue, constructor, contract, debugger, default, delegatecall, delete, do, else, emit, event, experimental, export, external, false, finally, for, function, gas, if, implements, import, in, indexed, instanceof, interface, internal, is, length, library, log0, log1, log2, log3, log4, memory, modifier, new, payable, pragma, private, protected, public, pure, push, require, return, returns, revert, selfdestruct, send, solidity, storage, struct, suicide, super, switch, then, this, throw, transfer, true, try, typeof, using, value, view, while, with, addmod, ecrecover, keccak256, mulmod, ripemd160, sha256, sha3}, 
    keywordstyle=[1]\color{blue}\bfseries,
    keywords=[2]{address, bool, byte, bytes, bytes1, bytes2, bytes3, bytes4, bytes5, bytes6, bytes7, bytes8, bytes9, bytes10, bytes11, bytes12, bytes13, bytes14, bytes15, bytes16, bytes17, bytes18, bytes19, bytes20, bytes21, bytes22, bytes23, bytes24, bytes25, bytes26, bytes27, bytes28, bytes29, bytes30, bytes31, bytes32, enum, int, int8, int16, int24, int32, int40, int48, int56, int64, int72, int80, int88, int96, int104, int112, int120, int128, int136, int144, int152, int160, int168, int176, int184, int192, int200, int208, int216, int224, int232, int240, int248, int256, mapping, string, uint, uint8, uint16, uint24, uint32, uint40, uint48, uint56, uint64, uint72, uint80, uint88, uint96, uint104, uint112, uint120, uint128, uint136, uint144, uint152, uint160, uint168, uint176, uint184, uint192, uint200, uint208, uint216, uint224, uint232, uint240, uint248, uint256, var, void, ether, finney, szabo, wei, days, hours, minutes, seconds, weeks, years},	
    keywordstyle=[2]\color{teal}\bfseries,
    alsoletter={.},
    keywords=[3]{block, blockhash, coinbase, difficulty, gaslimit, number, timestamp, msg, data, gas, msg.sender, sig, value, now, tx, gasprice, origin},	
    keywordstyle=[3]\color{violet}\bfseries,
    identifierstyle=\color{black},
    sensitive=false,
    comment=[l]{//},
    morecomment=[s]{/*}{*/},
    commentstyle=\color{orange}\ttfamily,
    stringstyle=\color{red}\ttfamily,
    morestring=[b]',
    morestring=[b]", 
    escapechar=|, 
    numbers=left, 
    basicstyle=\scriptsize\ttfamily,
    numberstyle=\scriptsize,
    xleftmargin=1em,
}

\newcommand{\securify}{\text{Securify}}
\newcommand{\may}{\pred{May}}
\newcommand{\must}{\pred{Must}}
\newcommand{\inst}[1]{\text{\color{blue} #1}}
\newcommand{\pred}[1]{\text{\color{violet}  \textit{#1}}}

\newcommand{\iio}{\textbf{InstIndOf}}
\newcommand{\vio}{\textbf{VarIndOf}}

\newcommand{\fedge}{-\!\fact\! \xrightarrow{}}
\newcommand{\qedge}{-\!\qact\! \xrightarrow{}}
\newcommand{\ds}{\textit{ds}}
\newcommand{\as}{\textit{as}}

\newcommand{\bsi}{\textit{BS}}
\newcommand{\conf}[2]{\langle #1,~#2 \rangle}
\newcommand{\sconf}[2]{\conf{#2}{#1}}
\newcommand{\nstep}[4]{\xrightarrow{#1}^*_{#2}\mid_{#4}^{#3}~}
\newcommand{\nmstep}[5]{\xrightarrow{#1}^{#5}_{#2}\mid_{#4}^{#3}~}
\newcommand{\bs}[2]{#1 \in \textit{BS}(#2)}
\newcommand{\notbs}[2]{#1 \not \in \textit{BS}(#2)}
\newcommand{\bsn}{\textit{BS}(N)}
\newcommand{\defset}[2]{#1 \in \textit{Def}(#2)}
\newcommand{\useset}[2]{#1 \in \textit{Use}(#2)}
\newcommand{\notuseset}[2]{#1 \not \in \textit{Use}(#2)}

\newcommand{\eqstatevar}[3]{#1(#3) = #2(#3)}
\newcommand{\sucnode}[1]{n_{#1}^{+1}}

\newcommand{\projectpath}[2]{#1 \downarrow_{#2}}
\newcommand{\projectpatheqlength}[3]{\mid \projectpath{#1}{#3} \mid~ =~ \mid \projectpath{#2}{#3} \mid}
\newcommand{\rv}[2]{\textit{rv}~#1~#2}

\newcommand{\transscd}[2]{#1 \xrightarrow{~}_{cd}^* #2}
\newcommand{\countnodes}[2]{\mid_{#2}^{#1}}

\usetikzlibrary{decorations.pathreplacing}
\newcommand{\mayrule}{\Leftarrow_{\may}}
\newcommand{\mustrule}{\Leftarrow_{\must}}

\newcommand{\IR}{\textsf{IR}\xspace}

\newcommand{\loc}{x}

\usepackage{xargs}
\usepackage{ stmaryrd }

\newcommand{\mem}{\textit{Mem}}
\newcommand{\mt}{\textit{X}^{m}.{D}}
\newcommand{\mk}{\textit{X}^{m}.{S }}

\newcommand{\horstify}{\textsc{HoRStify}\xspace}
\newcommand{\horst}{\textsc{HoRSt}\xspace}

\newtheorem{theorem}{Theorem}
\newtheorem*{theorem*}{Theorem}
\newtheorem{lemma}{Lemma}
\newtheorem{definition}{Definition}
\newtheorem{assumption}{Assumption}

\Crefname{assumption}{Assumption}{Assumptions}
\crefname{assumption}{assumption}{assumptions}

\begin{document}

\title{\horstify:
Sound Security Analysis of \\Smart Contracts
\thanks{
    This work has been supported by the Heinz Nixdorf Foundation through a Heinz Nixdorf Research Group (HN-RG) and funded by the Deutsche Forschungsgemeinschaft (DFG, German Research Foundation) under Germany's Excellence Strategy---EXC 2092 CASA---390781972, and through grant 389792660 as part of TRR~248---CPEC, see \url{https://perspicuous-computing.science}.
    }
}

\author{\IEEEauthorblockN{Sebastian Holler\IEEEauthorrefmark{3}\IEEEauthorrefmark{1}, Sebastian Biewer\IEEEauthorrefmark{2}, Clara Schneidewind\IEEEauthorrefmark{1}}

\IEEEauthorblockA{\IEEEauthorrefmark{1}\textit{Max-Planck-Institute for Security \& Privacy} \\
\textit{Universitätsstraße, Bochum, Germany} }
\IEEEauthorblockA{\IEEEauthorrefmark{2}\textit{Saarland University}, \IEEEauthorrefmark{3}\textit{Saarbrücken Graduate School of Computer Science} \\
\textit{Saarland Informatics Campus, Saarbr\"ucken, Germany} }
}

\maketitle

\begin{abstract}
    The cryptocurrency Ethereum is the most widely used execution platform for smart contracts. Smart contracts are distributed applications, which govern financial assets and, hence, can implement advanced financial instruments, such as decentralized exchanges or autonomous organizations (DAOs). 
    Their financial nature makes smart contracts an attractive attack target, as demonstrated by numerous exploits on popular contracts resulting in financial damage of millions of dollars.
    This omnipresent attack hazard motivates the need for sound static analysis tools, which assist smart contract developers in eliminating contract vulnerabilities a priori to deployment.
     
    Vulnerability assessment that is sound and insightful for EVM contracts is a formidable challenge because contracts execute low-level bytecode in a largely unknown and potentially hostile execution environment.
    So far, there exists no provably sound automated analyzer that allows for the verification of security properties based on program dependencies, even though prevalent attack classes fall into this category. 
    In this work, we present \horstify, the first automated analyzer for dependency properties of Ethereum smart contracts based on sound static analysis. \horstify grounds its soundness proof on a formal proof framework for static program slicing that we instantiate to the semantics of EVM bytecode.
    We demonstrate that \horstify is flexible enough to soundly verify the absence of famous attack classes such as timestamp dependency and, at the same time, performant enough to analyze real-world smart contracts. 
\end{abstract}

\begin{IEEEkeywords}
Ethereum, Smart Contract, Blockchain, Dependency Analysis, Security, Tool
\end{IEEEkeywords}

\section{Introduction}

Modern cryptocurrencies enable mutually mistrusting users to conduct financial operations without relying on a central trusted authority.
Foremost, the cryptocurrency Ethereum supports the trustless execution of arbitrary quasi Turing-complete
programs, so-called smart contracts~\cite{wood2014ethereum}, which manage money in the virtual currency Ether.

The expressiveness of smart contracts gives rise to 
a whole distributed financial ecosystem known as Decentralized Finance (DeFi), which encompasses a multitude of (financial) applications such as brokerages\cite{zanzi2020nsbchain,li2018blockchain}, decentralized exchanges~\cite{jain2022coin,adamik2018smartexchange,zima2017coincer} or decentralized autonomous organizations~\cite{sims2021decentralised,hassan2021decentralized}.
However, smart contracts have shown to be particularly prone to programming errors that lead to devastating financial losses~\cite{hacks}. 
These severe incidents can be attributed to different factors. 
First, smart contracts are agents that interact with a widely unpredictable and potentially hostile environment. 
Accounting for all possible environment behaviors adds a layer of complexity to smart contract development. 
Second, smart contracts manage real money. This financial nature makes them an extraordinarily lucrative attack target. 
Third, transactions in blockchain-based cryptocurrencies, like Ethereum, are inherently immutable. As a consequence, not only the effects of exploits are persistent, but also vulnerable smart contracts cannot be patched. 
Given this state of affairs, it is of utmost importance to preempt contract vulnerabilities a priori to contract deployment. 

Sound static analysis tools allow for reasoning about all possible runtime behaviors without deploying a contract on the blockchain. 
In this way, smart contract developers and users can reliably identify and eliminate harmful behavior before publishing or interacting with Ethereum smart contracts. 
However, as shown in recent works~\cite{schneidewind2020ethor,schneidewind2020good}, most automatic static analyzers for Ethereum smart contracts that promise soundness guarantees cannot live up to their soundness claims. 

To the best of our knowledge, the only tools targeting sound and automated static analyses of smart contract security properties are Securify~\cite{tsankov2018securify}, ZEUS~\cite{kalra2018zeus}, EtherTrust~\cite{GMS::CAV18}, NeuCheck~\cite{lu2019neucheck}, and eThor~\cite{schneidewind2020ethor}. The soundness claims of ZEUS, Securify, EtherTrust, and NeuCheck are systematically confuted in~\cite{schneidewind2020good} and~\cite{schneidewind2020ethor}.

The analysis tool eThor~\cite{schneidewind2020ethor} comes with a rigorous soundness proof but only supports the verification of reachability properties. 
While this is sufficient to characterize the absence of interesting attack classes, many other smart contract security properties do not fall within this property fragment. 
Grishenko et al.~\cite{grishchenko2018semantic} give a semantic characterization of security properties that characterize the absence of prominent classes of smart contract bugs. 
Most of these properties fall into the class of non-interference-style two-safety properties that we will refer to as \emph{dependency properties} and fall out of the scope of eThor's analysis. 
The only tool that, up to now, targeted the (sound) verification of dependency properties was the tool Securify~\cite{tsankov2018securify}---which was empirically shown unsound in~\cite{schneidewind2020ethor}.

\paragraph*{Our Contributions}
In this work, we revisit Securify's approach. 
In this course, we analyze the peculiar challenges in designing a sound static dependency analysis tool for Ethereum smart contracts.
We show how to overcome these obstacles with a principled approach based on rigorous formal foundations. 
Leveraging a formal proof framework for static program slicing~\cite{wasserrab2009pdg}, we design a provably sound dependency analysis for Ethereum smart contracts on the level of Ethereum Virtual Machine (EVM) bytecode. 
Finally, we give an implementation of the analyzer \horstify that performs the static dependency analysis via a logical encoding, which can be automatically solved by Datalog solvers.
We demonstrate how to use \horstify to automatically verify dependency properties on smart contracts, such as the ones defined in~\cite{grishchenko2018semantic}.
Concretely, we make the following contributions: 
\begin{itemize}
\item 
We study the root causes of the soundness issues of the state-of-the-art Ethereum smart contract analysis tool Securify~\cite{tsankov2018securify} that so far had only been reported through empirical evidence in~\cite{schneidewind2020ethor,schneidewind2020good}. 
In this course, we uncover new soundness problems in Securify's analysis, which we can show to affect real-world smart contracts.

\item We devise a new dependency analysis for EVM bytecode based on program slicing following the static program framework presented in~\cite{wasserrab2009pdg}.
\item We prove this dependency analysis to be sound with respect to a formal semantics of EVM bytecode. 
\item We show how to approach relevant smart contract security properties presented in~\cite{grishchenko2018semantic} with the dependency analysis.
\item We present \horstify, an automated prototype static analysis tool that implements the dependency analysis.
\item We demonstrate that \horstify overcomes the soundness issues of Securify while showing comparable performance and small precision loss on real-world smart contracts.
\end{itemize}

The remainder of the paper is organized as follows: 
~\Cref{sec:overview} overviews our approach;
~\Cref{sec:background-ethereum} introduces the necessary background on Ethereum smart contract execution; 
~\Cref{sec:challenges} discusses the challenges in designing sound static analysis tools for smart contract dependency analysis;
~\Cref{sec:background-slicing} introduces the slicing proof framework from~\cite{wasserrab2009pdg} that our analysis builds on; 
~\Cref{sec:theory} presents our static analysis based on program slicing and its soundness proof;
~\Cref{sec:evaluation} reports on our prototype implementation \horstify and its practical evaluation; and~\Cref{sec:conclusion} concludes the paper.
\label{sec:intro}

\section{Overview}
\label{sec:overview}
\begin{figure}
    \includegraphics[width=\columnwidth]{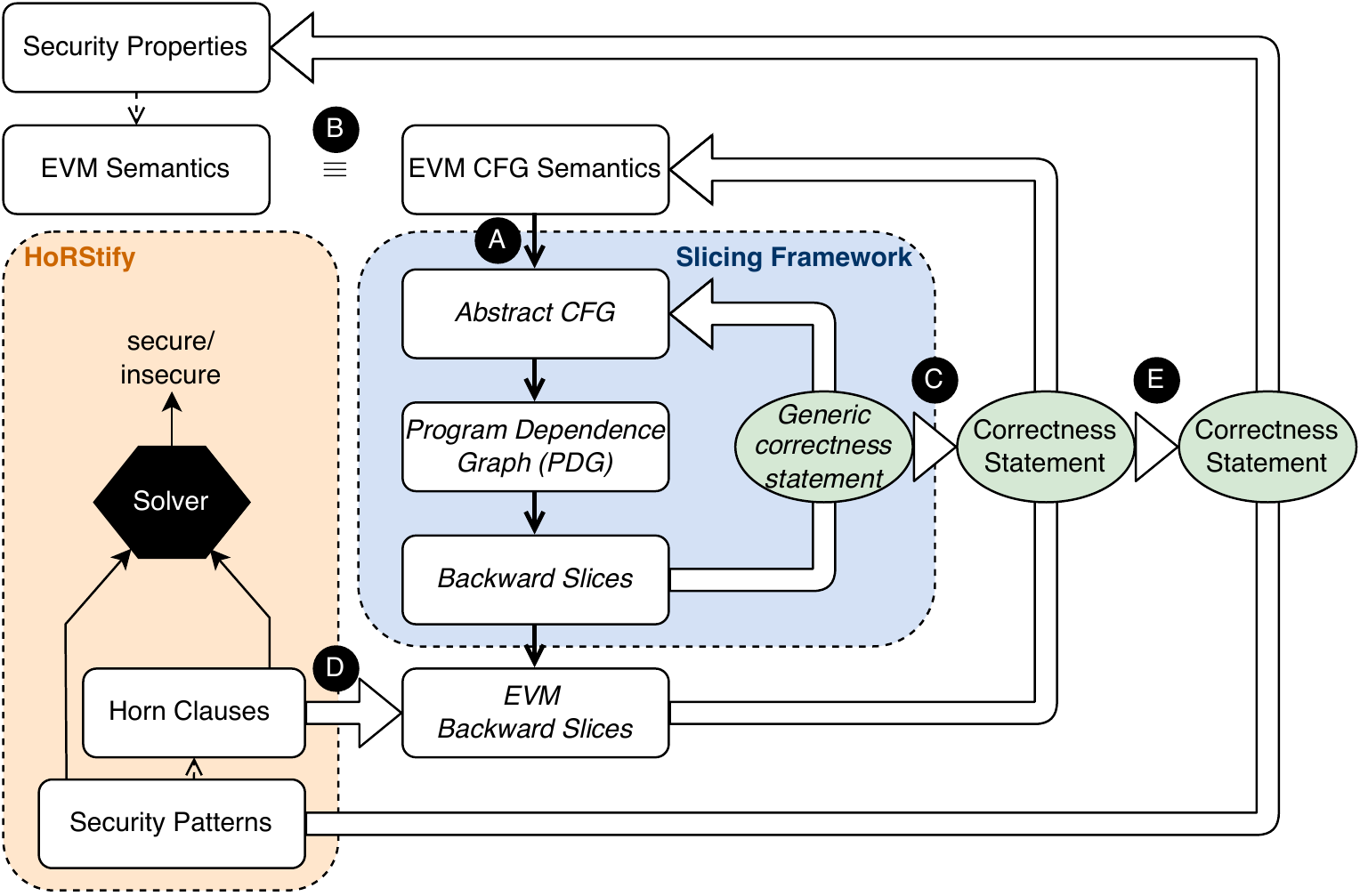}
    \caption{Overview on the formal guarantees of \horstify}
    \label{fig:overview}
\end{figure}

In this paper, we develop a dependency analysis tool for EVM bytecode that is designed in accordance with formal correctness statements providing overall soundness guarantees.
The correctness proof is modularized as depicted in \Cref{fig:overview}.
The core module is a generic proof framework~\cite{wasserrab2009pdg} for backward slicing using abstract control flow graphs (CFGs). 
In these CFGs, each node is annotated with all variables it reads and all variables it writes.
The backward slice of a node is a set containing all nodes that possibly influence the variables written in the respective node.
The framework extends the abstract CFG to a program dependence graph (PDG) by explicitly defining the data and control dependencies between the nodes.
For this PDG, the framework establishes a generic correctness statement for slicing: whenever a node influences another, the influencing node appears in the backward slice of the influenced node. 
To obtain the correctness result for a concrete programming language the abstract CFG representation is instantiated for a concrete program semantics. 

We instantiate the framework for EVM bytecode by devising a new EVM CFG semantics. 
We show (\textnode{A}) that the EVM CFG semantics satisfies all requirements for instantiation and (\textnode{B}) that it is equivalent to a formalisation of the EVM bytecode semantics.
From this, we obtain backward slicing for EVM contracts with a corresponding correctness statement (\textnode{C}).

For the actual analysis, we express dependencies 
in EVM contracts by means of dependency predicates which 
we characterize by (fixpoints over) a set of logical rules, given in the form of \emph{Constrained Horn Clauses (CHC)}. 
Most importantly, we show that if the backward slice of some program point contains some other program point, then the (potential) dependency between these two program points is also captured by the predicate encoding (\textnode{D}).  

From the EVM bytecode analyzer Securify~\cite{tsankov2018securify} we adopt the idea of defining so-called security patterns to soundly approximate the satisfaction (or violation) of a security property.
A \emph{security pattern} is a set of facts over dependency predicates, which characterize the form of dependencies that are ruled out by the pattern.
In contrast to Securify, our formal characterization of dependency predicates enables a correctness statement for the approximating behavior of the security patterns w.r.t. their corresponding property (\textnode{E}).

 Finally, we present the prototype tool \horstify that implements our dependency analysis and uses the Datalog engine Soufflé to perform the fixpoint computation and to check whether a security pattern is matched.
A pattern match guarantees (in)security w.r.t. the respective security property.

\paragraph*{Challenges}

The main challenge of designing a practical and sound dependency analysis for EVM bytecode is finding precise and performant abstractions that tame the complexity of EVM bytecode while maintaining soundness guarantees. 
As we will show in~\Cref{sec:challenges}, EVM bytecode's language design makes this task particularly hard: 
Non-standard language features introduce corner cases that are easily overlooked or make it necessary to enhance the analysis with custom optimizations that can lead to unsoundness when done in an ad-hoc manner.
As a consequence, it is of paramount importance to construct a sound analysis tool with formal foundations that are flexible enough to cover those subtleties. 

The slicing framework~\cite{wasserrab2009pdg} enables a modular soundness proof that separates the standard argument for the correctness of slicing from the characterization of program dependencies.
However, even though this reduces the proof effort, a naive instantiation of the framework would introduce a multitude of superfluous dependencies and hence lead to a highly imprecise analysis.
For this reason, the key challenge lies in the design of the  EVM CFG semantics. 
We will show how to approach these challenges with a solid theoretical foundation and by circumventing the bothersome technical hurdles without compromising the soundness of the analysis. 

\section{Background on Ethereum Smart Contracts}
\label{sec:background-ethereum}
\label{sec:whereverTheRouletteIs}
 Ethereum smart contracts are distributed applications that are jointly executed by the users of the Ethereum blockchain. 
In the following, we shortly overview the workings of Ethereum and the resulting particularities of the Ethereum smart contract execution environment.

\paragraph{Ethereum}
\label{subsec:ethereum}
The cryptocurrency Ethereum supports smart contracts via an account-based execution model. The global state of the system is given by accounts whose states are modified through the execution of transactions. 
All accounts have in common that they hold a balance in the currency Ether. An account can be either an  \emph{external account} that is owned by a user of the system and that solely supports user-authorized money transfers, or
a \emph{contract account} that manages its spending behavior autonomously by means of a program associated with the contract that may use its own persistent storage to provide advanced stateful functionalities. 

Users interact with accounts via transactions. 
Transactions either call existing accounts or create new contract accounts. A call transaction transfers an amount of money (that could be $0$) to the target account and triggers the execution of the account's code if the target is a contract account. 
A contract execution can modify the contract's persistent storage and potentially initiates further transactions. In this case, we speak of \emph{internal} transactions, as opposed to \emph{external} transactions, which are initiated by users on behalf of external accounts.

\paragraph{Smart Contract Languages}
Smart contracts are specified in \emph{EVM bytecode} and executed by the \emph{Ethereum Virtual Machine} (EVM).
EVM bytecode is a stack-based low-level language that supports standard instructions for stack manipulation, arithmetics, jumps, and memory access. 
On top, EVM's instruction set includes blockchain-specific opcodes, for example, to access transaction information and to initiate internal transactions.
While the EVM bytecode is technically Turing-complete, the execution of smart contracts is bounded by a transaction-specific resource limit. With each transaction, the originator sets this limit in the unit \emph{gas} and pays for it upfront. 
During the execution, instructions consume gas. The execution halts with an exception if running out of gas and reverts all effects of the prior execution. 

In practice, Ethereum smart contracts are written in high-level languages---foremost, Solidity~\cite{solidity}---and compiled to EVM bytecode. 
Solidity is an imperative language that mimics features of object-oriented languages like Java but supports additional primitives for accessing blockchain information and performing transactions.
For better readability, we will give examples using the Solidity syntax even though our analysis operates on EVM bytecode. 
We will introduce relevant Solidity language features throughout the paper when needed. 

\paragraph{Adversarial Execution Environment}
The blockchain environment poses novel challenges to the programmers of smart contracts. 
As opposed to programs that run locally, smart contracts are executed in an untrusted environment. This means, in particular, that certain system parameters cannot be fully trusted. 
A prominent example of this issue is Ethereum's block timestamp: 
In Ethereum's blockchain-based consensus mechanism, the system is advanced by appending a bulk of transactions grouped into a block to the blockchain, a distributed tamper-resistant data structure. 
These blocks are created by special system users, so-called miners. 
While all system users check that blocks only contain valid transactions, the correctness of a block's metadata cannot easily be verified. So is each block required to carry a timestamp, but due to the lack of synchronicity in the system, this timestamp can only be checked to lie within a plausible range.
This enables a miner to choose the value of the block timestamp freely within this range.
The following example illustrates how this peculiarity can be exploited in a smart contract:

\begin{lstlisting}[language=Solidity]
function spinWheel() private (uint) {
	return block.timestamp % 37;  }
\end{lstlisting}
The function \lstinline|spinWheel()| implements a spinning wheel 
that determines a random number between 0 and 36 based on the block timestamp (accessed via \lstinline|block.timestamp|). 
Based on such a function, a contract could implement a roulette game where players bet money on the outcome of the spinning wheel.
While the system timestamp may serve as a decent source of randomness for programs that run locally, this is not the case for smart contracts. 
A miner could easily tweak the timestamp of a block containing an invocation of the \lstinline|spinWheel| function and thereby influence its outcome.
In this way, a miner could ensure to win the roulette game themself.

\section{Challenges in Sound Dependency Analysis}
\label{sec:challenges}

As recently demonstrated in the literature~\cite{schneidewind2020good}, the sound analysis of Ethereum smart contracts is a challenging problem; most analysis tools aiming at provable soundness guarantees fall short of their goal. 
This can be mainly attributed to the non-standard language features of the EVM bytecode language and the unusual execution model of the EVM:
Smart contracts are executed in a (potentially) hostile environment, which can interact with, and even, schedule contracts.
The smart contract execution is dependent on the gas resource and the low-level EVM bytecode language features little static information. 
As a consequence, execution heavily depends on unknown runtime parameters, which makes it hard to reason statically about contract behaviors in a sound and reasonably precise and efficient manner. 
 This incentives the incorporation of ad-hoc optimizations, which increase the complexity of the analysis even further.
Consequently, it is crucial to establish rigorous formal foundations for EVM bytecode analysis and to align the implementation with these foundations.
In the following, we demonstrate how the lack of formal foundations affects the guarantees of the state-of-the-art analysis tool Securify~\cite{tsankov2018securify}.

\subsection{Securify}
The automated analyzer Securify is the only analysis tool up to now that aims at giving provable guarantees for dependency analyses of EVM bytecode contracts. 
%
    It decompiles the bytecode into a stackless \emph{intermediate representation} (\IR), where values are stored in variables in static single assignment (SSA) form rather than on a stack.
    Further, it 
    determines the CFG of the contract and encodes the transitive control and data flow dependencies between variables and program locations as a set of \emph{dependency predicates}.
  While it is not possible to specify arbitrary (security) properties in Securify, the tool allows for defining \emph{compliance patterns} and \emph{violation patterns} that serve as ``approximations'' for the satisfaction and, respectively, the violation of the property.
  These patterns are defined over the dependency predicates
  and can be checked automatically using the Datalog solver Soufflé~\cite{jordan2016souffle}.
  A \emph{compliance pattern is sound w.r.t.~a property}, if satisfying the pattern implies satisfaction of the property, and, analogously, a \emph{violation pattern is sound w.r.t.~a property} if satisfying the pattern implies violation of the property.
  If neither of the patterns is satisfied, the satisfaction of the property is inconclusive.
  Obviously, it cannot be that for the same contract and for the same property a sound compliance and violation pattern hold simultaneously.

\begin{figure}
    \small
	\begin{align*}
		\textit{compliance: } 
        &\pred{all}~ \inst{jump}(L_1,Y,\_),\inst{sstore}(L_2,\_,\_). \\ 
        &\quad \pred{MustFollow}(L_1,L_2) \land \pred{DetBy}(L_1,Y,\inst{caller}) \\
		\textit{violation: } 
        & \pred{some}~\inst{sstore}(L_1, X, \_). \\ 
        & \quad \neg \pred{MayDepOn}(X, \inst{caller}) \wedge \neg \pred{MayDepOn}(L_1, \inst{caller})
	\end{align*}
    \normalsize
	\caption{\textsf{Restricted Write} compliance and violation pattern \cite{tsankov2018securify}}
	\label{fig:rest_write_pattern}
    \vspace{-0.5cm}
\end{figure}

  An example of a security property is the \emph{restricted write} (RW) property for 
  storage locations.
  Intuitively, a contract satisfies RW, if for all 
  storage locations, there is at least one caller address that cannot write to this location.
  \Cref{fig:rest_write_pattern} shows a compliance\footnote{The Securify implementation contains two compliance patterns; one is shown in~\cite{tsankov2018securify}, the other one is shown in \Cref{fig:rest_write_pattern}. 
  } and violation pattern for RW~\cite{tsankov2018securify}.

The compliance pattern for RW states that for all conditional jump instructions at program location $L_1$ that branch on condition $Y$ ($\inst{jump}(L_1,Y,\_)$)
and for all storage write instructions at location $L_2$ ($\inst{sstore}(L_2,\_,\_)$) that are necessarily preceded by such jump instructions ($\pred{MustFollow}(L_1,L_2)$), 
it must hold that at location $L_1$ the condition $Y$ must be determined by the caller of the transaction ($\pred{DetBy}(L_1,Y,\inst{caller}$)). 
The violation pattern for RW states that there is some storage write instruction at location $L_1$ writing to storage address $X$ ($\inst{sstore}(L_1, X, \_)$) such that that neither the address $X$ nor the execution of the storage instruction at $L_1$ may depend on the caller of the transaction ($\neg \pred{MayDepOn}(X, \inst{caller}) \wedge \neg \pred{MayDepOn}(L_1, \inst{caller})$).

\subsection{Soundness Issues}
\label{sec:soundness-issues}

Even though Securify characterizes security properties and their corresponding compliance and violation patterns, no formal connection between patterns and properties is established. 
In particular, they do not prove the soundness of the patterns they propose in \cite{tsankov2018securify}  w.r.t. the properties they are supposed to approximate.
Doing so would require 
1) to prove that the dependency predicates imply semantic notions of independence (sound core analysis)
and 2) to prove that the semantic notions implied by the security patterns indeed imply the security properties (sound security patterns).
In the following, we use the example of the RW property to show how the absence of formal soundness arguments causes Securify to miss corner cases that undermine its soundness guarantees.

\subsubsection{Sound Core Analysis}
Securify does not draw a
connection between the dependency predicates and the EVM bytecode semantics. 
This leads to mismatches between the intuitions for the predicates and their definitions.

\paragraph{Must-analysis}\label{sec:mustAnalysis}
\begin{figure}
    \begin{lstlisting}
contract Start { bool test = false;
    function flipper() public {
        if (uint(msg.sender) * 0 == 0)  |\label{line:must}|
        {   test = !test; } } } |\label{line:must:assignment}|
\end{lstlisting}
    \caption{$\securify$ counterexample: must-analysis}
    \label{code:counter_must}
    \vspace{-0.5cm}
\end{figure}

Securify's dependency analysis and predicates can be attributed to one of two categories: 
 a \emph{may-analysis}  aims at over-approximating possible control and dataflow dependencies, encoded by \emph{may-predicates}, and
 a \emph{must-analysis}  aims at capturing dependencies and deducing \emph{must-predicates} that show a definite effect on the actual execution.
According to their usage in the security patterns, 
negated may-predicates imply a notion of independence, while must-predicates should imply a form of determination.
More precisely, it is stated that the must-predicate $\pred{DetBy}(L, Y, T)$ ``indicates that a different value of T guarantees that the value of Y changes.''~\cite{tsankov2018securify}
This guarantee, however, is violated in the contract shown in~\Cref{code:counter_must},
where Securify inferred that \lstinline|test| is determined by the caller although every caller can change the value of \lstinline|test|:   
The check of the conditional evaluates to \lstinline|true| for any value of \lstinline|msg.sender|, hence allowing every caller to write the \lstinline|test| field.
Still, Securify reports this contract to match the compliance pattern, indicating that the condition in line~\ref{line:must} would be determined by the caller.
The underlying reason for this problem is of substantial nature: 
The must-analysis under-approximates control flows but over-approximates data flows. 
More precisely, a variable $X$ is considered to be determined by a variable $Y$ if $Y$ occurs in the expression assigned to $X$. 
Since \lstinline|msg.sender| appears in the condition expression in line~\ref{line:must}, the condition is considered to be determined by \lstinline|msg.sender| even though it actually is independent of \lstinline|msg.sender|.
This treatment makes the must-analysis inherently unsound.

Due to this substantial mismatch between the intuition for the $\pred{DetBy}$ predicate and its implementation, it is unclear whether 
adjusting the implementation of the must-analysis such that it is sound, could result in a performant and precise analysis. 
So, in this work, we will focus on the may-analysis.

\paragraph{Memory Abstraction}

\begin{figure}
    \begin{lstlisting}
contract Start { bool test = false;   
    function storeTest(uint c) public {
        address[] memory a = new address[](7);
        for (uint i = 0; i < 7; i++) {
            a[i] = msg.sender;}
        if (a[0] != address(0)) {test = !test;} } }|\label{line:mem:caller-check}|
\end{lstlisting}
    \caption{$\securify$ counterexample: storage abstraction}
    \label{code:counter_storage}
    \vspace{-0.5cm}
\end{figure}

For establishing a sound may-analysis, it is crucial to overapproximate dependencies for all relevant system components that can interact with each other. 
In particular, this includes stack, memory and storage variables, because values are written from the stack to the local memory and persistent storage, and back. 
However, the addresses of memory and storage accesses are not statically known but specified on the stack. 
E.g., the EVM instruction $\MSTORE(x, y)$ denotes that the value in stack variable $y$ should be written to the address as given in stack variable $x$.
Consequently, the concrete memory address at which the value in $y$ will be stored may only be known at runtime. 
This poses a big challenge to static analysis since for precisely modeling the dependencies on different memory and storage cells, their accesses need to be known statically. 
Otherwise, the dependencies on all memory and storage cells would need to be merged, resulting in a substantial precision loss. 
In practice, memory and storage addresses can in most cases be precomputed by partial evaluation~\cite{tsankov2018securify}.
Hence, this preprocessing information can be used to enhance the analysis precision.

Securify implements this optimization in an unsound way, as illustrated by the example in~\Cref{code:counter_storage}. 
Here, function \lstinline|storeTest| locally
 defines a new  address array \lstinline|a| of size $7$ and initializes all its elements with the contract caller \lstinline|msg.sender|. 
The write access to the \lstinline|test| variable is restricted by the condition that the first array element \lstinline|a[0]| (which obviously contains the caller address) is not $0$. 
Consequently, the contract satisfies the RW property. 
Still, Securify certifies a violation of the RW pattern w.r.t. \lstinline{test}\footnote{This is indeed unsoundness and not imprecision: Securify guarantees that a property does not hold if the violation pattern matches. Only inconclusive cases (i.e., no compliance and no violation pattern matches) cause imprecision.}.
The example illustrates that the analysis does not consider that a memory address may be statically unknown at the point of writing but known at the point of reading. 
Since writing to the array is done in a loop, for the assignment \lstinline|a[i] = msg.sender| the memory address cannot be statically determined. 
For the condition in line~\ref{line:mem:caller-check}, in contrast, the memory address for \lstinline|a[0]| can be precomputed.
However, Securify fails to account for the fact that dependencies of an unknown memory access should propagate to all concrete memory addresses.

\paragraph{Reentrancy handling}

\begin{figure}
    \begin{lstlisting}
contract Start { bool test = false;
    address a; |\label{line:re:a}|
    function setAddress(address addr) public  
    {   a = addr;  }|\label{line:re:conflict}|  
    function flipper () public {
        try Start(this).setAddress(msg.sender) {  |\label{line:re:setAddress}|
            if (a != address(0)) { test = !test; } |\label{line:re:caller-check}|
        } catch { revert();  }  }  }
\end{lstlisting}
    \caption{$\securify$ counterexample: reentrancy handling}
    \label{code:counter_reentrancy}
    \vspace{-0.5cm}
\end{figure}

Smart contracts are reactive programs in the sense that they can transfer control to other contracts and are subject to \emph{reentrancy}, i.e., while awaiting the return of the other contract, this contract may call the waiting contract again. 
\Cref{code:counter_reentrancy} shows a simple case of reentrancy.
In this variant of \Cref{code:counter_must}, function \lstinline|flipper| calls the contract's function \lstinline|setAddress| within a new internal transaction\footnote{A reasonable contract would call a function of the same contract directly so that such a call would be translated to a $\JUMP$ by the compiler. The chosen syntax enforces that the function call will be translated to a $\CALL$ instead.}.
\lstinline|flipper| uses \lstinline|setAddress| to store the caller \lstinline|msg.sender| in the storage location \lstinline|a| (defined in line~\ref{line:re:a}).
Then, \lstinline|flipper| modifies the critical storage location \lstinline|test| if and only if the address stored in \lstinline|a| is not zero.
Ethereum contracts are executed non-concurrently, so the value of \lstinline|a| remains unchanged after line~\ref{line:re:setAddress} and before the evaluation of the condition in line~\ref{line:re:caller-check}.

Consequently, a caller with address $0$ can never write to the \lstinline|test| field and the contract satisfies the RW property w.r.t. \lstinline|test|. 
Still, Securify reports a match of the violation pattern for \lstinline|test|.
Inspection of the Securify code reveals that it does not model potential dependencies between arguments of external calls and storage locations accessible via reentrancy.

\paragraph{External call handling}
Aside from reentrancy, 
external calls may affect the local execution state in multiple ways. 

The success of an external call is indicated by placing a corresponding flag on the stack and the return value of the call (if existent) is written to a memory fragment that is specified as an argument to the call.
These effects may depend on the recipient and the arguments of the call. 
The example in~\Cref{fig:call-counter} illustrates how ignoring those dependencies causes an unsoundness in Securify:
\begin{figure}
\begin{lstlisting}
contract Check { 
	function testZero (address a) public {
	    assert (a == address(0)); } }	
contract Start {
	bool test = false; 
	address check = address(42);  
	function flip() public {
	    try Check(check).testZero(msg.sender){
	        test = !test; 
	    } catch {return;} } }
\end{lstlisting}
\caption{Securify counter example: external call handling}
\label{fig:call-counter}
\end{figure}
In this example, the sender check is outsourced to the method \lstinline|testZero| of another contract. 
The assignment of variable \lstinline|test| 
depends on 
whether \lstinline|testZero| returns without the \lstinline|assert| throwing an exception, which in turn depends on 
(the input data) \lstinline|msg.sender|. 
Hence, this contract satisfies the RW property.
Still, Securify reports a violation, since no dependencies between the input to the call and the call output are modeled.

\paragraph{Gas handling}

\begin{figure}
    \begin{lstlisting}
contract Start { bool test = false;
    function flipper () public {
        require(gasleft() > 10000);  |\label{line:gas:require}|
        bool flip = false;
        if (msg.sender == address(0)) { |\label{line:gas:condi}|
        {   while (gasleft() >= 5000)
                { flip = !flip; } }  |\label{line:gas:idnoc}|
        if (gasleft() < 5000) {test = flip;} } }    |\label{line:gas:testAssign}|
\end{lstlisting}
    \caption{$\securify$ counterexample: gas handling}
    \label{code:counter_gas}
    \vspace{-0.5cm}
\end{figure}

\Cref{code:counter_gas} shows a contract that indirectly restricts write access to storage \lstinline{test} by consuming the gas resource in a controlled way.
In line~\ref{line:gas:require}, the contract ensures that it is executed with a generous amount of gas; if not enough gas is available, the execution is aborted and no caller is able to write to \lstinline{test}.
The code between lines~\ref{line:gas:condi} and~\ref{line:gas:idnoc} essentially wastes masses of gas if the caller address is equal to 0, and, otherwise, consumes very little gas.
The crux of the contract is in line~\ref{line:gas:testAssign}: From the amount of gas that is left, the contract can determine if the caller's address is equal to 0---this is the case if and only if less than 5000 gas units are left.
Hence, depending on the amount of available gas, either no caller or only caller 0 can write to \lstinline{test}. 
So, there is always at least one caller that cannot write to \lstinline{test}---the contract satisfies the RW property. %
However, Securify reports a violation of this property. 
The reason for this wrong analysis result is that Securify does not track dependencies for the gas resource.

\subsubsection{Sound Security Properties}
Since the dependency predicates do not have a semantic characterization, the soundness of the security patterns w.r.t. their corresponding property cannot be proven. 
Indeed, Schneidewind et al.~\cite{schneidewind2020ethor} provide counter examples for the soundness of 13 out of the 17 security patterns given in~\cite{tsankov2018securify}.
Above that, the unsoundness of the RW property undeniably manifests in line~\ref{line:re:conflict} of the contract we constructed in \Cref{code:counter_reentrancy}.
For this example, Securify reports simultaneously(!) satisfaction of a compliance and a violation pattern for the RW property w.r.t. \lstinline{a}.
This refutes the claim that compliance and violation patterns constitute sufficient criteria for property compliance and violation, respectively.

\section{Analysis Foundations}
\label{sec:background-slicing}
To design a sound static analysis for EVM bytecode based on program slicing, we instantiate the slicing proof framework from~\cite{wasserrab2009pdg} with a formal bytecode semantics as defined in~\cite{grishchenko2018semantic}. 
Before discussing the instantiation in~\Cref{sec:theory}, we shortly overview both frameworks. 

\subsection{EVM bytecode semantics}
\label{subsec:semantics}

{The EVM semantics was formally defined in~\cite{grishchenko2018semantic} in form of a small-step semantics.} 
We use a \emph{linearized} representation of the semantics inspired by Securify, where the use of the stack is  replaced by the usage of local variables in SSA form.
We will call these variables \emph{stack variables} and, in the following, always refer to the linearized representation of the semantics.

Formally, the semantics of EVM bytecode is given by a small-step relation $\sstep{\transenv}{\exconf{\callstack}{\transeffects}}{\exconf{\callstack'}{\transeffects'}}$. 
The relation describes how a contract, whose execution state is given by a callstack $\callstack$, can progress to callstack $\callstack'$ under a transaction environment $\transenv$. 
The transaction environment $\transenv$ holds information about the external transaction that initiated execution. 
We let $\ssteps{\transenv}{\exconf{\callstack}{\transeffects}}{\exconf{\callstack'}{\transeffects'}}$ denote the reflexive transitive closure of the small-step relation and call the pair $(\transenv, \callstack)$ a \emph{configuration}.
The details of the components of the EVM configurations can be found in ~\cite{grishchenko2018semantic}.

\noindent The overall state of an external transaction execution is captured by a callstack $\callstack$. 
The elements of the callstack model the states of all (pending) internal transactions. 
Internal transactions can either be pending, as indicated by a regular execution state $\regstate{\mstate}{\exenv}{\gstate}$, or terminated. 
The state of a pending transaction encompasses, the current global state $\gstate$, the execution environment $\exenv$ and the machine state $\mstate$.
The global state $\gstate$ describes the state of all accounts of the system and is defined as a partial mapping between account addresses and account states. 
The execution environment $\exenv$, among others, contains the $\code$ of the currently executing contract.
We model the code of a contract as a function $\evmcontract$ that maps program counters to tuples $(\instruction(\vec{x}), \nextpc, \pre)$, where $\instruction$ denotes an opcode from the EVM instruction set, $\vec{x}$ is the vector of input and output (stack) variables to this opcode, and $\nextpc$ denotes the  program counter for the next instruction.
Further, we instrument each instruction with a list $\pre$ of precomputed values for the arguments $\vec{x}$. 
This instrumentation is only introduced for analysis purposes and does not affect the execution.

The machine state $\mstate$ captures the state of the local machine and holds the amount of gas ($\lgas$) available for execution, the program counter ($\lpc$), the local memory, and the state of the (linearized) stack variables ($\textit{s}$).

\paragraph{Small-step Rules}
We illustrate the working of the EVM bytecode semantics using the example of the $\ADD$ instruction.
This instruction takes two values as input and writes their sum back to its return variable.
\begin{mathpar}
\small\infer{
\arraypos{\access{\exenv}{\code}}{\access{\mstate}{\pc}}= (\ADD(r, a, b), \nextpc, \pre) \\
\access{\mstate}{\gas} \geq 3 \\
\mstate'= \dec{\update{\update{\mstate}{\stack}{\update{\access{\mstate}{\stack}}{r}{\access{\mstate}{\stack}(a)+\access{\mstate}{\stack}(b)}}}{\pc}{\nextpc}}{\gas}{3}}
{\ssteptrace{\transenv}{\cons{\regstatefull{\mstate}{\exenv}{\gstate}{\transeffects}}{\callstack}}{\cons{\regstatefull{\mstate'}{\exenv}{\gstate}{\transeffects}}{\callstack}}{
		\ADD(a, b)
	}}
\end{mathpar}
Given a sufficient amount of gas (here $3$ units), an $\ADD$ instruction with result (stack) variable $r$ and operand (stack) variables $a$ and $b$ writes the sum of the values of $a$ and $b$ to $r$ and advances the program counter to $\nextpc$.
These effects, as well as the subtraction of the gas cost, are reflected in the updated machine state $\mstate'$.

\paragraph{Security properties}
\label{subsec:properties}

Previous work~\cite{grishchenko2018semantic} has shown that there are several generic smart contract security properties, which are desirable irrespective of the individual contract logic. 
The properties formally defined in~\cite{grishchenko2018semantic} are integrity properties that aim at ruling out the influence of attacker behavior on sensitive contract actions, in particular, the spending of money.
These properties are e.g., the independence of a contract's spending behavior from miner-controlled parameters (as the block timestamp) or mutable contract state.
Further,~\cite{grishchenko2018semantic} introduces the notion of call integrity, which requires that the spending behavior of a contract is independent of the code of other smart contracts. 
Since call integrity is hard to verify in the presence of reentering exeutions, a proof strategy is devised that decomposes call integrity into one reachability property (single-entrancy) that restricts reentering executions and two local dependency properties. 
These local dependency properties ensure that the spending behavior of the contract does not depend on the return effects of calls to other (unknown) contracts (effect independence) or immediately on the code of such contracts (code independence). 

Focussing on integrity, the security properties from~\cite{grishchenko2018semantic} are given as non-interference-style notions. 
We illustrate this with the example of timestamp independence, a property that requires that the block timestamp cannot influence a contract's spending behavior and hence would rule out vulnerabilities as those in the roulette example:


\begin{definition}[Independence of the block timestamp]
    \label{def:timestampIndependence}
    A contract $C$ is independent of the block timestamp if for all reachable configurations $(\transenv, \cons{\annotate{s}{C}}{\callstack})$ it holds for all $\transenv'$ that 
    \vspace{-0.2cm}
\noindent
    {\small
    \begin{align*}
     \transenv \equalupto{\timestamp} \transenv' 
    ~\land \sstepstrace{\transenv}{\cons{\annotate{\exstate}{C}}{\callstack}}{\cons{\annotate{\exstate'}{C}}{\callstack}}{\pi}
    ~\land~\finalstate{\exstate'} \\
    \land~ \sstepstrace{\transenv'}{\cons{\annotate{\exstate}{c}}{\callstack}}{\cons{\annotate{\exstate''}{C}}{\callstack}}{\pi'} ~\land~ \finalstate{\exstate''} 
    \implies \project{\pi}{\filtercallscreates{C}} = \project{\pi'}{\filtercallscreates{C}}
    \end{align*}}
    \vspace{-0.5cm}
    \end{definition}

This definition requires that two executions of the contract $C$ starting in the same execution state $\annotate{\exstate}{C}$ and in transaction environments $\transenv$ and $\transenv'$ that are equal up to the block timestamp (denoted by $\transenv \equalupto{\timestamp} \transenv'$) exhibit the same calling behavior (captured by the call traces $\project{\pi}{\filtercallscreates{C}}$). 
Intuitively, this ensures that the contract $C$ may not perform different money transfers based on the block timestamp. 
The roulette example trivially violates this property since, based on the block timestamp, the prize will be paid out to a different user.
\subsection{Program Slicing}
\label{subsec:slicing}
Static program slicing is a method for capturing the dependencies between different program points (nodes) and variables in a program.
Intuitively, the program slice of some program node $n$ in a program $P$ consists of all those nodes $n'$ in $P$ that may affect the values of variables written in $n$. 
Program slices are constructed based on the program dependence graph (PDG) that models the control and data dependencies between the nodes of a program. 
In the following, we will review the static slicing framework by Wasserraab et al. \cite{wasserrab2009pdg}, which establishes a language-independent correctness result for slicing based on abstract control flow graphs (CFGs). 

\paragraph{Abstract control flow graph}
An abstract CFG is a language-agnostic representation of program semantics. Technically, an abstract CFG is parametrized by a set of program states $\astates$ and defined by a set of nodes (representing program points) and a set of directed edges between nodes. 
Edges may be of two different types: 
State-changing edges $n \fedge n'$ alter the program state $\cfgstate \in \astates$ by applying the function $f$ to $\cfgstate$ and predicate edges $n \qedge n'$ guard the transition between $n$ and $n'$ with the predicate $Q$ on the program state $\cfgstate$.
We write $n \xrightarrow{\edgelist}^* n'$ to denote that node $n$ can be reached $n'$ using the edges in the list $\edgelist$. 
Abstract CFG edges can be related to actual runs of the program
by lifting them to a small-step relation of the form $\acfgstep{\sconf{\cfgstate}{n}}{a}{\sconf{\cfgstate'}{n'}}$.


\REPLACEMfor{}{220707}{
\paragraph{Backward slices}}
{\paragraph{PDG and backward slices}}

    The PDG for a program consists of the same nodes as the CFG for this program and has edges that indicate data and control dependencies.
    To make data dependencies inferable,
    each node $n$ is annotated with a set of variables that are written (short \emph{Def set}, written $\defn(n)$) and a set of variables that are read by the outgoing edges of the node (short \emph{Use set}, written $\usen(n)$). 
    A node $n'$ is data dependent on node $n$ (written $n \xrightarrow{~}_{dd} n'$) if $n$ defines a variable $Y$ ($Y \in \defn(n)$), which is used by $n'$ ($Y \in \usen(n')$) and $n'$ is reachable from $n$ in the CFG without passing another node that defines $Y$. 
    A node $n'$ is (standard) control dependent on node $n$ (written $n \xrightarrow{~}_{cd} n'$) if $n'$ is reachable from $n$ in the CFG, but $n$ can as well reach the program's exit node without passing through $n'$ and all other nodes on the path from $n$ to $n'$ cannot reach the exit node without passing through $n'$.
    So intuitively, $n$ is the node at which the decision is made whether $n'$ will be executed or not.
    Based on the data and control flow edges of the PDG, the backward slice of a node $n$ (written $\bss{n}$) is defined as the set of all nodes $n'$ that can reach $n$ within the PDG.

\paragraph{Correctness statement}
The generic correctness statement for slicing proven in~\cite{wasserrab2009pdg} is stated as follows:
\begin{theorem}
    \label{proof:correctness_slicing}
    \textit{Correctness of Slicing Based on Paths \cite{wasserrab2009pdg}}
    {\small
    \begin{mathpar}
        \infer{
            \sconf{\cfgstate}{n} \xrightarrow{\as}^* \sconf{\cfgstate'}{n'} \\
        }
        {
            \exists~ as'.~ \sconf{\cfgstate'}{n} \xrightarrow{\as'}^*_{\textit{BS}(n')} \sconf{\cfgstate''}{n'} 
            ~ \wedge~
            \projectpath{\textit{as}}{\textit{BS}(n')} = as' \\
            ~\wedge~ (\forall~ \useset{V}{n'}. \eqstatevar{\cfgstate'}{\cfgstate''}{V})
        }
    \end{mathpar}}
\end{theorem}

Intuitively, the theorem states that whenever a node $n$ can reach some node $n'$ in the PDG ($\sconf{\cfgstate}{n} \xrightarrow{\as}^* \sconf{\cfgstate'}{n'}$), then removing all outgoing edges from nodes not in the backward slice of $n'$ ($\sconf{\cfgstate}{n} \xrightarrow{\as'}^*_{\textit{BS}(n')} \sconf{\cfgstate''}{n'}$) without altering the path through the PDG in any other way ($\projectpath{\textit{as}}{\textit{BS}(n')} = as'$) has no impact on $n'$. 
Having no impact on $n'$ means that variables used in $n'$ are assigned to the same values regardless of whether the edges have been removed or not ($\forall \useset{V}{n'}.~ \eqstatevar{\cfgstate'}{\cfgstate''}{V}$).
We call the PDG without the above-mentioned edges also \emph{sliced PDG} or \emph{sliced graph}. 

\section{Sound EVM Dependency Analysis}
\label{sec:theory}
In the following, we instantiate the slicing proof framework~\cite{wasserrab2009pdg} to accurately capture program dependencies of EVM smart contracts in terms of program slices. We then give a logical characterization of such program slices, which allows for the automatic computation of dependencies between different program points and variables with the help of a Datalog solver. 
The generic correctness statement of the slicing proof framework guarantees that the slicing-based dependencies soundly over-approximate all real program dependencies.
We show how to use this result to automatically verify relevant smart contract security properties such as the independence of the transaction environment and the independence of mutable account state as defined in~\cite{grishchenko2018semantic}.

\subsection{Instantiation of Slicing Proof Framework}
\label{subsec:instantiation}
\newcommand{\func}{\textit{func}}
\newcommand{\addresses}{\mathcal{A}}  
\newcommand{\annotations}{\addresses \times \{ \bot \}} 
\renewcommand{\none}{\textit{None}}
\renewcommand{\c}{\mathcal{C}}
\newcommand{\start}{\textit{start}}
\newcommand{\exit}{\textit{exit}}
\renewcommand{\curropcode}[2]{\omega_{\mathcal{C},#1}}

We instantiate the abstract CFG from the slicing framework with the linearized EVM semantics. 

The concrete layout of the instantiation heavily influences the resulting backward slices and the precision of the analysis. 
In the following, we sketch the most interesting aspects of our instantiation of the CFG components and how they contribute to the design of a precise dependency analysis.

\paragraph*{Preprocessing Information}
For a precise analysis, it is indispensable to preprocess contracts to aggregate as much statically obtainable information as possible---without compromising the soundness of the overall analysis.
For example, knowing the precise destination of jump instructions is crucial to reconstruct control flow precisely, and, moreover, this information usually can be easily reconstructed, especially, when contracts were compiled from a high-level language with structured control flow.

\ifextended
We require in the following that the preprocessed information is correct: 

\begin{definition}[Sound Preprocessing]
A contract $\precontract$ has sound preprocessing information if for  all execution states $\annotate{\exstate}{C}$ with an initial machine state running contract $C$ it holds that if
$\ssteps{\transenv}{\cons{\annotate{\exstate}{C}}{\callstack}}{\cons{\annotate{\exstate'}{C}}{\callstack}}$ 
then
\small{
\begin{align*}
\precontract(\access{\access{\exstate'}{\mstate}}{\pc}) = (\instruction(\vec{x}), \nextpc, \pre) 
\\
\Rightarrow \forall i \in [0, \size{\vec{x}} -1 ].~ 
\pre[i] = \optional{\access{\mstate}{\stack}(x_i)} ~\lor~ \pre[i] = \None
\end{align*}
}
\end{definition}
\fi

In the remainder, we assume that all existing preprocessing information is correct and sufficient to reconstruct the contract's CFG.
Recall that, formally, we consider a contract a function, such that for a program counter $\lpc$, $\precontract(\lpc) = (\instruction(\vec{x}), \lpc_{\text{next}}, \pre)$ where 
$\pre$ contains the preprocessing information for the instruction $\instruction(\vec{x})$:
for every $\vec{x}[i]$, $\pre[i]$ either holds a precomputed static value, or $\bot$ to indicate that no static value could be inferred.
Note that we restrict preprocessing to stack variables.
For our analysis, we are only interested in precomputed values for memory and storage locations and jump destinations.

\paragraph*{CFG States}
The edges of the CFG are labeled with state-changing functions or predicates on states.
For EVM bytecode programs, the CFG state $\cfgstate$ is partitioned into stack variables (denoted by $\stackvar{x}$), memory variables ($\memvar{x}{}$), storage variables ($\storvar{x}{}$) and local ($\locenvvar{x}$) and global ($\globenvvar{x}$) environmental variables.  
Memory and storage variables represent cells in the local memory, respectively the global storage of the contract under analysis.
Local environment variables contain the information of the execution environment that is specific to an internal transaction.
Global environmental variables denote environmental information whose accessibility is not limited to a single internal transaction, like the state of other contracts and the block timestamp. 
Environmental information that cannot be directly accessed during the execution (such as the storage of other contracts) is hidden in the dedicated global environmental variable $\globenvvar{\extenv}$.

\paragraph*{CFG Nodes, Edges \& Def and Use Sets}

To transform an EVM bytecode program into a CFG, we map every program counter $\pc$ to one or more nodes $(\lpc, i)$ in the CFG (where $i\in\NN$ is used to distinguish between multiple nodes for $\lpc$).
We call a node $(\lpc, 0)$ \emph{initial node (for $\lpc$)} and nodes $(\lpc, i)$ with $i > 0$ \emph{intermediate nodes (for $\lpc$)}.
%
%
Since the size of the callstack below the translated callstack element may influence the contract execution, the rule set defining the CFG transformation constructs a relation of the form  
$\cfgstep{\evmcontract, \cd}{\cfgnode{\lpc}{\cfgnodindex}}{\cfgact}{\cfgnode{\lpc'}{\cfgnodindex'}}$, where $\precontract$ is the contract for which the CFG is constructed, $\cd$ is the size of the callstack, and $a$ stands for either a $\qact$ action (for a predicate edge) or $\fact$ action (for a state-changing edge).
With every rule, we also provide Def and Use sets.
The Use sets contain all variables whose values are retrieved from the state $\cfgstate$ in the definition of the $Q$ predicate or $f$ function.
Similarly, the definition set contains all variables that are overwritten by the function $f$ (and is always empty for predicate edges).%

\begin{figure}
	\scriptsize
	\begin{mathpar}
		\infer{
			\precontract(\lpc) = (\JUMPI(\stackvar{x_1}, \stackvar{x_2}), \lpc_\text{next}, \pre) \\
			\stateupdate{\cfgstate \leftarrow \envvar{\lgas} \define \cfgstate[\envvar{\lgas}] - 10}
		}
		{\edgestep{}{\cfgnode{\lpc}{0}}{\cfgnode{\lpc}{1}}}\\
		{\Defs{\envvar{\lgas}} \\
				\Uses{\envvar{\lgas}}} \\
		\infer{
			\precontract(\lpc) = (\JUMPI(\stackvar{x_1}, \stackvar{x_2}), \lpc_\text{next}, \pre) \\
			\predupdate{\cfgstate[\stackvar{x_2}] = 0} 
		}
		{\predstep{}{\cfgnode{\lpc}{1}}{\cfgnode{\lpc_\text{next}}{0}}}\\
		{\Def{\emptyset} \\
			\Uses{\stackvar{x_2}}} 
	\end{mathpar}
	\caption{$\JUMPI$ abstract CFG instantiation}
	\label{fig:jumpi_exampl}
\end{figure}

\Cref{fig:jumpi_exampl} shows two exemplary rules for the conditional jump instruction \JUMPI{}.
The first argument to \JUMPI{} is the jump destination and the second argument is the condition variable that must be non-zero for the jump to happen. 
We only show rules for the case that the condition is false, i.e., the jump does not happen. 
The upper rule defines a state-changing edge that deducts the gas that has to be paid for a \JUMPI{} instruction.
Appropriately, both Def and Use sets contain the gas variable because the current gas value must be read from and the reduced value updated in state $\cfgstate$.
Note that the edge goes from the initial node for $\lpc$ to an intermediate node for $\lpc$, because a second step is necessary to decide whether the program should jump.
The second step, depicted by the lower rule, continues in the intermediate node for $\lpc$ and checks if the condition (in variable $\stackvar{x_2}$) is false (i.e., if it is zero) via a predicate edge.
In this case, the execution proceeds to the initial node representing $\lpc_\text{next}$.
$\stackvar{x_2}$ is the only variable used by $Q$, hence it is the only variable in the Use set.%

It can be shown that the CFG semantics and EVM semantics coincide via two simulation relations where every (multi-)step in the CFG semantics between initial nodes is simulated by a step of the bytecode semantics and vice versa.

\subsection{Core abstractions} 
\label{subsec:core_abst}
We review the most interesting aspects of the CFG semantics and how they lead to a precise dependency analysis. 
In this course, we will show how to overcome the challenges presented in~\Cref{sec:challenges}.

\paragraph{Gas abstraction}
In the EVM, the execution of instructions consumes gas. 
If the gas is not sufficient to finish the execution of a contract, it is aborted with an exception.
Modeling this behavior accurately would result in a very imprecise analysis, since, technically, every instruction would be control-dependant on all its preceding instructions. 
This is as the execution of an instruction depends on whether prior instructions led to an out-of-gas exception. 
However, in practice, users should only call contracts with a sufficient amount of gas since, otherwise, the contract execution exceptionally halts. 
For this reason, there exist static analysis tools for computing (sound) gas bounds~\cite{albert2021don} and even Solidity's online compiler
provides gas estimates for smart contract execution. 

Hence, for our analysis we assume that a contract does not run out of gas and do not model the corresponding behavior in the CFG semantics.
We remark that Securify also makes this assumption implicitly; we spell it out explicitly as follows:

	\begin{assumption}[Absence of local out-of-gas exceptions (informal)]
		\label{asm:gas}
		A contract execution \emph{does not exhibit local-out-of-gas exceptions} if each local exception can be attributed to the execution of an $\INVALID$ opcode.
	\end{assumption}

In contrast to Securify, we do not ignore gas entirely, but model the gas reduction for all instructions.
This allows capturing dependencies such as the one highlighted in \Cref{code:counter_gas} (and missed by Securify).
In the CFG, we always model the gas reduction as a separate edge involving an intermediate node (e.g., with the upper rule in \Cref{fig:jumpi_exampl}).
The Def set of one node contains only the gas variable, while the Def set of the other node only contains the (stack) variables involved in the actual instruction.
\begin{figure}
		\scriptsize
		\begin{mathpar}
			\infer{
				\precontract(\lpc) =  (\ADD(\stackvar{y},\stackvar{x_1}, \stackvar{x_2}), \lpc_\text{next}, \pre) \\
				\stateupdate{\cfgstate \leftarrow \envvar{\lgas} \define \cfgstate[\envvar{\lgas}] - 3} 
			}
			{\edgestep{\precontract, \cd}{\cfgnode{\lpc}{0}}{\cfgnode{\lpc}{1}}} \\
			{\Defs{\envvar{\lgas}} \\
				\Uses{\envvar{\lgas}}}
		\end{mathpar}
		\begin{mathpar}
			\infer{
				\precontract(\lpc) =  (\ADD(\stackvar{y},\stackvar{x_1}, \stackvar{x_2}), \lpc_\text{next}, \pre) \\
				\stateupdate{\cfgstate \leftarrow \stackvar{y} \define \cfgstate[\stackvar{x_1}] + \cfgstate[\stackvar{x_2}]}
			}
			{\edgestep{\precontract, \cd}{\cfgnode{\lpc}{1}}{\cfgnode{\lpc_\textit{next}}{0}}} \\
			{\Defs{\stackvar{y}} \\
				\Uses{\stackvar{x_1},\stackvar{x_2}}}
		\end{mathpar}
		\caption{CFG semantics rules for the $\ADD$ instruction.}
		\label{fig:rule-add}
	\end{figure}
An example for that is given by the (simplified) CFG  rules of the \ADD{} instruction 
in~\Cref{fig:rule-add}.
Technically, an \ADD{} instruction performs two types of state updates: it decreases the gas and performs addition on stack variables. 
Since those two state updates are independent, their execution can be split into two different nodes. 
As a consequence, the node $(\lpc, 1)$ is not data-dependent on nodes writing the gas variable.

\begin{figure}
\includegraphics[width=0.8\columnwidth]{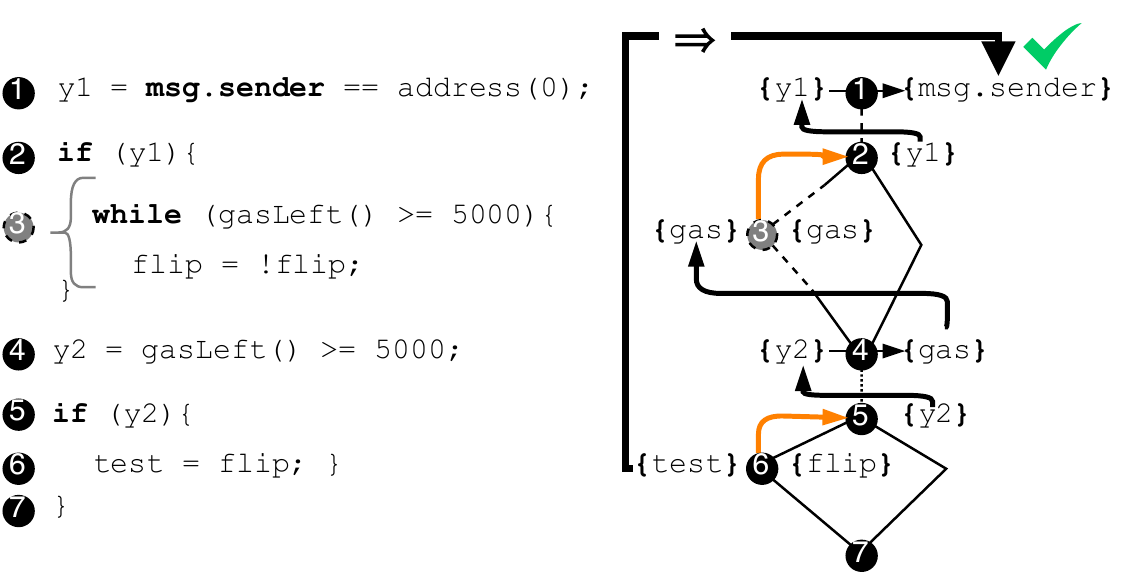}
\caption{Example control flow with gas dependencies. 
Def sets are given at the left of each node, Use sets at the right. 
Data dependencies are indicated by black arrows, control dependencies by orange ones.
}
\label{fig:gas-dependencies}
\end{figure}

Still, the gas abstraction is sound (under~\Cref{asm:gas}) and correctly captures the dependencies of the example in~~\Cref{code:counter_gas}: 
\Cref{fig:gas-dependencies} shows an incomplete and simplified CFG of the example in~\Cref{code:counter_gas} with annotated Def and Use sets.
The example illustrates how the CFG captures the dependency of the storage write (\lstinline|test = flip|) on the \lstinline|msg.sender| variable. 
The storage write in \textnode{6} is control dependant on the conditional \lstinline|y2| in \textnode{5}, and \textnode{5} depends on node \textnode{4} where \lstinline|y2| is defined.
\textnode{4} accesses the gas value, so a dependency between \textnode{4} and the gas nodes is established.
 Node \textnode{3} is one of these gas nodes (there are more not shown in the picture). The execution of \textnode{3} depends on condition \lstinline|y1| checked in \textnode{2}, so it is control dependant on \textnode{2}. Node \textnode{1} defines \lstinline|y1|, so \textnode{2} depends on \textnode{1}.
 Thus, there is a transitive dependency between writing to \lstinline|test| in \textnode{6} and reading \lstinline|msg.sender| in \textnode{1}.

\paragraph{Memory Abstraction}

To precisely model memory and storage accesses in a CFG, it is important to know statically as many memory and storage locations as possible.
Assume that such statical information is not available: then memory (or storage) cannot be separated into regions and all read and write operations introduce dependencies with the whole memory (or storage).
This would introduce many false dependencies.
During a preprocessing step, such static information can be inferred.
But, as demonstrated in \Cref{sec:challenges}, using preprocessed data may introduce unsoundness. 
This requires careful integration of preprocessing information into the CFG defining rules.
In the following we consider only memory variables; all ideas equally apply to storage variables.

\begin{figure}
	\includegraphics[width=\columnwidth]{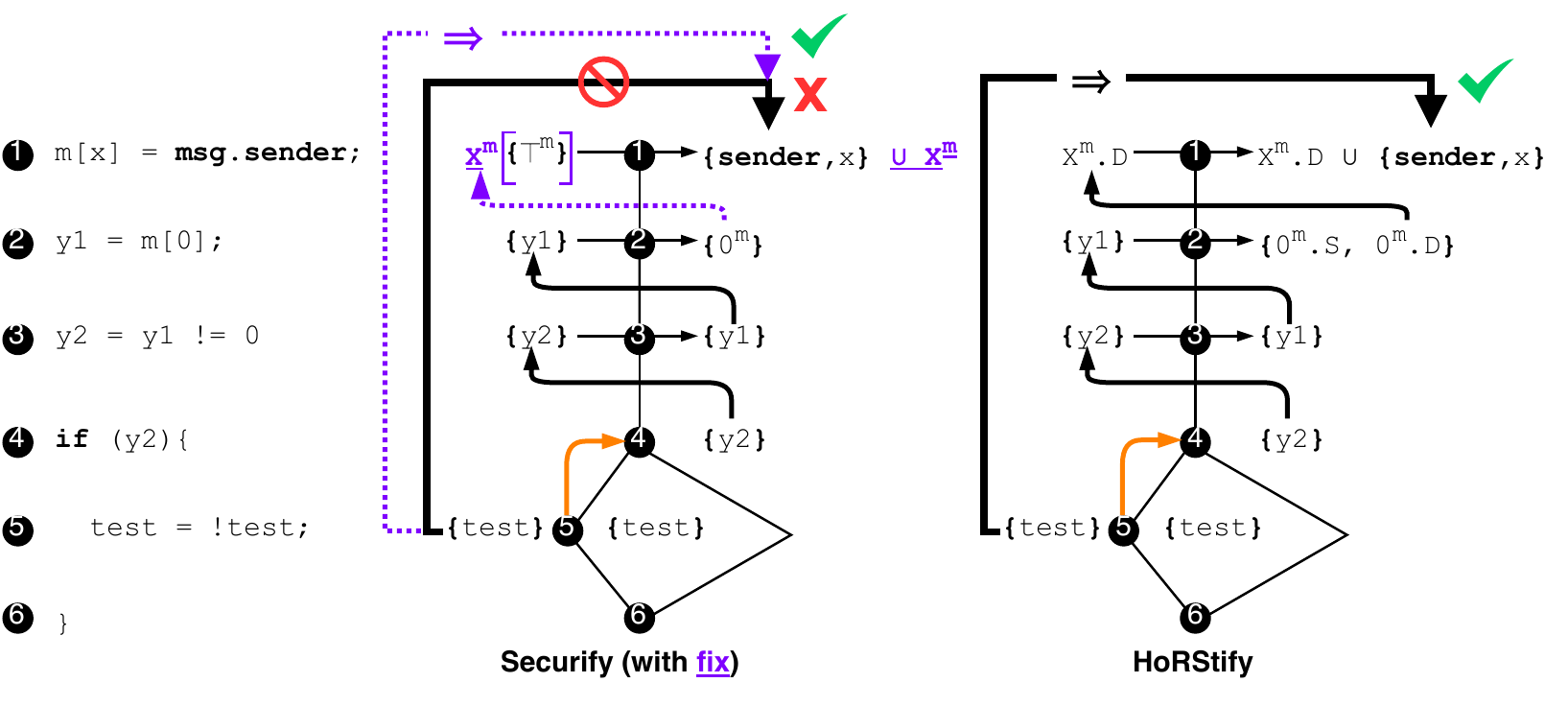}
\caption{Simplified version of contract in~\Cref{code:counter_storage} satisfying the RW property with PDGs depicting the dependencies modeled by Securify and \horstify.}
\label{fig:memory-example-dependent}
\end{figure}

We propose a, to the best of our knowledge, novel memory abstraction that is sound and provides high precision.
To position our approach between unsound and imprecise memory abstractions, we revisit \Cref{code:counter_storage} in a simplified version that is depicted as a CFG in \Cref{fig:memory-example-dependent}.
The black and solid line parts of the left CFG visualize how Securify misses the dependency between \lstinline|msg.sender| (\textnode{1}) and writing to \lstinline|test| (\textnode{5}). In Securify,
write accesses to unknown memory locations are assumed to write a special memory variable $\memvar{\top}{}$. 
However, when reading from a statically known memory location (as done in \textnode{2}), Securify does not consider that a value could have been written to this location when the location was not statically known, i.e., that the value could have been stored in $\memvar{\top}{}$: the Use set of \textnode{2} contains only $\memvar{0}{}$, but not $\memvar{\top}{}$.
A hypothetical fix for this unsoundness is to replace the variable $\memvar{\top}{}$ by the whole set $\memvar{X}{}$ of all memory variables.
This fix is depicted in violet in~\Cref{fig:memory-example-dependent}.
Now, the dependency of the read access in \textnode{2} to the write operation in \textnode{1} is naturally established.
One should notice, however, that this interpretation implies that the Use set of node \textnode{1} needs to contain all variables in $\memvar{X}{}$ as well: a new value is written to one unknown location, but for all other locations the value is ``copied'' from the existing memory cells, and hence, all these cells need to be included in the Use set. 
Even though fixing the soundness issue, this modeling would lead to an imprecise analysis as depicted in~\Cref{fig:memory-example-independent}.
\begin{figure}
	\includegraphics[width=\columnwidth]{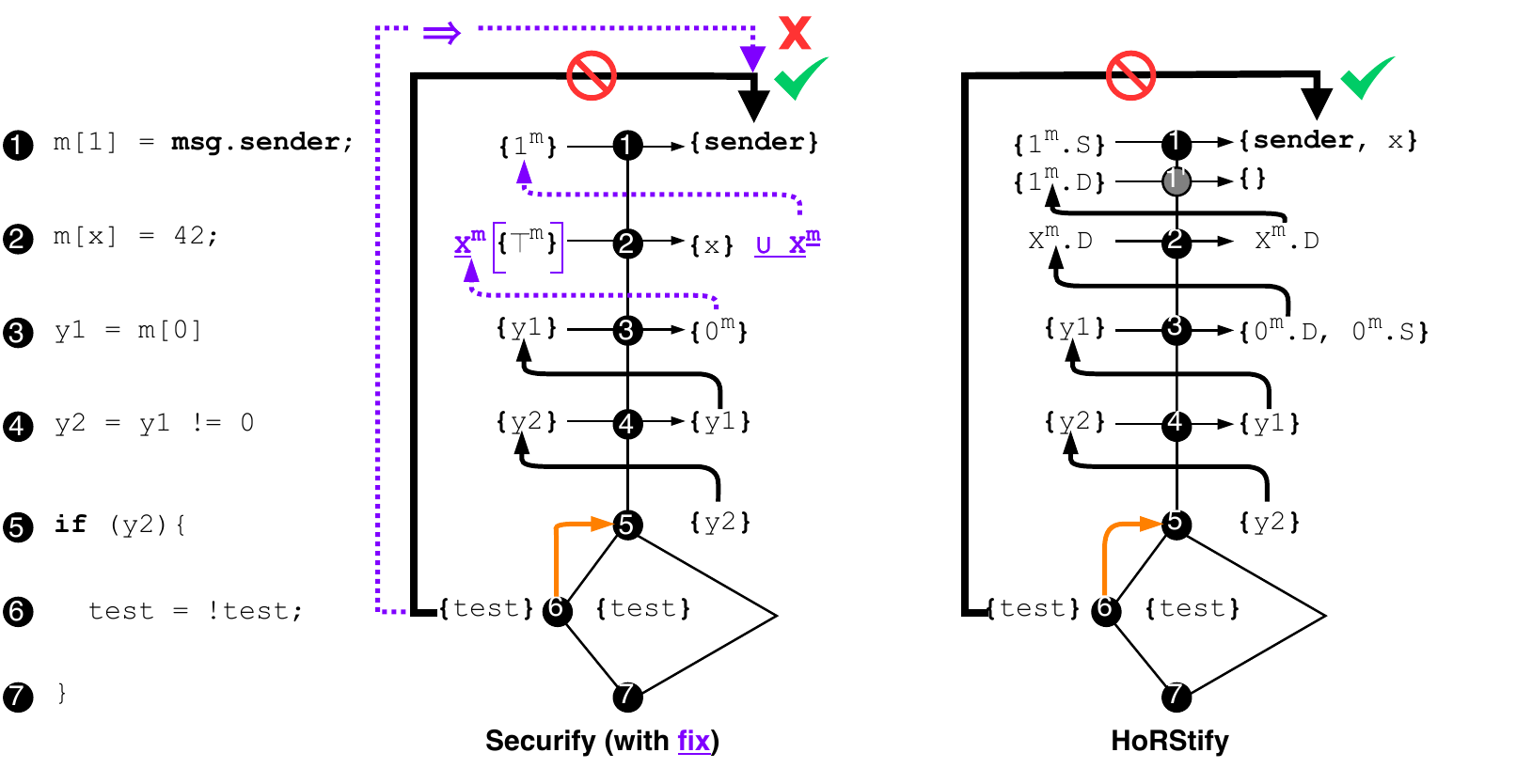}
	\caption{Contract violating the RW property with PDGs depicting the dependencies as modeled by Securify and \horstify.}
	\label{fig:memory-example-independent}
\end{figure}
%
This variant of \Cref{fig:memory-example-dependent} first writes \lstinline|msg.sender| to the known memory location \lstinline|1| in node \textnode{1} and then writes a value to an unknown memory location in node \textnode{2}. 
Since the condition \lstinline|y2| only depends on the value in memory location \lstinline|0| while \lstinline|msg.sender| was written to location \lstinline|1|, the final write to the \lstinline|test| variable in \textnode{6} does \textbf{not} depend on \lstinline|msg.sender|. 
However, the hypothetical fix of Securify infers a possible dependency between \textnode{6} and \lstinline|msg.sender| (shown in violet in the left CFG in~\Cref{fig:memory-example-independent}).
This imprecision is caused by interpreting a write to an unknown memory location as a write to possibly all memory locations as this requires the Use set in \textnode{2} to contain $\memvar{X}{}$. 
This creates a dependency between the assignment of location \lstinline|1| to \lstinline|msg.sender| in \textnode{1} and the memory access in \textnode{2}. 

Our memory abstraction is sound but more precise than the hypothetical fix above.
For every memory variable $x$ we use two sub-variables instead:
\emph{\memvarconcname} $\memvarconc{x}{}$ stores values that are assigned to $x$ when the memory location for $x$ is statically known, and \emph{\memvarabsname} $\memvarabs{x}{}$ stores values assigned to $x$ when $x$'s location is not statically known.~
During the execution, every write access to a memory variable $x$ stores the assigned value in $\memvarabs{x}{}$, unless the memory location for $x$ is statically known, in which case $\memvarconc{x}{}$ stores the value and $\memvarabs{x}{}$ is set to $\bot$. 
Correspondingly, when reading from a variable (regardless of the memory location being statically known or not), first, the value of the \memvarabsname is read, and only if it is $\bot$, the value of the \memvarconcname is taken.
We model this read access 
with the function

\noindent
{\small 
\begin{align*}
	\textit{load}~\cfgstate~\loc &= 
	\begin{cases}
		\cfgstate[\memvarconc{\loc}{}] \qquad \text{if } \cfgstate[\memvarabs{\loc}{}] = \None \\
		\cfgstate[\memvarabs{\loc}{}] \qquad \text{otherwise.} 
	\end{cases}
\end{align*}
}

This two-layered memory abstraction ensures that the execution is deterministic and that the read values coincide with those obtained during an execution without prior preprocessing. 
\begin{figure*}
	\scriptsize
	\begin{minipage}{0.33\textwidth}
	\begin{mathpar}
		\infer{
			\evmcontract(\lpc) = (\MLOAD(\stackvar{y}, \stackvar{x})), \lpc_\textit{next}, \pre) \\
			\pre[1] = \None \\
			\stateupdate{\cfgstate \leftarrow \stackvar{y} \define \textit{load}~\cfgstate~(\cfgstate[\stackvar{x}])} \\
		}
		{\edgestep{\transenv}{\cfgnode{\lpc}{0}}{\cfgnode{\lpc}{1}}}
		\\
		{\Defs{\stackvar{y}} \\
			\Use{\mt \cup \mk \cup \{\stackvar{x}\}}}
		
			\infer{
				\evmcontract(\lpc) = (\MSTORE(\stackvar{x_1}, \stackvar{x_2}), \lpc_\textit{next}, \pre) \\
				\pre[0] = \None \\ 
				\stateupdate{\cfgstate \leftarrow \varabs{\cfgstate[\stackvar{x_1}]} \define \cfgstate[\stackvar{x_2}]} \\
			}
			{\edgestep{\transenv}{\cfgnode{\lpc}{0}}{\cfgnode{\lpc}{1}}}
			\\
			{\Def{\mt} \\
				\Use{\mt \cup \{\stackvar{x_1}, \stackvar{x_2}\}}}

	\end{mathpar}
\end{minipage}
\begin{minipage}{0.33\textwidth}
\begin{mathpar}
	\infer{
		\evmcontract(\lpc) = (\MLOAD(\stackvar{y}, \stackvar{x})), \lpc_\textit{next}, \pre) \\
		\pre[1] = \optional{\memvar{\loc}{}} \\
		\stateupdate{\cfgstate \leftarrow \stackvar{y} \define \textit{load}~\cfgstate~\memvar{\loc}{}} \\
	}
	{\edgestep{\transenv}{\cfgnode{\lpc}{0}}{\cfgnode{\lpc}{1}}}
	\\
	{\Defs{\stackvar{y}} \\
		\Use{\{\memvarconc{\loc}{}, \memvarabs{\loc}{}\}}}

		\infer{
			\evmcontract(\lpc) = (\MSTORE(\stackvar{x_1}, \stackvar{x_2}), \lpc_\textit{next}, \pre) \\
			\pre[1] = \optional{\memvar{\loc}{}} \\
			\stateupdate{\cfgstate \leftarrow \varconc{\memvar{\loc}{}} \define \cfgstate[\stackvar{x_2}]} \\
		}
		{\edgestep{\transenv}{\cfgnode{\lpc}{0}}{\cfgnode{\lpc}{1}}}
		\\
		{\Defs{\memvarconc{\loc}{}} \\
			\Use{\{ \stackvar{x_2} \}}}

\end{mathpar}
\end{minipage}
\begin{minipage}{0.33\textwidth}
	\begin{mathpar}
		\infer{
			\evmcontract(\lpc) = (\MSTORE(\stackvar{x_1}, \stackvar{x_2}), \lpc_\textit{next}, \pre) \\
			\pre[0] = \optional{\memvar{\loc}{}} \\ 
			\stateupdate{\cfgstate \leftarrow \varabs{\memvar{\loc}{}} \define \bot} \\
		}
		{\edgestep{\transenv}{\cfgnode{\lpc}{1}}{\cfgnode{\lpc}{2}}}
		\\
		{\Defs{\memvarabs{\loc}{}} \\
			\Use{\emptyset}}
	\end{mathpar}
	\end{minipage}
	\caption{$\MLOAD$ memory abstraction instantiation}
	\label{fig:memory_abs}
\end{figure*}%
$\textit{load}$ is used in the inference rules in \Cref{fig:memory_abs} that define the memory read and write operations for the CFG semantics.
In these rules, we use $\mk$ for the set of all \memvarconcnames and $\mt$ for the set of all \memvarabsnames.
The leftmost \MSTORE{} rule is for the case that a value is written to a memory location that could not be statically inferred (i.e., $\pre[0] = \bot$). There, any of the memory variables from $\mt$ might be redefined, hence the Def set contains all variables in $\mt$.
As discussed for the hypothetical fix of Securify, also the Use set needs to include $\mt$, 
because we must not interrupt potential dependencies for memory cells that are not changed by this \MSTORE{} instruction.
%
An example 
for this is node \textnode{2} in~\Cref{fig:memory-example-independent} (right CFG).
The \memvarconcnames are not part of the Use set and hence not part of the value intermingling in \textnode{2}. 
This removes the imprecision that occurred in the proposed hypothetical fix above.
Still, the \MLOAD{} rules make sure that no dependencies to \memvarconcnames are missed by adding both \memvarabsnames and \memvarconcnames to the Use set.
This way, the connection between the memory location and the stored value is preserved; $\memvarabs{\loc_1}{}$ does not inherit any data dependencies from $\memvarconc{\loc_2}{}$ for locations $\loc_1 \neq \loc_2$.
An example for that is given in~\Cref{fig:memory-example-independent}, where memory location \lstinline|0| does not inherit the dependency from memory location \lstinline|1| written in \textnode{1}. 
This is thanks to the node splitting at \textnode{1} that breaks the propagation of dependencies on precomputed locations to dynamic ones.

\paragraph{Call Abstraction}
Contract calls in Ethereum trigger a multitude of possible (side) effects. 
When calling another account, the control flow is handed over to the code residing in this account.
This code may initiate further internal transactions, e.g., perform money transfers or even reenter the calling contract before reporting back the result to the callee.

This behavior poses a big challenge to sound static analysis since all possible effects of interactions with other (potentially unknown) contracts need to be over-approximated. 
Securify avoids this challenge by sacrificing soundness and ignoring all data dependencies arising from external calls (including effects of reentrancy) as demonstrated by the examples in~\Cref{code:counter_reentrancy} and~\Cref{fig:call-counter}.
In contrast, to give a sound and precise characterization of these dependencies, we first simplify the problem by restricting our analysis to a set of well-behaved smart contracts and then model the remaining dependencies in a fine-grained manner.

The class of smart contracts that we target are such contracts that cannot write storage variables in reentering executions. 
This restriction rules out race conditions on contract variables and as such is a highly-desirable property that can be easily achieved (e.g., by a strict local locking discipline). 
We call contracts satisfying this restriction \emph{store unreachable}:
\begin{assumption}[Store unreachability (informal)]
    \label{def:storeunreachability}
    A contract $\evmcontract$ is store unreachable if all its reentering executions cannot reach an $\SSTORE$ instruction.
\end{assumption}
The contract in~\Cref{code:counter_reentrancy} trivially violates store unreachability since the field \lstinline|a| can be written in a reentering execution.
This could be easily fixed by guarding each function with a lock that blocks reentering executions. 
Store unreachability is a local reachability property of the contract under analysis and as such falls in the scope of the sound analysis tool eThor~\cite{schneidewind2020ethor} and hence can be automatically verified.



Even when focussing on store unreachable contracts, the program dependencies induced by external calls are manifold and often subtle.
\Cref{fig:call-rule-external} shows one (slightly simplified) rule of the CFG semantics for external calls. 
As seen in the previous examples, node splitting is used to separate the dependencies of different variables. 
The rule displayed in~\Cref{fig:call-rule-external} gives one of the rules for setting a call's return value (written to the stack variable $\stackvar{y}$) and updating the external environment (represented by variable $\globenvvar{\extenv}$) according to the call effects.

\begin{figure}
    \scriptsize
    \begin{mathpar}
        \infer{
            \precontract(\pc) = (\CALL(\stackvar{y}, \stackvar{\lgas}, \stackvar{\recipient}, \stackvar{\valu}, \stackvar{\io}, \stackvar{\is}, \stackvar{\oo}, \stackvar{\os}), \pc', \pre)\\
            f_1 = \fun{\cfgstate}{\updatestate{\cfgstate}{\stackvar{y}}
            {\accesscfgstate{\applycall(\cfgstate, \precontract, \pc)}{\stackvar{y}}}} \\
            f_2 = \fun{\cfgstatefull}{\updatestate{\cfgstate}{\globenvvar{\cfgexternalpc{\pc'}}}
            {\accesscfgstate{\applycall(\cfgstate, \precontract, \pc)}{\globenvvar{\cfgexternalpc{\pc'}}}}}  \\
            f = \fun{\cfgstate}{f_2(f_1(\cfgstatefull))}\\
        }
        {\edgestep{\transenv}{\cfgnode{\pc}{0}}{\cfgnode{\pc}{1}}} \\
        {\Def{  \{ \stackvar{y}, \globenvvar{\cfgexternalpc{\pc'}}}\}} \\
            \Use{ \{ \stackvar{\lgas}, \stackvar{\recipient}, \stackvar{\valu}, \stackvar{\io}, \stackvar{\is}, \stackvar{\oo}, \stackvar{\os}, \locenvvar{\gaspc{\pc}}, \locenvvar{\activeaccount} \}
             ~\cup~ \memvar{X} ~\cup~ \globenvvar{X}~\cup~ \storvar{X}{\pc}
            } 
    \end{mathpar}
    \caption{Simplified CFG rule for the $\CALL$ opcode}
    \label{fig:call-rule-external}
\end{figure}

To obtain the updated CFG state after a call, the rule uses the function $\applycall$, which executes the internal transaction initiated by the \CALL{} opcode\footnote{We define $\applycall$ using the EVM semantics and hence can infer Def and Use sets from the corresponding EVM semantics rules.}.
The CFG state resulting from this execution is then used to describe the state updates (in the case of the given rule, the updates on the variables $\stackvar{y}$ and $\globenvvar{\extenv}$, as indicated by the Def set). 
Even though the whole CFG state $\cfgstate$ is taken as an argument by $\applycall$, not all variables in $\cfgstate$ can influence all aspects of the state after returning.
The variables that indeed may affect $\stackvar{y}$ and $\globenvvar{\extenv}$ are given in the Use set. 
More precisely, the result of a call may still depend on the global state, so all global environmental variables ($\globenvvar{X}$), as well as the global variables of the contract under analysis itself ($\storvar{X}{}$). 
Additionally, the execution of the called contract can be influenced by the parameters given to the call: 
The argument $\stackvar{\lgas}$ attributes to the amount of gas given to the call, $\stackvar{\recipient}$ gives the address of the recipient account and $\stackvar{\valu}$ the amount of money transferred with the call.
The arguments $\stackvar{\io}$ and $\stackvar{is}$ specify the memory fraction (offset and size) from which input data to the call is read and $\stackvar{\oo}$ and $\stackvar{\os}$ correspondingly define the memory fraction where the call's result data will be written.
In the given simplified rule, we consider that the concrete memory fragments could not be precomputed and hence all memory ($\memvar{X}{}$) could potentially be input data to the call.
The Use set also contains the calling account (as given in $\locenvvar{\activeaccount}$), since this information is made accessible during a call. 
Finally, the Use set contains the amount of gas that is available at the point of calling (given by $\locenvvar{\lgas}$) since this value may influence the amount of gas given to the call.

We want to highlight two forms of dependencies, which may erroneously be assumed to be ruled out by the assumption of store unreachability:
First, the Use set explicitly contains the storage variables ($\storvar{X}{}$) of the contract under analysis, even though we assume this contract to be store unreachable and (by the semantics) its storage variables cannot be accessed by any other contract.
Second, both the Def and the Use set contain the variable $\globenvvar{\extenv}$ that represents the external environment (in particular the state of other contract accounts). 
This implies that the rule in~\Cref{fig:call-rule-external} explicitly models information to be stored and retrieved from contract accounts during an external call.
In~\Cref{code:storage-call-example,code:external-example}, we illustrate the need for these dependencies by two examples.
\begin{figure}
\begin{lstlisting}
contract RetrieveSender { 
  function getTestSender() public returns (address) {
    try Test(msg.sender).getSender() returns (address a) { |\label{line:recall-test}|
      return a; }
    catch {return address(0); }}}

contract Test { 
  bool test = false;
  address sender; 
  RetrieveSender rs = RetrieveSender (address(42)); 
  function getSender () public returns (address) { 
    return sender;}
  function flip () public {  
    sender = msg.sender; 
    try rs.getTestSender() returns (address a)  { |\label{line:call-retrieve-sender}|
      if (a != address(0)){ |\label{line:a-check-address}|
        test = !test;}}
    catch {return; }}}
\end{lstlisting}
    \caption{Example: Reading storage variables during reentering execution.}
    \label{code:storage-call-example}
\end{figure}

\begin{figure}
\begin{lstlisting}
contract SaveAddr { 
  address addr = address(0); 
  function set(address a) public {
    addr = a; }
  function get( ) public returns (address) {return addr; }}

contract Test { 
  bool test = false;
  SaveAddr sa = SaveAddr (address(42)); 
  function flip () public {  
    try sa.set(msg.sender) { |\label{line:setAddr}|
      try c.get() returns (address a)  {  |\label{line:getAddr}|
        if (sa != address(0)){ |\label{line:a-check-addr-external}|
          test = !test; } }
      catch {return; }}
    catch {return; } } }
\end{lstlisting}
    \caption{Example: Propagating dependencies via an external contract account.}
    \label{code:external-example}
\end{figure}

The example in~\Cref{code:storage-call-example} shows how dependencies on a storage variable are introduced by reading a contract variable during a reentering execution. 
Note that store unreachability only assures that reentering executions can not write contract variables, but does not prevent read accesses.
The example gives another version of the \lstinline|Test| contract, which performs the check of \lstinline|msg.sender| in an indirect way: 
First, it writes \lstinline|msg.sender| to the contract variable \lstinline|sender|.
To read the variable again, a \lstinline|RetrieveSender| contract \lstinline|rs| is used as a proxy:
\footnote{Note that in Ethereum, a contract is identified by its address. 
In Solidity, the syntax \lstinline|RetrieveSender rs = RetrieveSender (address(42))| means that the contract at address 42 is assumed to be (of the type) RetrieveSender and accessible via variable \lstinline|rs|.} 
The \lstinline|Test| contract calls \lstinline|RetrieveSender|'s \lstinline|getTestSender| function (in line~\ref{line:call-retrieve-sender}), which in turn reenters \lstinline|Test| via its \lstinline|getSender| function (in line~\ref{line:recall-test}) to obtain the value of \lstinline|sender|.
This value is finally returned to contract \lstinline|Test|.
As a consequence, the return variable \lstinline|a| in line~\ref{line:a-check-address} contains the value of \lstinline|msg.sender|, and so the assignment of variable \lstinline|test| is dependent on \lstinline|msg.sender|.
This dependency, however, can only be tracked when considering that the contract's own storage variables may influence the return value of an external call.

The example in~\Cref{code:external-example} shows how dependencies can be propagated via another contract account. 
Note that store unreachability is a contract-specific property that only ensures that the contract under analysis is not written in reentering executions. 
The assumption does not restrict the storage modification of other contracts. 
The version of the \lstinline|Test| contract given in~\Cref{code:external-example} uses the contract \lstinline|SaveAddr| to propagate the value of \lstinline|msg.sender|. 
To this end, it first writes the value of \lstinline|msg.sender| into the \lstinline|addr| storage variable of the \lstinline|SaveAddr| contract \lstinline|sa| using the \lstinline|set| function (in line~\ref{line:setAddr}). 
Afterwards, it retrieves the value back by accessing \lstinline|c|'s storage via the \lstinline|get| function (in line~\ref{line:getAddr}). 
Consequently, the return variable \lstinline|a| contains the value of \lstinline|msg.sender| in line~\ref{line:a-check-addr-external} what makes the following write to \lstinline|test| dependent on that value.
This dependency can only be faithfully modeled when considering that an external call may change the state of other accounts, and may also be influenced by this state. 
This motivates why the $\globenvvar{\extenv}$ variable needs to be included in both the Def and the Use set of the rule in~\Cref{fig:call-rule-external}.



\subsection{Soundness Reasoning via Dependency Predicates}
\label{subsec:predicates}


Inspired by Securify, we define dependency predicates that can capture the data and control flow dependencies induced by the PDG (as given through the CFG semantics). 
They are inhabited via a set of logical rules (CHCs) $\rules{\evmcontract}$ that describe the data and control flow propagation through the PDG of a contract $\evmcontract$. 
More formally, the transitive closure of the $\evmcontract$'s PDG is computed as the least fixed point over $\rules{\evmcontract}$ (denoted by $\lfp{\rules{\evmcontract}}$).
Most prominently, $\lfp{\rules{\evmcontract}}$ includes the predicates $\pred{VarMayDepOn}$ and $\pred{InstMayDepOn}$.
Intuitively, $\pred{VarMayDepOn}(y,x)$ states that the value of variable $y$ may depend on the value of variable $x$ and $\pred{InstMayDepOn}(\node,x)$ says that the reachability of node $\node$ may depend on the value of variable $x$.
In the following, let $\node_x$ and $\node_y$ denote nodes that define variables $x$ and $y$, respectively.
The formal relation between dependency predicates and backward slices is captured by the following lemma: 

\begin{lemma}[Fixpoint Characterization of Backward Slices]
    \label{lem:fixed-point-backward}
Let $x$ and $y$ be variables and $\evmcontract$ be a contract. The following holds:
\begin{enumerate}
    \item \small{$(\exists \node_x~\node_y.~ \bs{\node_x}{\node_y}) \Rightarrow  \pred{VarMayDepOn}(y,x) \in \lfp{\rules{C}} $}
    \item \small{
    \(
        \begin{aligned}[t]
            (\exists ~\node~\node_\textit{if}~\node_x. ~ \node_\textit{if} \xrightarrow{~}_{cd} \node ~\land~ \bs{\node_x}{\node_\textit{if}}) \\
             \Rightarrow \pred{InstMayDepOn}(\node,x) \in \lfp{\rules{C}}
        \end{aligned}
      \)}
\end{enumerate}

\end{lemma}
\Cref{lem:fixed-point-backward} states 1) that whenever there is a node $\defnode{x}$ defining $x$ in the backward slice of a node $\defnode{y}$ defining $y$, then $\pred{VarMayDepOn}(y,x)$ is derivable from the CHCs in $\rules{\evmcontract}$ and 2) that whenever there is a node $\node_x$ defining $x$ in the backward slice of a node $\node_\textit{if}$ on which node $\node$ is control dependent then $\pred{InstMayDepOn}(\node,x)$ is derivable from $\rules{\evmcontract}$. 
The intuition behind statement 2) is that node $\node$ is controlled by $\node_\textit{if}$ (by the definition of standard control dependence), which means that $\node_\textit{if}$ is a branching node. 
$\bs{\node_x}{\node_\textit{if}}$ indicates that the branching condition of $\node_\textit{if}$ depends on variable $x$ and, hence, so does the reachability of $\node$.

Next, we give an explicit semantic characterization of the dependency predicates, which we prove sound using~\Cref{proof:correctness_slicing}. 
This explicit characterization enables us to compose \emph{security patterns} as a set of different facts over dependency predicates and to reason about them in a modular fashion.
As a consequence, we can show in~\Cref{sec:soundApproxPatterns} that checking the inclusion of security patterns in the least fixpoint of the rule set 
$\rules{\evmcontract}$ is sufficient to prove non-interference-style properties.
Concretely, we can characterize facts from the $\pred{VarMayDepOn}$ predicate as follows:
\begin{theorem}[Soundness of Dependency Predicates]
    \label{thm:soundness-predicates}
    \small{
\begin{align*}
    \forall x~y.~\pred{VarMayDepOn}(y,x) \not \in \lfp{\rules{C}} 
    \Rightarrow \independent{y}{x}
\end{align*}}
with $\independent{y}{x}$ given as:
\small{
\begin{align*}
    \forall \defnode{x}~i~\cfgstate_1~\cfgstate_2~\cfgstate_1'.~ 
        \cfgstate_1 \equalupto{x} \cfgstate_2 ~\land~ \istep{\nodesucc{\defnode{x}}}{\cfgstate_1}{\defnodes{y}}{i}{\node}{\cfgstate_1'} \\
        \Rightarrow \exists \cfgstate_2'.~ \istep{\nodesucc{\defnode{x}}}{\cfgstate_2}{\defnodes{y}}{i}{\node}{\cfgstate_2'}  ~\land~ \cfgstate_1'(y) = \cfgstate_2'(y)
\end{align*}
}
where 
$\nodesucc{\defnode{x}}$ denotes the unique successor node of $\defnode{x}$, and $\defnodes{y}$ the set of all nodes defining $y$.
$\istep{\nodesucc{\defnode{x}}}{\cfgstate_1}{\defnodes{y}}{i}{\node}{\cfgstate_1'}$ describes an execution from $\defnode{x}$ to $\node$ that passes exactly $i$ nodes defining $y$.
\end{theorem}

The theorem states that if $\pred{VarMayDepOn}(y,x)$ is not included in $\lfp{\rules{C}}$ then $y$ is independent of $x$ ($\independent{y}{x}$).
A variable $y$ is considered independent of $x$ if for any two configurations $\cfgstate_1$ and $\cfgstate_2$ that are equal up to $x$, and any execution starting at node $\nodesucc{\defnode{x}}$, the first node after $x$ is defined, passing $i$ nodes that define $y$, and ending in a node $\node$ at state $\cfgstate_1'$, one can find a matching execution from $\cfgstate_2$ that passes the same number of nodes defining $y$ and ends at node $\node$ in a state $\cfgstate_2'$ such that $\cfgstate_2'$ and $\cfgstate_1'$ agree on $y$.
This definition ensures loop sensitivity: it captures that during a looping execution, every individual occurrence of a node defining $y$ can be matched by the other execution---so that the values of $y$ agree whenever $y$ gets reassigned.
The proof of~\Cref{thm:soundness-predicates} uses~\Cref{lem:fixed-point-backward} and~\Cref{proof:correctness_slicing}.
For the full proof and a similar characterization of $\pred{InstMayDepOn}(i,x)$, we refer to\ifextended ~Appendix~\ref{sec:may-analysis}.\fi

\subsection{Sound Approximation of Security Properties}
\label{sec:soundApproxPatterns}
With \Cref{thm:soundness-predicates} we are able to formally connect dependency predicates and (independence-based) security properties. 
%
We take \emph{\tracenoninterference{}} as a concrete example, which comprises a whole class of non-interference-style security properties. 
Concretely, we consider \tracenoninterference\ w.r.t. a set of EVM configuration components $\inputcomponents$, which includes, for example, the block timestamp. A predicate $\outputproj$ defines \emph{instructions of interest}. If two executions of a contract $C$ start in configurations that differ only in the components in $\inputcomponents$, then the instructions of interest must coincide in the two traces that result from these executions.

	\begin{definition}[\Tracenoninterference]
	Let $\precontract$ be an EVM contract, $\inputcomponents$ be a set of components of EVM configurations and $\outputproj$ be a predicate on instructions. 
	Then \tracenoninterference{} of contract $\precontract$ w.r.t. $\inputcomponents$ and $\outputproj$ (written $\tracenoninter{\evmcontract}{\inputcomponents}{\outputproj}$) is defined as follows: 
	{\small
		\begin{align*}
			\tracenoninter{\evmcontract}{\inputcomponents}{\outputproj}
			&\define
			\forall ~\transenv ~\transenv' ~\exstate ~\exstate'~ t~ t' ~\pi. ~\pi'~ 
			\\
			&(\transenv, \exstate) \equalupto{\inputvars} (\transenv', \exstate') 
			\\
			&\Rightarrow \sstepstrace{\transenv}{\cons{\annotate{\exstate}{\evmcontract}}{\callstack}}{\cons{\annotate{t}{\evmcontract}}{\callstack}}{\pi} ~\land~ \finalstate{t}
			\\
			&\Rightarrow \sstepstrace{\transenv}{\cons{\annotate{\exstate'}{\evmcontract}}{\callstack}}{\cons{\annotate{t'}{\evmcontract}}{\callstack}}{\pi'} ~\land~ \finalstate{t'}
			\\
			&\Rightarrow \project{\pi}{\outputproj} = \project{\pi'}{\outputproj}
		\end{align*}
	}
		where $\project{\pi}{\outputproj}$ denotes the trace filtered by $\outputproj$, so containing only the instructions satisfying $\outputproj$.
		
	\end{definition}

	The dependency properties defined in~\cite{grishchenko2018semantic} can be expressed in terms of \tracenoninterference.
	E.g., the timestamp independence property in~\Cref{def:timestampIndependence} is captured as an instance of
	\tracenoninterference{} as follows:
	$$\tracenoninter{\evmcontract}{\{\access{\transenv}{\timestamp}\}}{\fun{\instruction}{\instruction = \CALL}}$$

	We show that we can give a sufficient criterion for \tracenoninterference{} in terms of dependency predicates. 
	More precisely, we give a set $\patnoninter{\inputcomponents}{\outputproj}{\precontract}$ of facts, such that $\patnoninter{\inputcomponents}{\outputproj}{\rules{\precontract}} \cap \lfp{\rules{\evmcontract}} = \emptyset$ implies $\tracenoninter{\evmcontract}{\inputcomponents}{\outputproj}$. 
	Practically, this means that we can prove $\tracenoninter{\evmcontract}{\inputcomponents}{\outputproj}$ by computing the least fixpoint over the CHCs $\rules{\precontract}$ (e.g., using a datalog engine) and then check whether it contains any fact from $\patnoninter{\inputcomponents}{\outputproj}{\precontract}$.
	For components in $\inputcomponents$, we assume a function $\componentToVar$ that maps components of the EVM semantic domain to CFG variables.
	The dependency predicates constituting a security pattern for \tracenoninterference{} are defined as
	{\small\begin{align*}
		\patnoninter{\inputcomponents}{\outputproj}{\precontract}\define
		&\{  \pred{InstMayDepOn}(\lpc, \componentToVar(\inputcomponent)) ~|~ \inputcomponent \in \inputcomponents \\
			&~\land~ \precontract(\lpc) = \instruction(\vec{x}, \nextpc, \pre)
			~\land~ \outputproj(\instruction)
			\} \\
		&\cup 
		\{  \pred{VarMayDepOn}(x_i, \componentToVar(\inputcomponent)) ~|~ \inputcomponent \in \inputcomponents \land \lpc\in\textsf{dom}({\precontract}) \\
		&~\land~ C(\lpc) = (\instruction(\vec{x}, \nextpc, \pre))
		~\land~ \outputproj(\instruction)
		~\land~ x_i \in \vec{x} \qquad\!\!
		\}.
	\end{align*}} 
	The following theorem shows that $\patnoninter{\inputcomponents}{\outputproj}{\precontract}$ is a security pattern for \tracenoninterference{}:
	\begin{theorem}[Soundness of \tracenoninterference]
		\label{thm:soundness-tni}
		Let $\precontract$ be a contract, $\inputcomponents$ a set of components, and $\outputproj$ an instruction-of-interest predicate.
		Then it holds that
		%
		{\small 
		\begin{align*}
			(\forall p\in \patnoninter{\inputcomponents}{\outputproj}{\precontract}.~ p \not \in \lfp{\rules{\precontract}}) 
			\Rightarrow \tracenoninter{\precontract}{\inputcomponents}{\outputproj}.
		\end{align*}}
	\end{theorem}

The absence of facts from $\patnoninter{\inputcomponents}{\outputproj}{\precontract}$ in $\lfp{\rules{\evmcontract}}$ ensures that the reachability of all instructions satisfying $\outputproj$ is independent of variables  representing components in $\inputcomponents$ and that all arguments $x_i$ of such instructions are independent of $\inputcomponent$ as well. 
These independences imply \tracenoninterference{} since they ensure that in two executions starting in configurations equal up to $\inputcomponents$, all instructions satisfying $\outputproj$ are executed in the same order (otherwise their reachability would depend on $\inputcomponents$) and with the same arguments (otherwise their argument variables would depend on $\inputcomponents$).
Consequently, such executions produce the same traces, when only considering instructions satisfying $\outputproj$. 
A full proof of~\Cref{thm:soundness-tni} can be found in\ifextended~Appendix~\ref{sec:tni}. \fi

\subsection{Discussion}
In this section, we presented a sound analysis pipeline for checking security properties for linearized EVM bytecode contracts by means of reasoning about dependencies between variables or instructions.
While our work was inspired by Securify~\cite{tsankov2018securify}, 
we developed new formal foundations for the dependency analysis of EVM bytecode contracts and in this way revealed several sources of unsoundness in the analysis of Securify.
Further, we provide soundness proofs for the analysis pipeline end-to-end\ifextended. \fi
The key pillars of the soundness proof are
\begin{inparaenum}[i)]
\item that our EVM CFG semantics satisfies all conditions to be used with the slicing framework~\cite{wasserrab2009pdg}, 
\item that the EVM linearized bytecode semantics and the CFG semantics are equivalent,
\item that our set of CHCs encodes an over-approximation of dependencies in an EVM contract, and 
\item that the generic security pattern $\patnoninter{\inputcomponents}{\outputproj}{\precontract}$ is a sound approximation of \tracenoninterference{}.
\end{inparaenum}
The proofs are valid under assumptions that are clearly stated in this paper.
For \Cref{asm:gas,def:storeunreachability} we point out the existence of other sound tools~\cite{schneidewind2020ethor,albert2021don} that can check these assumptions.

We assume that EVM smart contracts are provided in a (stack-less) linearized form.
Transforming into such a representation from a stack-based one is a well-studied problem~\cite{leung1999static} and a standard step performed by most static analysis tools~\cite{grech2019gigahorse,tsankov2018securify}.
Up to this requirement, our analysis is parametric with respect to other preprocessing steps.
More precisely, our analysis pipeline is sound for contracts with sound preprocessing information, and hence, in particular, for contracts without any preprocessing information but jump destinations needed for the CFG (cf. \Cref{subsec:instantiation}).
This gives the flexibility, to enhance the precision of the analysis through the incorporation of soundly precomputed values and makes the design of sound preprocessing an orthogonal problem. 
There exist already works on soundly precomputing jump destinations for EVM bytecode~\cite{grishchenko2020static}, which are to be complemented with other precomputing steps in the future.

\section{Evaluation}
\label{sec:evaluation}
The focus of this paper is on the theoretical foundations of a sound dependency analysis of smart contracts.
However, we demonstrate the practicality of the presented approach by developing the prototype analyzer \horstify.
We do not implement the logical rules from~\Cref{subsec:predicates} directly in Soufflé (as done by Securify), but encode them in the \horst specification language~\cite{schneidewind2020ethor}.
The \horst language is a high-level language for the specification of CHCs. 
By introducing this additional abstraction layer, we get a close correspondence between our theoretical rules and their actual implementation and, hence, anticipate a lower risk of implementation mistakes that may invalidate soundness claims in the implementation. 

\horstify accepts as input a set of dependency facts encoding the security patterns specified in the \horst language and Ethereum smart contracts in the EVM bytecode format.
It first invokes Securify's decompiler to transform the contract into 
a linearized representation and does some lightweight preprocessing to obtain the precomputable values (cf.~\Cref{subsec:instantiation}).
Then, \horstify uses our formal specification of the CFG construction rules and the \horst framework to create a Soufflé executable for the analysis and invokes it.


To reduce the risks of implementation mistakes, we proceeded in two steps.
First, we encoded Securify's RW violation pattern in the \horst language to execute HoRStify with this pattern and the contracts in \Cref{code:counter_storage,code:counter_gas,fig:call-counter}\footnote{We did not consider the contract in \Cref{code:counter_must} since it concerns the must-analysis and the contract in \Cref{code:counter_reentrancy}, which violates \Cref{def:storeunreachability}.}. 
In contrast to Securify, \horstify correctly determines that these contracts do not satisfy the RW violation pattern.
In addition to these corner cases, we successfully evaluated \horstify on Securify's internal test suite 
involving 25 contracts.

\newcommand{\securifyspecificiy}{\textsf{S}_{\textit{Sec}}}
\newcommand{\horstifyspecificiy}{\textsf{S}_{\textit{Hor}}}
\newcommand{\tnhy}{\textit{tn}_{\textit{Hor}}}
\newcommand{\tnsy}{\textit{tn}_{\textit{Sec}}}
\newcommand{\fnhy}{\textit{fn}_{\textit{Hor}}}
\newcommand{\fnsy}{\textit{fn}_{\textit{Sec}}}
\newcommand{\tphy}{\textit{tp}_{\textit{Hor}}}
\newcommand{\tpsy}{\textit{tp}_{\textit{Sec}}}
\newcommand{\fphy}{\textit{fp}_{\textit{Hor}}}
\newcommand{\fpsy}{\textit{fp}_{\textit{Sec}}}
\newcommand{\nsucccontracts}{\size{\textit{dataset}}}
\newcommand{\tnh}{2}
\newcommand{\fnh}{0}

\newcommand{\goodmismatches}{29}

  \begin{table}
    \begin{center}
      {
        \footnotesize
        \begin{center}
          \begin{tabular}{ccccc}
            \toprule
            \multirow{2}{*}{contracts} & \multirow{2}{*}{errors} & \multirow{2}{*}{timeouts} & contracts & \multirow{2}{*}{$\varnothing$ time (ms)} \\
            & & & \textbackslash (errors $\cup$ timeouts) & \\
            \midrule
            \multirow{2}{*}{720} & \textbf{H} 34 & \textbf{H} 46& \multirow{2}{*}{634} &  \textbf{H} 7055\\
            & \textbf{S} 34 & \textbf{S} 30& & \textbf{S} 3107\\
            \bottomrule \\
          \end{tabular}
        \end{center}
      }
      \caption{Large-scale evaluation of \horstify{} (\textbf{H}) and Securify (\textbf{S}).}
      \label{tab:eval}
    \end{center}
  \end{table}

Next, we conduct a large-scale evaluation of \horstify and Securify on real-world contracts. 
To this end, we use the sanitized dataset from~\cite{schneidewind2020ethor} that consists of 720 distinct smart contracts from the Ethereum blockchain.
We compare the performance of Securify and \horstify on this dataset for both the RW pattern and for timestamp independence (TS) as defined in~\Cref{sec:soundApproxPatterns}. 
We manually inspect all contracts on which Securify and \horstify report a different result.

  \begin{figure}
    \begin{center}
    \includegraphics[width=\columnwidth]{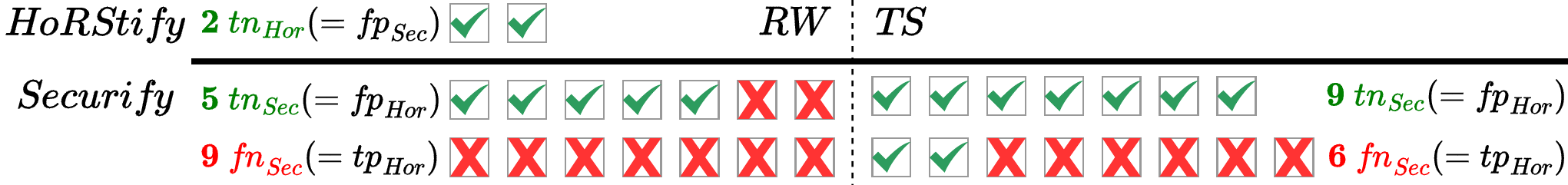}
    \end{center}
    \caption[Classification of mismatching results]{Classification of mismatching results of \horstify{} (above) and Securify (below) for the RW (left) and TS (right) property\footnotemark. 
    Ticks indicate correct matches (\textit{tn}) and crosses wrong matches (\textit{fn}) of the respective tool.
    $\tnhy$/$\tnsy$, $\fnhy$/$\fnsy$, $\tphy$/$\tpsy$, $\fphy$/$\fpsy$ denote true negatives, false negatives, true positives, and true negatives of \horstify/Securify, respectively.}
    \label{fig:mismatches}
  \end{figure}

\footnotetext{For TS we only consider the 165 contracts from the dataset containing a TIMESTAMP opcode, as Securify labels other contracts as trivially secure. The manual classification is a conservative best-effort estimate.}

\Cref{tab:eval} shows the evaluation results.
The average execution time of \horstify is approximately 2.3 times longer than for Securify.
Consequently, \horstify suffers from more timeouts than Securify; the execution of both tools is aborted after one minute.
  \Cref{fig:mismatches} visualizes the manual classification for those smart contracts where \horstify{} and Securify disagree. 
  There are only two contracts where \horstify{} matches the corresponding pattern, but Securify does not. 
  Recall that for a sound tool, a pattern match indicates the discovery of provable independencies that imply either property violation (RW) or compliance (TS).
  An erroneous pattern match by \horstify would present a soundness issue (false negative).
  We carefully examined the two examples and could confirm them not to constitute false negatives of \horstify{} but false positives of Securify ($\fpsy$), unveiling an imprecision of Securify.
  This seems surprising since our analysis generally tracks more dependencies than the one of Securify. 
  However, while \horstify implements standard control dependence to encode control dependencies (e.g., to compute join points after loops), Securify implements a less precise custom algorithm. 

  The contracts where Securify matches a pattern, but \horstify{} does not, can either reveal soundness issues (false negatives) of Securify ($\fnsy$) or a precision loss (false positives) of \horstify ($\fphy$). 
  Indeed, in the \goodmismatches{} contracts that are flagged only by Securify, we find both cases (as shown at the bottom of~\Cref{fig:mismatches}), as we will illustrate with two examples:

\begin{figure}[t]
\begin{lstlisting}
contract RNG {
  mapping (address => uint) nonces;
  uint public last;
  function RandomNumber() returns(uint) {
    return RandomNumberFromSeed(
      uint(sha3(block.number))^uint(sha3(now))
      ^uint(msg.sender)^uint(tx.origin)); }
  function RandomNumberFromSeed(uint seed) returns(uint) {
    nonces[msg.sender]++;
    last = seed^(uint(sha3(block.blockhash(block.number),
                              nonces[msg.sender]))
                    *0x000b0007000500030001);
    return last; }
  function Guess(uint _guess) returns (bool) {
    if (RandomNumber() == _guess) {
      if (!msg.sender.send(this.balance)) throw; |\label{lottery-ln:send}|
        RandomNumberGuessed(_guess, msg.sender);
        return true; }
    return false; } }
\end{lstlisting}
\caption{Lottery Contract 0xaed5a41450b38fc0ea0f6f203a985653fe187d9c}
\label{fig:rnd-contract}
\end{figure}

\Cref{fig:rnd-contract} shows a (slightly shortened) version of a contract classified as safe for TS according to Securify, but that HoRStify (correctly) reports as vulnerable.
It is a lottery contract that pays out a user who manages to guess a random number (function \lstinline|Guess|). 
The random number is generated from blockchain and transaction-specific values, including the timestamp (accessed via \lstinline|now| in \lstinline|RandomNumberFromSeed|).
Hence, the payout in line~\ref{lottery-ln:send} is not independent of the timestamp. 
Securify fails to detect this dependency due to its unsound memory abstraction (as described in~\Cref{sec:soundness-issues}): 
As Ethreum's hash function (\lstinline|sha3|) reads input from the local memory, the timestamp is written to the memory where its dependencies are lost.

\begin{figure}[t]
\begin{lstlisting}
contract lottery{
  address[] public tickets;
  function buyTicket(){
    if (msg.value != 1/10) throw;
    if (msg.value == 1/10)
      tickets.push(msg.sender); |\label{lottery2-ln:push}|
      address(0x88a1e54971b31974b2be4d9c67546abbd0a3aa8e)
        .send(msg.value/40);
    if (tickets.length >= 5) runLottery(); } |\label{lottery2-ln:if}|
  function runLottery() internal {
    tickets[addmod(now, 0, 5)].send((1/1000)*95);
    runJackpot();}
  function runJackpot() internal {
    if(addmod(now, 0, 150) == 0)
      tickets[addmod(now, 0, 5)].send(this.balance);
    delete tickets; } } |\label{lottery2-ln:delete}|
    \end{lstlisting}
    \caption{Lottery contract 0xe120100349a0b1BF826D2407E519D75C2Fe8f859}
    \label{fig:fp-contract}
\end{figure}

\REPLACESfor{}{230124}{The soundness of \horstify comes at the prize of some false positives. 
A typical example is depicted in~\Cref{fig:fp-contract}. }
{\Cref{fig:fp-contract} shows an example of a false positive for \horstify.}
The contract implements a lottery where users can register (via \lstinline|buyTicket|) and whenever 5 users were registered, one of them is selected as a winner. 
Despite the obvious timestamp dependency, the contract shows RW violations, which \horstify fails to prove.\footnote{Note that this is not a soundness issue since the soundness of \horstify ensures that independencies can be proven. In the case of violation patterns as RW the independence constitutes an unwanted effect and hence, we can only use it to prove the vulnerability of a contract, not its safety.}
E.g., the \lstinline|tickets| array is updated without performing a check on the sender. 
\horstify does not detect this vulnerability due to its sound storage abstraction: 
In line \ref{lottery2-ln:push}, the caller (\lstinline|msg.sender|) is appended to the \lstinline|tickets| array. 
Since the array position to which \lstinline|msg.sender| will be added cannot be statically known, \horstify needs to assume \lstinline|msg.sender| to be written to any position. 
When checking the size of \lstinline|tickets| in line~\ref{lottery2-ln:if}, the condition is considered dependent on \lstinline|msg.sender| (because in the abstraction, \lstinline|msg.sender| is considered to potentially affect all storage locations, including the one containing the array size). 
Thus, the \lstinline|delete| operation in line~\ref{lottery2-ln:delete} is considered dependent on \lstinline|msg.sender|. 
One should notice, that only the unsoundness of Securify's storage abstraction, enables Securify to correctly detect the RW violation in this case.

Overall, based on our evaluation results, we can bound the precision loss of \horstify w.r.t. Securify. 
More concretely, when considering that Securify has a specificity\footnote{
    The specificity is a standard precision measure and is calculated as $\frac{\textit{tn}}{\textit{tn} + \textit{fp}}$
    }
of $\securifyspecificiy$ on the full dataset, then one can easily show that it holds for the specificity $\horstifyspecificiy$ of \horstify that $\horstifyspecificiy \geq \securifyspecificiy + \frac{\tnhy - \tnsy}{\nsucccontracts}$ 
\REPLACESfor{}{230124}{where $\tnhy$ are the true negatives for \horstify found within the manually inspected mismatching contracts, and $\tnsy$ are the true negatives for Securify found within these contracts.}
{where $\tnhy$ are the true negatives for \horstify, and $\tnsy$ are the true negatives for Securify found within the manually inspected mismatching contracts.}
Inserting the results from~\Cref{fig:mismatches}, we can show that $\horstifyspecificiy$ can be at most $0.5$ percentage points less than $\securifyspecificiy$ for RW on the given dataset and at most $5.4$ percent points less for TS. 

We refer to \href{https://horstify.org}{\texttt{horstify.org}} for more information about \horstify.

\section{Related Work}
\label{sec:realated_work}
Existing approaches to enforce the correctness of Ethereum smart contracts can be broadly categorized into analyses at design time and analyses at runtime.
The latter include methods like runtime monitoring~\cite{ellul2018runtime,wang2019vultron} or information flow control mechanisms~\cite{cecchetti2021compositional}.
Such dynamic analysis approaches, however, have limited applicability to the Ethereum blockchain, since they either require fundamental updates to the workings of the EVM or impose tremendous costs in terms of gas.
Static analyses, in contrast, verify smart contracts at design time before they become immutable objects on the blockchain. 
Most static analyzers are bug-finding tools (such as Oyente~\cite{luu2016making}, EthBMC~\cite{frank2020ethbmc}, and Maian~\cite{nikolic2018finding}) 
that aim to reduce the number of contracts that are wrongly claimed to be buggy (false positives). 
To this end, these tools usually rely on the symbolic execution of the contract under analysis.
The dual objective of bug-finding is to prove a smart contract secure. 
Analyzers following this objective do not only aim at producing a low number of false negatives in practice but to give provable guarantees for their analysis result, e.g., that a contract flagged as safe is guaranteed to enjoy a corresponding security property.
The only example of a tool, which comes with a provable soundness claim, so far, is the analyzer eThor~\cite{schneidewind2020ethor}, whose analysis relies on abstract interpretation.

Symbolic execution and abstract interpretation have in common to target properties that can be decided for a finite prefix of a single (yet arbitrary) execution trace of a smart contract (so-called \emph{reachability properties}). 
    However, many generic security properties for smart contracts (as defined in~\cite{grishchenko2018semantic}) require comparing \emph{two} execution traces from different initial configurations and fall into the broader category of \emph{2-safety properties}. 
    To check 2-safety properties with tools whose analysis is limited to reachability properties (such as eThor) requires an overapproximation of the original property in terms of reachability. 
    But finding such a meaningful over-approximation, which does not result in an intolerable precision loss, is not always possible.
    In~\cite{grishchenko2018semantic}, it is, e.g., shown how to overapproximate the call integrity 2-safety property (characterizing the absence of reentrancy attacks) by a reachability property (single-entrancy) and two other properties, which are captured by our notion of trace noninterference.
    However, trace noninterference properties still concern two execution traces and hence cannot be verified using eThor.
    \horstify (inspired by the unsound Securify tool~\cite{tsankov2018securify}) devises a different analysis technique, which immediately accommodates the analysis of trace noninterference. 
    As opposed to the analysis underlying eThor, this technique does not allow for verifying general reachability properties, but a special class of 2-safety properties (including trace noninterference).
    \horstify and eThor, hence, can be seen as complementing tools that target incomparable property classes.
    The call integrity property falls neither in the scope of eThor nor \horstify{}, but its overapproximation decomposes it into trace noninterference properties (within the scope of \horstify) and a reachability property (within the scope of eThor). 
    Other generic security properties from~\cite{grishchenko2018semantic} for characterizing the independence of miner-controlled parameters (including timestamp independence) immediately constitute trace noninterference properties and as such can be analyzed by \horstify but not by eThor.

More complex properties involving both universal and existential quantification of execution traces~\cite{ClarksonS08,DBLP:conf/esop/DArgenioBBFH17} cannot be checked by either \horstify or eThor.

\section{Conclusion}
\label{sec:conclusion}
\ifextended
In this work, we present the first provably sound static dependency analysis for EVM bytecode. 
Taking up the approach of the state-of-the-art static analyzer Securify~\cite{tsankov2018securify}, we uncover conceptual soundness issues of the tool, so we replace the underlying analysis and spelled out formal soundness guarantees.
The soundness proof of our analysis relies on the proof framework from~\cite{wasserrab2009pdg} for static program slicing, which we instantiate for EVM bytecode. 
The slicing framework can capture the notion of may-dependence, whereas we elucidated that the must-dependence promoted by Securify raises soundness questions already at the conceptional basis.
Although we removed support for must-dependence, we could show that the resulting analysis is flexible enough to soundly characterize relevant smart contract security properties such as timestamp dependence. 
Finally, we demonstrate the practicality of the approach by providing the prototypical analyzer HoRStify. HoRStify encodes the slicing-based dependency analysis as logical rules that can be automatically solved by the Datalog solver Soufflé; it can verify real-world smart contracts, and even though being provable sound, shows performance comparable to Securify. 
\fi

\bibliographystyle{plain}
\bibliography{Bibliography.bib}

\ifextended
\newpage
\onecolumn

\begin{appendices}

\section{Sound Dependency Analysis for EVM bytecode} 
\subsection{Instantiation of the Slicing Framework}

\subsubsection{State transformation}
\label{appendix:subsec:state-trans}
 We formally define the state $\cfgstate$ of the EVM as used in the CFG semantics. 
Afterward, we define the state transformation functions $\tocfgstate$ and $\toevmstate$ that convert between the different EVM state representations. 

\paragraph*{EVM state}
We revisit the formal definition of the EVM state as given in~\cite{grishchenko2018semantic}. 

In the following, we will use $\BB$ to denote the set $\{0,1\}$ of bits and accordingly $\BB^{x}$ for sets of bitstrings of size $x$. 
We further let $\integer{x}$ denote the set of non-negative integers representable by $x$ bits and allow for implicit conversion between those two representations (assuming bitstrings to represent a big-endian encoding of natural numbers). 
In addition, we will use the notation $\arrayof{X}$ (resp. $\stackof{X}$) for arrays (resp. lists) of elements from the set $X$. We use standard notations for operations on arrays and lists. 
In particular we write $\arraypos{a}{\pos}$ to access position $\pos \in [1, \size{a} - 1]$ of array $a \in \arrayof{X}$ and $\arrayinterval{a}{\downv}{\upv}$ to access the subarray of size $\upv - \downv$ from position $\downv \in  [1, \size{a} - 1]$ to $\upv \in  [1, \size{a} - 1]$. In case that $\downv > \upv$ this operation results in the empty array $\emptyarray$. 
In addition, we write $\concat{a_1}{a_2}$ for the concatenation of two arrays $a_1, a_2 \in \arrayof{X}$. 

In the following formalization, we will make use of bytearrays $b \in \bytearray$. To this end, we will assume functions $\bitstringtobytearray{(\cdot)} \in \BB^x \to \bytearray$ and $\bytearraytobitstring{(\cdot)} \in \bytearray \to \BB^x$ to chunk bitstrings with size dividable by $8$ to bytearrays and vice versa.
To denote the zero byte, we write $0^8$ and, accordingly, for an array of zero bytes of size n, we write $0^{8\cdot n}$.  

For lists, we denote the empty list by $\nil$ and write $\cons{x}{\textit{xs}}$ for placing element $x \in X$ on top of list $\textit{xs} \in \stackof{X}$. 
In addition, we write $\concatstack{\textit{xs}}{\textit{ys}}$ for concatenating lists $\textit{xs}, \textit{ys} \in \stackof{X}$. 

We let $\addresses$ denote the set of $160$-bit addresses ($\BB^{160}$). 

In Figure~\ref{fig:grammar'} we give a full grammar for call stacks: 
\begin{figure*}[h]
\begin{mathpar}
\begin{array}{rlclll}
\text{Call stacks} & \callstacks & \ni & \callstack & \define & \cons{\excstate}{\callstackplain} ~|~  \cons{\haltstate{\gstate}{d}{g}{\transeffects}}{\callstackplain} ~|~ \callstackplain  \\
\text{Plain call stacks} & 
\callstacksplain & \ni & \callstackplain & \define & \cons{\regstate{\mstate}{\exenv}{\gstate}}{\callstackplain} \\
\text{Machine states} & \mstates & \ni & \mstate & \define & \smstate{\lgas}{\lpc}{m}{i}{s} \\
\text{Execution environments} & \exenvs & \ni & \exenv & \define & \sexenv{\textit{actor}}{\textit{input}}{\textit{sender}}{\textit{value}}{\textit{code}}\\
\text{Global states} & \gstates  & \ni &\gstate & &  \\
\text{Account states} & \accounts & \ni & \accountv & \define & \accountstate{n}{b}{\textit{code}}{\textit{stor}} ~|~ \bot   \\ 
\text{Transaction environments} & \transenvs & \ni & \transenv & \define & (o, \textit{prize}, H) \\
\text{Block headers} & \blockheaders & \ni & \blockheader & \define & (\parent, \beneficiary, \difficulty, \blocknumber, \gaslimit, \timestamp)\\
\\
\text{Notations:} 
&\multicolumn{5}{c}{
d \in \bytearray, \quad g \in \integer{256}, \quad \transeffects \in \teffects, \quad 
o \in \addresses, \quad \textit{prize} \in \integer{256}, \quad H \in \blockheaders} \\
&\multicolumn{5}{c}{
\lgas \in \integer{256}, \quad \lpc \in \integer{256}, \quad m \in \integer{256} \to \integer{256} \quad i \in \integer{256}, \quad s \in \integer{8} \to \integer{256}}  \\
&\multicolumn{5}{c}{
\textit{sender} \in \addresses \quad \textit{input} \in \bytearray \quad \textit{sender} \in \addresses \quad \textit{value} \in \integer{256} \quad \textit{code} \in \bytearray} \\
&\multicolumn{5}{c}{
b \in \integer{256} \quad \stor \in \integer{256} \to \integer{256} \quad L \in \sequenceof{\logevents} \quad \suicideset \subseteq \addresses \quad 
\gstates = \addresses \to \accounts 
} \\
&\multicolumn{5}{c}{
    \parent \in \integer{256} \quad \beneficiary \in \addresses \quad \difficulty \in \integer{256}} \\
&\multicolumn{5}{c}{    
    \blocknumber \integer{256} \quad \gaslimit \in \integer{256} \quad \timestamp \in \integer{256}
}
\end{array}
\end{mathpar}
\caption{Grammar for calls stacks and transaction environments}
\label{fig:grammar'}
\end{figure*}


Note that the grammar was slightly adapted to account for the fact that the local stack is assumed to be precomputed to local variables.
Further, the intermediate representation assumes all memory accesses to be aligned, meaning that memory acceses only occure at addresses that are multiples of $32$.

\paragraph*{CFG state}
We formally define the state for the CFG semantics of EVM bytecode. 
The formal definition of the CFG state is given as follows:

\begin{align*}
    \text{State} && \cfgstate &\define (\cfgstack, \cfgmemconc, \cfgmemabs, \cfgstorconc, \cfgstorabs, \cfglocenv, \cfgglobenv) \\
    \text{Stack} && \cfgstack &\in  \integer{8} \to \integer{256} \\
    \text{Local Static Memory} && \cfgmemconc & \in \integer{256} \to \integer{256}  \\
    \text{Local Dynamic Memory} && \cfgmemabs & \in \integer{256} \to \integer{256} \cup \{ \bot \} \\
    \text{Global Static Storage} && \cfgstorconc & \in  \integer{256} \to \integer{256} \\
    \text{Global Dynamic Storage} && \cfgstorabs & \in  \integer{256} \to \integer{256} \cup \{ \bot \} \\
    \text{Local Environment} && \cfglocenv & \define (\gaspc{\lpc}, \msizepc{\lpc}, \textit{actor}, \textit{input}, \textit{sender}, \textit{value}) \\
    \text{Global Environment} && \cfgglobenv & \define (\parent, \beneficiary, \difficulty, \blocknumber, \gaslimit, \timestamp, \originator, \gasprize, \cfgexternalpc{\lpc}) \\
    \text{External Global Environment} && \cfgexternal & \define (b, n, \gstate)
\end{align*}

We will treat the CFG state as a heterogeneous mapping and write 
$\accesscfgstate{\cfgstate}{\stackvar{x}}$ for $\access{\cfgstate}{\cfgstack}(x)$; 
$\accesscfgstate{\cfgstate}{\memvarconc{x}{\lpc}}$ for $\access{\cfgstate}{\cfgmemconc}(x)$; 
$\accesscfgstate{\cfgstate}{\memvarabs{x}{\lpc}}$ for $\access{\cfgstate}{\cfgmemabs}(x)$; 
$\accesscfgstate{\cfgstate}{\storvarconc{x}{\lpc}}$ for $\access{\cfgstate}{\cfgstorconc}(x)$; 
$\accesscfgstate{\cfgstate}{\storvarabs{x}{\lpc}}$ for $\access{\cfgstate}{\cfgstorabs}(x)$;
$\accesscfgstate{\cfgstate}{\locenvvar{x}}$ for $\access{\access{\cfgstate}{\cfglocenv}}{x}$; 
and $\accesscfgstate{\cfgstate}{\globenvvar{x}}$ for $\access{\access{\cfgstate}{\cfgglobenv}}{x}$.
In particular, we will treat (static and dynamic) memory and storage locations as variables and will write 
$\memvarabs{X}{}$ for the set of all dynamic memory locations,
$\memvarconc{X}{}$ for the set of all static memory locations, 
$\storvarabs{X}{}$ for the set of all dynamic storage locations,
and $\storvarconc{X}{}$ for the set of all static storage locations.
Further, we use $\memvar{X}{}$ and $\storvar{X}{}$ to denote the (distinct) sets of all memory and, respectively, storage locations. 

The state is partitioned according to the granularity of the analysis. 
All components whose dependencies are explicitly tracked occur on the top level. 

\paragraph{State transformation}

In the following we will assume the $\load$ function to be defined on both memory locations $\memvar{x}{\lpc}$ and storage locations $\storvar{x}{\lpc}$.  
\begin{align*}
	\load~\cfgstate~\loc &= 
	\begin{cases}
		\cfgstate[\varabs{\loc_{}}] \qquad \text{if } \cfgstate[\varconc{\loc_{}}] = \None \\
		\cfgstate[\varconc{\loc_{}}] \qquad \text{otherwise.} 
	\end{cases}
\end{align*}
where $\loc \in \memvar{X}{} \cup \storvar{X}{}$

Using this, the translation between the different state types can be defined as follows: 

{\small
\begin{align*}
\tocfgstate(\transenv, \exstate) \define 
\begin{cases}
    (\cfgstate, \precontract, \lpc) 
& \regstate{\mstate}{\exenv}{\gstate} = \exstate 
~\land~ \cfgstack = \access{\mstate}{\stack}
~\land~ \cfgmemconc = \fun{(i, \lpc)}{\access{\mstate}{\memo}(i)}
~\land~  \cfgmemabs = \fun{(i, \lpc)}{\bot} \\
&~\land~ \cfgstorconc = \fun{(i, \lpc)}{\access{\gstate(\access{\exenv}{\activeaccount})}{\stor}(i)}
~\land~  \cfgstorabs = \fun{(i, \lpc)}{\bot}
\\
& ~\land~ \cfglocenv = (\fun{\lpc}{\access{\mstate}{\gas}}, \fun{\lpc}{\access{\mstate}{i}}, \access{\exenv}{\activeaccount}, \access{\exenv}{\activeaccount}, \access{\exenv}{\sender}, \access{\exenv}{\valu}) \\
& ~\land~ \cfgexternal = \fun{\lpc}{(\access{\gstate(\access{\exenv}{\activeaccount})}{\balance}, \access{\gstate(\access{\exenv}{\activeaccount})}{\nonce}, \gstate)}
~\land~  (\originator, \gasprize, \blockheader) = \transenv \\
& ~\land~ (\parent, \beneficiary, \difficulty, \blocknumber, \gaslimit, \timestamp) = \blockheader \\
& ~\land~ \cfgglobenv = (\parent, \beneficiary, \difficulty, \blocknumber, \gaslimit, \timestamp, \originator, \gasprize, \cfgexternal) \\
& ~\land~ \cfgstate = (\cfgstack, \cfgmemconc, \cfgmemabs, \cfgstorconc, \cfgstorabs, \cfglocenv, \cfgglobenv) 
~\land~ \precontract =   \access{\exenv}{\activecode} \\
&~\land~ \lpc = \access{\mstate}{\pc} 
\end{cases}
\end{align*} }

Note that we will usually write $\cfgstate = \tocfgstate(\transenv, \exstate)$ to implicitely drop the reconstructed contract $\evmcontract$ and program counter $\lpc$.

{\small
\begin{align*}
\toevmstate(\cfgstate, \precontract, \lpc) = 
\begin{cases}
(\transenv, \exstate) & 
 (\cfgstack, \cfgmemconc, \cfgmemabs, \cfgstorconc, \cfgstorabs, \cfglocenv, \cfgglobenv) = \cfgstate \\
&~\land~  (\parent, \beneficiary, \difficulty, \blocknumber, \gaslimit, \timestamp, \originator, \gasprize, \cfgexternal) = \cfgglobenv \\
& ~\land (b, n, \gstate') = \cfgexternal \\
& ~\land~ (\lgas, i, \textit{actor}, \textit{input}, \textit{sender}, \textit{value}) = \cfglocenv \\
& ~\land~ \code = \fun{\lpc}{(\access{\precontract(\lpc)}{\instruction}, \access{\precontract(\lpc)}{\nextpc})}\\
& ~\land~ \mstate = (\lgas, \lpc, \fun{x}{\load~\cfgstate~\memvar{x}{\lpc}}, i, \cfgstack)
~\land~ \exenv = (\activeaccount, \inputdata, \sender, \tvalue, \code) \\
&~\land~ \gstate = \update{\gstate'}{\activeaccount}{(b, n, \code, \fun{x}{\load ~\cfgstate~ \storvar{x}{\lpc}})}  \\
&~\land~ \blockheader =  (\parent, \beneficiary, \difficulty, \blocknumber, \gaslimit, \timestamp)
~\land~ \transenv = (\originator, \gasprize, \blockheader) 
\end{cases}
\end{align*}
}

For defining the EVM CFG semantics, we assume $\cfgstatefull = \cfgstate \uplus \cfgstatecopy$ to represent the heterogeneous mapping that maps two copies of the variables from $\cfgstate$ to their respective values. 
We will denote the copy of variable $x$ from $\cfgstate$ in $\cfgstatecopy$ as $\tempvar{x}$. 
Variables $\tempvar{x}$ in $\cfgstatecopy$ function as temporal variables and are initially set to $\bot$.
We denote with $\cfgstatecopybot$ the partial mapping that maps all variables $\tempvar{x}$ to $\bot$.
Temporal variables are only needed to track individual variable dependencies for opcodes that are initiating internal transactions.
In this case, many variables are updated simultanously and to distinguish the different dependencies in a fine-grained manner, the updates are first written into temporal variables and copied to their corresponding variables in $\cfgstate$ only later.

Note that for state updates only touching variables in $\cfgstate$, by convention we write $\fun{\cfgstate}{\langle \textit{exp} \rangle}$ while we use $\fun{\cfgstatefull}{\langle \textit{exp} \rangle}$ to denote state updates that may also touch temporal variables.

 \subsubsection{CFG semantics}
 \label{appendix:subsec:cfg-semantics}

We closely follow the semantic rules given in~\cite{grishchenko2018semantic} and group the rules whenever possible.
\paragraph{Binary Stack Operations}

We first give the rules  for binary stack operations
We define 
\begin{align*}
\binops \define \{ \ADD,  \SUB, \LT,  \GT,  \EQ, \AND, \OR, \XOR, \SLT,  \SGT, \MUL, \DIV, \SDIV, \\
\MOD, \SMOD, \SIGNEXTEND, \BYTE \} 
\end{align*}
and 
\begin{align*}
\binopcost{\binopv} = 
\begin{cases}
3 & \binopv \in  \{\ADD, \SUB, \LT, \GT, \SLT, \SGT, \EQ, \AND, \OR, \XOR, \BYTE \}  \\
5 &\binopv \in \{\MUL, \DIV, \SDIV, \MOD, \SMOD, \SIGNEXTEND \} \\
\end{cases}
\end{align*}
and 
\begin{align*}
\binopfun{\binopv} = 
\begin{cases}
\lambda (a,b).\, a + b \mod 2^{256} & \binopv = \ADD \\
\lambda (a,b).\,a - b \mod 2^{256} & \binopv = \SUB \\
\lambda (a,b).\, \cond{a < b}{1}{0} & \binopv = \LT \\
\lambda (a,b). \, \cond{a > b}{1}{0} & \binopv = \GT \\
\lambda (a,b). \, \cond{\signed{a} < \signed{b}}{1}{0} & \binopv = \SLT \\
\lambda (a,b). \, \cond{\signed{a} > \signed{b}}{1}{0} & \binopv = \SGT \\
\lambda (a,b). \, \cond{a = b}{1}{0} & \binopv = \EQ \\
\lambda (a,b).\, a \bitand b & \binopv = \AND\\
\lambda (a,b).\, a \bitor b & \binopv = \OR \\
\lambda (a,b).\, a \bitxor b & \binopv = \XOR \\
\lambda (a,b).\, a \cdot b \mod 2^{256} & \binopv = \MUL \\
\lambda (a,b).\, \cond{(b = 0)}{0}{\lfloor a \div b \rfloor}  & \binopv = \DIV\\
\lambda (a,b).\, \cond{(b=0)}{0}{a \mod b}  & \binopv = \MOD\\
\lambda (a,b).\, (b=0)?~0~:~(a= 2^{255}\land \signed{b} = -1)?~2^{256}~: \\
~\textit{let}~x = \signed{a} \div \signed{b} ~\textit{in}~ \unsigned{(\signof{x} \cdot \lfloor | x| \rfloor)}  & \binopv = \SDIV\\
\lambda (a,b).\, \cond{(b=0)}{0}{\unsigned{(\signof{a} \cdot |a| \mod |b|)}} &\binopv = \SMOD\\
\lambda (o, b). \, \cond{(o \geq 32)}{0}{\concat{\extract{b}{8 \cdot o}{8 \cdot o + 7}}{0^{248}}} & \binopv = \BYTE \\ 
\lambda (a, b). \, \textit{let}~x = 256-8(a+1) ~\textit{in} \\
~\textit{let}~ s = \arraypos{b}{x} ~\textit{in}~ \concat{s^x}{\extract{b}{x}{255}} & \binopv = \SIGNEXTEND 
\end{cases}
\end{align*}

where $\signof{\cdot}: \sinteger{x} \to \{-1, 1\}$ is defined as 
\begin{align*}
\signof{x} = 
\begin{cases}
1 & x \geq 0 \\
0 & \text{otherwise}
\end{cases}
\end{align*}
and $\bitand$, $\bitor$ and $\bitxor$ are bitwise and, or and xor, respectively.

Exceptions to the normal binary operations are the exponentiation as this instruction uses non-constant costs and the computation of the Keccack-256 hash. 

The CFG rules for these operations are given as follows: 

\begin{mathpar}
    \infer{
        \precontract(\pc) =  (\instruction(\stackvar{y},\stackvar{x_1}, \stackvar{x_2}), \pc', \pre) \\
        \instruction \in \binops\\
        f = \fun{\cfgstate}{\updatestate{\cfgstate}{\stackvar{y}}{\binopfun{\instruction}(\cfgstate[\stackvar{x_1}], \cfgstate[\stackvar{x_2}])}}
    }
    {\edgestep{\precontract, \cd}{\cfgnode{\pc}{0}}{\cfgnode{\pc}{1}}} \\
    {\Defs{\stackvar{y}} \\
        \Uses{\stackvar{x_1},\stackvar{x_2}}}
\end{mathpar}
\begin{mathpar}
    \infer{
        \precontract(\pc) =  (\instruction(\stackvar{y},\stackvar{x_1}, \stackvar{x_2}), \pc', \pre) \\
        \instruction \in \binops\\
        f = \fun{\cfgstate}{\updatestate{\cfgstate}{\envvar{\gaspc{\pc'}}}{\cfgstate[\envvar{\gaspc{\pc}}] - \binopcost{\instruction}}} 
    }
    {\edgestep{\precontract, \cd}{\cfgnode{\pc}{1}}{\cfgnode{\pc'}{0}}} \\
    {\Defs{\locenvvar{\gaspc{\pc'}}} \\
        \Uses{\locenvvar{\gaspc{\pc}}}}
\end{mathpar}

Exceptions to the normal binary operations are the exponentiation as this instruction uses non-constant costs and the computation of the Keccack-256 hash. 
We give their CFG semantics rules separately: 

\begin{mathpar}
    \infer{
        \precontract(\pc) =  (\EXP(\stackvar{y},\stackvar{x_1}, \stackvar{x_2}), \pc', \pre) \\
        f = \fun{\cfgstate}{\updatestate{\cfgstate}{\stackvar{y}}{\cfgstate[\stackvar{x_1}]^{\cfgstate[\stackvar{x_2}]}  \mod 2^{256}}} \\
    }
    {\edgestep{\precontract, \cd}{\cfgnode{\pc}{0}}{\cfgnode{\pc}{1}}} \\
    {\Defs{\stackvar{y}} \\
        \Use{ \{ \stackvar{x_1},\stackvar{x_2}} \} }
\end{mathpar}
\begin{mathpar}
    \infer{
        \precontract(\pc) =  (\EXP(\stackvar{y},\stackvar{x_1}, \stackvar{x_2}), \pc', \pre) \\
        \costs = \fun{\cfgstate}{\cond{(\cfgstate[\stackvar{x_2}] = 0)}{10}{10 + 10* (1 + \left \lfloor \log_{256}{\cfgstate[\stackvar{x_2}]} \right \rfloor)}} \\
        f = \fun{\cfgstate}{\updatestate{\cfgstate}{\envvar{\gaspc{\pc'}}}{\cfgstate[\envvar{\gaspc{\pc}}] - \costs(\cfgstate)}}  \\    }
    {\edgestep{\precontract, \cd}{\cfgnode{\pc}{1}}{\cfgnode{\pc'}{0}}} \\
    {\Defs{\locenvvar{\gaspc{\pc'}}} \\
        \Uses{\locenvvar{\gaspc{\pc}}, \stackvar{x_2}}}
\end{mathpar}


\begin{mathpar}
    \infer{
        \precontract(\pc) =  (\SHA(\stackvar{y},\stackvar{x_1}, \stackvar{x_2}), \pc', \pre) \\
        v =  \fun{\cfgstate}{\load_{m}~ \accesscfgstate{\cfgstate}{\stackvar{x_1}}~\accesscfgstate{\cfgstate}{\stackvar{x_2}}} \\
        f = \fun{\cfgstate}{\updatestate{\cfgstate}{\stackvar{y}}{\keccak{v(\cfgstate)}}}
    }
    {\edgestep{\precontract, \cd}{\cfgnode{\pc}{0}}{\cfgnode{\pc}{1}}} \\
    {\Defs{\stackvar{y}} \\
        \Use{ \{ \stackvar{x_1},\stackvar{x_2} \} ~\cup~ \memvarabs{X}{} ~\cup~ \memvarconc{X}{} }}
\end{mathpar}
\begin{mathpar}
    \infer{
        \precontract(\pc) =  (\SHA(\stackvar{y},\stackvar{x_1}, \stackvar{x_2}), \pc', \pre) \\
        \pos = \fun{\cfgstate}{\accesscfgstate{\cfgstate}{\stackvar{x_1}}}\\
        \siz = \fun{\cfgstate}{\accesscfgstate{\cfgstate}{\stackvar{x_2}}}\\
        \aw = \fun{\cfgstate}{\memext{\accesscfgstate{\cfgstate}{\locenvvar{\msizepc{\pc}}}}{\pos(\cfgstate)}{\siz(\cfgstate)}}\\ 
        \costs = \fun{\cfgstate}{\costmem{\accesscfgstate{\cfgstate}{\locenvvar{\msizepc{\pc}}}}{\aw(\cfgstate)} + 30 + 6 \cdot \left \lceil \frac{\siz(\cfgstate)}{32} \right \rceil} \\
        f = \fun{\cfgstate}{\updatestate{\cfgstate}{\envvar{\gaspc{\pc'}}}{\cfgstate[\envvar{\gaspc{\pc}}] - \costs(\cfgstate)}}  \\    }
    {\edgestep{\precontract, \cd}{\cfgnode{\pc}{1}}{\cfgnode{\pc}{2}}} \\
    {\Defs{\locenvvar{\gaspc{\pc}}} \\
        \Uses{\locenvvar{\gaspc{\pc'}}, \locenvvar{\msizepc{\pc}}, \stackvar{x_1}, \stackvar{x_2}}}
\end{mathpar}
\begin{mathpar}
    \infer{
        \precontract(\pc) =  (\SHA(\stackvar{y},\stackvar{x_1}, \stackvar{x_2}), \pc', \pre) \\
        \pos = \fun{\cfgstate}{\accesscfgstate{\cfgstate}{\stackvar{x_1}}}\\
        \siz = \fun{\cfgstate}{\accesscfgstate{\cfgstate}{\stackvar{x_2}}}\\
        \aw = \fun{\cfgstate}{\memext{\accesscfgstate{\cfgstate}{\locenvvar{\msizepc{\pc}}}}{\pos(\cfgstate)}{\siz(\cfgstate)}}\\ 
        f = \fun{\cfgstate}{\updatestate{\cfgstate}{\envvar{\msizepc{\pc'}}}{\aw(\cfgstate)}}  \\    }
    {\edgestep{\precontract, \cd}{\cfgnode{\pc}{2}}{\cfgnode{\pc'}{0}}} \\
    {\Defs{\locenvvar{\msizepc{\pc'}}} \\
        \Uses{\locenvvar{\msizepc{\pc}},  \stackvar{x_1}, \stackvar{x_2}}}
\end{mathpar}

Where $\load_{m}~\cfgstate~o~s$ is defined as 
\begin{align*}
    \load_{m}~\cfgstate~o~s \define 
    \begin{cases}
        0 & s = 0\\
        (\load~\cfgstate~\memvar{o}{\lpc}) * 256^{(s-1)} + \load~\cfgstate~(o+1)~(s-1)  & s > 0
    \end{cases}
\end{align*}

\paragraph{Unary Stack Operations}
\begin{mathpar}
    \infer{
        \precontract(\pc) =  (\ISZERO(\stackvar{y},\stackvar{x}), \pc', \pre) \\
        r = \fun{\cfgstate}{\cond{(\accesscfgstate{\cfgstate}{\stackvar{x}}=0)}{1}{0} } \\
        f = \fun{\cfgstate}{\updatestate{\cfgstate}{\stackvar{y}}{r(\cfgstate)}}
    }
    {\edgestep{\precontract, \cd}{\cfgnode{\pc}{0}}{\cfgnode{\pc}{1}}} \\
    {\Defs{\stackvar{y}} \\
        \Uses{\stackvar{x}}}
\end{mathpar}
\begin{mathpar}
    \infer{
        \precontract(\pc) =  (\ISZERO(\stackvar{y},\stackvar{x}), \pc', \pre) \\        
        f = \fun{\cfgstate}{\updatestate{\cfgstate}{\envvar{\gaspc{\pc'}}}{\cfgstate[\envvar{\gaspc{\pc}}] - 3}} 
    }
    {\edgestep{\precontract, \cd}{\cfgnode{\pc}{1}}{\cfgnode{\pc'}{0}}} \\
    {\Defs{\locenvvar{\gaspc{\pc'}}} \\
        \Uses{\locenvvar{\gaspc{\pc}}}}
\end{mathpar}
\begin{mathpar}
    \infer{
        \precontract(\pc) =  (\NOT(\stackvar{y},\stackvar{\stackvar{x}}), \pc', \pre) \\
        r = \fun{\cfgstate}{\bitneg(\accesscfgstate{\cfgstate}{x})} \\
        f = \fun{\cfgstate}{\updatestate{\cfgstate}{\stackvar{y}}{r(\cfgstate)}}
    }
    {\edgestep{\precontract, \cd}{\cfgnode{\pc}{0}}{\cfgnode{\pc}{1}}} \\
    {\Defs{\stackvar{y}} \\
        \Uses{\stackvar{x}}}
\end{mathpar}
\begin{mathpar}
    \infer{
        \precontract(\pc) =  (\NOT(\stackvar{y},\stackvar{x}), \pc', \pre) \\        
        f = \fun{\cfgstate}{\updatestate{\cfgstate}{\envvar{\gaspc{\pc'}}}{\cfgstate[\envvar{\gaspc{\pc}}] - 3}} 
    }
    {\edgestep{\precontract, \cd}{\cfgnode{\pc}{1}}{\cfgnode{\pc'}{0}}} \\
    {\Defs{\locenvvar{\gaspc{\pc'}}} \\
        \Uses{\locenvvar{\gaspc{\pc}}}}
\end{mathpar}
where $\bitneg$ is bitwise negation. 

\paragraph{Ternary Stack Operations}

\begin{mathpar}
    \infer{
        \precontract(\pc) =  (\ADDMOD(\stackvar{y},\stackvar{x_1}, \stackvar{x_2}, \stackvar{x_3}), \pc', \pre) \\
        a = \fun{\cfgstate}{\accesscfgstate{\cfgstate}{\stackvar{x_1}}} \\
        b = \fun{\cfgstate}{\accesscfgstate{\cfgstate}{\stackvar{x_2}}} \\
        c = \fun{\cfgstate}{\accesscfgstate{\cfgstate}{\stackvar{x_3}}} \\
        r = \fun{\cfgstate}{\cond{(c(\cfgstate)=0)}{0}{(a(\cfgstate) + b(\cfgstate)) \mod c(\cfgstate)}} \\
        f = \fun{\cfgstate}{\updatestate{\cfgstate}{\stackvar{y}}{r(\cfgstate)}}
    }
    {\edgestep{\precontract, \cd}{\cfgnode{\pc}{0}}{\cfgnode{\pc}{1}}} \\
    {\Defs{\stackvar{y}} \\
        \Uses{\stackvar{x_1},\stackvar{x_2}, \stackvar{x_3}}}
\end{mathpar}
\begin{mathpar}
    \infer{
        \precontract(\pc) =  (\ADDMOD(\stackvar{y},\stackvar{x_1}, \stackvar{x_2}, \stackvar{x_3}), \pc', \pre) \\
        f = \fun{\cfgstate}{\updatestate{\cfgstate}{\envvar{\gaspc{\pc'}}}{\cfgstate[\envvar{\gaspc{\pc}}] - 8}} 
    }
    {\edgestep{\precontract, \cd}{\cfgnode{\pc}{1}}{\cfgnode{\pc'}{0}}} \\
    {\Defs{\locenvvar{\gaspc{\pc'}}} \\
        \Uses{\locenvvar{\gaspc{\pc}}}}
\end{mathpar}

\begin{mathpar}
    \infer{
        \precontract(\pc) =  (\MULMOD(\stackvar{y},\stackvar{x_1}, \stackvar{x_2}, \stackvar{x_3}), \pc', \pre) \\
        a = \fun{\cfgstate}{\accesscfgstate{\cfgstate}{\stackvar{x_1}}} \\
        b = \fun{\cfgstate}{\accesscfgstate{\cfgstate}{\stackvar{x_2}}} \\
        c = \fun{\cfgstate}{\accesscfgstate{\cfgstate}{\stackvar{x_3}}} \\
        r = \fun{\cfgstate}{\cond{(c(\cfgstate)=0)}{0}{(a(\cfgstate) \cdot b(\cfgstate)) \mod c(\cfgstate)}} \\
        f = \fun{\cfgstate}{\updatestate{\cfgstate}{\stackvar{y}}{r(\cfgstate)}}
    }
    {\edgestep{\precontract, \cd}{\cfgnode{\pc}{0}}{\cfgnode{\pc}{1}}} \\
    {\Defs{\stackvar{y}} \\
        \Uses{\stackvar{x_1},\stackvar{x_2}, \stackvar{x_3}}}
\end{mathpar}
\begin{mathpar}
    \infer{
        \precontract(\pc) =  (\MULMOD(\stackvar{y},\stackvar{x_1}, \stackvar{x_2}, \stackvar{x_3}), \pc', \pre) \\
        f = \fun{\cfgstate}{\updatestate{\cfgstate}{\envvar{\gaspc{\pc'}}}{\cfgstate[\envvar{\gaspc{\pc}}] - 8}} 
    }
    {\edgestep{\precontract, \cd}{\cfgnode{\pc}{1}}{\cfgnode{\pc'}{0}}} \\
    {\Defs{\locenvvar{\gaspc{\pc}}} \\
        \Uses{\locenvvar{\gaspc{\pc}}}}
\end{mathpar}

\paragraph{Accessing the execution environment}

\begin{mathpar}
    \infer{
        \precontract(\pc) =  (\ADDRESS(\stackvar{y}), \pc', \pre) \\
        r = \fun{\cfgstate}{\accesscfgstate{\cfgstate}{\locenvvar{\activeaccount}}} \\
        f = \fun{\cfgstate}{\updatestate{\cfgstate}{\stackvar{y}}{r(\cfgstate)}}
    }
    {\edgestep{\precontract, \cd}{\cfgnode{\pc}{0}}{\cfgnode{\pc}{1}}} \\
    {\Defs{\stackvar{y}} \\
        \Uses{\locenvvar{\activeaccount}}}
\end{mathpar}

Most instructions for accessing the execution environment have the same gas cost. For this reason we summarize the rule for gas substraction for them.

\begin{mathpar}
    \infer{
        \precontract(\pc) =  (\instruction(\stackvar{y}), \pc', \pre) \\    
        \instruction \in \{ \ADDRESS, \CALLER, \CALLVALUE, \CODESIZE, \CALLDATASIZE \} \\    
        f = \fun{\cfgstate}{\updatestate{\cfgstate}{\envvar{\gaspc{\pc'}}}{\cfgstate[\envvar{\gaspc{\pc}}] - 2}} 
    }
    {\edgestep{\precontract, \cd}{\cfgnode{\pc}{1}}{\cfgnode{\pc'}{0}}} \\
    {\Defs{\locenvvar{\gaspc{\pc'}}} \\
        \Uses{\locenvvar{\gaspc{\pc}}}}
\end{mathpar}

\begin{mathpar}
    \infer{
        \precontract(\pc) =  (\CALLER(\stackvar{y}), \pc', \pre) \\
        r = \fun{\cfgstate}{\accesscfgstate{\cfgstate}{\locenvvar{\sender}}} \\
        f = \fun{\cfgstate}{\updatestate{\cfgstate}{\stackvar{y}}{r(\cfgstate)}}
    }
    {\edgestep{\precontract, \cd}{\cfgnode{\pc}{0}}{\cfgnode{\pc}{1}}} \\
    {\Defs{\stackvar{y}} \\
        \Uses{\locenvvar{\sender}}}
\end{mathpar}

\begin{mathpar}
    \infer{
        \precontract(\pc) =  (\CALLVALUE(\stackvar{y}), \pc', \pre) \\
        r = \fun{\cfgstate}{\accesscfgstate{\cfgstate}{\locenvvar{\tvalue}}} \\        
        f = \fun{\cfgstate}{\updatestate{\cfgstate}{\stackvar{y}}{r(\cfgstate)}}
    }
    {\edgestep{\precontract, \cd}{\cfgnode{\pc}{0}}{\cfgnode{\pc}{1}}} \\
    {\Defs{\stackvar{y}} \\
        \Uses{\locenvvar{\tvalue}}}
\end{mathpar}

\begin{mathpar}
    \infer{
        \precontract(\pc) =  (\CODESIZE(\stackvar{y}), \pc', \pre) \\
        r = \fun{\cfgstate}{\size{\accesscfgstate{\cfgstate}{\locenvvar{\code}}}} \\
        f = \fun{\cfgstate}{\updatestate{\cfgstate}{\stackvar{y}}{r(\cfgstate)}}
    }
    {\edgestep{\precontract, \cd}{\cfgnode{\pc}{0}}{\cfgnode{\pc}{1}}} \\
    {\Defs{\stackvar{y}} \\
        \Uses{\locenvvar{\code}}}
\end{mathpar}

\begin{mathpar}
    \infer{
        \precontract(\pc) =  (\CALLDATASIZE(\stackvar{y}), \pc', \pre) \\
        r = \fun{\cfgstate}{\size{\accesscfgstate{\cfgstate}{\locenvvar{\inputdata}}}} \\
        f = \fun{\cfgstate}{\updatestate{\cfgstate}{\stackvar{y}}{r(\cfgstate)}}
    }
    {\edgestep{\precontract, \cd}{\cfgnode{\pc}{0}}{\cfgnode{\pc}{1}}} \\
    {\Defs{\stackvar{y}} \\
        \Uses{\locenvvar{\inputdata}}}
\end{mathpar}

We give individual rules for accessing the code and input data. 

The $\CALLDATALOAD$ instruction accesses a word of the call data at a specified position

\begin{mathpar}
    \infer{
        \precontract(\pc) =  (\CALLDATALOAD(\stackvar{y},\stackvar{x}), \pc', \pre) \\
        a = \fun{\cfgstate} = \accesscfgstate{\cfgstate}{\stackvar{x}}\\
        d = \fun{\cfgstate}{\accesscfgstate{\cfgstate}{\locenvvar{\inputdata}}} \\ 
        \siz = \fun{\cfgstate}{\size{d(\cfgstate)}}\\ 
        k = \fun{\cfgstate}{\cond{(\siz(\cfgstate) - a(\cfgstate) < 0)}{0} {\mini{\siz(\cfgstate) - a(\cfgstate)}{32}}}\\
        v = \fun{\cfgstate}{\arraypos{d(\cfgstate)}{a(\cfgstate), a(\cfgstate) + k(\cfgstate) - 1}} \\
        r = \fun{\cfgstate}{\concat{v(\cfgstate)}{0^{256 - k(\cfgstate)\cdot 8}}} \\
        f = \fun{\cfgstate}{\updatestate{\cfgstate}{\stackvar{y}}{r(\cfgstate)}}
    }
    {\edgestep{\precontract, \cd}{\cfgnode{\pc}{0}}{\cfgnode{\pc}{1}}} \\
    {\Defs{\stackvar{y}} \\
    \Uses{\locenvvar{\code}}}
\end{mathpar}
\begin{mathpar}
    \infer{
        \precontract(\pc) =  (\CALLDATALOAD(\stackvar{y},\stackvar{x}), \pc', \pre) \\        
        f = \fun{\cfgstate}{\updatestate{\cfgstate}{\envvar{\gaspc{\pc'}}}{\cfgstate[\envvar{\gaspc{\pc}}] - 3}} 
    }
    {\edgestep{\precontract, \cd}{\cfgnode{\pc}{1}}{\cfgnode{\pc'}{0}}} \\
    {\Defs{\locenvvar{\gaspc{\pc'}}} \\
        \Uses{\locenvvar{\gaspc{\pc}}}}
\end{mathpar}

\begin{mathpar}
    \infer{
        \precontract(\pc) =  (\CALLDATACOPY(\stackvar{y},\stackvar{x_1}, \stackvar{x_2}, \stackvar{x_3}), \pc', \pre) \\
        \pos_\memo = \fun{\cfgstate}{\accesscfgstate{\cfgstate}{\stackvar{x_1}}}\\
        \pos_\data = \fun{\cfgstate}{\accesscfgstate{\cfgstate}{\stackvar{x_2}}}\\
        \siz = \fun{\cfgstate}{\accesscfgstate{\cfgstate}{\stackvar{x_3}}}\\
        \datav = \fun{\cfgstate}{\accesscfgstate{\cfgstate}{\locenvvar{\inputdata}}} \\ 
        k = \fun{\cfgstate}{\cond{(\size{\datav(\cfgstate)} - \pos_\data(\cfgstate) < 0}{0}{\mini{\size{\datav(\cfgstate)}- \pos_\data(\cfgstate)}{\siz(\cfgstate)}}}\\
        d' = \fun{\cfgstate}{\arraypos{\datav(\cfgstate)}{\pos_\data(\cfgstate), \pos_\data(\cfgstate) + k(\cfgstate) - 1}}\\
        d = \fun{\cfgstate}{\concat{d'(\cfgstate)}{0^{8 \cdot(\siz(\cfgstate) - k(\cfgstate))}}}\\
        f = \fun{\cfgstate}{\updatevarset{\cfgstate}{\memvarabs{i}{\pc'}}{d(\cfgstate)[i]}{i}{[\pos_\memo(\cfgstate), \pos_\memo(\cfgstate) + \siz(\cfgstate) -1]}}
    }
    {\edgestep{\precontract, \cd}{\cfgnode{\pc}{0}}{\cfgnode{\pc}{1}}} \\
    {\Def{\memvarabs{X}{\pc'}} \\
    \Uses{\locenvvar{\inputdata}, \stackvar{x_1}, \stackvar{x_2}, \stackvar{x_3}}}
\end{mathpar}
\begin{mathpar}
    \infer{
        \precontract(\pc) =  (\CALLDATACOPY(\stackvar{y},\stackvar{x_1}, \stackvar{x_2}, \stackvar{x_3}), \pc', \pre) \\
        \pos_\memo = \fun{\cfgstate}{\accesscfgstate{\cfgstate}{\stackvar{x_1}}}\\
        \pos_\data = \fun{\cfgstate}{\accesscfgstate{\cfgstate}{\stackvar{x_2}}}\\
        \siz = \fun{\cfgstate}{\accesscfgstate{\cfgstate}{\stackvar{x_3}}}\\
        \aw = \fun{\cfgstate}{\memext{\accesscfgstate{\cfgstate}{\locenvvar{\msizepc{\pc}}}}{\pos_\memo(\cfgstate)}{\siz(\cfgstate)}} \\
        \costs = \fun{\cfgstate}{\costmem{\accesscfgstate{\cfgstate}{\locenvvar{\msizepc{\pc}}}}{\aw(\cfgstate)}+3 + 3 \cdot \left \lceil \frac{\siz(\cfgstate)}{32} \right \rceil}\\     
        f = \fun{\cfgstate}{\updatestate{\cfgstate}{\envvar{\gaspc{\pc'}}}{\cfgstate[\envvar{\gaspc{\pc}}] - \costs(\cfgstate)}} 
    }
    {\edgestep{\precontract, \cd}{\cfgnode{\pc}{1}}{\cfgnode{\pc'}{2}}} \\
    {\Defs{\locenvvar{\gaspc{\pc'}}} \\
        \Uses{\locenvvar{\gaspc{\pc}}, \locenvvar{\msizepc{\pc}}, \stackvar{x_1}, \stackvar{x_2}, \stackvar{x_3}}}
\end{mathpar}
\begin{mathpar}
    \infer{
        \precontract(\pc) =  (\CALLDATACOPY(\stackvar{y},\stackvar{x_1}, \stackvar{x_2}, \stackvar{x_3}), \pc', \pre) \\
        \pos_\memo = \fun{\cfgstate}{\accesscfgstate{\cfgstate}{\stackvar{x_1}}}\\
        \pos_\data = \fun{\cfgstate}{\accesscfgstate{\cfgstate}{\stackvar{x_2}}}\\
        \siz = \fun{\cfgstate}{\accesscfgstate{\cfgstate}{\stackvar{x_3}}}\\
        \aw = \fun{\cfgstate}{\memext{\accesscfgstate{\cfgstate}{\locenvvar{\msizepc{\pc}}}}{\pos_\memo(\cfgstate)}{\siz(\cfgstate)}} \\
        f = \fun{\cfgstate}{\updatestate{\cfgstate}{\envvar{\msizepc{\pc'}}}{\aw(\cfgstate)}} 
    }
    {\edgestep{\precontract, \cd}{\cfgnode{\pc}{2}}{\cfgnode{\pc'}{0}}} \\
    {\Defs{\locenvvar{\msizepc{\pc'}}} \\
        \Uses{\locenvvar{\msizepc{\pc}}, \stackvar{x_1}, \stackvar{x_3}}}
\end{mathpar}

The rules for copying a fraction of the code to memory ($\CODECOPY$) are similar: 

\begin{mathpar}
    \infer{
        \precontract(\pc) =  (\CODECOPY(\stackvar{y},\stackvar{x_1}, \stackvar{x_2}, \stackvar{x_3}), \pc', \pre) \\
        \pos_\memo = \fun{\cfgstate}{\accesscfgstate{\cfgstate}{\stackvar{x_1}}}\\
        \pos_\code= \fun{\cfgstate}{\accesscfgstate{\cfgstate}{\stackvar{x_2}}}\\
        \siz = \fun{\cfgstate}{\accesscfgstate{\cfgstate}{\stackvar{x_3}}}\\
        \datav = \fun{\cfgstate}{\accesscfgstate{\cfgstate}{\locenvvar{\activecode}}} \\ 
        k = \fun{\cfgstate}{\cond{(\size{\datav(\cfgstate)} - \pos_\code(\cfgstate) < 0}{0}{\mini{\size{\datav(\cfgstate)}- \pos_\data(\cfgstate)}{\siz(\cfgstate)}}}\\
        d' = \fun{\cfgstate}{\arraypos{\datav(\cfgstate)}{\pos_\code(\cfgstate), \pos_\code(\cfgstate) + k(\cfgstate) - 1}}\\
        d = \fun{\cfgstate}{\concat{d'(\cfgstate)}{\STOP^{\siz(\cfgstate) - k(\cfgstate)}}}\\
        f = \fun{\cfgstate}{\updatevarset{\cfgstate}{\memvarabs{i}{\pc'}}{d(\cfgstate)[i]}{i}{[\pos_\memo(\cfgstate), \pos_\memo(\cfgstate) + \siz(\cfgstate) -1]}}
    }
    {\edgestep{\precontract, \cd}{\cfgnode{\pc}{0}}{\cfgnode{\pc}{1}}} \\
    {\Def{\memvarabs{X}{\pc'}} \\
    \Uses{\locenvvar{\activecode}, \stackvar{x_1}, \stackvar{x_2}, \stackvar{x_3}}}
\end{mathpar}
\begin{mathpar}
    \infer{
        \precontract(\pc) =  (\CODECOPY(\stackvar{y},\stackvar{x_1}, \stackvar{x_2}, \stackvar{x_3}), \pc', \pre) \\
        \pos_\memo = \fun{\cfgstate}{\accesscfgstate{\cfgstate}{\stackvar{x_1}}}\\
        \pos_\code = \fun{\cfgstate}{\accesscfgstate{\cfgstate}{\stackvar{x_2}}}\\
        \siz = \fun{\cfgstate}{\accesscfgstate{\cfgstate}{\stackvar{x_3}}}\\
        \aw = \fun{\cfgstate}{\memext{\accesscfgstate{\cfgstate}{\locenvvar{\msizepc{\pc}}}}{\pos_\memo(\cfgstate)}{\siz(\cfgstate)}} \\
        \costs = \fun{\cfgstate}{\costmem{\accesscfgstate{\cfgstate}{\locenvvar{\msizepc{\pc}}}}{\aw(\cfgstate)}+3 + 3 \cdot \left \lceil \frac{\siz(\cfgstate)}{32} \right \rceil}\\     
        f = \fun{\cfgstate}{\updatestate{\cfgstate}{\envvar{\gaspc{\pc'}}}{\cfgstate[\envvar{\gaspc{\pc}}] - \costs(\cfgstate)}} 
    }
    {\edgestep{\precontract, \cd}{\cfgnode{\pc}{1}}{\cfgnode{\pc'}{2}}} \\
    {\Defs{\locenvvar{\gaspc{\pc'}}} \\
        \Uses{\locenvvar{\gaspc{\pc}}, \locenvvar{\msizepc{\pc}}, \stackvar{x_1}, \stackvar{x_2}, \stackvar{x_3}}}
\end{mathpar}
\begin{mathpar}
    \infer{
        \precontract(\pc) =  (\CODECOPY(\stackvar{y},\stackvar{x_1}, \stackvar{x_2}, \stackvar{x_3}), \pc', \pre) \\
        \pos_\memo = \fun{\cfgstate}{\accesscfgstate{\cfgstate}{\stackvar{x_1}}}\\
        \pos_\code = \fun{\cfgstate}{\accesscfgstate{\cfgstate}{\stackvar{x_2}}}\\
        \siz = \fun{\cfgstate}{\accesscfgstate{\cfgstate}{\stackvar{x_3}}}\\
        \aw = \fun{\cfgstate}{\memext{\accesscfgstate{\cfgstate}{\locenvvar{\msizepc{\pc}}}}{\pos_\memo(\cfgstate)}{\siz(\cfgstate)}} \\
        f = \fun{\cfgstate}{\updatestate{\cfgstate}{\envvar{\msizepc{\pc'}}}{\aw(\cfgstate)}} 
    }
    {\edgestep{\precontract, \cd}{\cfgnode{\pc}{2}}{\cfgnode{\pc'}{0}}} \\
    {\Defs{\locenvvar{\msizepc{\pc'}}} \\
        \Uses{\locenvvar{\msizepc{\pc}}, \stackvar{x_1}, \stackvar{x_3}}}
\end{mathpar}

Note that the rules for $\CALLDATACOPY$ and $\CODECOPY$ could easily be refined to account for preprocessing information.

\paragraph{Accessing the transaction environment}

\begin{mathpar}
    \infer{
        \precontract(\pc) =  (\ORIGIN(\stackvar{y}), \pc', \pre) \\
        r = \fun{\cfgstate}{\accesscfgstate{\cfgstate}{\globenvvar{\origin}}} \\
        f = \fun{\cfgstate}{\updatestate{\cfgstate}{\stackvar{y}}{r(\cfgstate)}}
    }
    {\edgestep{\precontract, \cd}{\cfgnode{\pc}{0}}{\cfgnode{\pc}{1}}} \\
    {\Defs{\stackvar{y}} \\
        \Uses{\globenvvar{\origin}}}
\end{mathpar}

Most instructions for accessing the transaction environment have the same gas cost. For this reason we summarize the rule for gas substraction for them.

\begin{mathpar}
    \infer{
        \precontract(\pc) =  (\instruction(\stackvar{y}), \pc', \pre) \\   
        \instruction \in \{ \ORIGIN, \GASPRICE, \COINBASE, \TIMESTAMP, \NUMBER, \GASLIMIT, \DIFFICULTY \} \\ 
        f = \fun{\cfgstate}{\updatestate{\cfgstate}{\envvar{\gaspc{\pc'}}}{\cfgstate[\envvar{\gaspc{\pc}}] - 2}} 
    }
    {\edgestep{\precontract, \cd}{\cfgnode{\pc}{1}}{\cfgnode{\pc'}{0}}} \\
    {\Defs{\locenvvar{\gaspc{\pc'}}} \\
        \Uses{\locenvvar{\gaspc{\pc}}}}
\end{mathpar}

\begin{mathpar}
    \infer{
        \precontract(\pc) =  (\GASPRICE(\stackvar{y}), \pc', \pre) \\
        r = \fun{\cfgstate}{\accesscfgstate{\cfgstate}{\globenvvar{\gasprize}}} \\
        f = \fun{\cfgstate}{\updatestate{\cfgstate}{\stackvar{y}}{r(\cfgstate)}}
    }
    {\edgestep{\precontract, \cd}{\cfgnode{\pc}{0}}{\cfgnode{\pc}{1}}} \\
    {\Defs{\stackvar{y}} \\
        \Uses{\globenvvar{\gasprize}}}
\end{mathpar}

\begin{mathpar}
    \infer{
        \precontract(\pc) =  (\COINBASE(\stackvar{y}), \pc', \pre) \\
        r = \fun{\cfgstate}{\accesscfgstate{\cfgstate}{\globenvvar{\beneficiary}}} \\
        f = \fun{\cfgstate}{\updatestate{\cfgstate}{\stackvar{y}}{r(\cfgstate)}}
    }
    {\edgestep{\precontract, \cd}{\cfgnode{\pc}{0}}{\cfgnode{\pc}{1}}} \\
    {\Defs{\stackvar{y}} \\
        \Uses{\globenvvar{\beneficiary}}}
\end{mathpar}

\begin{mathpar}
    \infer{
        \precontract(\pc) =  (\TIMESTAMP(\stackvar{y}), \pc', \pre) \\
        r = \fun{\cfgstate}{\accesscfgstate{\cfgstate}{\globenvvar{\timestamp}}} \\
        f = \fun{\cfgstate}{\updatestate{\cfgstate}{\stackvar{y}}{r(\cfgstate)}}
    }
    {\edgestep{\precontract, \cd}{\cfgnode{\pc}{0}}{\cfgnode{\pc}{1}}} \\
    {\Defs{\stackvar{y}} \\
        \Uses{\globenvvar{\timestamp}}}
\end{mathpar}

\begin{mathpar}
    \infer{
        \precontract(\pc) =  (\NUMBER(\stackvar{y}), \pc', \pre) \\
        r = \fun{\cfgstate}{\accesscfgstate{\cfgstate}{\globenvvar{\blocknumber}}} \\
        f = \fun{\cfgstate}{\updatestate{\cfgstate}{\stackvar{y}}{r(\cfgstate)}}
    }
    {\edgestep{\precontract, \cd}{\cfgnode{\pc}{0}}{\cfgnode{\pc}{1}}} \\
    {\Defs{\stackvar{y}} \\
        \Uses{\globenvvar{\blocknumber}}}
\end{mathpar}

\begin{mathpar}
    \infer{
        \precontract(\pc) =  (\GASLIMIT(\stackvar{y}), \pc', \pre) \\
        r = \fun{\cfgstate}{\accesscfgstate{\cfgstate}{\globenvvar{\gaslimit}}} \\
        f = \fun{\cfgstate}{\updatestate{\cfgstate}{\stackvar{y}}{r(\cfgstate)}}
    }
    {\edgestep{\precontract, \cd}{\cfgnode{\pc}{0}}{\cfgnode{\pc}{1}}} \\
    {\Defs{\stackvar{y}} \\
        \Uses{\globenvvar{\gaslimit}}}
\end{mathpar}

\begin{mathpar}
    \infer{
        \precontract(\pc) =  (\DIFFICULTY(\stackvar{y}), \pc', \pre) \\
        r = \fun{\cfgstate}{\accesscfgstate{\cfgstate}{\globenvvar{\difficulty}}} \\
        f = \fun{\cfgstate}{\updatestate{\cfgstate}{\stackvar{y}}{r(\cfgstate)}}
    }
    {\edgestep{\precontract, \cd}{\cfgnode{\pc}{0}}{\cfgnode{\pc}{1}}} \\
    {\Defs{\stackvar{y}} \\
        \Uses{\globenvvar{\difficulty}}}
\end{mathpar}

\begin{mathpar}
    \infer{
        \precontract(\pc) =  (\BLOCKHASH(\stackvar{y}, \stackvar{x}), \pc', \pre) \\
        n = \fun{\cfgstate}{\accesscfgstate{\cfgstate}{\stackvar{x}}} \\
        r = \fun{\cfgstate}{\funP{\accesscfgstate{\cfgstate}{\globenvvar{\parent}}}{n(\cfgstate)}{0}} \\
        f = \fun{\cfgstate}{\updatestate{\cfgstate}{\stackvar{y}}{r(\cfgstate)}}
    }
    {\edgestep{\precontract, \cd}{\cfgnode{\pc}{0}}{\cfgnode{\pc}{1}}} \\
    {\Defs{\stackvar{y}} \\
        \Uses{\globenvvar{\parent}}}
\end{mathpar}

where the function $\funP{h}{n}{a}$ tries to access the block with number $n$ by traversing the block chain starting from $h$ until the counter $a$ reaches the limit of $256$ or the genesis block is reached. 
\begin{align*}
\funP{h}{n}{a} \define 
\begin{cases}
0 & n > \access{h}{\blocknumberc} \lor a = 256 \lor h = 0 \\
h & n = \access{h}{\blocknumberc} \\
\funP{\access{h}{\parentc}}{n}{a + 1} & \text{otherwise}
\end{cases}
\end{align*}

\begin{mathpar}
    \infer{
        \precontract(\pc) =  (\BLOCKHASH(\stackvar{y}), \pc', \pre) \\   
        f = \fun{\cfgstate}{\updatestate{\cfgstate}{\envvar{\gaspc{\pc'}}}{\cfgstate[\envvar{\gaspc{\pc}}] - 20}} 
    }
    {\edgestep{\precontract, \cd}{\cfgnode{\pc}{1}}{\cfgnode{\pc'}{0}}} \\
    {\Defs{\locenvvar{\gaspc{\pc'}}} \\
        \Uses{\locenvvar{\gaspc{\pc}}}}
\end{mathpar}

\paragraph{Accessing the global state}

\begin{mathpar}
    \infer{
        \precontract(\pc) =  (\BALANCE(\stackvar{y}, \stackvar{x}), \pc', \pre) \\
        a = \fun{\cfgstate}{\accesscfgstate{\cfgstate}{\stackvar{x}}} \\
        r = \fun{\cfgstate}{
            \cond{
                (\access{\accesscfgstate{\cfgstate}{\globenvvar{\cfgexternalpc{\pc}}}}{\gstate}(a(\cfgstate) \mod{} 2^{160}) = \accountstate{\accountnoncev}{\accountbalancev}{\accountstorv}{\accountcodev})}{\accountbalancev}{0}}\\
        f = \fun{\cfgstate}{\updatestate{\cfgstate}{\stackvar{y}}{r(\cfgstate)}}
    }
    {\edgestep{\precontract, \cd}{\cfgnode{\pc}{0}}{\cfgnode{\pc}{1}}} \\
    {\Defs{\stackvar{y}} \\
        \Uses{\globenvvar{\cfgexternalpc{\pc}}}}
\end{mathpar}

\begin{mathpar}
    \infer{
        \precontract(\pc) =  (\BALANCE(\stackvar{y}, \stackvar{x}), \pc', \pre) \\
        f = \fun{\cfgstate}{\updatestate{\cfgstate}{\envvar{\gaspc{\pc'}}}{\cfgstate[\envvar{\gaspc{\pc}}] - 400}} 
    }
    {\edgestep{\precontract, \cd}{\cfgnode{\pc}{1}}{\cfgnode{\pc'}{0}}} \\
    {\Defs{\locenvvar{\gaspc{\pc'}}} \\
        \Uses{\locenvvar{\gaspc{\pc}}}}
\end{mathpar}

\begin{mathpar}
    \infer{
        \precontract(\pc) =  (\EXTCODESIZE(\stackvar{y}, \stackvar{x}), \pc', \pre) \\
        a = \fun{\cfgstate}{\accesscfgstate{\cfgstate}{\stackvar{x}}} \\
        r = \fun{\cfgstate}{\size{
                (\access{\left(\access{\accesscfgstate{\cfgstate}{\globenvvar{\cfgexternalpc{\pc}}}}{\gstate}(a(\cfgstate) \mod{} 2^{160})\right)}{\accountcode}}}\\
        f = \fun{\cfgstate}{\updatestate{\cfgstate}{\stackvar{y}}{r(\cfgstate)}}
    }
    {\edgestep{\precontract, \cd}{\cfgnode{\pc}{0}}{\cfgnode{\pc}{1}}} \\
    {\Defs{\stackvar{y}} \\
        \Uses{\globenvvar{\cfgexternalpc{\pc}}}}
\end{mathpar}

\begin{mathpar}
    \infer{
        \precontract(\pc) =  (\EXTCODESIZE(\stackvar{y}, \stackvar{x}), \pc', \pre) \\        
        f = \fun{\cfgstate}{\updatestate{\cfgstate}{\envvar{\gaspc{\pc'}}}{\cfgstate[\envvar{\gaspc{\pc}}] - 700}} 
    }
    {\edgestep{\precontract, \cd}{\cfgnode{\pc}{1}}{\cfgnode{\pc'}{0}}} \\
    {\Defs{\locenvvar{\gaspc{\pc'}}} \\
        \Uses{\locenvvar{\gaspc{\pc}}}}
\end{mathpar}

\begin{mathpar}
    \infer{
        \precontract(\pc) =  (\EXTCODECOPY(\stackvar{y},\stackvar{x_1}, \stackvar{x_2}, \stackvar{x_3}, \stackvar{x_4}), \pc', \pre) \\
        a = \fun{\cfgstate}{\accesscfgstate{\cfgstate}{\stackvar{x_1}}} \\
        \pos_\memo = \fun{\cfgstate}{\accesscfgstate{\cfgstate}{\stackvar{x_2}}}\\
        \pos_\code= \fun{\cfgstate}{\accesscfgstate{\cfgstate}{\stackvar{x_3}}}\\
        \siz = \fun{\cfgstate}{\accesscfgstate{\cfgstate}{\stackvar{x_4}}}\\
        \datav = \fun{\cfgstate}{\access{\left(\access{\accesscfgstate{\cfgstate}{\globenvvar{\cfgexternalpc{\pc}}}}{\gstate}(a \mod{} 2^{160}) \right)}{\accountcode}} \\ 
        k = \fun{\cfgstate}{\cond{(\size{\datav(\cfgstate)} - \pos_\code(\cfgstate) < 0}{0}{\mini{\size{\datav(\cfgstate)}- \pos_\data(\cfgstate)}{\siz(\cfgstate)}}}\\
        d' = \fun{\cfgstate}{\arraypos{\datav(\cfgstate)}{\pos_\code(\cfgstate), \pos_\code(\cfgstate) + k(\cfgstate) - 1}}\\
        d = \fun{\cfgstate}{\concat{d'(\cfgstate)}{\STOP^{\siz(\cfgstate) - k(\cfgstate)}}}\\
        f = \fun{\cfgstate}{\updatevarset{\cfgstate}{\memvarabs{i}{\pc'}}{d(\cfgstate)[i]}{i}{[\pos_\memo(\cfgstate), \pos_\memo(\cfgstate) + \siz(\cfgstate) -1]}}
    }
    {\edgestep{\precontract, \cd}{\cfgnode{\pc}{0}}{\cfgnode{\pc}{1}}} \\
    {\Def{\memvarabs{X}{\pc'}} \\
    \Uses{\globenvvar{\cfgexternalpc{\pc}}, \stackvar{x_1}, \stackvar{x_2}, \stackvar{x_3}, \stackvar{x_4}}}
\end{mathpar}
\begin{mathpar}
    \infer{
        \precontract(\pc) =  (\EXTCODECOPY(\stackvar{y},\stackvar{x_1}, \stackvar{x_2}, \stackvar{x_3}, \stackvar{x_4}), \pc', \pre) \\
        a = \fun{\cfgstate}{\accesscfgstate{\cfgstate}{\stackvar{x_1}}} \\
        \pos_\memo = \fun{\cfgstate}{\accesscfgstate{\cfgstate}{\stackvar{x_2}}}\\
        \pos_\code= \fun{\cfgstate}{\accesscfgstate{\cfgstate}{\stackvar{x_3}}}\\
        \siz = \fun{\cfgstate}{\accesscfgstate{\cfgstate}{\stackvar{x_4}}}\\
        \aw = \fun{\cfgstate}{\memext{\accesscfgstate{\cfgstate}{\locenvvar{\msizepc{\pc}}}}{\pos_\memo(\cfgstate)}{\siz(\cfgstate)}} \\     
        \costs = \fun{\cfgstate}{\costmem{\accesscfgstate{\cfgstate}{\locenvvar{\msizepc{\pc}}}}{\aw(\cfgstate)}+700 + 3 \cdot \left \lceil \frac{\siz(\cfgstate)}{32} \right \rceil}\\ 
        f = \fun{\cfgstate}{\updatestate{\cfgstate}{\envvar{\gaspc{\pc'}}}{\cfgstate[\envvar{\gaspc{\pc}}] - \costs(\cfgstate)}} 
    }
    {\edgestep{\precontract, \cd}{\cfgnode{\pc}{1}}{\cfgnode{\pc'}{2}}} \\
    {\Defs{\locenvvar{\gaspc{\pc'}}} \\
        \Uses{\locenvvar{\gaspc{\pc}}, \locenvvar{\msizepc{\pc}}, \stackvar{x_2}, \stackvar{x_3}, \stackvar{x_4}}}
\end{mathpar}
\begin{mathpar}
    \infer{
        \precontract(\pc) =  (\EXTCODECOPY(\stackvar{y},\stackvar{x_1}, \stackvar{x_2}, \stackvar{x_3}, \stackvar{x_4}), \pc', \pre) \\
        a = \fun{\cfgstate}{\accesscfgstate{\cfgstate}{\stackvar{x_1}}} \\
        \pos_\memo = \fun{\cfgstate}{\accesscfgstate{\cfgstate}{\stackvar{x_2}}}\\
        \pos_\code= \fun{\cfgstate}{\accesscfgstate{\cfgstate}{\stackvar{x_3}}}\\
        \siz = \fun{\cfgstate}{\accesscfgstate{\cfgstate}{\stackvar{x_4}}}\\
        \aw = \fun{\cfgstate}{\memext{\accesscfgstate{\cfgstate}{\locenvvar{\msizepc{\pc}}}}{\pos_\memo(\cfgstate)}{\siz(\cfgstate)}} \\
        f = \fun{\cfgstate}{\updatestate{\cfgstate}{\envvar{\msizepc{\pc'}}}{\aw(\cfgstate)}} 
    }
    {\edgestep{\precontract, \cd}{\cfgnode{\pc}{2}}{\cfgnode{\pc'}{0}}} \\
    {\Defs{\locenvvar{\msizepc{\pc'}}} \\
        \Uses{\locenvvar{\msizepc{\pc}}, \stackvar{x_2}, \stackvar{x_4}}}
\end{mathpar}

\paragraph{Stack Operations}
Since we assume the code to be in SSA form, all stack operations have been replaced by assignments.

\begin{mathpar}
    \infer{
        \precontract(\pc) =  (\ASSIGN(\stackvar{y},\stackvar{x}), \pc', \pre) \\
        r = \fun{\cfgstate}{\accesscfgstate{\cfgstate}{\stackvar{x}}} \\
        f = \fun{\cfgstate}{\updatestate{\cfgstate}{\stackvar{y}}{r(\cfgstate)}}
    }
    {\edgestep{\precontract, \cd}{\cfgnode{\pc}{0}}{\cfgnode{\pc}{1}}} \\
    {\Defs{\stackvar{y}} \\
        \Uses{\stackvar{x}}}
\end{mathpar}
\begin{mathpar}
    \infer{
        \precontract(\pc) =  (\ASSIGN(\stackvar{y},\stackvar{x}), \pc', \pre) \\        
        f = \fun{\cfgstate}{\updatestate{\cfgstate}{\envvar{\gaspc{\pc'}}}{\cfgstate[\envvar{\gaspc{\pc}}] - 3}} 
    }
    {\edgestep{\precontract, \cd}{\cfgnode{\pc}{1}}{\cfgnode{\pc'}{0}}} \\
    {\Defs{\locenvvar{\gaspc{\pc'}}} \\
        \Uses{\locenvvar{\gaspc{\pc}}}}
\end{mathpar}

\paragraph{Jump Instructions}
For the case of jump instructions by assumption, the jump destination has been precomputed. 
Since the JUMP instruction has no other effect than updating the program counter, we only need to add a rule for updating the gas and stepping to the next program counter: 

\begin{mathpar}
    \infer{
        \precontract(\pc) =  (\JUMP(\stackvar{y},\stackvar{x}), \pc', \pre) \\ 
        \pre[0] = \optional{\pc'} \\       
        f = \fun{\cfgstate}{\updatestate{\cfgstate}{\envvar{\gaspc{\pc'}}}{\cfgstate[\envvar{\gaspc{\pc}}] - 8}} 
    }
    {\edgestep{\precontract, \cd}{\cfgnode{\pc}{0}}{\cfgnode{\pc'}{0}}} \\
    {\Defs{\locenvvar{\gaspc{\pc'}}} \\
        \Uses{\locenvvar{\gaspc{\pc}}}}
\end{mathpar}

For the conditional jump instruction, in addition to deducing the gas, the next program counter needs to be decided based on the condition.

We first give the rules for updating the gas value:

\begin{mathpar}
    \infer{
        \precontract(\pc) = (\JUMPI(\stackvar{y}, \stackvar{x_1}, \stackvar{x_2}), \pc', \pre) \\
        \pre[0] = \optional{\pc''} \\     
        f = \fun{\cfgstate}{\updatestate{\cfgstate}{\envvar{\gaspc{\pc'}}}{\cfgstate[\envvar{\gaspc{\pc}}] - 10}} 
    }
    {\edgestep{\precontract, \cd}{\cfgnode{\pc}{0}}{\cfgnode{\pc}{1}}} \\
    {\Defs{\locenvvar{\gaspc{\pc'}}} \\
        \Uses{\locenvvar{\gaspc{\pc}}}}
\end{mathpar}

Finally, we give the rules for branching: 
\begin{mathpar}
    \infer{
        \precontract(\pc) = (\JUMPI(\stackvar{y}, \stackvar{x_1}, \stackvar{x_2}), \pc', \pre) \\
        \pre[0] = \optional{\pc''} \\     
        Q = \fun{\cfgstate}{\accesscfgstate{\cfgstate}{\stackvar{x_2}} = 0}
    }
    {\predstep{\precontract, \cd}{\cfgnode{\pc}{1}}{\cfgnode{\pc'}{0}}} \\
    {\Def{\emptyset} \\
        \Uses{\stackvar{x_2}}}
\end{mathpar}

\begin{mathpar}
    \infer{
        \precontract(\pc) = (\JUMPI(\stackvar{y}, \stackvar{x_1}, \stackvar{x_2}), \pc', \pre) \\
        \pre[0] = \optional{\pc''} \\     
        Q = \fun{\cfgstate}{\accesscfgstate{\cfgstate}{\stackvar{x_2}} \neq 0}
    }
    {\predstep{\precontract, \cd}{\cfgnode{\pc}{1}}{\cfgnode{\pc''}{0}}} \\
    {\Def{\emptyset} \\
        \Uses{\stackvar{x_2}}}
\end{mathpar}

\paragraph{Memory Instructions}

	\begin{mathpar}
		\infer{
			\precontract(\pc) = (\MLOAD(\stackvar{y}, \stackvar{x})), \pc', \pre) \\
            \pre[1] = \none \\
            f = \fun{\cfgstate}{\updatestate{\cfgstate}{\stackvar{y}}{\load~\cfgstate~\memvar{\accesscfgstate{\cfgstate}{\stackvar{x}}}{\pc}}} 
		}
		{\edgestep{\transenv}{\cfgnode{\pc}{0}}{\cfgnode{\pc}{1}}}
		\\
		{\Defs{\stackvar{y}} \\
			\Use{\memvarabs{X}{\pc} \cup \memvarconc{X}{\pc}}}
    \end{mathpar}
    \begin{mathpar}
        \infer{
            \precontract(\pc) = (\MLOAD(\stackvar{y}, \stackvar{x})), \pc', \pre) \\
            \pre[1] = \optional{\loc} \\
            f = \fun{\cfgstate}{\updatestate{\cfgstate}{\stackvar{y}}{\load~\cfgstate~\memvar{\loc}{\pc}}} 
        }
        {\edgestep{\transenv}{\cfgnode{\pc}{0}}{\cfgnode{\pc}{1}}}
        \\
        {\Defs{\stackvar{y}} \\
            \Use{\{\memvarconc{\loc}{\pc}, \memvarabs{\loc}{\pc}}\}}
    \end{mathpar}
\begin{mathpar}
    \infer{
        \precontract(\pc) = (\MLOAD(\stackvar{y}, \stackvar{x})), \pc', \pre) \\
        a = \fun{\cfgstate}{\accesscfgstate{\cfgstate}{\stackvar{x}}} \\
        \aw = \fun{\cfgstate}{\memext{\accesscfgstate{\cfgstate}{\locenvvar{\msizepc{\pc}}}}{a(\cfgstate)}{32}} \\
        \costs = \fun{\cfgstate}{\costmem{\accesscfgstate{\cfgstate}{\locenvvar{\msizepc{\pc}}}}{\aw(\cfgstate)} + 3} \\
        f = \fun{\cfgstate}{\updatestate{\cfgstate}{\locenvvar{\gaspc{\pc'}}}{\accesscfgstate{\cfgstate}{\gaspc{\pc'}}- \costs(\cfgstate)}}    
        }
    {\edgestep{\precontract, \cd}{\cfgnode{\pc}{1}}{\cfgnode{\pc}{2}}} \\
    {\Defs{\locenvvar{\gaspc{\pc'}}} \\
        \Uses{\locenvvar{\gaspc{\pc}}, \locenvvar{\msizepc{\pc}}, \stackvar{x}}}
\end{mathpar}
\begin{mathpar}
    \infer{
        \precontract(\pc) = (\MLOAD(\stackvar{y}, \stackvar{x})), \pc', \pre) \\
        a = \fun{\cfgstate}{\accesscfgstate{\cfgstate}{\stackvar{x}}} \\
        \aw = \fun{\cfgstate}{\memext{\accesscfgstate{\cfgstate}{\locenvvar{\msizepc{\pc}}}}{a(\cfgstate)}{32}} \\
        f = \fun{\cfgstate}{\updatestate{\cfgstate}{\locenvvar{\msizepc{\pc'}}}{\aw(\cfgstate)}} 
    }
    {\edgestep{\precontract, \cd}{\cfgnode{\pc}{2}}{\cfgnode{\pc'}{0}}} \\
    {\Defs{\locenvvar{\msizepc{\pc'}}} \\
        \Uses{\locenvvar{\msizepc{\pc}}, \stackvar{x}}}
\end{mathpar}
Note that we do not distinguish between the $\MLOAD$ and the $\MLOAD$byte instruction, because we assume the preprocessing to already check for consistent memory accesses.

    \begin{mathpar}
			\infer{
				\precontract(\pc) = (\MSTORE(\stackvar{x_1}, \stackvar{x_2}), \pc', \pre) \\
                \pre[1] = \None \\
                f = \fun{\cfgstate}{\updatestate{\cfgstate}{\memvarabs{\left (\accesscfgstate{\cfgstate}{\stackvar{x_1}} \right)}{\pc'}}{\accesscfgstate{\cfgstate}{\stackvar{x_2}}}}
			}
			{\edgestep{\transenv}{\cfgnode{\pc}{\startvalue +2}}{\cfgnode{\pc}{2}}}
			\\
			{\Def{\memvarabs{X}{\pc'}} \\
				\Use{\memvarabs{X}{\pc} \cup \{\stackvar{x_1}, \stackvar{x_2}\}}}
    \end{mathpar}
\begin{mathpar}
		\infer{
			\precontract(\pc) = (\MSTORE(\stackvar{x_1}, \stackvar{x_2}), \pc', \pre) \\
            \pre[0] = \optional{\loc} \\
            f = \fun{\cfgstate}{\updatestate{\cfgstate}{\memvarconc{\loc}{\pc'}}{\accesscfgstate{\cfgstate}{\stackvar{x_2}}}}
		}
		{\edgestep{\transenv}{\cfgnode{\pc}{\startvalue +2 }}{\cfgnode{\pc}{1}}}
		\\
		{\Defs{\memvarconc{\loc}{\pc'}} \\
			\Use{\{ \stackvar{x_2} \}}}		
\end{mathpar}
	\begin{mathpar}
		\infer{
			\precontract(\pc) = (\MSTORE(\stackvar{x_1}, \stackvar{x_2}), \pc', \pre) \\
            \pre[0] = \optional{\loc} \\
            f = \fun{\cfgstate}{\updatestate{\cfgstate}{\memvarabs{\loc}{\pc'}}{\bot}}
		}
		{\edgestep{\transenv}{\cfgnode{\pc}{1}}{\cfgnode{\pc}{2}}}
		\\
		{\Defs{\memvarabs{\loc}{\pc'}} \\
			\Use{\emptyset}}
	\end{mathpar}
    \begin{mathpar}
        \infer{
			\precontract(\pc) = (\MSTORE(\stackvar{x_1}, \stackvar{x_2}), \pc', \pre) \\
            a = \fun{\cfgstate}{\accesscfgstate{\cfgstate}{\stackvar{x_1}}} \\
            \aw = \fun{\cfgstate}{\memext{\accesscfgstate{\cfgstate}{\locenvvar{\msizepc{\pc}}}}{a(\cfgstate)}{32}} \\
            \costs = \fun{\cfgstate}{\costmem{\accesscfgstate{\cfgstate}{\locenvvar{\msizepc{\pc}}}}{\aw(\cfgstate)} + 3} \\
            f = \fun{\cfgstate}{\updatestate{\cfgstate}{\locenvvar{\gaspc{\pc'}}}{\accesscfgstate{\cfgstate}{\gaspc{\pc'}}- \costs(\cfgstate)}}          
            }
        {\edgestep{\precontract, \cd}{\cfgnode{\pc}{2}}{\cfgnode{\pc}{\startvalue + 5}}} \\
        {\Defs{\locenvvar{\gaspc{\pc'}}} \\
            \Uses{\locenvvar{\gaspc{\pc}}, \locenvvar{\msizepc{\pc}}, \stackvar{x_1}}}
    \end{mathpar}
    \begin{mathpar}
        \infer{
			\precontract(\pc) = (\MSTORE(\stackvar{x_1}, \stackvar{x_2}), \pc', \pre) \\
            a = \fun{\cfgstate}{\accesscfgstate{\cfgstate}{\stackvar{x_1}}} \\
            \aw = \fun{\cfgstate}{\memext{\accesscfgstate{\cfgstate}{\locenvvar{\msizepc{\pc}}}}{a(\cfgstate)}{32}} \\
            f = \fun{\cfgstate}{\updatestate{\cfgstate}{\locenvvar{\msizepc{\pc'}}}{\aw(\cfgstate)}} 
        }
        {\edgestep{\precontract, \cd}{\cfgnode{\pc}{\startvalue + 5}}{\cfgnode{\pc'}{0}}} \\
        {\Defs{\locenvvar{\msizepc{\pc'}}} \\
            \Uses{\locenvvar{\msizepc{\pc}}, \stackvar{x_1}}}
    \end{mathpar}

\paragraph{Storage Instructions}
The storage instructions closely resemble the instructions for memory access:

\begin{mathpar}
    \infer{
        \precontract(\pc) = (\SLOAD(\stackvar{y}, \stackvar{x})), \pc', \pre) \\
        \pre[1] = \none \\
        f = \fun{\cfgstate}{\updatestate{\cfgstate}{\stackvar{y}}{\load~\cfgstate~\storvar{\accesscfgstate{\cfgstate}{\stackvar{x}}}{\pc}}} 
    }
    {\edgestep{\transenv}{\cfgnode{\pc}{0}}{\cfgnode{\pc}{1}}}
    \\
    {\Defs{\stackvar{y}} \\
        \Use{\storvarabs{X}{\pc} \cup \storvarconc{X}{\pc}}}
\end{mathpar}
\begin{mathpar}
    \infer{
        \precontract(\pc) = (\SLOAD(\stackvar{y}, \stackvar{x})), \pc', \pre) \\
        \pre[1] = \optional{\loc} \\
        f = \fun{\cfgstate}{\updatestate{\cfgstate}{\stackvar{y}}{\load~\cfgstate~\storvar{\loc}{\pc}}} 
    }
    {\edgestep{\transenv}{\cfgnode{\pc}{0}}{\cfgnode{\pc}{1}}}
    \\
    {\Defs{\stackvar{y}} \\
        \Use{\{\storvarconc{\loc}{\pc}, \storvarabs{\loc}{\pc}}\}}
\end{mathpar}
\begin{mathpar}
\infer{
    \precontract(\pc) = (\SLOAD(\stackvar{y}, \stackvar{x})), \pc', \pre) \\
    a = \fun{\cfgstate}{\accesscfgstate{\cfgstate}{\stackvar{x}}} \\
    f = \fun{\cfgstate}{\updatestate{\cfgstate}{\locenvvar{\gaspc{\pc'}}}{\accesscfgstate{\cfgstate}{\gaspc{\pc'}}- 200}} 
}
{\edgestep{\precontract, \cd}{\cfgnode{\pc}{1}}{\cfgnode{\pc'}{0}}} \\
{\Defs{\locenvvar{\gaspc{\pc'}}} \\
    \Uses{\locenvvar{\gaspc{\pc}},}}
\end{mathpar}

\begin{mathpar}
    \infer{
        \precontract(\pc) = (\SSTORE(\stackvar{x_1}, \stackvar{x_2}), \pc', \pre) \\
        \pre[1] = \None \\
        f = \fun{\cfgstate}{\updatestate{\cfgstate}{\storvarabs{\left (\accesscfgstate{\cfgstate}{\stackvar{x_1}} \right)}{\pc'}}{\accesscfgstate{\cfgstate}{\stackvar{x_2}}}}
    }
    {\edgestep{\transenv}{\cfgnode{\pc}{\startvalue +2}}{\cfgnode{\pc}{2}}}
    \\
    {\Def{\storvarabs{X}{\pc'}} \\
        \Use{\storvarabs{X}{\pc} \cup \{\stackvar{x_1}, \stackvar{x_2}\}}}
\end{mathpar}
\begin{mathpar}
\infer{
    \precontract(\pc) = (\SSTORE(\stackvar{x_1}, \stackvar{x_2}), \pc', \pre) \\
    \pre[0] = \optional{\loc} \\
    f = \fun{\cfgstate}{\updatestate{\cfgstate}{\storvarconc{\loc}{\pc'}}{\accesscfgstate{\cfgstate}{\stackvar{x_2}}}}
}
{\edgestep{\transenv}{\cfgnode{\pc}{\startvalue +2 }}{\cfgnode{\pc}{1}}}
\\
{\Defs{\storvarconc{\loc}{\pc'}} \\
    \Use{\{ \stackvar{x_2} \}}}		
\end{mathpar}
\begin{mathpar}
\infer{
    \precontract(\pc) = (\SSTORE(\stackvar{x_1}, \stackvar{x_2}), \pc', \pre) \\
    \pre[0] = \optional{\loc} \\
    f = \fun{\cfgstate}{\updatestate{\cfgstate}{\storvarabs{\loc}{\pc'}}{\bot}}
}
{\edgestep{\transenv}{\cfgnode{\pc}{1}}{\cfgnode{\pc}{2}}}
\\
{\Defs{\storvarabs{\loc}{\pc'}} \\
    \Use{\emptyset}}
\end{mathpar}
\begin{mathpar}
\infer{
    \precontract(\pc) = (\SSTORE(\stackvar{x_1}, \stackvar{x_2}), \pc', \pre) \\
    a = \fun{\cfgstate}{\accesscfgstate{\cfgstate}{\stackvar{x_1}}} \\
    b = \fun{\cfgstate}{\accesscfgstate{\cfgstate}{\stackvar{x_2}}} \\
    \costs = \fun{\cfgstate}{\cond{(b(\cfgstate) \neq 0 \land (\load~\cfgstate~\storvar{(\accesscfgstate{\cfgstate}{a(\cfgstate)})}{\pc} = 0)}{20000}{5000}} \\
    f = \fun{\cfgstate}{\updatestate{\cfgstate}{\locenvvar{\gaspc{\pc'}}}{\accesscfgstate{\cfgstate}{\gaspc{\pc'}}- \costs(\cfgstate)}}
} 
{\edgestep{\precontract, \cd}{\cfgnode{\pc}{2}}{\cfgnode{\pc'}{0}}} \\
{\Defs{\locenvvar{\gaspc{\pc'}}} \\
    \Use{ \{ \locenvvar{\gaspc{\pc}}, \stackvar{x_1} \} ~\cup~ \memvarabs{X}{\pc}} ~\cup~ \memvarconc{X}{\pc}}
\end{mathpar}

\paragraph{Accessing the machine state}

\begin{mathpar}
    \infer{
        \precontract(\pc) =  (\GAS(\stackvar{y}), \pc', \pre) \\
        r = \fun{\cfgstate}{\accesscfgstate{\cfgstate}{\locenvvar{\gaspc{\pc}}}} \\
        f = \fun{\cfgstate}{\updatestate{\cfgstate}{\stackvar{y}}{r(\cfgstate)}}
    }
    {\edgestep{\precontract, \cd}{\cfgnode{\pc}{0}}{\cfgnode{\pc}{1}}} \\
    {\Defs{\stackvar{y}} \\
        \Uses{\locenvvar{\gaspc{\pc}}}}
\end{mathpar}

Most instructions for accessing the machine state have the same gas cost. For this reason we summarize the rule for gas substraction for them.

\begin{mathpar}
    \infer{
        \precontract(\pc) =  (\instruction(\stackvar{y}), \pc', \pre) \\    
        \instruction \in \{ \GAS, \PC, \MSIZE \} \\    
        f = \fun{\cfgstate}{\updatestate{\cfgstate}{\envvar{\gaspc{\pc'}}}{\cfgstate[\envvar{\gaspc{\pc}}] - 2}} 
    }
    {\edgestep{\precontract, \cd}{\cfgnode{\pc}{1}}{\cfgnode{\pc'}{0}}} \\
    {\Defs{\locenvvar{\gaspc{\pc'}}} \\
        \Uses{\locenvvar{\gaspc{\pc}}}}
\end{mathpar}

\begin{mathpar}
    \infer{
        \precontract(\pc) =  (\MSIZE(\stackvar{y}), \pc', \pre) \\
        r = \fun{\cfgstate}{\accesscfgstate{\cfgstate}{\locenvvar{\msizepc{\pc}}}} \\
        f = \fun{\cfgstate}{\updatestate{\cfgstate}{\stackvar{y}}{r(\cfgstate)}}
    }
    {\edgestep{\precontract, \cd}{\cfgnode{\pc}{0}}{\cfgnode{\pc}{1}}} \\
    {\Defs{\stackvar{y}} \\
        \Uses{\locenvvar{\msizepc{\pc}}}}
\end{mathpar}

\begin{mathpar}
    \infer{
        \precontract(\pc) =  (\PC(\stackvar{y}), \pc', \pre) \\
        r = \fun{\cfgstate}{\pc} \\
        f = \fun{\cfgstate}{\updatestate{\cfgstate}{\stackvar{y}}{r(\cfgstate)}}
    }
    {\edgestep{\precontract, \cd}{\cfgnode{\pc}{0}}{\cfgnode{\pc}{1}}} \\
    {\Defs{\stackvar{y}} \\
        \Use{\emptyset}}
\end{mathpar}

\paragraph{Logging}
Since we do not model the logged events, for the logging instructions we only need to model the effect on gas and local memory.

\begin{mathpar}
    \infer{
        \precontract(\pc) =  (\LOG{n}(\stackvar{x_1}, \stackvar{x_2}, \dots, \stackvar{x_{n+2}}), \pc', \pre) \\
        \pos_\mem = \fun{\cfgstate}{\accesscfgstate{\cfgstate}{\stackvar{x_1}}}\\
        \siz = \fun{\cfgstate}{\accesscfgstate{\cfgstate}{\stackvar{x_2}}} \\
        \aw = \fun{\cfgstate}{\memext{\accesscfgstate{\cfgstate}{\locenvvar{\msizepc{\pc}}}}{\pos_\memo(\cfgstate)}{\siz(\cfgstate)}} \\ 
        \costs = \fun{\cfgstate}{\costmem{\accesscfgstate{\cfgstate}{\locenvvar{\msizepc{\pc}}}}{\aw(\cfgstate)} + 375 + 8 \cdot \siz(\cfgstate) + n  \cdot 375} \\ 
        f = \fun{\cfgstate}{\updatestate{\cfgstate}{\locenvvar{\gaspc{\pc'}}}{\accesscfgstate{\cfgstate}{\locenvvar{\gaspc{\pc}}}- \costs(\cfgstate)}}
        }
    {\edgestep{\precontract, \cd}{\cfgnode{\pc}{0}}{\cfgnode{\pc}{1}}} \\
    {\Defs{\locenvvar{\gaspc{\pc'}}} \\
        \Uses{\locenvvar{\gaspc{\pc}}, \locenvvar{\msizepc{\pc}}, \stackvar{x_1}, \stackvar{x_2}}}
\end{mathpar}

\begin{mathpar}
    \infer{
        \precontract(\pc) =  (\LOG{n}(\stackvar{x_1}, \stackvar{x_2}, \dots, \stackvar{x_{n+2}}), \pc', \pre) \\
        \pos_\mem = \fun{\cfgstate}{\accesscfgstate{\cfgstate}{\stackvar{x_1}}}\\
        \siz = \fun{\cfgstate}{\accesscfgstate{\cfgstate}{\stackvar{x_2}}} \\
        \aw = \fun{\cfgstate}{\memext{\accesscfgstate{\cfgstate}{\locenvvar{\msizepc{\pc}}}}{\pos_\memo(\cfgstate)}{\siz(\cfgstate)}} \\ 
        f = \fun{\cfgstate}{\updatestate{\cfgstate}{\locenvvar{\msizepc{\pc'}}}{\aw(\cfgstate)}}
        }
    {\edgestep{\precontract, \cd}{\cfgnode{\pc}{1}}{\cfgnode{\pc'}{0}}} \\
    {\Defs{\locenvvar{\msizepc{\pc'}}} \\
        \Uses{\locenvvar{\msizepc{\pc}}, \stackvar{x_1}, \stackvar{x_2}}}
\end{mathpar}

\paragraph{Halting instructions}

\begin{mathpar}
    \infer{
        \precontract(\pc) =  (\RETURN(\stackvar{x_1}, \stackvar{x_2}, \pc', \pre) \\
        \pos_\mem = \fun{\cfgstate}{\accesscfgstate{\cfgstate}{\stackvar{x_1}}}\\
        \siz = \fun{\cfgstate}{\accesscfgstate{\cfgstate}{\stackvar{x_2}}} \\
        \aw = \fun{\cfgstate}{\memext{\accesscfgstate{\cfgstate}{\locenvvar{\msizepc{\pc}}}}{\pos_\memo(\cfgstate)}{\siz(\cfgstate)}} \\ 
        \costs = \fun{\cfgstate}{\costmem{\accesscfgstate{\cfgstate}{\locenvvar{\msizepc{\pc}}}}{\aw(\cfgstate)}} \\ 
        f = \fun{\cfgstate}{\updatestate{\cfgstate}{\locenvvar{\gaspc{\pc'}}}{\accesscfgstate{\cfgstate}{\locenvvar{\gaspc{\pc}}}- \costs(\cfgstate)}}
        }
    {\edgestep{\precontract, \cd}{\cfgnode{\pc}{0}}{\cfgnode{\pc}{1}}} \\
    {\Defs{\locenvvar{\gaspc{\pc'}}} \\
        \Uses{\locenvvar{\gaspc{\pc}}, \locenvvar{\msizepc{\pc}}, \stackvar{x_1}, \stackvar{x_2}}}
\end{mathpar}

\begin{mathpar}
    \infer{
        \precontract(\pc) =  (\RETURN(\stackvar{x_1}, \stackvar{x_2}, \pc', \pre) \\
        \pos_\mem = \fun{\cfgstate}{\accesscfgstate{\cfgstate}{\stackvar{x_1}}}\\
        \siz = \fun{\cfgstate}{\accesscfgstate{\cfgstate}{\stackvar{x_2}}} \\
        \aw = \fun{\cfgstate}{\memext{\accesscfgstate{\cfgstate}{\locenvvar{\msizepc{\pc}}}}{\pos_\memo(\cfgstate)}{\siz(\cfgstate)}} \\ 
        f = \fun{\cfgstate}{\updatestate{\cfgstate}{\locenvvar{\msizepc{\pc'}}}{\aw(\cfgstate)}}
        }
    {\edgestep{\precontract, \cd}{\cfgnode{\pc}{1}}{\haltnode}} \\
    {\Defs{\locenvvar{\msizepc{\pc'}}} \\
        \Uses{\locenvvar{\msizepc{\pc}}, \stackvar{x_1}, \stackvar{x_2}}}
\end{mathpar}

\begin{mathpar}
    \infer{
        \precontract(\pc) =  (\STOP, \pc', \pre) \\ 
        f = \fun{\cfgstate}{\cfgstate}
        }
    {\edgestep{\precontract, \cd}{\cfgnode{\pc}{0}}{\haltnode}} \\
    {\Def{\emptyset} \\
        \Use{\emptyset}}
\end{mathpar}

\begin{mathpar}
    \infer{
        \precontract(\pc) =  (\INVALID, \pc', \pre) \\ 
        f = \fun{\cfgstate}{\cfgstate}
        }
    {\edgestep{\precontract, \cd}{\cfgnode{\pc}{0}}{\exceptionnode}} \\
    {\Def{\emptyset} \\
        \Use{\emptyset}}
\end{mathpar}

The $\SUICIDE$ instruction allows for the self destruction of the executing account.
Since we only model the execution of a single contract and since the execution halts with the execution of $\SUICIDE$, we do not model the deletion of the contract itself. 
However, the gas cost of the $\SUICIDE$ instruction, as well as the balances of accounts are affected by the execution of the $\SUICIDE$ instruction. 

\begin{mathpar}
    \infer{
        \precontract(\pc) =  (\SUICIDE(\stackvar{x}), \pc', \pre) \\ 
        a_\textit{ben} = \fun{\cfgstate}{\accesscfgstate{\cfgstate}{\stackvar{x}}} \\
        a = \fun{\cfgstate}{a_\textit{ben}(\cfgstate) \mod{} 2^{160}} \\
        \gstate = \fun{\cfgstate}{\access{\accesscfgstate{\cfgstate}{\locenvvar{\cfgexternalpc{\pc}}}}{\gstate}} \\ 
        \phi = \fun{\cfgstate}{\gstate(\cfgstate)(a) = \bot} \\ 
        f_1 = \fun{\cfgstate}{
            \updatestate{\cfgstate}
            {\globenvvar{\cfgexternalpc{\pc'}}}
            {\update
                {\accesscfgstate{\cfgstate}{\globenvvar{\cfgexternalpc{\pc}}}}
                {\gstate}
                { \updategstate
                    {\gstate(\cfgstate)}
                    {\accesscfgstate{\cfgstate}{\locenvvar{\activeaccount}}}
                    {\update{\gstate(\cfgstate)}{\balance}{0}}
                }
            }
        } \\
        f_2 = \fun{\cfgstate}{
            \updatestate
            {\cfgstate}
            {\globenvvar{\cfgexternalpc{\pc'}}}
            {\update
                {\accesscfgstate{\cfgstate}{\globenvvar{\cfgexternalpc{\pc}}}}
                {\gstate}
                {\updategstate
                    {\gstate(\cfgstate)}
                    {a(\cfgstate)}
                    {\inc
                        {\gstate(\cfgstate)}
                        {\balance}
                        {\access{\gstate(\cfgstate)(\accesscfgstate{\cfgstate}{\locenvvar{\activeaccount}})}{\balance}}
                    } 
                }
            }
        } \\ 
        f_3 = \fun{\cfgstate}{
            \updatestate
            {\cfgstate}
            {\globenvvar{\cfgexternalpc{\pc'}}}
            {
                \update
                    {\accesscfgstate{\cfgstate}{\globenvvar{\cfgexternalpc{\pc}}}}
                    {\gstate}
                    {\updategstate
                        {\gstate(\cfgstate)}
                        {a(\cfgstate)}
                        {\account{0}{\access{\gstate(\cfgstate)(\accesscfgstate{\cfgstate}{\locenvvar{\activeaccount}})}{\balance}}{\lam{x}{0}}{\emptyarray}}
                    }
            }
        }\\
        f = \fun{\cfgstate}{\cond{(\phi(\cfgstate))}{f_1(f_2(\cfgstate))}{f_1(f_3(\cfgstate))}}
        }
    {\edgestep{\precontract, \cd}{\cfgnode{\pc}{0}}{\cfgnode{\pc}{1}}} \\
    {\Defs{\globenvvar{\cfgexternalpc{\pc'}}} \\
        \Uses{\globenvvar{\cfgexternalpc{\pc}}, \stackvar{x}, \locenvvar{\activeaccount}}}
\end{mathpar}
\begin{mathpar}
    \infer{
        \precontract(\pc) =  (\SUICIDE(\stackvar{x}), \pc', \pre) \\ 
        a_\textit{ben} = \fun{\cfgstate}{\accesscfgstate{\cfgstate}{\stackvar{x}}} \\
        a = \fun{\cfgstate}{a_\textit{ben}(\cfgstate) \mod{} 2^{160}} \\
        \gstate = \fun{\cfgstate}{\access{\accesscfgstate{\cfgstate}{\locenvvar{\cfgexternalpc{\pc}}}}{\gstate}} \\ 
        \phi = \fun{\cfgstate}{\gstate(\cfgstate)(a) = \bot} \\ 
        \costs = \fun{\cfgstate}{\cond{(\phi(\cfgstate))}{37000}{5000}} \\
        f = \fun{\cfgstate}{\updatestate{\cfgstate}{\locenvvar{\gaspc{\pc'}}}{\accesscfgstate{\cfgstate}{\locenvvar{\gaspc{\pc}}} - \costs(\cfgstate)}}
        }
    {\edgestep{\precontract, \cd}{\cfgnode{\pc}{1}}{\haltnode}} \\
    {\Defs{\locenvvar{\gaspc{\pc'}}} \\
        \Uses{\globenvvar{\cfgexternalpc{\pc}}, \stackvar{x}, \locenvvar{\gaspc{\pc}}}}
\end{mathpar}

\paragraph{Rules for transaction initiating instructions.}
Since we will exclude the execution of $\DELEGATECALL$ and $\CALLCODE$ statements by assumption, we will only give the CFG rules for $\CALL$, $\STATICCALL$, and $\CREATE$. 

\paragraph*{\CALL}

The definition of the $\CALL$ rule comes with small technical difficulties. 
In order to give a precise analysis, the different changes in the state that are triggered by the call need to be done in separate nodes whenever they have different dependencies. 
While e.g., the output data fragment can be influenced by the input to the call as well as by the global environment, the memory outside of the output fragment will simply be propagated. 
Similarly, the number of active words $\msize$ only depends on the arguments specifying the input and output memory fragment, but not on any other inputs or the global environment. 

To model these dependencies accurately, the corresponding updates need to happen in different nodes. 
However, we can only characterize the overall call effects on a state (using the EVM small-step semantics).

We first define the function $\applycall$ that mimics the effects of a function call on a CFG state $\cfgstate$.
For simplicity, we define the function here as a relation. 
However, this relation is functional given that 
$\precontract(\pc) = (\CALL(\stackvar{y}, \stackvar{\lgas}, \stackvar{\recipient}, \stackvar{\valu}, \stackvar{\io}, \stackvar{\is}, \stackvar{\oo}, \stackvar{\os}), \pc', \pre) \lor \ \precontract(\pc) = (\STATICCALL(\stackvar{y}, \stackvar{\lgas}, \stackvar{\recipient}, \stackvar{\io}, \stackvar{\is}, \stackvar{\oo}, \stackvar{\os}), \pc', \pre)$.

\begin{mathpar}
    \infer 
    {  
        (\precontract(\pc) = (\CALL(\stackvar{y}, \stackvar{\lgas}, \stackvar{\recipient}, \stackvar{\valu}, \stackvar{\io}, \stackvar{\is}, \stackvar{\oo}, \stackvar{\os}), \pc', \pre) \lor \ \precontract(\pc) = (\STATICCALL(\stackvar{y}, \stackvar{\lgas}, \stackvar{\recipient}, \stackvar{\io}, \stackvar{\is}, \stackvar{\oo}, \stackvar{\os}), \pc', \pre)) \\
    (\transenv, \exstate) = \toevmstate(\cfgstate, \precontract, \pc) \\
    \transactionstep{\transenv}{\cons{\exstate}{\callstack}}{\cons{\exstate'}{\callstack}} \\
    \cfgstate' = \tocfgstate(\transenv, \exstate') \\
    }
    {\applycall(\cfgstate, \precontract, \pc) = \cfgstate'} 
\end{mathpar}

Note that $\applycall$ operates on a CFG state without temporal variables (since this is what $\toevmstate$ is expecting and $\tocfgstate$ is returning). 
When using $\applycall$ in the following rules in conjunction with full CFG states $\cfgstatefull$, we will write $\restrict{\cfgstatefull}{\cfgstate}$ to denote the restriction of $\cfgstatefull$ to the set of non-temporal variables ($\domain{\cfgstate}$).


The rules update the different state components one after the other, grouping those updates together that have the same dependencies. 
Technically, for propagating pc-indexed state components, we need to update all variables in individual notes. 
To this end, first, the temporal variables in $\cfgstatecopy$ are set to the values of $\cfgstate'$ (the state after applying the call). 
Finally, the variables of $\cfgstate$ are (one by one) updated to the values of $\cfgstatecopy$ and $\cfgstatecopy$ is set to $\cfgstatecopybot$.


We now define the CFG rules for the $\CALL$ instruction.

For the precise treatment of dependencies, the other state updates are treated differently depending on the preprocessing information. 
We always need to consider all possible effects on the state. 
More precisely, the effect on the local return memory fraction, on the return value $y$, on the external environment, on the gas, and on the active words and memory. 
The effects on the return value and the external environment depend on the same variables (which determine the overall outcome of the call):
\begin{itemize}
\item the arguments to the call 
\item the input memory fragment (as specified by the arguments to the call)
\item the current amount of gas available 
\item the global environment (including the state of all other accounts and all globally accessible values) 
\item the global variables of the contract (since those may be read during reentrancy)
\end{itemize}

The return memory fragment after the call, in addition, may depend on the previous values in this memory fraction. 
This is, because in the case that the call was unsuccessful (returned with an exception), the return memory fragment stays unchanged.

The gas value after the call (in addition to the call outcome that influences the amount of gas refunded) depends, also, on the active words in memory (since this influences the costs for memory access for writing to the return memory fragment).

Finally, the active words in memory only depend on the location and size of the input and return memory fragment and the previous number of active words in memory. 

We give rules for these four different forms of dependencies. 
To track the dependencies precisely, we first write the updated values into the corresponding temporal variables and only update the original variables later. 
This is required so that we can use the $\applycall$ function on the original state to obtain the updated values separately.
Since we are considering purely functional state updates, we cannot simply save the original result of the $\applycall$ function, but need to recompute it at every node in order to obtain the needed values. 
To denote that we are updating the full CFG state, we write the state update function as $\fun{\cfgstatefull}(\langle \textit{exp} \rangle)$.
To refer to the restriction of the $\cfgstatefull$ to the non-temporal variables, we write $\restrict{\cfgstatefull}{\cfgstate}$. 


We first define the rules for the updates of $y$ and the external environment. 
These rules are different because the $\STATICCALL$ does not change the static environment. 

\begin{mathpar}
    \infer{
    \precontract(\pc) = (\CALL(\stackvar{y}, \stackvar{\lgas}, \stackvar{\recipient}, \stackvar{\valu}, \stackvar{\io}, \stackvar{\is}, \stackvar{\oo}, \stackvar{\os}), \pc', \pre) \land \startmemoffset = 4 \\
        \pre[\startmemoffset] = \optional{\loc_\io} \\
        \pre[\startmemoffset+1] = \optional{\loc_\is} \\
        f_1 = \fun{\cfgstatefull}{\updatestate{\cfgstatefull}{\tempvar{\stackvar{y}}}
        {\accesscfgstate{\applycall(\restrict{\cfgstatefull}{\cfgstate}, \precontract, \pc)}{\stackvar{y}}}} \\
        f_2 = \fun{\cfgstatefull}{\updatestate{\cfgstatefull}{\tempvar{\globenvvar{\cfgexternalpc{\pc'}}}}
        {\accesscfgstate{\applycall(\restrict{\cfgstatefull}{\cfgstate}, \precontract, \pc)}{\globenvvar{\cfgexternalpc{\pc'}}}}}  \\
        f = \fun{\cfgstatefull}{f_2(f_1(\cfgstatefull))}\\
    }
    {\edgestep{\transenv}{\cfgnode{\pc}{0}}{\cfgnode{\pc}{1}}} \\
    {\Def{  \{ \tempvar{\stackvar{y}}, \tempvar{\globenvvar{\cfgexternalpc{\pc'}}}\}}} \\
        \Use{
        \{\memvarconc{\loc}{\pc} ~|~ \loc \in [\loc_\io, \loc_\io + \loc_\is - 1] \}
            ~\cup~ \{\memvarabs{\loc}{\pc} ~|~ \loc \in [\loc_\io,\loc_\io + \loc_\is - 1] \} \\
         ~\cup~ \{ \stackvar{\lgas}, \stackvar{\recipient}, \stackvar{\valu}, \locenvvar{\gaspc{\pc}}, \locenvvar{\activeaccount} \}
         ~\cup~ \globenvvar{X}~\cup~ \storvar{X}{\pc}
         ~\cup \{ \stackvar{\oo} ~|~ \pre[\startmemoffset+2] = \optional{\loc_\oo}  \}  ~\cup \{ \stackvar{\os} ~|~ \pre[\startmemoffset+3] = \optional{\loc_\os}  \} 
        } 
\\
    \infer{
        \precontract(\pc) = (\CALL(\stackvar{y}, \stackvar{\lgas}, \stackvar{\recipient}, \stackvar{\valu}, \stackvar{\io}, \stackvar{\is}, \stackvar{\oo}, \stackvar{\os}), \pc', \pre) \land \startmemoffset = 4 \\
        (\pre[\startmemoffset] = \bot ~\lor~ \pre[\startmemoffset+1] = \bot) \\ 
        f_1 = \fun{\cfgstatefull}{\updatestate{\cfgstatefull}{\tempvar{\stackvar{y}}}
        {\accesscfgstate{\applycall(\restrict{\cfgstatefull}{\cfgstate}, \precontract, \pc)}{\stackvar{y}}}} \\
        f_2 = \fun{\cfgstatefull}{\updatestate{\cfgstatefull}{\tempvar{\globenvvar{\cfgexternalpc{\pc'}}}}
        {\accesscfgstate{\applycall(\restrict{\cfgstatefull}{\cfgstate}, \precontract, \pc)}{\globenvvar{\cfgexternalpc{\pc'}}}}}  \\
        f = \fun{\cfgstatefull}{f_2(f_1(\cfgstatefull))}\\
    }
    {\edgestep{\transenv}{\cfgnode{\pc}{0}}{\cfgnode{\pc}{1}}} \\
    {\Def{  \{ \tempvar{\stackvar{y}}, \tempvar{\globenvvar{\cfgexternalpc{\pc'}}}\}}} \\
        \Use{ \memvarconc{X}{\pc} ~\cup~ \memvarabs{X}{\pc}
         ~\cup~ \{ \stackvar{\lgas}, \stackvar{\recipient}, \stackvar{\valu}, \locenvvar{\gaspc{\pc}}, \locenvvar{\activeaccount} \} 
         ~\cup~ \globenvvar{X}~\cup~ \storvar{X}{\pc} \\
         ~\cup \{ \stackvar{\io} ~|~ \pre[\startmemoffset] = \optional{\loc_\io}  \}  ~\cup \{ \stackvar{\is} ~|~ \pre[\startmemoffset+1] = \optional{\loc_\is}  \} 
         ~\cup \{ \stackvar{\oo} ~|~ \pre[\startmemoffset+2] = \optional{\loc_\oo}  \}  ~\cup \{ \stackvar{\os} ~|~ \pre[\startmemoffset+3] = \optional{\loc_\os}  \} 
        } 
\end{mathpar}

\begin{mathpar}
    \infer{
         \precontract(\pc) = (\STATICCALL(\stackvar{y}, \stackvar{\lgas}, \stackvar{\recipient}, \stackvar{\io}, \stackvar{\is}, \stackvar{\oo}, \stackvar{\os}), \pc', \pre)
         \land \startmemoffset = 3 \\
        \pre[\startmemoffset] = \optional{\loc_\io} \\
        \pre[\startmemoffset+1] = \optional{\loc_\is} \\
        f_1 = \fun{\cfgstatefull}{\updatestate{\cfgstatefull}{\tempvar{\stackvar{y}}}
        {\accesscfgstate{\applycall(\restrict{\cfgstatefull}{\cfgstate}, \precontract, \pc)}{\stackvar{y}}}} \\
        f_2 = \fun{\cfgstatefull}{\updatestate{\cfgstatefull}{\tempvar{\globenvvar{\cfgexternalpc{\pc'}}}}
        {\accesscfgstate{\applycall(\restrict{\cfgstatefull}{\cfgstate}, \precontract, \pc)}{\globenvvar{\cfgexternalpc{\pc'}}}}}  \\
        f = \fun{\cfgstatefull}{f_2(f_1(\cfgstatefull))}\\
    }
    {\edgestep{\transenv}{\cfgnode{\pc}{0}}{\cfgnode{\pc}{1}}} \\
    {\Def{  \{ \tempvar{\stackvar{y}} \}}} \\
        \Use{
        \{\memvarconc{\loc}{\pc} ~|~ \loc \in [\loc_\io, \loc_\io + \loc_\is - 1] \}
            ~\cup~ \{\memvarabs{\loc}{\pc} ~|~ \loc \in [\loc_\io,\loc_\io + \loc_\is - 1] \} \\
         ~\cup~ \{ \stackvar{\lgas}, \stackvar{\recipient}, \stackvar{\valu}, \locenvvar{\gaspc{\pc}}, \locenvvar{\activeaccount} \}
         ~\cup~ \globenvvar{X}~\cup~ \storvar{X}{\pc}
         ~\cup \{ \stackvar{\oo} ~|~ \pre[\startmemoffset+2] = \optional{\loc_\oo}  \}  ~\cup \{ \stackvar{\os} ~|~ \pre[\startmemoffset+3] = \optional{\loc_\os}  \} 
        } 
\\
    \infer{
         \precontract(\pc) = (\STATICCALL(\stackvar{y}, \stackvar{\lgas}, \stackvar{\recipient}, \stackvar{\io}, \stackvar{\is}, \stackvar{\oo}, \stackvar{\os}), \pc', \pre)
         \land \startmemoffset = 3 \\
        (\pre[\startmemoffset] = \bot ~\lor~ \pre[\startmemoffset+1] = \bot) \\ 
        f_1 = \fun{\cfgstatefull}{\updatestate{\cfgstatefull}{\tempvar{\stackvar{y}}}
        {\accesscfgstate{\applycall(\restrict{\cfgstatefull}{\cfgstate}, \precontract, \pc)}{\stackvar{y}}}} \\
        f_2 = \fun{\cfgstatefull}{\updatestate{\cfgstatefull}{\tempvar{\globenvvar{\cfgexternalpc{\pc'}}}}
        {\accesscfgstate{\applycall(\restrict{\cfgstatefull}{\cfgstate}, \precontract, \pc)}{\globenvvar{\cfgexternalpc{\pc'}}}}}  \\
        f = \fun{\cfgstatefull}{f_2(f_1(\cfgstatefull))}\\
    }
    {\edgestep{\transenv}{\cfgnode{\pc}{0}}{\cfgnode{\pc}{1}}} \\
    {\Def{  \{ \tempvar{\stackvar{y}} \}}} \\
        \Use{ \memvarconc{X}{\pc} ~\cup~ \memvarabs{X}{\pc}
         ~\cup~ \{ \stackvar{\lgas}, \stackvar{\recipient}, \stackvar{\valu}, \locenvvar{\gaspc{\pc}}, \locenvvar{\activeaccount} \} 
         ~\cup~ \globenvvar{X}~\cup~ \storvar{X}{\pc} \\
         ~\cup \{ \stackvar{\io} ~|~ \pre[\startmemoffset] = \optional{\loc_\io}  \}  ~\cup \{ \stackvar{\is} ~|~ \pre[\startmemoffset+1] = \optional{\loc_\is}  \} 
         ~\cup \{ \stackvar{\oo} ~|~ \pre[\startmemoffset+2] = \optional{\loc_\oo}  \}  ~\cup \{ \stackvar{\os} ~|~ \pre[\startmemoffset+3] = \optional{\loc_\os}  \} 
        } 
\end{mathpar}

Next, we define the rules for the update of the gas value:

\begin{mathpar}
        \infer{
            (\precontract(\pc) = (\CALL(\stackvar{y}, \stackvar{\lgas}, \stackvar{\recipient}, \stackvar{\valu}, \stackvar{\io}, \stackvar{\is}, \stackvar{\oo}, \stackvar{\os}), \pc', \pre) \land \startmemoffset = 4 \\
            \lor \ \precontract(\pc) = (\STATICCALL(\stackvar{y}, \stackvar{\lgas}, \stackvar{\recipient}, \stackvar{\io}, \stackvar{\is}, \stackvar{\oo}, \stackvar{\os}), \pc', \pre)
            \land \startmemoffset = 3) \\
                \pre[\startmemoffset] = \optional{\loc_\io} \\
                \pre[\startmemoffset+1] = \optional{\loc_\is} \\
                f = \fun{\cfgstatefull}{\updatestate{\cfgstatefull}{\tempvar{\locenvvar{\gaspc{\pc'}}}}{\accesscfgstate{\applycall(\restrict{\cfgstatefull}{\cfgstate}, \precontract, \pc)}{\locenvvar{\gaspc{\pc'}}}}}
             }
            {\edgestep{\transenv}{\cfgnode{\pc}{1}}{\cfgnode{\pc}{2}}}
            \\
            {\Defs{\locenvvar{\tempvar{\gaspc{\pc'}}}}} \\
            \Use{
                \{\memvarconc{\loc}{\pc} ~|~ \loc \in [\loc_\io, \loc_\io + \loc_\is - 1] \}
                ~\cup~ \{\memvarabs{\loc}{\pc} ~|~ \loc \in [\loc_\io,\loc_\io + \loc_\is - 1] \} \\
             ~\cup~ \{ \stackvar{\lgas}, \stackvar{\recipient}, \stackvar{\valu}, \locenvvar{\gaspc{\pc}}, \locenvvar{\msizepc{\pc}}, \locenvvar{\activeaccount} \}~\cup~ \globenvvar{X}~\cup~ \storvar{X}{\pc}
             ~\cup \{ \stackvar{\oo} ~|~ \pre[\startmemoffset+2] = \optional{\loc_\oo}  \}  ~\cup \{ \stackvar{\os} ~|~ \pre[\startmemoffset+3] = \optional{\loc_\os}  \} 
            } 
            \\
         \infer{
            (\precontract(\pc) = (\CALL(\stackvar{y}, \stackvar{\lgas}, \stackvar{\recipient}, \stackvar{\valu}, \stackvar{\io}, \stackvar{\is}, \stackvar{\oo}, \stackvar{\os}), \pc', \pre) \land \startmemoffset = 4 \\
            \lor \ \precontract(\pc) = (\STATICCALL(\stackvar{y}, \stackvar{\lgas}, \stackvar{\recipient}, \stackvar{\io}, \stackvar{\is}, \stackvar{\oo}, \stackvar{\os}), \pc', \pre)
            \land \startmemoffset = 3) \\
                (\pre[\startmemoffset] = \bot ~\lor~  \pre[\startmemoffset+1] = \bot) \\
                f = \fun{\cfgstatefull}{\updatestate{\cfgstatefull}{\tempvar{\locenvvar{\gaspc{\pc'}}}}{\accesscfgstate{\applycall(\restrict{\cfgstatefull}{\cfgstate}, \precontract, \pc)}{\locenvvar{\gaspc{\pc'}}}}}
             }
            {\edgestep{\transenv}{\cfgnode{\pc}{1}}{\cfgnode{\pc}{2}}}
            \\
            {\Defs{\tempvar{\locenvvar{\gaspc{\pc'}}}}} \\
            \Use{
                \memvarconc{X}{\pc} ~\cup~ \memvarabs{X}{\pc}
             ~\cup~ \{ \stackvar{\lgas}, \stackvar{\recipient}, \stackvar{\valu}, \locenvvar{\gaspc{\pc}}, \locenvvar{\msizepc{\pc}}, \locenvvar{\activeaccount} \}~\cup~ \globenvvar{X}~\cup~ \storvar{X}{\pc} \\
             ~\cup \{ \stackvar{\io} ~|~ \pre[\startmemoffset] = \optional{\loc_\io}  \}  ~\cup \{ \stackvar{\is} ~|~ \pre[\startmemoffset+1] = \optional{\loc_\is}  \} 
             ~\cup \{ \stackvar{\oo} ~|~ \pre[\startmemoffset+2] = \optional{\loc_\oo}  \}  ~\cup \{ \stackvar{\os} ~|~ \pre[\startmemoffset+3] = \optional{\loc_\os}  \} 
            }
\end{mathpar}

Next, we give the rule for the update of the active words in memory: 
\begin{mathpar}
            \infer{
                (\precontract(\pc) = (\CALL(\stackvar{y}, \stackvar{\lgas}, \stackvar{\recipient}, \stackvar{\valu}, \stackvar{\io}, \stackvar{\is}, \stackvar{\oo}, \stackvar{\os}), \pc', \pre) \land \startmemoffset = 4 \\
                \lor \ \precontract(\pc) = (\STATICCALL(\stackvar{y}, \stackvar{\lgas}, \stackvar{\recipient}, \stackvar{\io}, \stackvar{\is}, \stackvar{\oo}, \stackvar{\os}), \pc', \pre)
                \land \startmemoffset = 3) \\
                f = \fun{\cfgstatefull}{\updatestate{\cfgstatefull}{\tempvar{\locenvvar{\msizepc{\pc'}}}}{\accesscfgstate{\applycall(\restrict{\cfgstatefull}{\cfgstate}, \precontract, \pc)}{\locenvvar{\msizepc{\pc'}}}}}
             }
            {\edgestep{\transenv}{\cfgnode{\pc}{2}}{\cfgnode{\pc}{3}}}
            \\
            {\Defs{\tempvar{\locenvvar{\msizepc{\pc'}}}}} \\
            \Use{ \{ \locenvvar{\msizepc{\pc}}\} 
            ~\cup \{ \stackvar{\io} ~|~ \pre[\startmemoffset] = \optional{\loc_\io}  \}  ~\cup \{ \stackvar{\is} ~|~ \pre[\startmemoffset+1] = \optional{\loc_\is}  \} 
            ~\cup \{ \stackvar{\oo} ~|~ \pre[\startmemoffset+2] = \optional{\loc_\oo}  \}  ~\cup \{ \stackvar{\os} ~|~ \pre[\startmemoffset+3] = \optional{\loc_\os}  \} 
            }
\end{mathpar}

Finally, we give the rules for the memory update. 
These rules are the most interesting ones since they differ heavily depending on the available pre-processing information.

We first consider the case that the input and the result memory fragment are known. 
In this case, the values of memory locations $\memvarconc{i}{\pc'}$ are assigned in individual nodes. 
The node splitting, in this case, allows for precise treatment, since it only needs to be considered that the value at memory location $i$ may depend on the input to the call or the previous value at exactly this location $i$. 
For assigning the $\bot$ value to dynamic memory locations $\memvarabs{i}{\pc'}$, this distinction is not needed since no dependencies are propagated in the first place. 
\begin{mathpar}
        \infer{
            (\precontract(\pc) = (\CALL(\stackvar{y}, \stackvar{\lgas}, \stackvar{\recipient}, \stackvar{\valu}, \stackvar{\io}, \stackvar{\is}, \stackvar{\oo}, \stackvar{\os}), \pc', \pre) \land \startmemoffset = 4 \\
            \lor \ \precontract(\pc) = (\STATICCALL(\stackvar{y}, \stackvar{\lgas}, \stackvar{\recipient}, \stackvar{\io}, \stackvar{\is}, \stackvar{\oo}, \stackvar{\os}), \pc', \pre)
            \land \startmemoffset = 3) \\
            \pre[\startmemoffset] = \optional{\loc_\io} \\
            \pre[\startmemoffset+1] = \optional{\loc_\is} \\
            \pre[\startmemoffset+2] = \optional{\loc_\oo} \\
            \pre[\startmemoffset+3] = \optional{\loc_\os} \\
            f = \fun{\cfgstatefull}{ \updatestate
            {\cfgstatefull}
            {\tempvar{\memvarconc{i}{\pc'}}}
            {\load~\applycall(\restrict{\cfgstatefull}{\cfgstate},\precontract,\pc)~\memvar{i}{\pc'}}} \\
            i \in [\loc_\oo, \loc_\oo + \loc_\os - 1]\\
        }
        {\edgestep{\transenv}{\cfgnode{\pc}{3 + (i - \loc_\oo)}}{\cfgnode{\pc}{3 + (i - \loc_\oo) + 1}}} \\
        {\Def{ \{\tempvar{\memvarconc{i}{\pc'}} \}} } \\
            \Use{
            \{\memvarconc{\loc}{\pc} ~|~ \loc \in [\loc_\io, \loc_\io + \loc_\is - 1] \} \\
                ~\cup~ \{\memvarabs{\loc}{\pc} ~|~ \loc \in [\loc_\io,\loc_\io + \loc_\is - 1] \} \\
             ~\cup~ \{ \stackvar{\lgas}, \stackvar{\recipient}, \stackvar{\valu}, \locenvvar{\gaspc{\pc}}, \locenvvar{\activeaccount}  \}~\cup~ \globenvvar{X}~\cup~ \storvar{X}{\pc}
             ~\cup~ \{ \memvarconc{i}{\pc}, \memvarabs{i}{\pc} \} 
            } \\
            \infer{
                (\precontract(\pc) = (\CALL(\stackvar{y}, \stackvar{\lgas}, \stackvar{\recipient}, \stackvar{\valu}, \stackvar{\io}, \stackvar{\is}, \stackvar{\oo}, \stackvar{\os}), \pc', \pre) \land \startmemoffset = 4 \\
                \lor \ \precontract(\pc) = (\STATICCALL(\stackvar{y}, \stackvar{\lgas}, \stackvar{\recipient}, \stackvar{\io}, \stackvar{\is}, \stackvar{\oo}, \stackvar{\os}), \pc', \pre)
                \land \startmemoffset = 3) \\
                \pre[\startmemoffset] = \optional{\loc_\io} \\
                \pre[\startmemoffset+1] = \optional{\loc_\is} \\
                \pre[\startmemoffset+2] = \optional{\loc_\oo} \\
                \pre[\startmemoffset+3] = \optional{\loc_\os} \\
                f = \fun{\cfgstatefull}{\updatevarset{\cfgstatefull}{\tempvar{\memvarabs{x}{\pc'}}}{\bot}{x}{[\memvar{\loc_\oo}{}, \memvar{\loc_\oo}{} + \memvar{\loc_\os}{} - 1]}}
            }
            {\edgestep{\transenv}{\cfgnode{\pc}{3 + \loc_\os}}{\cfgnode{\pc}{4+\loc_\os}}} \\
            {\Def{  \{\tempvar{\memvarabs{\loc}{\pc'}} ~|~ \loc \in [\loc_\oo, \loc_\oo + \loc_\os - 1] \} }} \\
            \Use{\emptyset}
\end{mathpar}

After updating the memory, it is still required to set the values of all updated variables to the corresponding temporal variables and to set the temporal variables back to $\bot$.
To determine the right offset for the different cases, we define a function $\getcallnodeoffset$ that given the precomputed memory fragments outputs the corresponding node offsets.
More precisely, it computes the number of intermediate nodes required for the node splitting in the different cases. 

\begin{align*}
    \getcallnodeoffset(\io, \is, \oo, \os) =
    \begin{cases}
        \loc_\os & \io = \some{\loc_\io} ~\land~ \is = \some{\loc_{\is}} ~\land~  \oo = \some{\loc_\oo}~\land~\os = \some{\loc_\os} \\
        \maxint{256} & \io = \some{\loc_\io} ~\land~ \is = \some{\loc_{\is}} ~\land~  (\oo = \bot~\lor~\os = \bot) \\
        1 & (\io = \bot ~\lor~  \is = \bot)  ~\land~ \oo = \optional{\loc_\oo} ~\land \os = \optional{\loc_\os}  \\
        0 & \textit{otherwise}
    \end{cases}
\end{align*}


We give the rules for the stack/external environment, gas and active words in memory only once, since they are the same for all cases (for different node offsets).

\begin{mathpar}
    \infer
    {
        (\precontract(\pc) = (\CALL(\stackvar{y}, \stackvar{\lgas}, \stackvar{\recipient}, \stackvar{\valu}, \stackvar{\io}, \stackvar{\is}, \stackvar{\oo}, \stackvar{\os}), \pc', \pre) \land \startmemoffset = 4 \\
         \lor \ \precontract(\pc) = (\STATICCALL(\stackvar{y}, \stackvar{\lgas}, \stackvar{\recipient}, \stackvar{\io}, \stackvar{\is}, \stackvar{\oo}, \stackvar{\os}), \pc', \pre)
         \land \startmemoffset = 3) \\
        \restartcounter = 4 + \getcallnodeoffset(\pre[\startmemoffset], \pre[\startmemoffset+1], \pre[\startmemoffset+2], \pre[\startmemoffset+3]) \\
        f = \fun{\cfgstatefull}{ 
            \updatestate{(
                \updatestate
                {\cfgstatefull}
                {\stackvar{y}}
                {\accesscfgstate{\cfgstatefull}{\tempvar{\stackvar{y}}}})}
            {\envvar{\extenv}}
            {\accesscfgstate{\cfgstatefull}{\tempvar{\envvar{\extenv}}}} 
            }\\
    }
    {
        \edgestep{\transenv}{\cfgnode{\pc}{\restartcounter}}{\cfgnode{\pc}{\restartcounter+1}}
    } \\
    {\Def{ \{ \stackvar{y}, \envvar{\extenv}} \} } \\
    \Use{ \{ \tempvar{\stackvar{y}}, \tempvar{\envvar{\extenv}}\} }
\end{mathpar}

\begin{mathpar}
    \infer
    {
        (\precontract(\pc) = (\CALL(\stackvar{y}, \stackvar{\lgas}, \stackvar{\recipient}, \stackvar{\valu}, \stackvar{\io}, \stackvar{\is}, \stackvar{\oo}, \stackvar{\os}), \pc', \pre) \land \startmemoffset = 4 \\
         \lor \ \precontract(\pc) = (\STATICCALL(\stackvar{y}, \stackvar{\lgas}, \stackvar{\recipient}, \stackvar{\io}, \stackvar{\is}, \stackvar{\oo}, \stackvar{\os}), \pc', \pre)
         \land \startmemoffset = 3) \\
        \restartcounter = 4 + \getcallnodeoffset(\pre[\startmemoffset], \pre[\startmemoffset+1], \pre[\startmemoffset+2], \pre[\startmemoffset+3]) \\
        f = \fun{\cfgstatefull}{ 
                \updatestate
                {\cfgstatefull}
                {\envvar{\gaspc{}}}
                {\accesscfgstate{\cfgstatefull}{\tempvar{\envvar{\gaspc{}}}}}}
    }
    {
        \edgestep{\transenv}{\cfgnode{\pc}{\restartcounter + 1}}{\cfgnode{\pc}{\restartcounter+2}}
    } \\
    {\Def{ \{\envvar{\gaspc{}}} \} } \\
    \Use{ \{ \tempvar{\envvar{\gaspc{}}}\} }
\end{mathpar}

\begin{mathpar}
    \infer
    {
        (\precontract(\pc) = (\CALL(\stackvar{y}, \stackvar{\lgas}, \stackvar{\recipient}, \stackvar{\valu}, \stackvar{\io}, \stackvar{\is}, \stackvar{\oo}, \stackvar{\os}), \pc', \pre) \land \startmemoffset = 4 \\
         \lor \ \precontract(\pc) = (\STATICCALL(\stackvar{y}, \stackvar{\lgas}, \stackvar{\recipient}, \stackvar{\io}, \stackvar{\is}, \stackvar{\oo}, \stackvar{\os}), \pc', \pre)
         \land \startmemoffset = 3) \\
        \restartcounter = 4 + \getcallnodeoffset(\pre[\startmemoffset], \pre[\startmemoffset+1], \pre[\startmemoffset+2], \pre[\startmemoffset+3]) \\
        f = \fun{\cfgstatefull}{ 
                \updatestate
                {\cfgstatefull}
                {\envvar{\msizepc{}}}
                {\accesscfgstate{\cfgstatefull}{\tempvar{\envvar{\msizepc{}}}}}}
    }
    {
        \edgestep{\transenv}{\cfgnode{\pc}{\restartcounter + 2}}{\cfgnode{\pc}{\restartcounter+3}}
    } \\
    {\Def{ \{\envvar{\msizepc{}}} \} } \\
    \Use{ \{ \tempvar{\envvar{\msizepc{}}}\} }
\end{mathpar}

Next, the values of the updated memory locations need to be carried over one by one. 
The number of nodes needed for that again depends on the pre-computed values for the input and output memory fraction.

\begin{mathpar}
    \infer{
        (\precontract(\pc) = (\CALL(\stackvar{y}, \stackvar{\lgas}, \stackvar{\recipient}, \stackvar{\valu}, \stackvar{\io}, \stackvar{\is}, \stackvar{\oo}, \stackvar{\os}), \pc', \pre) \land \startmemoffset = 4 \\
         \lor \ \precontract(\pc) = (\STATICCALL(\stackvar{y}, \stackvar{\lgas}, \stackvar{\recipient}, \stackvar{\io}, \stackvar{\is}, \stackvar{\oo}, \stackvar{\os}), \pc', \pre)
         \land \startmemoffset = 3) \\
        \restartcounter = 4 + \getcallnodeoffset(\pre[\startmemoffset], \pre[\startmemoffset+1], \pre[\startmemoffset+2], \pre[\startmemoffset+3]) \\
        \pre[\startmemoffset+2] = \optional{\loc_\oo} \\
        \pre[\startmemoffset+3] = \optional{\loc_\os} \\
        f = \fun{\cfgstatefull}{ \updatestate
        {\cfgstatefull}
        {\memvarconc{i}{\pc'}}
        {\accesscfgstate{\cfgstatefull}{\tempvar{\memvarconc{i}{\pc'}}}}} \\
        i \in [\loc_\oo, \loc_\oo + \loc_\os - 1]\\
    }
    {\edgestep{\transenv}{\cfgnode{\pc}{\restartcounter + 3 + (i - \loc_\oo)}}{\cfgnode{\pc}{\restartcounter + 3 + (i - \loc_\oo) + 1}}} \\
    {\Def{ \{\memvarconc{i}{\pc'} \}} } \\
        \Use{ \{\tempvar{\memvarconc{i}{\pc'}} \} } \\
    \infer{
        (\precontract(\pc) = (\CALL(\stackvar{y}, \stackvar{\lgas}, \stackvar{\recipient}, \stackvar{\valu}, \stackvar{\io}, \stackvar{\is}, \stackvar{\oo}, \stackvar{\os}), \pc', \pre) \land \startmemoffset = 4 \\
         \lor \ \precontract(\pc) = (\STATICCALL(\stackvar{y}, \stackvar{\lgas}, \stackvar{\recipient}, \stackvar{\io}, \stackvar{\is}, \stackvar{\oo}, \stackvar{\os}), \pc', \pre)
         \land \startmemoffset = 3) \\
            \restartcounter = 4 + \getcallnodeoffset(\pre[\startmemoffset], \pre[\startmemoffset+1], \pre[\startmemoffset+2], \pre[\startmemoffset+3]) \\
            \pre[\startmemoffset+2] = \optional{\loc_\oo} \\
            \pre[\startmemoffset+3] = \optional{\loc_\os} \\
            f = \fun{\cfgstatefull}{\updatevarset{\cfgstatefull}{\memvarabs{x}{\pc'}}{\accesscfgstate{\cfgstatefull}{\tempvar{\memvarabs{x}{\pc'}}}}{x}{[\memvar{\loc_\oo}{}, \memvar{\loc_\oo}{} + \memvar{\loc_\os}{} - 1]}}
        }
        {\edgestep{\transenv}{\cfgnode{\pc}{\restartcounter + 3 + \loc_\os}}{\cfgnode{\pc}{\restartcounter + \loc_\os + 4}}} \\
        {\Def{  \{\memvarabs{\loc}{\pc'} ~|~ \loc \in [\loc_\oo, \loc_\oo + \loc_\os - 1] \} }} \\
            \Use{  \{\tempvar{\memvarabs{\loc}{\pc'}} ~|~ \loc \in [\loc_\oo, \loc_\oo + \loc_\os - 1] \} }\\
\end{mathpar}

Finally, all temporal variables are set to $\bot$. 
This is the same for all cases, irrespective of the pre-computed values for the input and output memory fraction.

\begin{mathpar}
\infer
{
    (\precontract(\pc) = (\CALL(\stackvar{y}, \stackvar{\lgas}, \stackvar{\recipient}, \stackvar{\valu}, \stackvar{\io}, \stackvar{\is}, \stackvar{\oo}, \stackvar{\os}), \pc', \pre) \land \startmemoffset = 4 \\
    \lor \ \precontract(\pc) = (\STATICCALL(\stackvar{y}, \stackvar{\lgas}, \stackvar{\recipient}, \stackvar{\io}, \stackvar{\is}, \stackvar{\oo}, \stackvar{\os}), \pc', \pre)
    \land \startmemoffset = 3) \\
    \restartcounter = 2 * (4 + \getcallnodeoffset(\pre[\startmemoffset], \pre[\startmemoffset+1], \pre[\startmemoffset+2], \pre[\startmemoffset+3])) \\
    f = \fun{\cfgstatefull}{\updatevarset{\cfgstatefull}{\tempvar{x}}{\bot}{\tempvar{x}}{\domain{\cfgstatecopy}}}
}
{\edgestep{\transenv}{\cfgnode{\pc}{\restartcounter}}{\cfgnode{\pc'}{0}}} \\
{\Def{\domain{\cfgstatecopy}}} \\
\Use{ \emptyset} 
\end{mathpar}

In the case that only the input memory is known, the dependencies need to be propagated to the whole memory (as potential output). 
Note that in this case, we still gain precision by node splitting since, otherwise, we would need to propagate the dependencies of the whole memory to the whole memory again. 
By node splitting, we ensure that only the dependencies of the input memory fragment are propagated to the whole memory.

\begin{mathpar}
    \infer{
        (\precontract(\pc) = (\CALL(\stackvar{y}, \stackvar{\lgas}, \stackvar{\recipient}, \stackvar{\valu}, \stackvar{\io}, \stackvar{\is}, \stackvar{\oo}, \stackvar{\os}), \pc', \pre) \land \startmemoffset = 4 \\
         \lor \ \precontract(\pc) = (\STATICCALL(\stackvar{y}, \stackvar{\lgas}, \stackvar{\recipient}, \stackvar{\io}, \stackvar{\is}, \stackvar{\oo}, \stackvar{\os}), \pc', \pre)
         \land \startmemoffset = 3) \\
        \pre[\startmemoffset] = \optional{\loc_\io} \\
        \pre[\startmemoffset+1] = \optional{\loc_\is} \\
        (\pre[\startmemoffset+2] = \bot ~\lor~ \pre[\startmemoffset+3] = \bot) \\
        f = \fun{\cfgstatefull}{ \updatestate
        {\cfgstatefull}
        {\tempvar{\memvarabs{i}{\pc'}}}
        {\load~\applycall(\restrict{\cfgstatefull}{\cfgstate},\precontract,\pc)~\memvar{i}{\pc'}}} \\
        i \in \integer{256}
    }
    {\edgestep{\transenv}{\cfgnode{\pc}{3 + i }}{\cfgnode{\pc}{3 + i + 1}}} \\
    {\Def{ \{\tempvar{\memvarabs{i}{\pc'}} \} }} \\
        \Use{
        \{\memvarconc{\loc}{\pc} ~|~ \loc \in [\loc_\io, \loc_\io + \loc_\is - 1] \}
            ~\cup~ \{\memvarabs{\loc}{\pc} ~|~ \loc \in [\loc_\io,\loc_\io + \loc_\is - 1] \} \\
         ~\cup~ \{ \stackvar{\lgas}, \stackvar{\recipient}, \stackvar{\valu}, \locenvvar{\gaspc{\pc}}, \locenvvar{\activeaccount}  \}~\cup~ \globenvvar{X} ~\cup~ \storvar{X}{\pc}
         ~\cup~ \{ \memvarconc{i}{\pc}, \memvarabs{i}{\pc} \} 
         ~\cup \{ \stackvar{\oo} ~|~ \pre[\startmemoffset+2] = \optional{\loc_\oo}  \}  ~\cup \{ \stackvar{\os} ~|~ \pre[\startmemoffset+3] = \optional{\loc_\os}  \} 
        } 
\end{mathpar}

Again, the values of the temporary memory variables need to be one-by-one written to the non-temporal variables:

\begin{mathpar}
    \infer{
        (\precontract(\pc) = (\CALL(\stackvar{y}, \stackvar{\lgas}, \stackvar{\recipient}, \stackvar{\valu}, \stackvar{\io}, \stackvar{\is}, \stackvar{\oo}, \stackvar{\os}), \pc', \pre) \land \startmemoffset = 4 \\
         \lor \ \precontract(\pc) = (\STATICCALL(\stackvar{y}, \stackvar{\lgas}, \stackvar{\recipient}, \stackvar{\io}, \stackvar{\is}, \stackvar{\oo}, \stackvar{\os}), \pc', \pre)
         \land \startmemoffset = 3) \\
        \restartcounter = 4 + \getcallnodeoffset(\pre[\startmemoffset], \pre[\startmemoffset+1], \pre[\startmemoffset+2], \pre[\startmemoffset+3]) \\
        \pre[\startmemoffset] = \optional{\loc_\io} \\
        \pre[\startmemoffset+1] = \optional{\loc_\is} \\
        (\pre[\startmemoffset+2] = \bot ~\lor~ \pre[\startmemoffset+3] = \bot) \\
        f = \fun{\cfgstatefull}{ \updatestate
        {\cfgstatefull}
        {\memvarabs{i}{\pc'}}
        {\accesscfgstate{\cfgstatefull}{\tempvar{\memvarabs{i}{}}}}} \\
        i \in \integer{256}
    }
    { 
        \edgestep{\transenv}{\cfgnode{\pc}{\restartcounter + 3 + i }}{\cfgnode{\pc}{\restartcounter + 3 + i + 1}}
        } \\
    {\Def{ \{ \memvarabs{i}{\pc'} \} }} \\
        \Use{ \{ \tempvar{{\memvarabs{i}{\pc'}}}\} 
        } 
\end{mathpar}

In the case that only the output memory is known, the dependencies from the whole memory need to be propagated, but only to a small memory fraction. 
In this scenario, node splitting does not help since anyway each affected memory node gets already the dependencies from the whole memory assigned:

\begin{mathpar}
    \infer{
        (\precontract(\pc) = (\CALL(\stackvar{y}, \stackvar{\lgas}, \stackvar{\recipient}, \stackvar{\valu}, \stackvar{\io}, \stackvar{\is}, \stackvar{\oo}, \stackvar{\os}), \pc', \pre) \land \startmemoffset = 4 \\
         \lor \ \precontract(\pc) = (\STATICCALL(\stackvar{y}, \stackvar{\lgas}, \stackvar{\recipient}, \stackvar{\io}, \stackvar{\is}, \stackvar{\oo}, \stackvar{\os}), \pc', \pre)
         \land \startmemoffset = 3) \\
        (\pre[\startmemoffset] = \bot ~\lor~  \pre[\startmemoffset+1] = \bot) \\
        \pre[\startmemoffset+2] = \optional{\loc_\oo} \\
        \pre[\startmemoffset+3] = \optional{\loc_\os} \\
        f = \fun{\cfgstatefull}{ \updatevarset
        {\cfgstatefull}
        {\tempvar{\memvarconc{i}{\pc'}}}
        {\load~\applycall(\restrict{\cfgstatefull}{\cfgstate},\precontract,\pc)~\memvar{i}{\pc'}}
        {i}
        {[\loc_\oo, \loc_\oo + \loc_\os - 1]}}
    }
    {\edgestep{\transenv}{\cfgnode{\pc}{3}}{\cfgnode{\pc}{4}}} \\
    {\Def{ \{\tempvar{\memvarconc{\loc}{\pc'}} ~|~ \loc \in [\loc_\oo, \loc_\oo + \loc_\os - 1] \} }} \\
        \Use{
            \memvarconc{X}{\pc}~\cup~ \memvarabs{X}{\pc} 
         ~\cup~ \{ \stackvar{\lgas}, \stackvar{\recipient}, \stackvar{\valu}, \locenvvar{\gaspc{\pc}}, \locenvvar{\activeaccount}  \}~\cup~ \globenvvar{X}~\cup~ \storvar{X}{\pc}
         ~\cup \{ \stackvar{\io} ~|~ \pre[\startmemoffset] = \optional{\loc_\io}  \}  ~\cup \{ \stackvar{\is} ~|~ \pre[\startmemoffset+1] = \optional{\loc_\is}  \} 
        } \\
        \infer{
            (\precontract(\pc) = (\CALL(\stackvar{y}, \stackvar{\lgas}, \stackvar{\recipient}, \stackvar{\valu}, \stackvar{\io}, \stackvar{\is}, \stackvar{\oo}, \stackvar{\os}), \pc', \pre) \land \startmemoffset = 4 \\
            \lor \ \precontract(\pc) = (\STATICCALL(\stackvar{y}, \stackvar{\lgas}, \stackvar{\recipient}, \stackvar{\io}, \stackvar{\is}, \stackvar{\oo}, \stackvar{\os}), \pc', \pre)
            \land \startmemoffset = 3) \\
            (\pre[\startmemoffset] = \bot ~\lor~  \pre[\startmemoffset+1] = \bot) \\
            \pre[\startmemoffset+2] = \optional{\loc_\oo} \\
            \pre[\startmemoffset+3] = \optional{\loc_\os} \\
            f = \fun{\cfgstatefull}{\updatevarset{\cfgstatefull}{\tempvar{\memvarabs{x}{\pc'}}}{\bot}{x}{[\loc_\oo,\loc_\oo + \loc_\os - 1]}}
        }
        {\edgestep{\transenv}{\cfgnode{\pc}{4}}{\cfgnode{\pc}{5}}} \\
        {\Def{  \{\tempvar{\memvarabs{\loc}{\pc'}} ~|~ \loc \in [\loc_\oo, \loc_\oo + \loc_\os - 1] \} }} \\
        \Use{\emptyset}
\end{mathpar}

Similar to the previous cases, the temporal memory variables are carried over afterward:

\begin{mathpar}
    \infer{
        (\precontract(\pc) = (\CALL(\stackvar{y}, \stackvar{\lgas}, \stackvar{\recipient}, \stackvar{\valu}, \stackvar{\io}, \stackvar{\is}, \stackvar{\oo}, \stackvar{\os}), \pc', \pre) \land \startmemoffset = 4 \\
         \lor \ \precontract(\pc) = (\STATICCALL(\stackvar{y}, \stackvar{\lgas}, \stackvar{\recipient}, \stackvar{\io}, \stackvar{\is}, \stackvar{\oo}, \stackvar{\os}), \pc', \pre)
         \land \startmemoffset = 3) \\
        \restartcounter = 4 + \getcallnodeoffset(\pre[\startmemoffset], \pre[\startmemoffset+1], \pre[\startmemoffset+2], \pre[\startmemoffset+3]) \\
        (\pre[\startmemoffset] = \bot ~\lor~  \pre[\startmemoffset+1] = \bot) \\
        \pre[\startmemoffset+2] = \optional{\loc_\oo} \\
        \pre[\startmemoffset+3] = \optional{\loc_\os} \\
        f = \fun{\cfgstatefull}{ \updatevarset
        {\cfgstatefull}
        {\memvarconc{i}{\pc'}}
        {\accesscfgstate{\cfgstatefull}{\tempvar{\memvarconc{i}{\pc'}}}}
        {i}
        {[\loc_\oo, \loc_\oo + \loc_\os - 1]}}
    }
    {\edgestep{\transenv}{\cfgnode{\pc}{\restartcounter + 3}}{\cfgnode{\pc}{\restartcounter + 4}}} \\
    {\Def{ \{\memvarconc{\loc}{\pc'} ~|~ \loc \in [\loc_\oo, \loc_\oo + \loc_\os - 1] \} }} \\
        \Use{ 
            \{\tempvar{\memvarconc{\loc}{\pc'}} ~|~ \loc \in [\loc_\oo, \loc_\oo + \loc_\os - 1] \} 
        } \\
        \infer{
            (\precontract(\pc) = (\CALL(\stackvar{y}, \stackvar{\lgas}, \stackvar{\recipient}, \stackvar{\valu}, \stackvar{\io}, \stackvar{\is}, \stackvar{\oo}, \stackvar{\os}), \pc', \pre) \land \startmemoffset = 4 \\
            \lor \ \precontract(\pc) = (\STATICCALL(\stackvar{y}, \stackvar{\lgas}, \stackvar{\recipient}, \stackvar{\io}, \stackvar{\is}, \stackvar{\oo}, \stackvar{\os}), \pc', \pre)
            \land \startmemoffset = 3) \\            \restartcounter = 4 + \getcallnodeoffset(\pre[\startmemoffset], \pre[\startmemoffset+1], \pre[\startmemoffset+2], \pre[\startmemoffset+3]) \\
            (\pre[\startmemoffset] = \bot ~\lor~  \pre[\startmemoffset+1] = \bot) \\
            \pre[\startmemoffset+2] = \optional{\loc_\oo} \\
            \pre[\startmemoffset+3] = \optional{\loc_\os} \\
            f = \fun{\cfgstatefull}{\updatevarset{\cfgstatefull}{\memvarabs{x}{\pc'}}{\accesscfgstate{\cfgstatefull}{\tempvar{\memvarabs{{x}}{\pc'}}}}{x}{[\loc_\oo,\loc_\oo + \loc_\os - 1]}}
        }
        {\edgestep{\transenv}{\cfgnode{\pc}{\restartcounter + 4}}{\cfgnode{\pc}{\restartcounter + 5}}} \\
        {\Def{  \{\memvarabs{\loc}{\pc'} ~|~ \loc \in [\loc_\oo, \loc_\oo + \loc_\os - 1] \} }} \\
        \Use{ \{\tempvar{\memvarabs{\loc}{\pc'}} ~|~ \loc \in [\loc_\oo, \loc_\oo + \loc_\os - 1] \} }
\end{mathpar}

Finally, if neither input nor result memory fraction can be determined the whole memory needs to be considered the whole memory.
We can characterize this by a single rule as follows: 

\begin{mathpar}
    \infer{
        (\precontract(\pc) = (\CALL(\stackvar{y}, \stackvar{\lgas}, \stackvar{\recipient}, \stackvar{\valu}, \stackvar{\io}, \stackvar{\is}, \stackvar{\oo}, \stackvar{\os}), \pc', \pre) \land \startmemoffset = 4 \\
         \lor \ \precontract(\pc) = (\STATICCALL(\stackvar{y}, \stackvar{\lgas}, \stackvar{\recipient}, \stackvar{\io}, \stackvar{\is}, \stackvar{\oo}, \stackvar{\os}), \pc', \pre)
         \land \startmemoffset = 3) \\
        (\pre[\startmemoffset] = \bot ~\lor~  \pre[\startmemoffset+1] = \bot) \\
        (\pre[\startmemoffset+2] = \bot ~\lor~ \pre[\startmemoffset+3] = \bot) \\
        f = \fun{\cfgstatefull}{ \updatevarset
        {\cfgstatefull}
        {\tempvar{\memvarabs{i}{\pc'}}}
        {\load~\applycall(\restrict{\cfgstatefull}{\cfgstate},\precontract,\pc)~\memvar{i}{\pc'}}
        {i}
        {[\accesscfgstate{\cfgstatefull}{\stackvar{\oo}}, \accesscfgstate{\cfgstatefull}{\stackvar{\oo}} + \accesscfgstate{\cfgstatefull}{\stackvar{\oo}} - 1]}}
    }
    {\edgestep{\transenv}{\cfgnode{\pc}{3}}{\cfgnode{\pc}{4}}} \\
    {\Def{ \tempvar{\memvarabs{X}{\pc'} }}} \\
        \Use{
        \memvarconc{X}{\pc}~\cup~ \memvarabs{X}{\pc}
         ~\cup~ \{ \stackvar{\lgas}, \stackvar{\recipient}, \stackvar{\valu}, \locenvvar{\gaspc{\pc}}, \locenvvar{\activeaccount}  \}~\cup~ \globenvvar{X}~\cup~ \storvar{X}{\pc}
         ~\cup \{ \stackvar{\io} ~|~ \pre[\startmemoffset] = \optional{\loc_\io}  \}  ~\cup \{ \stackvar{\is} ~|~ \pre[\startmemoffset+1] = \optional{\loc_\is}  \} 
         ~\cup \{ \stackvar{\oo} ~|~ \pre[\startmemoffset+2] = \optional{\loc_\oo}  \}  ~\cup \{ \stackvar{\os} ~|~ \pre[\startmemoffset+3] = \optional{\loc_\os}  \} 
        } 
\end{mathpar}

The rules for carrying over the temporal memory variables are as follows: 

\begin{mathpar}
    \infer{
        (\precontract(\pc) = (\CALL(\stackvar{y}, \stackvar{\lgas}, \stackvar{\recipient}, \stackvar{\valu}, \stackvar{\io}, \stackvar{\is}, \stackvar{\oo}, \stackvar{\os}), \pc', \pre) \land \startmemoffset = 4 \\
         \lor \ \precontract(\pc) = (\STATICCALL(\stackvar{y}, \stackvar{\lgas}, \stackvar{\recipient}, \stackvar{\io}, \stackvar{\is}, \stackvar{\oo}, \stackvar{\os}), \pc', \pre)
         \land \startmemoffset = 3) \\
        \restartcounter = 4 + \getcallnodeoffset(\pre[\startmemoffset], \pre[\startmemoffset+1], \pre[\startmemoffset+2], \pre[\startmemoffset+3]) \\
        (\pre[\startmemoffset] = \bot ~\lor~  \pre[\startmemoffset+1] = \bot) \\
        (\pre[\startmemoffset+2] = \bot ~\lor~ \pre[\startmemoffset+3] = \bot) \\
        f = \fun{\cfgstatefull}{ \updatevarset
        {\cfgstatefull}
        {\memvarabs{i}{\pc'}}
        {\accesscfgstate{\cfgstatefull}{\tempvar{\memvarabs{i}{\pc'}}}}
        {i}
        {[\accesscfgstate{\cfgstatefull}{\stackvar{\oo}}, \accesscfgstate{\cfgstatefull}{\stackvar{\oo}} + \accesscfgstate{\cfgstatefull}{\stackvar{\oo}} - 1]}}
    }
    {\edgestep{\transenv}{\cfgnode{\pc}{\restartcounter + 3}}{\cfgnode{\pc}{\restartcounter + 4}}} \\
    {\Def{\memvarabs{X}{\pc'} }} \\
        \Use{
           \tempvar{\memvarabs{X}{\pc'}}
        } 
\end{mathpar}

It is important to note that the function $\applycall(\cfgstate, \precontract, \pc)$ in all of the different rules returns the same result since only temporal variables are altered before each call of $\applycall$ (and those are not used by $\applycall(\cfgstate, \precontract, \pc)$). 

Note that we make here use of the distinction between local and global environment variables.
Intuitively, the global environment variables can be accessed by other contracts as well and hence may influence the outcome of the call. 
Consequently, they need to be included in the Use sets of all the rules applying the effects of the call to the state. 
The local environment variable $\locenvvar{\gas}$ plays a special role in that the current amount of gas may influence the amount of gas given to the call and hence also the outcome of the execution. 
For this reason, $\locenvvar{\gas}$ needs to be included in the Use set.

We give the rules for the $\CREATE$ opcode in a similar fashion:

We first devise a rule for the application of a create transaction: 

\begin{mathpar}
    \infer 
    {  
        (\precontract(\pc) = (\CREATE(\stackvar{y}, \stackvar{\valu}, \stackvar{\io}, \stackvar{\is}, \pc', \pre)) 
        \lor \precontract(\pc) = (\CREATETWO(\stackvar{y}, \stackvar{\valu}, \stackvar{\salt} \stackvar{\io}, \stackvar{\is}, \pc', \pre))\\
    (\transenv, \exstate) = \toevmstate(\cfgstate, \precontract, \pc) \\
    \transactionstep{\transenv}{\cons{\exstate}{\callstack}}{\cons{\exstate'}{\callstack}} \\
    \cfgstate' = \tocfgstate(\transenv, \exstate') \\
    }
    {\applycreate(\cfgstate, \precontract, \pc) = \cfgstate'} 
\end{mathpar}

As opposed to call instructions, create instructions do not expect a return value written to memory, but only the resulting address of the created account is written to the stack. 
However, we need to consider that in case an exception occurs, $0$ is written to the stack. 
An exception can occur if the execution of the initialization code causes an exception. 
Since the behavior of the initialization code may again depend on the environment, the value written to the stack can be dependent on the global environment. 

We, hence, can again summarize the rules for the return value $\stackvar{y}$ and the external environment.
Similar to the $\CALL$ rules, all updates are first done to temporal variables and only later transferred to the actual ones.

\begin{mathpar}
    \infer{
        \precontract(\pc) = (\CREATE(\stackvar{y}, \stackvar{\valu}, \stackvar{\io}, \stackvar{\is}, \pc', \pre))\\
        \startmemoffset = 2 \\ 
        \pre[\startmemoffset] = \optional{\loc_\io} \\
        \pre[\startmemoffset+1] = \optional{\loc_\is} \\
        f_1 = \fun{\cfgstate}{\updatestate{\cfgstatefull}{\tempvar{\stackvar{y}}}
        {\accesscfgstate{\applycreate(\restrict{\cfgstatefull}{\cfgstate}, \precontract, \pc)}{\stackvar{y}}}} \\
        f_2 = \fun{\cfgstatefull}{\updatestate{\cfgstatefull}{\tempvar{\globenvvar{\cfgexternalpc{\pc'}}}}
        {\accesscfgstate{\applycreate(\restrict{\cfgstatefull}{\cfgstate}, \precontract, \pc)}{\globenvvar{\cfgexternalpc{\pc'}}}}}  \\
        f = \fun{\cfgstatefull}{f_2(f_1(\cfgstatefull))}\\
    }
    {\edgestep{\transenv}{\cfgnode{\pc}{0}}{\cfgnode{\pc}{1}}} \\
    {\Def{  \{ \tempvar{\stackvar{y}}, \tempvar{\globenvvar{\cfgexternalpc{\pc'}}}\}}} \\
        \Use{
        \{\memvarconc{\loc}{\pc} ~|~ \loc \in [\loc_\io, \loc_\io + \loc_\is - 1] \}
            ~\cup~ \{\memvarabs{\loc}{\pc} ~|~ \loc \in [\loc_\io,\loc_\io + \loc_\is - 1] \} \\
         ~\cup~ \{ \stackvar{\valu}, \locenvvar{\gaspc{\pc}}, \locenvvar{\activeaccount} \}~\cup~ \globenvvar{X}~\cup~ \storvar{X}{\pc}
        } 
\end{mathpar}
\begin{mathpar}
    \infer{
        \precontract(\pc) = (\CREATE(\stackvar{y}, \stackvar{\valu}, \stackvar{\io}, \stackvar{\is}, \pc', \pre))\\
        \startmemoffset = 2 \\ 
        (\pre[\startmemoffset] = \bot ~\lor~ \pre[\startmemoffset+1] = \bot) \\ 
        f_1 = \fun{\cfgstatefull}{\updatestate{\cfgstatefull}{\tempvar{\stackvar{y}}}
        {\accesscfgstate{\applycreate(\restrict{\cfgstatefull}{\cfgstate}, \precontract, \pc)}{\stackvar{y}}}} \\
        f_2 = \fun{\cfgstatefull}{\updatestate{\cfgstatefull}{\tempvar{\globenvvar{\cfgexternalpc{\pc'}}}}
        {\accesscfgstate{\applycreate(\restrict{\cfgstatefull}{\cfgstate}, \precontract, \pc)}{\globenvvar{\cfgexternalpc{\pc'}}}}}  \\
        f = \fun{\cfgstatefull}{f_2(f_1(\cfgstatefull))}\\
    }
    {\edgestep{\transenv}{\cfgnode{\pc}{0}}{\cfgnode{\pc}{1}}} \\
    {\Def{  \{ \tempvar{\stackvar{y}}, \tempvar{\globenvvar{\cfgexternalpc{\pc'}}}\}}} \\
        \Use{ \memvarconc{X}{\pc} ~\cup~ \memvarabs{X}{\pc}
         ~\cup~ \{ \stackvar{\valu}, \locenvvar{\gaspc{\pc}}, \locenvvar{\activeaccount} \} 
         ~\cup~ \globenvvar{X}~\cup~ \storvar{X}{\pc} \\
         ~\cup \{ \stackvar{\io} ~|~ \pre[\startmemoffset] = \optional{\loc_\io}  \}  ~\cup \{ \stackvar{\is} ~|~ \pre[\startmemoffset+1] = \optional{\loc_\is}  \} 
        } 
\end{mathpar}

We give the rules for the gas computation: 
\begin{mathpar}
    \infer{
        \precontract(\pc) = (\CREATE(\stackvar{y}, \stackvar{\valu}, \stackvar{\io}, \stackvar{\is}, \pc', \pre))\\
        \startmemoffset = 2 \\ 
            \pre[\startmemoffset] = \optional{\loc_\io} \\
            \pre[\startmemoffset+1] = \optional{\loc_\is} \\
            f = \fun{\cfgstatefull}{\updatestate{\cfgstatefull}{\tempvar{\locenvvar{\gaspc{\pc'}}}}{\accesscfgstate{\applycreate(\restrict{\cfgstatefull}{\cfgstate}, \precontract, \pc)}{\locenvvar{\gaspc{\pc'}}}}}
         }
        {\edgestep{\transenv}{\cfgnode{\pc}{1}}{\cfgnode{\pc}{2}}}
        \\
        {\Defs{\tempvar{\locenvvar{\gaspc{\pc'}}}}} \\
        \Use{
            \{\memvarconc{\loc}{\pc} ~|~ \loc \in [\loc_\io, \loc_\io + \loc_\is - 1] \}
            ~\cup~ \{\memvarabs{\loc}{\pc} ~|~ \loc \in [\loc_\io,\loc_\io + \loc_\is - 1] \} \\
         ~\cup~ \{ \stackvar{\valu}, \locenvvar{\gaspc{\pc}}, \locenvvar{\msizepc{\pc}}, \locenvvar{\activeaccount} \}~\cup~ \globenvvar{X}~\cup~ \storvar{X}{\pc}
        } 
\end{mathpar}
\begin{mathpar}
     \infer{
        \precontract(\pc) = (\CREATE(\stackvar{y}, \stackvar{\valu}, \stackvar{\io}, \stackvar{\is}, \pc', \pre))\\
        \startmemoffset = 2 \\ 
            (\pre[\startmemoffset] = \bot ~\lor~  \pre[\startmemoffset+1] = \bot) \\
            f = \fun{\cfgstatefull}{\updatestate{\cfgstatefull}{\tempvar{\locenvvar{\gaspc{\pc'}}}}{\accesscfgstate{\applycreate(\restrict{\cfgstatefull}{\cfgstate}, \precontract, \pc)}{\locenvvar{\gaspc{\pc'}}}}}
         }
        {\edgestep{\transenv}{\cfgnode{\pc}{1}}{\cfgnode{\pc}{2}}}
        \\
        {\Defs{\locenvvar{\tempvar{\gaspc{\pc'}}}}} \\
        \Use{
            \memvarconc{X}{\pc} ~\cup~ \memvarabs{X}{\pc}
         ~\cup~ \{ \stackvar{\valu}, \locenvvar{\gaspc{\pc}}, \locenvvar{\msizepc{\pc}}, \locenvvar{\activeaccount} \}~\cup~ \globenvvar{X}~\cup~ \storvar{X}{\pc} \\
         ~\cup \{ \stackvar{\io} ~|~ \pre[\startmemoffset] = \optional{\loc_\io}  \}  ~\cup \{ \stackvar{\is} ~|~ \pre[\startmemoffset+1] = \optional{\loc_\is}  \} 
        }
\end{mathpar}

Next, we give the rule for the update of the active words in memory: 
\begin{mathpar}
            \infer{
                \precontract(\pc) = (\CREATE(\stackvar{y}, \stackvar{\valu}, \stackvar{\io}, \stackvar{\is}, \pc', \pre))\\
                \startmemoffset = 2 \\ 
                f = \fun{\cfgstatefull}{\updatestate{\cfgstatefull}{\tempvar{\locenvvar{\msizepc{\pc'}}}}{\accesscfgstate{\applycreate(\restrict{\cfgstatefull}{\cfgstate}, \precontract, \pc)}{\locenvvar{\msizepc{\pc'}}}}}
             }
            {\edgestep{\transenv}{\cfgnode{\pc}{2}}{\cfgnode{\pc}{3}}}
            \\
            {\Defs{\tempvar{\locenvvar{\msizepc{\pc'}}}}} \\
            \Use{ \{ \locenvvar{\msizepc{\pc}}\} 
            ~\cup \{ \stackvar{\io} ~|~ \pre[\startmemoffset] = \optional{\loc_\io}  \}  ~\cup \{ \stackvar{\is} ~|~ \pre[\startmemoffset+1] = \optional{\loc_\is}  \} 
            }
\end{mathpar}

We give the rules for writing the temporal variables into the actual ones one by one: 

\begin{mathpar}
    \infer{
        \precontract(\pc) = (\CREATE(\stackvar{y}, \stackvar{\valu}, \stackvar{\io}, \stackvar{\is}, \pc', \pre))\\
        \startmemoffset = 2 \\ 
        f_1 = \fun{\cfgstate}{\updatestate{\cfgstatefull}{\stackvar{y}}
        {\accesscfgstate{\cfgstatefull}{\tempvar{\stackvar{y}}}}} \\
        f_2 = \fun{\cfgstatefull}{\updatestate{\cfgstatefull}{\globenvvar{\cfgexternalpc{\pc'}}}
        {\accesscfgstate{\cfgstatefull}{\tempvar{\globenvvar{\cfgexternalpc{\pc'}}}}}}  \\
        f = \fun{\cfgstatefull}{f_2(f_1(\cfgstatefull))}\\
    }
    {\edgestep{\transenv}{\cfgnode{\pc}{3}}{\cfgnode{\pc}{4}}} \\
    {\Def{  \{ \stackvar{y}, \globenvvar{\cfgexternalpc{\pc'}}\}}} \\
        \Use{
            \{ \tempvar{\stackvar{y}}, \tempvar{\globenvvar{\cfgexternalpc{\pc'}}}\}
        } 
\end{mathpar}

\begin{mathpar}
    \infer{
        \precontract(\pc) = (\CREATE(\stackvar{y}, \stackvar{\valu}, \stackvar{\io}, \stackvar{\is}, \pc', \pre))\\
        \startmemoffset = 2 \\ 
        f = \fun{\cfgstate}{\updatestate{\cfgstatefull}{\locenvvar{\gaspc{\pc'}}}
        {\accesscfgstate{\cfgstatefull}{\tempvar{\locenvvar{\gaspc{\pc'}}}}}} \\
    }
    {\edgestep{\transenv}{\cfgnode{\pc}{4}}{\cfgnode{\pc}{5}}} \\
    {\Def{  \{ \locenvvar{\gaspc{\pc'}}\}}} \\
        \Use{
            \{ \tempvar{\locenvvar{\gaspc{\pc'}}} \}
        } 
\end{mathpar}

\begin{mathpar}
    \infer{
        \precontract(\pc) = (\CREATE(\stackvar{y}, \stackvar{\valu}, \stackvar{\io}, \stackvar{\is}, \pc', \pre))\\
        \startmemoffset = 2 \\ 
        f = \fun{\cfgstate}{\updatestate{\cfgstatefull}{\locenvvar{\msizepc{\pc'}}}
        {\accesscfgstate{\cfgstatefull}{\tempvar{\locenvvar{\msizepc{\pc'}}}}}} \\
    }
    {\edgestep{\transenv}{\cfgnode{\pc}{5}}{\cfgnode{\pc}{6}}} \\
    {\Def{  \{ \locenvvar{\msizepc{\pc'}}\}}} \\
        \Use{
            \{ \tempvar{\locenvvar{\msizepc{\pc'}}} \}
        } 
\end{mathpar}

Finally, all temporal variables are set to $\bot$ again: 

\begin{mathpar}
    \infer
    {
        \precontract(\pc) = (\CREATE(\stackvar{y}, \stackvar{\valu}, \stackvar{\io}, \stackvar{\is}, \pc', \pre))\\
        \startmemoffset = 2 \\ 
        f = \fun{\cfgstatefull}{\updatevarset{\cfgstatefull}{\tempvar{x}}{\bot}{\tempvar{x}}{\domain{\cfgstatecopy}}}
    }
    {\edgestep{\transenv}{\cfgnode{\pc}{6}}{\cfgnode{\pc'}{0}}} \\
    {\Def{\domain{\cfgstatecopy}}} \\
    \Use{ \emptyset} 
    \end{mathpar}

Finally, the instruction $\CREATETWO$ operates in a similar fashion as $\CREATE$ with the main difference being that the newly created contract is assigned an address that can be predetermined. To this end, $\CREATETWO$ takes an additional argument $\salt$, which together with the creation code determines the address.
Correspondingly, the rules for $\CREATETWO$ closely follow those of $\CREATE$.  

We can again summarize the rules for the return value $\stackvar{y}$ and the external environment.
Similar to the $\CREATE$ and $\CALL$ rules, all updates are first done to temporal variables and only later transferred to the actual ones.

\begin{mathpar}
    \infer{
        \precontract(\pc) = (\CREATETWO(\stackvar{y}, \stackvar{\valu}, \stackvar{\salt}, \stackvar{\io}, \stackvar{\is}, \pc', \pre))\\
        \startmemoffset = 2 \\ 
        \pre[\startmemoffset] = \optional{\loc_\io} \\
        \pre[\startmemoffset+1] = \optional{\loc_\is} \\
        f_1 = \fun{\cfgstate}{\updatestate{\cfgstatefull}{\tempvar{\stackvar{y}}}
        {\accesscfgstate{\applycreate(\restrict{\cfgstatefull}{\cfgstate}, \precontract, \pc)}{\stackvar{y}}}} \\
        f_2 = \fun{\cfgstatefull}{\updatestate{\cfgstatefull}{\tempvar{\globenvvar{\cfgexternalpc{\pc'}}}}
        {\accesscfgstate{\applycreate(\restrict{\cfgstatefull}{\cfgstate}, \precontract, \pc)}{\globenvvar{\cfgexternalpc{\pc'}}}}}  \\
        f = \fun{\cfgstatefull}{f_2(f_1(\cfgstatefull))}\\
    }
    {\edgestep{\transenv}{\cfgnode{\pc}{0}}{\cfgnode{\pc}{1}}} \\
    {\Def{  \{ \tempvar{\stackvar{y}}, \tempvar{\globenvvar{\cfgexternalpc{\pc'}}}\}}} \\
        \Use{
        \{\memvarconc{\loc}{\pc} ~|~ \loc \in [\loc_\io, \loc_\io + \loc_\is - 1] \}
            ~\cup~ \{\memvarabs{\loc}{\pc} ~|~ \loc \in [\loc_\io,\loc_\io + \loc_\is - 1] \} \\
         ~\cup~ \{ \stackvar{\valu}, \stackvar{\salt}, \locenvvar{\gaspc{\pc}}, \locenvvar{\activeaccount} \}~\cup~ \globenvvar{X}~\cup~ \storvar{X}{\pc}
        } 
\end{mathpar}
\begin{mathpar}
    \infer{
        \precontract(\pc) = (\CREATETWO(\stackvar{y}, \stackvar{\valu}, \stackvar{\salt}, \stackvar{\io}, \stackvar{\is}, \pc', \pre))\\
        \startmemoffset = 2 \\ 
        (\pre[\startmemoffset] = \bot ~\lor~ \pre[\startmemoffset+1] = \bot) \\ 
        f_1 = \fun{\cfgstatefull}{\updatestate{\cfgstatefull}{\tempvar{\stackvar{y}}}
        {\accesscfgstate{\applycreate(\restrict{\cfgstatefull}{\cfgstate}, \precontract, \pc)}{\stackvar{y}}}} \\
        f_2 = \fun{\cfgstatefull}{\updatestate{\cfgstatefull}{\tempvar{\globenvvar{\cfgexternalpc{\pc'}}}}
        {\accesscfgstate{\applycreate(\restrict{\cfgstatefull}{\cfgstate}, \precontract, \pc)}{\globenvvar{\cfgexternalpc{\pc'}}}}}  \\
        f = \fun{\cfgstatefull}{f_2(f_1(\cfgstatefull))}\\
    }
    {\edgestep{\transenv}{\cfgnode{\pc}{0}}{\cfgnode{\pc}{1}}} \\
    {\Def{  \{ \tempvar{\stackvar{y}}, \tempvar{\globenvvar{\cfgexternalpc{\pc'}}}\}}} \\
        \Use{ \memvarconc{X}{\pc} ~\cup~ \memvarabs{X}{\pc}
         ~\cup~ \{ \stackvar{\valu}, \stackvar{\salt} \locenvvar{\gaspc{\pc}}, \locenvvar{\activeaccount} \} 
         ~\cup~ \globenvvar{X}~\cup~ \storvar{X}{\pc} \\
         ~\cup \{ \stackvar{\io} ~|~ \pre[\startmemoffset] = \optional{\loc_\io}  \}  ~\cup \{ \stackvar{\is} ~|~ \pre[\startmemoffset+1] = \optional{\loc_\is}  \} 
        } 
\end{mathpar}

We give the rules for the gas computation: 
\begin{mathpar}
    \infer{
        \precontract(\pc) = (\CREATETWO(\stackvar{y}, \stackvar{\valu}, \stackvar{\salt}, \stackvar{\io}, \stackvar{\is}, \pc', \pre))\\
        \startmemoffset = 2 \\ 
            \pre[\startmemoffset] = \optional{\loc_\io} \\
            \pre[\startmemoffset+1] = \optional{\loc_\is} \\
            f = \fun{\cfgstatefull}{\updatestate{\cfgstatefull}{\tempvar{\locenvvar{\gaspc{\pc'}}}}{\accesscfgstate{\applycreate(\restrict{\cfgstatefull}{\cfgstate}, \precontract, \pc)}{\locenvvar{\gaspc{\pc'}}}}}
         }
        {\edgestep{\transenv}{\cfgnode{\pc}{1}}{\cfgnode{\pc}{2}}}
        \\
        {\Defs{\tempvar{\locenvvar{\gaspc{\pc'}}}}} \\
        \Use{
            \{\memvarconc{\loc}{\pc} ~|~ \loc \in [\loc_\io, \loc_\io + \loc_\is - 1] \}
            ~\cup~ \{\memvarabs{\loc}{\pc} ~|~ \loc \in [\loc_\io,\loc_\io + \loc_\is - 1] \} \\
         ~\cup~ \{ \stackvar{\valu}, \stackvar{\salt}, \locenvvar{\gaspc{\pc}}, \locenvvar{\msizepc{\pc}}, \locenvvar{\activeaccount} \}~\cup~ \globenvvar{X}~\cup~ \storvar{X}{\pc}
        } 
\end{mathpar}
\begin{mathpar}
     \infer{
        \precontract(\pc) = (\CREATETWO(\stackvar{y}, \stackvar{\valu}, \stackvar{\salt}, \stackvar{\io}, \stackvar{\is}, \pc', \pre))\\
        \startmemoffset = 2 \\ 
            (\pre[\startmemoffset] = \bot ~\lor~  \pre[\startmemoffset+1] = \bot) \\
            f = \fun{\cfgstatefull}{\updatestate{\cfgstatefull}{\tempvar{\locenvvar{\gaspc{\pc'}}}}{\accesscfgstate{\applycreate(\restrict{\cfgstatefull}{\cfgstate}, \precontract, \pc)}{\locenvvar{\gaspc{\pc'}}}}}
         }
        {\edgestep{\transenv}{\cfgnode{\pc}{1}}{\cfgnode{\pc}{2}}}
        \\
        {\Defs{\locenvvar{\tempvar{\gaspc{\pc'}}}}} \\
        \Use{
            \memvarconc{X}{\pc} ~\cup~ \memvarabs{X}{\pc}
         ~\cup~ \{ \stackvar{\valu}, \stackvar{\salt}, \locenvvar{\gaspc{\pc}}, \locenvvar{\msizepc{\pc}}, \locenvvar{\activeaccount} \}~\cup~ \globenvvar{X}~\cup~ \storvar{X}{\pc} \\
         ~\cup \{ \stackvar{\io} ~|~ \pre[\startmemoffset] = \optional{\loc_\io}  \}  ~\cup \{ \stackvar{\is} ~|~ \pre[\startmemoffset+1] = \optional{\loc_\is}  \} 
        }
\end{mathpar}

Next, we give the rule for the update of the active words in memory: 
\begin{mathpar}
            \infer{
                \precontract(\pc) = (\CREATETWO(\stackvar{y}, \stackvar{\valu}, \stackvar{\salt}, \stackvar{\io}, \stackvar{\is}, \pc', \pre))\\
                \startmemoffset = 2 \\ 
                f = \fun{\cfgstatefull}{\updatestate{\cfgstatefull}{\tempvar{\locenvvar{\msizepc{\pc'}}}}{\accesscfgstate{\applycreate(\restrict{\cfgstatefull}{\cfgstate}, \precontract, \pc)}{\locenvvar{\msizepc{\pc'}}}}}
             }
            {\edgestep{\transenv}{\cfgnode{\pc}{2}}{\cfgnode{\pc}{3}}}
            \\
            {\Defs{\tempvar{\locenvvar{\msizepc{\pc'}}}}} \\
            \Use{ \{ \locenvvar{\msizepc{\pc}}\} 
            ~\cup \{ \stackvar{\io} ~|~ \pre[\startmemoffset] = \optional{\loc_\io}  \}  ~\cup \{ \stackvar{\is} ~|~ \pre[\startmemoffset+1] = \optional{\loc_\is}  \} 
            }
\end{mathpar}

We give the rules for writing the temporal variables into the actual ones one by one: 

\begin{mathpar}
    \infer{
        \precontract(\pc) = (\CREATETWO(\stackvar{y}, \stackvar{\valu}, \stackvar{\salt}, \stackvar{\io}, \stackvar{\is}, \pc', \pre))\\
        \startmemoffset = 2 \\ 
        f_1 = \fun{\cfgstate}{\updatestate{\cfgstatefull}{\stackvar{y}}
        {\accesscfgstate{\cfgstatefull}{\tempvar{\stackvar{y}}}}} \\
        f_2 = \fun{\cfgstatefull}{\updatestate{\cfgstatefull}{\globenvvar{\cfgexternalpc{\pc'}}}
        {\accesscfgstate{\cfgstatefull}{\tempvar{\globenvvar{\cfgexternalpc{\pc'}}}}}}  \\
        f = \fun{\cfgstatefull}{f_2(f_1(\cfgstatefull))}\\
    }
    {\edgestep{\transenv}{\cfgnode{\pc}{3}}{\cfgnode{\pc}{4}}} \\
    {\Def{  \{ \stackvar{y}, \globenvvar{\cfgexternalpc{\pc'}}\}}} \\
        \Use{
            \{ \tempvar{\stackvar{y}}, \tempvar{\globenvvar{\cfgexternalpc{\pc'}}}\}
        } 
\end{mathpar}

\begin{mathpar}
    \infer{
        \precontract(\pc) = (\CREATETWO(\stackvar{y}, \stackvar{\valu}, \stackvar{\salt}, \stackvar{\io}, \stackvar{\is}, \pc', \pre))\\
        \startmemoffset = 2 \\ 
        f = \fun{\cfgstate}{\updatestate{\cfgstatefull}{\locenvvar{\gaspc{\pc'}}}
        {\accesscfgstate{\cfgstatefull}{\tempvar{\locenvvar{\gaspc{\pc'}}}}}} \\
    }
    {\edgestep{\transenv}{\cfgnode{\pc}{4}}{\cfgnode{\pc}{5}}} \\
    {\Def{  \{ \locenvvar{\gaspc{\pc'}}\}}} \\
        \Use{
            \{ \tempvar{\locenvvar{\gaspc{\pc'}}} \}
        } 
\end{mathpar}

\begin{mathpar}
    \infer{
        \precontract(\pc) = (\CREATETWO(\stackvar{y}, \stackvar{\valu}, \stackvar{\salt}, \stackvar{\io}, \stackvar{\is}, \pc', \pre))\\
        \startmemoffset = 2 \\ 
        f = \fun{\cfgstate}{\updatestate{\cfgstatefull}{\locenvvar{\msizepc{\pc'}}}
        {\accesscfgstate{\cfgstatefull}{\tempvar{\locenvvar{\msizepc{\pc'}}}}}} \\
    }
    {\edgestep{\transenv}{\cfgnode{\pc}{5}}{\cfgnode{\pc}{6}}} \\
    {\Def{  \{ \locenvvar{\msizepc{\pc'}}\}}} \\
        \Use{
            \{ \tempvar{\locenvvar{\msizepc{\pc'}}} \}
        } 
\end{mathpar}

Finally, all temporal variables are set to $\bot$ again: 

\begin{mathpar}
    \infer
    {
        \precontract(\pc) = (\CREATETWO(\stackvar{y}, \stackvar{\valu}, \stackvar{\salt}, \stackvar{\io}, \stackvar{\is}, \pc', \pre))\\
        \startmemoffset = 2 \\ 
        f = \fun{\cfgstatefull}{\updatevarset{\cfgstatefull}{\tempvar{x}}{\bot}{\tempvar{x}}{\domain{\cfgstatecopy}}}
    }
    {\edgestep{\transenv}{\cfgnode{\pc}{6}}{\cfgnode{\pc'}{0}}} \\
    {\Def{\domain{\cfgstatecopy}}} \\
    \Use{ \emptyset} 
    \end{mathpar}

\paragraph{From CFG semantics to Logical Rules}

We illustrate how the logical rules describing the PDG derived from the CFG semantics are constructed. 

The CFG semantics describes the PDG by giving control dependencies (via the CFG) and data dependencies via the Def and Use sets. 
Each transition rule introduces data dependencies from all variables in the Def set to all variables in the Use set. 
The node splitting allows for enhancing precision since assignments of several variables that do not share the same Use set can be distinguished in a more fine-grained manner. 

For translating the CFG rules into dependency predicates, it is simply required to model the resulting data and control flow dependencies. 
However, we need to introduce a further abstraction step to account for the fact that Def and Use sets may be infinite (or at least unreasonably large, assuming that memory and storage locations can be represented by 256 bits). 
More precisely, we will introduce a symbolic variable $\top$, which we will use to summarize memory and storage variables.
Intuitively, $\top$ when used for modeling variable access (in the Use set) will represent the union of all dynamic and static memory (or storage) variables ($\memvarabs{X}{} \cup \memvarconc{X}{}$ or $\storvarabs{X}{} \cup \storvarconc{X}{}$, respectively).
When used to model writing variables (in the Def set), $\top$ will represent all dynamic memory (or storage) variables ($\memvarabs{X}{}$ or $\storvarabs{X}{}$, respectively).

To model this, we will assume the following types for our predicate domains:
\begin{align*}
\loctype &\define \integer{256} \cup \{ \top \} \\
\tagtype &\define  \instructions
\end{align*}
where $\instructions$ is the set of all instructions.
Intuitively, $\loctype$ encodes the type of all (symbolic) storage locations and $\tagtype$ encodes the types of so-called \emph{tags}. Tags model those variables on which we explicitly want to track dependencies. 
For the scope of this work, we will only track dependencies on static environment variables, which we will (for simplicity) all represent by the opcodes that access these variables. For this reason, we define $\tagtype$ to consist of the set of all instructions $\instructions$. 

To capture the dependencies as induced by the Def and Use sets, we define local data dependency predicates that describe the data dependencies between variables at specific nodes.

As opposed to directly specifying the Def and Use sets, these predicates enumerate all pairs of variables in the Def and the Use set (the cross-product between them).
In this way, we do not need to make the subnodes at a given program counter (as given in the CFG semantics) explicit. 
These subnodes result from node splitting and only aim for separating the dependencies for different nodes in the Def set. 
Consequently, we can easily mimic this effect by directly modeling dependencies between variables as they are induced by the Def and Use set at a given subnode. 

For efficiency reasons, we consider different variable types and devise predicates that describe the local dependencies between these types (as induced by the CFG nodes for a specific program counter). 
This results in improved performance since it enables the underlying datalog solver to compute several smaller fixpoints (for each variable type) instead of a big fixpoint (that captures the dependencies for all variables).

More precisely, we define the following predicates (indexed by the program counter) for the different combinations of variable types as follows, where their name indicates the corresponding type (\pred{Var} for stack variables, \pred{Mem} for memory variables, $\pred{Store}$ for storage variables, $\pred{Gas}$ for local environmental variable $\locenvvar{\lgas}$, $\pred{Msize}$ for the local environmental variable $\locenvvar{\msize}$, and $\pred{External}$ for the global environmental variables $\globenvvar{\extenv}$, and $\pred{Source}$ the static local and global environmental variables).

\begin{minipage}[t]{0.24\columnwidth}
\begin{align*}
    \pred{VarVar}_{\pc} &\subseteq \integer{256} \times \integer{256}  \\
    \pred{VarMem}_{\pc} &\subseteq \integer{256} \times \loctype \\
    \pred{VarStor}_{\pc} &\subseteq \integer{256} \times \loctype \\
    \pred{VarExternal}_{\pc} &\subseteq \integer{256} \\
    \pred{VarGas}_{\pc} &\subseteq \integer{256} \\
    \pred{VarSource}_{\pc} &\subseteq \integer{256} \times \tagtype \\
    \\ 
    \pred{StoreVar}_{\pc}  &\subseteq \loctype \times \integer{256} 
    \\
\end{align*}
\end{minipage}
\begin{minipage}[t]{0.24\columnwidth}
    \begin{align*}
    \pred{MemMem}_{\pc}  &\subseteq \loctype \times \loctype  \times \loctype \\
    \pred{MemVar}_{\pc}  &\subseteq \loctype \times \integer{256}  \\
    \pred{MemExternal}_{\pc} &\subseteq \loctype \\
    \pred{MemGas}_{\pc} &\subseteq \loctype \\
    \pred{MemMsize}_{\pc} &\subseteq \loctype \\
    \pred{MemStore}_{\pc} &\subseteq \loctype \times \loctype \\
    \pred{MemSource}_{\pc} &\subseteq \loctype \times \tagtype \\
\end{align*}
\end{minipage}
\begin{minipage}[t]{0.24\columnwidth}
\begin{align*}
    \pred{GasMem}_{\pc}  &\subseteq \loctype \\
    \pred{GasVar}_{\pc}  &\subseteq \integer{256}  \\
    \pred{GasExternal}_{\pc} &\subseteq \BB \\
    \pred{GasMsize}_{\pc} &\subseteq \BB\\
    \pred{GasStore}_{\pc} &\subseteq \loctype \\
    \pred{GasSource}_{\pc} &\subseteq \tagtype \\
    \\
    \pred{MsizeVar}_{\pc}  &\subseteq \integer{256}  \\
\end{align*}
\end{minipage}
\begin{minipage}[t]{0.24\columnwidth}
\begin{align*}
    \pred{ExternalMem}_{\pc}  &\subseteq \loctype \\
    \pred{ExternalVar}_{\pc}  &\subseteq \integer{256}  \\
    \pred{ExternalGas}_{\pc} &\subseteq \BB \\
    \pred{ExternalMsize}_{\pc} &\subseteq \BB\\
    \pred{ExternalStore}_{\pc} &\subseteq \loctype \\
    \pred{ExternalSource}_{\pc} &\subseteq \tagtype \\
\end{align*}
\end{minipage}

Where $\langle\textit{write}\rangle \langle \textit{read} \rangle$ indicates for $\textit{write} \in \{ \pred{Var}, \pred{Mem}, \pred{Store}, \pred{Gas}. \pred{Msize}, \pred{External} \}$ and $\textit{read} \in \{ \pred{Var}, \pred{Mem}, \pred{Store}, \pred{Gas}. \pred{Msize}, \pred{External}, \pred{Source} \}$ that variable kind $\textit{write}$ is written and variable kind $\textit{read}$ is read. 

E.g., $\pred{VarMem}(x, y)$ indicates that stack variable $\stackvar{x}$ is written dependent on dynamic and static memory variables $\{ \memvarconc{y}{}, \memvarabs{y}{} \}$ and $\pred{VarMem}(x, \top)$ indicates that stack variable $\stackvar{x}$ depends on all static and dynamic memory variables ($\memvarconc{X}{} \cup \memvarabs{X}{}$).
Similarly, $\pred{MemVar}(x, y)$ indicates that the static memory location $\memvarconc{x}{}$ depends on stack variable $\stackvar{y}$, and $\pred{MemVar}(\top, y)$ indicates that all dynamic memory variables $\memvarabs{X}{}$ depend on stack variable $\stackvar{y}$.
Note that the variable $\top$ is used in the symbolic fashion described above.

A special case of this symbolic treatment is the predicate $\pred{MemMem}$, which takes three arguments to give a more fine-grained symbolic modeling of whole memory intervals:
The first position of $\pred{MemMem}$ specifies a memory location, and the two next positions specify a memory interval, which is given by its start offset and size. 
Start offset and size can again be of type $\loctype$, such that $\pred{MemMem}_\pc(x, \top, \top)$ indicates that the (static) memory variable $\memvarconc{x}{}$ depends on all static and dynamic memory variables and $\pred{MemMem}_\pc(x, i, s)$ indicates that $\memvar{x}{}$ depends on all static and dynamic memory variables starting at memory position $i$ until $i + s -1$.

Note that $\textit{write}$ can never be $\pred{Source}$ since static local and global environment variables can never be written. Similarly, other \textit{write}-\textit{read} combinations are omitted for cases that never occur (e.g., for store unreachable contracts, the only way to write the contract's storage is the $\SSTORE$ opcode, that allows for storing a stack variable. Consequently, within a CFG node, a storage variable can only depend on the stack variables, so the predicate $\pred{StoreVar}$ is sufficient to capture all local dependencies of storage variables).

The local dependency predicates can be simply inhabited by rules that closely follow the CFG semantics: 
For each program counter $\pc$, instruction-specific rules are generated that reflect the dependencies induced by the Def and Use sets of the subnodes of $\pc$. 
We give the example for the $\MLOAD$ instruction: 

\begin{align*}
    \{ \top
    \Rightarrow~ 
    &\pred{VarMem}(y, \top) ~|~  
    & \precontract(\pc) = (\MLOAD(\stackvar{y}, \stackvar{x}, \pc', \pre) ~\land~ \pre[0] = \bot \\
    \top
    \Rightarrow~ 
    &\pred{MsizeVar}_\pc(x) \\
    \top
    \Rightarrow~ 
    &\pred{GasVar}_\pc(x) \\
    \top
    \Rightarrow~ 
    &\pred{GasMsize}_\pc(\top) \} \\
    \\
    \{ \top
    \Rightarrow~ 
    &\pred{VarMem}(y, v) ~|~  
    & \precontract(\pc) = (\MLOAD(\stackvar{y}, \stackvar{x}, \pc', \pre) ~\land~ \pre[0] = v \\
    \top
    \Rightarrow~ 
    &\pred{GasMsize}_\pc(\top) \}
\end{align*}
The given rules describe the dependencies induced by the corresponding CFG rules: 
The result variable $\stackvar{y}$ either depends on all static and dynamic memory locations (if the memory location is unknown) or on the specific static and dynamic memory locations $\{ \memvarconc{v}{}, \memvarabs{v}{} \}$. 
The gas $\locenvvar{\lgas}$ depends on the value of the active words in memory and on the stack variable $\stackvar{x}$, which holds the memory position. Similarly, $\locenvvar{\msize}$ depends on $\stackvar{x}$. 
Note that we do not explicitly model that $\locenvvar{\lgas}$ and $\locenvvar{\msize}$ always depends on themselves since this is always the case, and hence we account for this by generic propagation rules, which always propagate gas and active word dependencies to the next program counter.

In addition to the local dependency predicates there exist special predicates that indicate that a variable is written in the first place: 
\begin{align*}
    \pred{MsizeWrite}_{\pc}  &\subseteq \BB  \\
    \pred{ExternalWrite}_{\pc}  &\subseteq \BB  \\
\end{align*}
These predicates encode that the corresponding variable (here $\locenvvar{\msize}$ or $\globenvvar{\extenv}$) is written at a specific program counter $\pc$. 
We need these predicates for expressing the interaction between data and control dependence (for building backward slices):
If a variable $x$ is written at a certain node $\node'$, which is control-dependent on another node $\node$, and $x$ is read at another node $\node''$ (without being overwritten before),
then $\node''$ is data dependent on $\node'$ and by transitivity, $\node''$ depends on $\node$ (via one data dependency and one control dependency edge).
For this reason, it is important to model when a variable is written.
When making Def and Use sets fully explicit, this is easy to see, however, in our modeling, we immediately consider the cross-product from the Def and Use sets (at a program counter). 
Consequently, it can happen that there are no entries for certain variables in the Def Set (if there is no variable in the Use set). 
So, we need to cover these cases explicitly.

For all other variables (but $\locenvvar{\msize}$ or $\globenvvar{\extenv}$), we can use existing predicates as indicators for writing. 
E.g., whenever a variable is written the $\pred{VarSource}$ predicate is inhabited. 
Similarly, whenever memory or storage variables are written $\pred{MemVar}$ or, respectively $\pred{StorVar}$ are inhabited for the corresponding variables. 
Gas is written at any program counter.

To model the transitive dependencies (as induced by the PDG), 
we use the local data dependencies (as modeled by the local dependency predicates above) and the control dependencies (pre-computed according to the definition of standard control dependence) and build their transitive closure. 
The control dependence is available via a predicate $\pred{Controls}$. 
First, the transitive closure for control dependence is modeled via the predicate $\pred{MayControls}_{\pc} \subseteq \integer{256} \times \integer{256}$. 
Intuitively, $\pred{MayControls}_\pc(\lpc_b, x_b)$ means that the program counter $\pc$ is transitively controled by the program counter $\lpc_b$ where at $\lpc_b$ there is a brach instruction ($\JUMPI$) with condition stack variable $\stackvar{x_b}$. 

Next, we define fixed point rules, which inhabit the following transitive closure predicates for program dependence: 

\begin{align*}
    \pred{VarMayDependOn} & \subseteq \NN \times \tagtype \\
    \pred{MemMayDependOn}_\pc & \subseteq \loctype \times \tagtype \\
    \pred{StorMayDependOn}_\pc & \subseteq \loctype \times \tagtype \\
    \pred{MsizeDependOn}_\pc & \subseteq \tagtype \\
    \pred{GasDependOn}_\pc & \subseteq \tagtype \\
    \pred{ExternalDependOn}_\pc & \subseteq \tagtype \\
\end{align*}

Intuitively, $\pred{VarMayDependOn}(x, t)$ denotes that variable $x$ may depend on tag $t$, so that a node $\node'$ where $t$ is in the Use set, is in the backward slice of the (unique) node where $x$ is written. 
Note that the tag $t$ represents a static environment variable.

We index the predicates (with exception of \pred{VarMayDependOn}) by the program counter to precisely characterize data dependence: 
A node $\node'$ is considered data dependant on another node $\node$ if $\node$ defines a variable $x$ that is used by $\node'$ and $\node'$ is reachable from $\node$ without passing through another node defining $x$. 
Since we aim at staying within a characterization of dependencies that only uses grounded Horn clauses, we cannot simply express the second requirement (namely that no other node defining $x$ should be passed).
Instead, we explicitly formulate rules propagating dependencies in the case that at a certain program counter a variable is not (over)written. 
E.g., we can formulate a generic rule for gas propagation, since gas is updated at every program counter (and hence is contained in the Def and Use set).

Note that dependencies of stack variables ($\pred{VarMayDependOn}$), as opposed to the other dependency predicates, is not indexed by the concrete node. 
This is because the contract is assumed to be in SSA form (for stack variables), and those should hence only appear at only a single program location.

Note that (similar to Securify), we currently only explicitly track transitive dependencies on local and global static environment variables (of type $\tagtype$). E.g., we can express that a stack variable transitively depends on, e.g., the block timestamp, but not that it transitively depends on e.g., a specific memory variable (however, it is, of course, captured that dependencies on global static environment variables can be introduced through dependencies on other variables).
The analysis can easily be extended to track further dependencies explicitly by adding rules introducing the corresponding dependencies to the corresponding $\langle \textit{write}\rangle\pred{Source}$ predicate. 

The rules inhabiting the fixed point predicates are fairly standard. 
They are slightly complicated by the fact that we consider different variable types with different predicates so we need to consider data dependencies described by all the different local dependence predicates.
Considering all possible combinations, introduces a slight overhead in rules (as compared to having a single predicate for variable types), but results in better performance, since it splits the fixed point computations into several smaller fixpoints.
Further, there are some subtleties to consider for our symbolic treatment of dynamic memory locations and memory intervals (as discussed above).

We illustrate this by the fixpoint rules for the $\pred{MemMayDependOn}$ predicate given in~\Cref{fig:memmaydependon-data,fig:memmaydependon-control}.

\begin{figure}
\begin{align}
    \{ 
   &\pred{MemSource}_{\pc'}(\ell, t) \Rightarrow~ 
    \pred{MemMayDependOn}_{\pc'}(\ell, t), 
    & ~|~  \precontract(\pc) = (\instruction(\vec{x}), \pc', \pre)  \label{hc:mem-source}\\
   \nonumber\\
    &\pred{MemVar}_{\pc'}(\ell, v) 
    ~\land~ \pred{VarMayDependOn}(v, t) \label{hc:mem-var1}\\
    &\Rightarrow~ 
    \pred{MemMayDependOn}_{\pc'}(\ell, t),  \nonumber\\
    \nonumber\\
    &\pred{MemVar}_{\pc'}(\ell, v) 
    ~\land~ \ell \neq \ell'
    ~\land~ \pred{MemMayDependOn}_{\pc}(\ell', t) \label{hc:mem-var2} \\
    &\Rightarrow~ 
    \pred{MemMayDependOn}_{\pc'}(\ell', t),  \nonumber \\
    \nonumber\\
    &\pred{MemGas}_{\pc'}(\ell) 
    ~\land~ \pred{GasMayDependOn}(t)  \label{hc:mem-gas}\\
    &\Rightarrow~ 
    \pred{MemMayDependOn}_{\pc'}(\ell, t),  \nonumber\\
    \nonumber\\
    &\pred{MemMsize}_{\pc'}(\ell) 
    ~\land~ \pred{MsizeMayDependOn}(t)  \label{hc:mem-msize}\\
    &\Rightarrow~ 
    \pred{MemMayDependOn}_{\pc'}(\ell, t),   \nonumber\\
    \nonumber\\
    &\pred{MemExternal}_{\pc'}(\ell) 
    ~\land~ \pred{ExternalMayDependOn}(t)  \label{hc:mem-external}\\
    &\Rightarrow~ 
    \pred{MemMayDependOn}_{\pc'}(\ell, t),    \nonumber\\
    \nonumber\\
    &\pred{MemStore}_{\pc'}(\ell, \top) 
    ~\land~ \pred{StoreMayDependOn}(\ell_S, t)  \label{hc:mem-stor1}\\
    &\Rightarrow~ 
    \pred{MemMayDependOn}_{\pc'}(\ell, t),    \nonumber\\
    \nonumber\\
    &\pred{MemStore}_{\pc'}(\ell, \ell_S) 
    ~\land~ \ell_S \in \integer{256}
    ~\land~ \pred{StoreMayDependOn}(\ell_S, t)  \label{hc:mem-stor2}\\
    &\Rightarrow~ 
    \pred{MemMayDependOn}_{\pc'}(\ell, t),    \nonumber\\
    \nonumber\\
    &\pred{MemStore}_{\pc'}(\ell, \ell_S) 
    ~\land~ \ell_S \in \integer{256}
    ~\land~ \pred{StoreMayDependOn}(\top, t)  \label{hc:mem-stor3}\\
    &\Rightarrow~ 
    \pred{MemMayDependOn}_{\pc'}(\ell, t),    \nonumber\\
    \nonumber\\
    &\pred{MemMayDependOn}_{\pc}(\top, t)  \label{hc:mem-top1}\\
    &\Rightarrow~ 
    \pred{MemMayDependOn}_{\pc'}(\top, t),    \nonumber\\
    \nonumber\\
    &\pred{MemMem}_{\pc'}(\ell, \top, \ell_s) 
    ~\land~ \pred{MemMayDependOn}_{\pc}(\ell', t) \label{hc:mem-top2}\\
    &\Rightarrow~ 
    \pred{MemMayDependOn}_{\pc'}(\ell, t),    \nonumber\\
    \nonumber\\
    &\pred{MemMem}_{\pc'}(\ell, \ell_o, \ell_s) 
    ~\land~ \ell_o \in \integer{256}
    ~\land~ \ell_s \in \integer{256}  \label{hc:mem-mem1}\\
    &\land i \geq \ell_o ~\land~ i < \ell_o + \ell_s 
    ~\land~ \pred{MemMayDependOn}_{\pc}(i, t)    \nonumber\\
    &\Rightarrow~ 
    \pred{MemMayDependOn}_{\pc'}(\ell, t),   \nonumber\\
    \nonumber\\
    &\pred{MemMem}_{\pc'}(\ell, \ell_o, \ell_s) 
    ~\land~ \ell_o \in \integer{256}
    ~\land~ \ell_s \in \integer{256}  \label{hc:mem-mem2}\\
    &\land~ \pred{MemMayDependOn}_{\pc}(\top, t)    \nonumber\\
    &\Rightarrow~ 
    \pred{MemMayDependOn}_{\pc'}(\ell, t),    \nonumber\\
    \nonumber\\
    &\pred{MemMayDependOn}_{\pc}(\ell, t)
    ~\land~ \pred{NoReassignMem}(\pc)  \label{hc:mem-noreassign}\\
    &\Rightarrow~ 
    \pred{MemMayDependOn}_{\pc'}(\ell, t),    \nonumber\\
     \}    \nonumber
\end{align}
   \caption{Horn Clauses describing the data flow dependencies captured by the $\pred{MemMayDependOn}$ predicate.} 
   \label{fig:memmaydependon-data}
\end{figure}

\begin{figure}
\begin{align}
    \{ 
    &\pred{MemVar}_{\pc}(\ell, v_1) 
    ~\land~ \pred{MayControls}_{\pc}(\pc', v_2)
    ~\land~ \pred{VarMayDependOn}(v_2, t) & ~|~  \pc \in \domain{\precontract} \\
    &\Rightarrow~ 
    \pred{MemMayDependOn}_{\pc}(\ell, t)  \nonumber
    \} 
\end{align}
\caption{Horn Clauses describing the control flow dependencies captured by the $\pred{MemMayDependOn}$ predicate.} 
\label{fig:memmaydependon-control}
\end{figure}

\Cref{fig:memmaydependon-data} shows the rules for describing transitive data dependencies for memory locations (captured by the $\pred{MemMayDependOn}$ predicate). 
To this end, there are rules for all local dependency predicates, which indicate that the local memory is written. 
Intuitively, the rules model the data dependencies introduced by the nodes at $\pc'$, by propagating dependencies known for the previous program counter $\pc'$. 
Rule~\ref{hc:mem-source} simply introduces dependencies from the $\pred{MemSource}$ predicates (constituting the base case).
The rules for propagating dependencies from variables(\ref{hc:mem-var1}), gas (\ref{hc:mem-gas}), active words in memory (\ref{hc:mem-msize}), and the external environment (\ref{hc:mem-external}) are fully standard. 
We need to consider that for all memory locations which are not overwritten, the dependencies from the previous program counter ($\pc$) are propagated. 
To this end, we define the predicate $\pred{NoReassignMem}$ that contains those opcodes that do not overwrite any memory location. 
It only contains all instructions but $\MSTORE$ and the copy operations. 
For all other operations, all previous dependencies are propagated (by rule~\ref{hc:mem-noreassign}).
For the overwriting operations, we need to consider that they do not write all memory variables.
In particular, $\MSTORE$ may only write a single memory location, so for all other memory locations, the dependencies shall be propagated (this is done by rule~\ref{hc:mem-var2}).
\footnote{Note that technically, we would need to have a similar rule for the copy operations, but since we anyway overapproximate them to only write the $\top$ variable, the generic $\top$-propagation rule (\ref{hc:mem-top1}) captures this case.}
Additionally, we have a general rule that always propagates the dependencies of the symbolic memory position $\top$ (\ref{hc:mem-top1}). 
This rule accounts for the fact that no opcode overwrites all (dynamic) memory locations. 

Finally, the most involved rules are those that involve symbolic reads from memory/storage locations. 
For those, we need to consider that reading from a static location, also always implies reading from $\top$. 
This is shown by rules~\ref{hc:mem-stor1},~\ref{hc:mem-stor2} and~\ref{hc:mem-stor3}, which account for the influences of storage locations on memory locations (as given by the local dependency predicate $\pred{MemStore}$): 
Here $\pred{MemStore}_{\pc'}(\ell, \top)$ indicates that memory variable $\ell$ depends on any static or dynamic storage variable. For this reason, any of their dependencies are propagated.
Similarly, $\pred{MemStore}_{\pc'}(\ell, \ell_S)$ (for $\ell_S \in \integer{256}$) indicates that the storage location from which is read is statically known. 
In this case, both the dependencies of the static storage locations within this interval are propagated, as well as the dependencies from $\top$ (indicating the corresponding dynamic storage locations).
\footnote{Note that in the implementation we currently omit this last rule, since, anyway, $\pred{MemStore}$ can only be inhabited with value $\top$ in the $\CALL$-like rules.}. 

The treatment of the $\pred{MemMem}$ predicate is similar: 
Rule~\ref{hc:mem-mem1} considers the case that potentially the whole memory is read ($\pred{MemMem}_{\pc'}(\ell, \top, \ell_s)$), rule~\ref{hc:mem-mem2} considers the case where the specified interval that is read is concretely known and all dependencies from the corresponding static locations are read, and rule~\ref{hc:mem-mem2} ensures that also in this case all rules from the dynamic location ($\top$) are propagated. 

Finally,~\Cref{fig:memmaydependon-control} shows the rule that models the influence of the control flow on the dependencies of memory variables (captured by $\pred{MemMayDependOn}$). 
The rule states that a memory location $\ell$ depends on tag $t$ (at $\pc$) if the memory is written at $\pc$ (indicated by $\pred{MemVar}_{\pc}(\ell, v_1) $) and there is another program counter $\pc'$, which controls $\pc$ and at $\pc'$ (which needs to be a branch instruction), the conditional variable is $v_2$, which again depends (transitively) on $t$.

\subsubsection{Equivalence proof}
We first introduce preliminary notions for the equivalence theorem. 

For reasoning about contract executions that span several internal transactions, we introduce the notion of contract annotation as used in~\cite{schneidewind2020ethor}.
For the sake of simplicity, in the main body of the paper, we annotated execution states only with the contract code $\evmcontract$. However, it gives more flexibility to characterize a contract as a pair $c = (\contractaddress, \evmcontract)$ where $\contractaddress$ is the address of the contract and $\evmcontract$ is its code.
This allows to distinguish executions of different contracts that share the same code. 
In the following, we will use the simplified annotation when sufficient and otherwise use the full annotation. 

We recall the notion of strong consistency from~\cite{schneidewind2020ethor}.

\begin{definition}[Annotation consistency]
    An execution state $\exstate$ is consistent with contract annotation $c$ if the following two conditions hold
    \begin{enumerate}
    \item
    $\isregular{\exstate} \implies \access{\access{\exstate}{\exenv}}{\activeaccount} = \access{c}{\addr}$
    \item
    $\isregular{\exstate} \lor \ishalt{\exstate} \implies \access{\access{\exstate}{\gstate}(\access{c}{\addr})}{\code} = \access{c}{\code}$
    \end{enumerate}
    where $\isregular{\cdot}$ and $\ishalt{\cdot}$ are predicates on execution states indicating whether they are regular execution states or halting states, respectively.
    \end{definition}

    \begin{definition}[Strong annotation consistency]
        An execution state $\exstate$ is strongly consistent with contract annotation $c$ (written $\strongconsistent{\exstate}{c}$) if it is consistent with $c$ and additionally
        \begin{align*}
        \isregular{\exstate} \implies \access{\access{\exstate}{\exenv}}{\code} = \access{c}{\code}
        \end{align*}
        \end{definition}

Intuitively, a contract annotation $c$ being strongly consistent with execution state $\exstate$ requires that $\exstate$ executes the contract as it resides in the global state.

Note that the function $\toevmstate(\cfgstate, \precontract, \lpc)$ maps states in the CFG semantics $\cfgstate$ to execution states that are strongly consistent with $\precontract$. 
\begin{lemma}
Let $(\transenv, \exstate) = \toevmstate(\cfgstate, \precontract, \lpc)$. 
Then $\exstate$ is strongly consistent with $(\access{\access{\exstate}{\exenv}}{\activeaccount}, \fun{\lpc}{(\access{\precontract(\lpc)}{\instruction}, \access{\precontract(\lpc)}{\pre})})$.
\end{lemma}
\begin{proof}
Trivially follows from the definition of $\toevmstate$ since $\access{\exstate}{\gstate}$ is set to hold code $\toevmstate(\cfgstate, \precontract, \lpc)$ at address $\access{\access{\exstate}{\exenv}}{\activeaccount}$. 
\end{proof}

We formally define the notion of a transaction step $\transactionstep{\transenv}{\cons{\annotate{\exstate}{C}}{\callstack}}{\cons{\annotate{\exstate'}{C}}{\callstack}}$
and the the relation $\combistep{\transenv}{\cons{\annotate{\exstate}{C}}{\callstack}}{\cons{\annotate{\exstate'}{C}}{\callstack}}$. 

\begin{definition}[Transaction step]
    Let $\callstack$ be a callstack, $\exstate$, $\exstate'$ be execution states and $\transenv$ be a transaction environment. Further let $\evmcontract$ be a contract. Then
\begin{align*}
    \transactionstep{\transenv}{\cons{\annotate{\exstate}{C}}{\callstack}}{\cons{\annotate{\exstate'}{C}}{\callstack}}
    \define \exists \exstate^* \exstate^\dagger \evmcontract^*.~ \sstep{\transenv}{\cons{\annotate{\exstate}{C}}{\callstack}}{\cons{\annotate{\exstate^*}{\evmcontract*}}{\cons{\annotate{\exstate}{C}}{\callstack}}}
    \rightarrow^*  \cons{\annotate{\exstate^\dagger}{\evmcontract*}}{\cons{\annotate{\exstate}{C}}{\callstack}} \rightarrow \cons{\annotate{\exstate'}{C}}{\callstack}
\end{align*}
\end{definition}
A transaction step describes the execution of a transaction being initiated in $\exstate$ (since in the execution step thereafter the element $\exstate^*$ is added to the call stack) and ends in $\exstate'$ (since this is the execution state immediately after removing the additional stack element). Note that this definition excludes that the execution might have returned before and have triggered another internal transaction since the execution state $\exstate$ on the stack would have otherwise changed (at least due to a decrease in gas).

\begin{definition}[Medium step]
    Let $\callstack$ be a callstack, $\exstate$, $\exstate'$ be execution states and $\transenv$ be a transaction environment. Further let $\evmcontract$ be a contract. Then
\begin{align*}
    \combistep{\transenv}{\cons{\annotate{\exstate}{C}}{\callstack}}{\cons{\annotate{\exstate'}{C}}{\callstack}}
    \define \sstep{\transenv}{\cons{\annotate{\exstate}{C}}{\callstack}}{\cons{\annotate{\exstate'}{C}}{\callstack}} 
    ~\lor~ \transactionstep{\transenv}{\cons{\annotate{\exstate}{C}}{\callstack}}{\cons{\annotate{\exstate'}{C}}{\callstack}}
\end{align*}
\end{definition}

Due to the two-layered memory abstraction, multiple states in the CFG semantics represent a single state in the EVM semantics. 
Consequently, we define a notion of equivalence on CFG states that takes this into account: 

\begin{definition}[CFG state equivalence]
Two CFG state $\cfgstate$ and $\cfgstate'$ are considered equivalent (written $\cfgstate \cfgstatequiv{} \cfgstate'$) if the following holds: 
\begin{align*}
    \access{\cfgstate}{\cfgstack} =  \access{\cfgstate'}{\cfgstack} 
    ~\land~ \access{\cfgstate}{\cfglocenv} = \access{\cfgstate'}{\cfglocenv}  
    ~\land~ \access{\cfgstate}{\cfgglobenv} = \access{\cfgstate'}{\cfgglobenv}  
    ~\land~ \forall \memvar{x}{}. ~\load ~ \cfgstate ~ \memvar{x}{\lpc} =  ~ \load ~\cfgstate' ~ \memvar{x}{\lpc}  
    ~\land~ \forall \storvar{x}{}. ~\load ~ \cfgstate ~ \storvar{x}{\lpc} =  ~ \load ~\cfgstate' ~ \storvar{x}{\lpc}  \\
\end{align*}
\end{definition}

We state some basic properties on CFG state equivalences:
\begin{lemma}[CFG state equivalence properties]
    \label{lem:cfgstatequiv-props}
The following hold: 
\begin{itemize}
\item $\cfgstate \cfgstatequiv{} \updatestate{\cfgstate}{x}{\load~\cfgstate~ x}$
\item $\cfgstate \cfgstatequiv{} \updatestate{(\updatestate{\cfgstate}{x}{\load~\cfgstate~ x})}{x}{\bot}$
\end{itemize}
\end{lemma}
\begin{proof}
Follows immediately from the definition $\cfgstatequiv{}$ and $\load$. 
\end{proof}

Most importantly, equivalent CFG states will be mapped to the same EVM states: 
\begin{lemma}
    \label{lem:cfgstateequiv-exstate}
For all contracts $\precontract$, all program counters $\lpc$ and all CFG state $\cfgstate$ $\cfgstate'$ it holds that
\begin{align*}
\cfgstate \cfgstatequiv{} \cfgstate' \Leftrightarrow \toevmstate(\cfgstate, \precontract, \lpc) = \toevmstate(\cfgstate', \precontract, \lpc)
\end{align*}
\end{lemma}
\begin{proof}
Follows immediately from the definitions of $\toevmstate$, $\cfgstatequiv{}$ and $\load$. 
\end{proof}

We formally define the medium step version of the CFG semantics: 
\begin{align*}
    \cfgmedstep{\precontract, \cd}{\cfgconfig{\cfgnode{\lpc}{0}}{\cfgstate}}{\cfgconfig{\node}{\cfgstate'}}
    \define 
    \exists ~n~ (\cfgstate_i)_{i \in [0, n]}~ (\cfgact_i)_{i \in [0, n]}.~ 
    \precontract, \cd \vDash  
\cfgconfig{\cfgnode{\lpc}{0}}{\cfgstate}   
\left ( \cfgactedge{\cfgact_i} 
\cfgconfig{\cfgnode{\lpc}{i}}{\cfgstate_i}
\right )_{i \in [0, n-1]}
\cfgactedge{\cfgact_n} \cfgconfig{\node}{\cfgstate'} 
\end{align*}

We now state the equivalence statement.
In particular, we explicitly state the assumptions on the execution (excluding exceptions and reentering storage modification). 
Further, we consider the cases where exception or halting states are entered. 
\begin{theorem}[Equivalence of EVM and CFG semantics]
    \label{thm:equivalence-evm-cfg}
Let $\evmcontract$ be a store unreachable contract with sound preprocessing information.
Then the following holds:
\begin{enumerate}
    \item 
Let $\combistep{\transenv}{\cons{\annotate{\exstate}{\evmcontract}}{\callstack}}{\cons{\annotate{\exstate'}{\evmcontract}}{\callstack}}$ be an execution of contract $\evmcontract$ that does not exhibit local out-of-gas exceptions and let $\exstate$ be strongly consistent with $\evmcontract$. Then either
\begin{enumerate}
\item $\exstate' = \regstate{\mstate'}{\exenv'}{\gstate'}$ and $\cfgmedstep{\precontract, \size{\callstack}}{\cfgconfig{\cfgnode{\access{\mstate}{\pc}}{0}}{\tocfgstate(\transenv, \exstate) \uplus \cfgstatecopybot}}{\cfgconfig{\cfgnode{\access{\mstate'}{\pc}}{0}}{\cfgstate' \uplus \cfgstatecopybot}}$ for some $\cfgstate'$ with $\cfgstate' \cfgstatequiv{} \tocfgstate(\transenv,\exstate')$
\item $\exstate' = \excstate$ and $\cfgmedstep{\precontract, \size{\callstack}}{\cfgconfig{\cfgnode{\access{\mstate}{\pc}}{0}}{\tocfgstate(\transenv, \exstate) \uplus \cfgstatecopybot}}{\cfgconfig{\exceptionnode}{\tocfgstate(\transenv, \exstate) \uplus \cfgstatecopybot}}$
\item $\exstate' =  \haltstate{\gstate'}{\lgas}{\datav}{}$ and $\cfgmedstep{\precontract, \size{\callstack}}{\cfgconfig{\cfgnode{\access{\mstate}{\pc}}{0}}{\tocfgstate(\transenv, \exstate) \uplus \cfgstatecopybot}}{\cfgconfig{\haltnode}{\tocfgstate(\transenv, \exstate) \uplus \cfgstatecopybot}}$
\end{enumerate}
%
\item Let $\cfgmedstep{\precontract, \cd}{\cfgconfig{\cfgnode{\lpc}{0}}{\cfgstate \uplus \cfgstatecopybot}}{\cfgconfig{\node}{\cfgstatefull}}$. Then either 
\begin{enumerate}
    \item 
$n = \cfgnode{\lpc'}{0}$ and $\cfgstatefull = \cfgstate' \uplus \cfgstatecopybot$ and for $(\transenv, \exstate) = \toevmstate(\cfgstate, \precontract, \lpc)$ and $(\transenv', \exstate') = \toevmstate(\cfgstate', \precontract, \lpc')$ it holds that $\transenv = \transenv'$ and for all $\callstack$ s.t. $\size{\callstack} = \cd$ it holds that either $\combistep{\transenv}{\cons{\annotate{\exstate}{\evmcontract}}{\callstack}}{\cons{\annotate{\exstate'}{\evmcontract}}{\callstack}}$
 or $\sstep{\transenv}{\cons{\annotate{\exstate}{\evmcontract}}{\callstack}}{\cons{\annotate{\excstate}{\evmcontract}}{\callstack}}$ and $\evmcontract[\lpc] \neq \INVALID$. 
 \item $n = \exceptionnode$ and for $(\transenv, \exstate) = \toevmstate(\cfgstate, \precontract, \lpc)$ and for all $\callstack$ s.t. $\size{\callstack} = \cd$ it holds that $\sstep{\transenv}{\cons{\annotate{\exstate}{\evmcontract}}{\callstack}}{\cons{\annotate{\excstate}{\evmcontract}}{\callstack}}$ and $\evmcontract[\lpc] = \INVALID$
 \item $n = \haltnode$ and for $(\transenv, \exstate) = \toevmstate(\cfgstate, \precontract, \lpc)$ and for all $\callstack$ s.t. $\size{\callstack} = \cd$ it holds that $\sstep{\transenv}{\cons{\annotate{\exstate}{\evmcontract}}{\callstack}}{\cons{\annotate{\haltstate{\access{\exstate}{\gstate}}{\lgas}{\datav}{}}{\evmcontract}}{\callstack}}$ for some $\lgas$ and $\datav$.
\end{enumerate}
\end{enumerate}
\end{theorem}


Note that the CFG semantics only aims at modeling a single contract execution. 
In particular, it does not consider the effects that a finalized execution may have on the caller (e.g., it does not model that global state variables are reverted in case that the execution halted exceptionally).

\begin{proof}
    We prove the two directions of the proof separately: 
    \begin{itemize}
        \item [$\Rightarrow$] Let $\combistep{\transenv}{\cons{\annotate{\exstate}{\evmcontract}}{\callstack}}{\cons{\annotate{\exstate'}{\evmcontract}}{\callstack}}$ be an execution of contract $\evmcontract$ that does not exhibit local out-of-gas exceptions and let $\exstate$ be strongly consistent with $\evmcontract$.
        We do a case distinction on $\exstate'$
        \begin{enumerate}
            \item $\exstate' = \regstate{\mstate'}{\exenv'}{\gstate'}$. We do a further case distinction on $\combistep{\transenv}{\cons{\annotate{\exstate}{\evmcontract}}{\callstack}}{\cons{\annotate{\exstate'}{\evmcontract}}{\callstack}}$. 
            \begin {itemize}
            \item The execution step was a local step and hence $\sstep{\transenv}{\cons{\annotate{\exstate}{\evmcontract}}{\callstack}}{\cons{\annotate{\exstate'}{\evmcontract}}{\callstack}}$. 
            In this case we know that $\access{\exstate}{\exenv} = \exenv'$ and $\access{\exstate}{\gstate} = \gstate'$. The proof trivially follows by case distinction over the instruction in $\access{\access{\exstate}{\exenv}}{\activecode}[\access{\access{\exstate}{\mstate}}{\pc}]$ using the fact that due to strong consistency $\access{\access{\exstate}{\exenv}}{\activecode} = C$. 
            \item The execution step was a call step and hence $\transactionstep{\transenv}{\cons{\annotate{\exstate}{C}}{\callstack}}{\cons{\annotate{\exstate'}{C}}{\callstack}}$. 
            By the definition of the CFG rule for calling which leverages the notion of a transaction step on the EVM semantics, we only need to argue that the changes on $\cfgstate'$ (in the corresponding rule sequences) preserve equivalence. This immediately follows from~\Cref{lem:cfgstatequiv-props}.  
             \end{itemize}
            \item $\exstate' = \excstate$. Since the execution does not exhibit local out of gas exceptions we know by~\Cref{asm:gas} that in this case $\evmcontract(\access{\access{\exstate}{\mstate}}{\pc}) = \INVALID$. This case, hence, follows immediately from the CFG semantics rule for the $\INVALID$ opcode.
            \item $\exstate' = \haltstate{\gstate}{\lgas}{\datav}{}$. This case follows immediately from the CFG semantics rules for the halting instructions $\STOP$ and $\RETURN$.  
        \end{enumerate}
        \item [$\Leftarrow$] Let $\cfgmedstep{\precontract, \cd}{\cfgconfig{\cfgnode{\lpc}{0}}{\cfgstate \uplus \cfgstatecopybot}}{\cfgconfig{\node}{\cfgstatefull}}$. 
        The proof follows by a simple case distinction on $\evmcontract[\lpc]$ using the CFG semantics rules for the individual instructions and taking advantage of~\Cref{lem:cfgstateequiv-exstate} to reason about the equality of the execution states. 
    \end{itemize}
\end{proof}

Note that we did not need to make use of the assumption that the contract is store unreachable for the equivalence proof. 
This is since this requirement is only needed to prove the consistency of the Def and Use sets. 

We revisit the assumption of store unreachability using full contract annotations: 
\begin{assumption}[Store unreachability]
    \label{asm:storeunreachability-full}
    A contract $c$ is store unreachable
    if $\{ \DELEGATECALL, \CALLCODE\} \cap \access{c}{\contractcode} = \emptyset$ and 
    for all regular execution states $\regstate{\mstate}{\exenv}{\gstate}$ that are strongly consistent with $c$, it holds that for all transaction environments $\transenv$ and all callstacks $\callstack$
    \par
    \nobreak
    {
    \noindent
    \begin{align*}
    \neg \exists \exstate, &\callstack'. \,
    \ssteps{\transenv}{\cons{\annotate{\regstate{\mstate}{\exenv}{\gstate}}{c}}{\callstack}}{\cons{\annotate{\exstate}{c}}{\concatstack{\callstack'}{\callstack}}} \\
    &~\land~ \size{\callstack'} > 0 ~ \land ~
    \access{c}{\contractcode}(\access{\access{\exstate}{\mstate}}{\pc}) = 
    (\instruction(\vec{x}), \nextpc) 
    ~\land~ \instruction \in \storeinstructions
    \end{align*}
    }%
    Where the set $\storeinstructions$ of store instructions is defined as
    \par
    \nobreak
    {
    \noindent
    \[\storeinstructions =  \{\SSTORE \}\]
    }%
    \end{assumption}

The key property following from the store unreachability is the following:

\begin{lemma}[Store unreachability implies global storage preservation]
Let $\evmcontract$ be a store unreachable contract and $\exstate$ an execution state that is strongly consistent with $\evmcontract$.
Then for all callstacks $\callstack$ and execution states $\exstate'$
\begin{align*}
\transactionstep{\transenv}{\cons{\annotate{\exstate}{c}}{\callstack}}{\cons{\annotate{\exstate'}{c}}{\callstack}}
\Rightarrow \access{\access{\exstate}{\gstate}(\access{\access{\exstate}{\exenv}}{\activeaccount})}{\stor} = \access{\access{\exstate'}{\gstate}(\access{\access{\exstate'}{\exenv}}{\activeaccount})}{\stor}
\end{align*}
\end{lemma}
\begin{proof}
    Assume towards contradiction that $\access{\access{\exstate}{\gstate}(\access{\access{\exstate}{\exenv}}{\activeaccount})}{\stor} \neq \access{\access{\exstate'}{\gstate}(\access{\access{\exstate'}{\exenv}}{\activeaccount})}{\stor}$. 
    Then there must have been a callstack $\callstack^*$ and an execution state $\exstate^*$ and a contract $\evmcontract^*$ such that 
    $\ssteps{\transenv}{\cons{\annotate{\exstate}{\evmcontract}}{\callstack}}{\cons{\annotate{\exstate^*}{c*}}{\concatstack{\callstack^*}{\cons{\annotate{\exstate}{c}}{\callstack}}}}
    \rightarrow \concatstack{\cons{\annotate{\exstate^+}{c+}}{\callstack^+}}{\cons{\annotate{\exstate}{c}}{\callstack}} \rightarrow^* \cons{\annotate{\exstate'}{c}}{\callstack}$
    such that $\access{\access{\exstate^*}{\gstate}(\access{\access{\exstate}{\exenv}}{\activeaccount})}{\stor} \neq \access{\access{\exstate^+}{\gstate}(\access{\access{\exstate}{\exenv}}{\activeaccount})}{\stor}$.
    Since storage can only be altered via the $\SSTORE$ instruction (which is a local instruction), we know that $c^* = c^+$ and $\callstack^* = \callstack^+$ and $\access{\exstate^*}{\exenv} = \access{\exstate^+}{\exenv}$ and $\access{\access{\exstate^*}{\exenv}}{\activecode}(\access{\access{\exstate^*}{\mstate}}{\pc}) = (\SSTORE(x, y), \nextpc)$.
    Further, since one can only write the storage of the active account, we know that $\access{\access{\exstate^*}{\exenv}}{\activeaccount} = \access{c}{\contractaddress}$. 
    Due to strong consistency (which is preserved during execution for contacts without $\DELEGATECALL$ and $\CALLCODE$), we hence know that $\access{\access{\exstate^*}{\exenv}}{\activecode} = \access{c}{\contractcode}$ and so also $c^* = c$. 
    From~\Cref{asm:storeunreachability-full}, we can derive a contradition.
\end{proof}


\subsection{Correctness of May Analysis}
\label{sec:may-analysis}

We start by introducing definitions and lemmas needed for our soundness claims.
\begin{definition}[Backward Slice of Set]
    $$ \bsn = \bigcup\limits_{n \in N} \textit{BS}(n) $$
\end{definition}

\begin{definition}[Execution steps]
    \label{def:cfg_execution}
    \textit{Execution steps with states are defined as follows} \\
    \begin{mathpar}
        \infer
        { n \qedge n' \\ Q(\cfgstate)}
        {\acfgconfigstep{\sconf{\cfgstate}{n}}{\qact}{\sconf{\cfgstate}{n'}}}
        
        \infer
        {n \fedge n'	\\
            \cfgstate' = f(\cfgstate) }
        {\acfgconfigstep{\sconf{\cfgstate}{n}}{\fact}{\sconf{\cfgstate'}{n'}}}
    \end{mathpar}
    and $\xrightarrow{}^*$ denotes the transitive closure.
\end{definition}

\begin{definition}[Deterministic successor function]
    \label{def:det_suc_func}
    We define the function \enquote{$\ds$} as
    \begin{align*}
        \ds(n) =
        \begin{cases}
            n'		& \text{if } \exists n'. \; n \fedge n' \\
            n'	& \text{if } \exists n'. \; n \qedge n' \land~ Q~1
        \end{cases}
    \end{align*}
\end{definition}

\begin{definition}[Sliced execution steps]
    \label{def:cfg_execution_sliced}
    \begin{mathpar}
        \infer{
            \sconf{\cfgstate}{n} \xrightarrow{a} \sconf{\cfgstate'}{n'} \\
            \bs{n}{N}
        }
        {
            \sconf{\cfgstate}{n} \xrightarrow{a}_{\bsn} \sconf{\cfgstate'}{n'}
        }
        \and
        \infer{
            \sconf{\cfgstate}{n} \xrightarrow{a} \sconf{\cfgstate'}{n'} \\
            \notbs{n}{N} \\
            n''= \ds(n)
        }
        {
            \sconf{\cfgstate}{n} \xrightarrow{\tau}_{\bsn} \sconf{\cfgstate}{n''}
        }
    \end{mathpar}	
\end{definition}

\begin{definition}[Sliced execution paths]
    \label{def:cfg_execution_sliced_path}
    We define the sliced execution path $\xrightarrow{as}_{\bsn}^*$ as the transitive $\tau$ closure of \Cref{def:cfg_execution_sliced} and require that $\xrightarrow{as}_{\bsn}^*$ must end at a node in $\bsn$.
    We write $\xrightarrow{as}_{\bsn}^m$ if path \enquote{$as$} in $\xrightarrow{as}_{\bsn}^*$ has exactly $m$ observable steps.
\end{definition}

\begin{definition}[Up-to equality]
    \label{def:quasi_equal}
    Let \textit{Var} be the set of all variables used in a contract $\precontract$. States $\cfgstate_1$ and $\cfgstate_2$ are equal except for variable $X$ if and only if 
    $$ \forall V \in \textit{Var} / \{X\}. ~ \eqstatevar{\cfgstate_1}{\cfgstate_2}{V}$$
    We write $\cfgstate_1 \equalupto{X} \cfgstate_2$ in that case.
\end{definition}

\begin{definition}[Deterministic successor]
    \label{def:det_suc_not}
    Node $n^{+1}$ is defined as the successor of node $n$ in a contract $\precontract$ if it is the sole successor of $n$ and undefined otherwise:
    $$n^{+1} = 
    \begin{cases}
        n' & \text{if } \forall n',~n''.~ n \rightarrow n' \land~ n \rightarrow n'' \implies n' = n'' \\
        \textit{undefined} & \text{otherwise}
    \end{cases}
    $$
\end{definition}

\begin{definition}[Step-indexed execution]
    \label{def:counting_defs}
    We define $\nstep{\as}{}{N}{i}$ by 
    \begin{mathpar}
        \infer{
            \conf{n}{\cfgstate} \xrightarrow{\as}^* \conf{n'}{\cfgstate'} \\
            \mid \projectpath{\as}{N}  \mid = i
        }
        {
            \conf{n}{\cfgstate} \nstep{\as}{}{N}{i} \conf{n'}{\cfgstate'}
        }
    \end{mathpar}
    where $\projectpath{}{N}$ filters out the actions where the source node is not in $N$.
\end{definition}

\begin{definition}[Relevant variables \cite{wasserrab2009pdg}]
    \label{def:rel_vars}
    Relevant variables of backward slice of $S$ at node $n$ are defined as follows
    \begin{mathpar}
        \infer{
            n \nstep{\as}{}{N_V}{0} n' \\
            \bs{n'}{S} \\
            \useset{V}{n'}
        }
        {
            V \in \rv{S}{n}
        }
    \end{mathpar}
\end{definition}

\begin{lemma}[Slice usage set property]
    \label{lemma:no_x_use_in_bs}
    \begin{align*}
        &\phantom{\implies } (\forall n_X. ~ \notbs{n_X}{N_Y}) 
        \implies \sconf{\cfgstate}{\sucnode{X}} \xrightarrow{~}_{\bsi(N_Y)}^{m} \sconf{\_}{n^+} 
        \\
        &\implies (\forall k \leq m.~ \sconf{\cfgstate}{\sucnode{X}} \xrightarrow{ }^k_{\bsi(N_Y)} \sconf{\_}{n^*} 
        \implies \notuseset{X}{n^*})
    \end{align*}
\end{lemma}


The lemma states that for all paths to nodes $n^+$ of the backward slice $N_Y$, it holds that each step ($n^*$) in the backward slice towards $n^+$ does not use $X$.


\begin{proof}
    We assume $\notbs{n_X}{N_Y}$ and $\sconf{\cfgstate}{\sucnode{X}} \xrightarrow{~}_{\bsi(N_Y)}^{m} \sconf{\_}{n^+} $. \\
    Towards contradiction we assume there would be a $k$ such that 
    $$ \sconf{\cfgstate_1}{\sucnode{X}} \xrightarrow{ }^k_{\bsi(N_Y)} \sconf{\_}{n^*} \wedge \useset{X}{n^*}$$
    Then we define $n^*_X$ with $l \le k$ such that
    $$ \sconf{\cfgstate_1}{\sucnode{X}} \xrightarrow{ }^l_{\bsi(N_Y)} \sconf{\cfgstate_1^*}{n_X^*} \wedge \defset{X}{n_X^*}$$
    where $n_X^*$ identifies the last node where $X$ was defined on the path to node $n^*$ or if such an $l$ does not exists, we know
    $$ \sconf{\_}{n_X} \xrightarrow{ } \sconf{\cfgstate_1}{\sucnode{X}} \nstep{}{\bsi(N_Y)}{N_X}{0} \sconf{\_}{n^*} $$
    This means that $X$ is either defined along the path or at the predecessor of the start node since we start at a successor of a definition of $X$.
    This node $n_X^*$ ($n_X$) is then by definition in the backward slide of $N_Y$ since $X$ is used at node $n^*$ in the backward slice and it was the closest definition.
    This contradicts the assumption $\notbs{n_X^*}{N_Y}$.
\end{proof}

The proofs use the Slicing framework's soundness claim for backward slices of a set of nodes instead of single nodes.
\begin{theorem}[Correctness of Slicing Based on Paths and Sets \cite{wasserrab2008towards}]
    \label{proof:correctness_slicing_sets}
    \begin{mathpar}
        \infer{
            \sconf{\cfgstate}{n} \xrightarrow{\as}^* \sconf{\cfgstate'}{n'} \\
            n' \in S \\
        }
        {
            \exists~ \as'.~ \sconf{\cfgstate}{n} \xrightarrow{\as'}^*_{\bsi(S)} \sconf{\cfgstate''}{n'} ~
            \wedge~ (\forall~ \useset{V}{n'}. \eqstatevar{\cfgstate'}{\cfgstate''}{V})~ \wedge~
            \projectpath{\textit{\as}}{\textit{BS}(S)} = \textit{as}'
        }
    \end{mathpar}
    
\end{theorem}

The following lemma lifts sliced executions while preserving the number of visits at nodes in the backward slice.

\begin{lemma}
    \label{lemma:VMDP_big}
    All executions within the backward slice of node-set $N$ can be mapped to real executions with the same number of visits in any subset of the backward slice and all relevant variables concerning the final node $n'$ of the execution were computed correctly.
    \begin{align*}
        &\phantom{\implies } n' \in N \implies N' \subseteq N
        \implies \sconf{\cfgstate}{n} \nstep{\as'}{\bsi(N)}{N'}{i} \sconf{\cfgstate'}{n'} \\
        &\implies \exists \as, \cfgstate''. \sconf{\cfgstate}{n} \nstep{\as}{}{N'}{i} \sconf{\cfgstate''}{n'} 
        \wedge \forall V \in rv ~ N ~ n'.~ \eqstatevar{\cfgstate'}{\cfgstate''}{V}
    \end{align*}
\end{lemma}
\begin{proof}
    We first show that path $\as$ exists outside of the backward slice: 
    $$ (1)~ \sconf{\cfgstate}{n} \xrightarrow{\as} \sconf{\cfgstate''}{n'}$$
    Path $\as$ is constructed by extending the sliced path $\as'$ with nodes outside of the backward slice.
    By \Cref{def:cfg_execution_sliced} and since every step in $\sconf{\cfgstate}{n} \xrightarrow{as} \sconf{\cfgstate'}{n'}$ would be deterministic by the slicing framework requirements, we can reconstruct an unique original path $\as$ by adding nodes along the path $\as'$ where $\tau$ edges were followed.
    
    Therefore, we construct the following path $\as$ where all nodes $\overline{n_i}$ with a line on top were not part of the original sliced path.
    \begin{center}
        \begin{tikzpicture}
            \node (AS) at (-0.3,0) {$(2)\quad \as~ =$};
            \node (A) at (1,0) {$n$};
            \node (B) at (3,0) {$n_1$};
            \node (C) at (5,0) {$\overline{n_1}$};
            \node (D) at (7,0) {$\overline{n_i}$};
            \node (E) at (9,0) {$n'$};
            \node (DOTS) at (6,0) {\textbf{\ldots}};
            \draw[->] (A) to node[above] {$\as'_1$} (B);
            \draw[->,dotted] (B) to node[above] {$\overline{\as'_1}$} (C);
            \draw[->] (D) to node[above] {$\as'_{i+1}$} (E);
        \end{tikzpicture}
    \end{center}
    where the original path was
    \begin{center}
        \begin{tikzpicture}
            \node (AS) at (-0.3,0) {$\phantom{(2)\quad}\as'~ =$};
            \node (A) at (1,0) {$n$};
            \node (B) at (3,0) {$n_1$};
            \node[label=below:{$\_$}] (C) at (5,0) {};
            \node[label=below:{$\_$}] (D) at (7,0) {};
            \node (E) at (9,0) {$n'$};
            \node (DOTS) at (6,0) {\textbf{\ldots}};
            \draw[->] (A) to node[above] {$\as'_1$} (B);
            \draw[->,dotted] (B) to node[above] {$\tau*$} (C);
            \draw[->] (D) to node[above] {$\as'_{i+1}$} (E);
        \end{tikzpicture}
    \end{center}
    The construction described above is feasible by the definition of $\xrightarrow{as}_{\bsn}^*$ and satisfies (1).
    We remark that path $\as'$ is an ordered sublist of the original sliced-path $\as$ as nodes were only inserted while construction in step (2) but not removed:
    $$ (3)~ \as' \subseteq_{od} \as$$
    It is left to be shown that path $\as$ has the same number of visits in the set $N'$ as the path $\as'$ in the sliced graph:
    $$ (4)~ \sconf{\cfgstate}{n} \nstep{\as}{}{N'}{i} \sconf{\cfgstate''}{n'}$$
    Recall that by assumption, all nodes in $N'$ are part of $N$ and therefore elements of the backward slice of $\bsi(N)$.
    By construction of (2) and (3), we know that the number of visits at $N'$ is the same in both paths and (4) holds. 
    We now show that for final state $\cfgstate''$ from (4) and the final state $\cfgstate'$ from the sliced-path it holds that
    $$ (5)~ \forall V \in rv ~ N ~ n'.~ \eqstatevar{\cfgstate'}{\cfgstate''}{V} $$ 
    All nodes where any variable $V$ in $rv ~ N ~ n'$ is defined are part of the backward slice by \Cref{def:rel_vars}.
    Therefore, no node of this kind was added or deleted in the construction of $\as$ and they were already present in $\as'$. \\
    With that, we can use the correctness statement of slicing to show that all relevant variables reassigned along $\as'$ have the same value in $\as$.
    Non-reassigned variables are not changed and stay the same in both states.
    Therefore, all relevant variables have the same value in $\cfgstate'$ and $\cfgstate''$.
    
    The statement follows from (4) and (5).
\end{proof}

\paragraph{Variable Independence}

\begin{lemma}[Variable Dependency Predicate]
    \label{lem:pred_slicing_app}
    Let $\precontract$ be a contract and $X$ and $Y$ be variables thereof. Then it holds that 
    \begin{align*}
           	\pred{VarMayDepOn}(Y,X) \not \in 	\lfp{\rules{\precontract}} &\Rightarrow \forall n_Y, n_X.~ \notbs{n_X}{n_Y}
     \end{align*}
\end{lemma}

\begin{definition}[Variable Independence]
    \label{def:var_ind_app}
    A variable $Y$ is independent of a variable $X$ if and only if for all states $\cfgstate_1$, $\cfgstate_2$ and $\cfgstate_1'$ it holds that 
    \begin{align*}
        &\forall n_X, n_Y, i.~
        \cfgstate_1 \equalupto{X} \cfgstate_2 
        \wedge ~\conf{\sucnode{X}}{\cfgstate_1} \nstep{}{}{N_Y}{i}  \conf{\sucnode{Y}}{\cfgstate_1'} \\
        &\implies \exists \cfgstate_2'. ~ \conf{\sucnode{X}}{\cfgstate_2} \nstep{}{}{N_Y}{i} \conf{\sucnode{Y}}{\cfgstate_2'} \wedge  \cfgstate_1'(Y) = \cfgstate_2'(Y) 
    \end{align*}
    We will refer to this definition by $\vio(n,X)$.
\end{definition}

\begin{theorem}[Soundness of Variable Independence]
    \label{app:thm:soundness-independence}
    Let $X$ and $Y$ be variables and $\notbs{n_X}{n_Y}$ for all nodes $n_X$, $n_Y$. 
    Then variable $Y$ is independent of variable $X$. 
\end{theorem}

\begin{proof}
    Let $C$ be a contract and $\precontract$ be a contract with sound preprocessing information that is consistent with $C$.
    We assume that $(1)~ \notbs{n_X}{n_Y}$ holds for all nodes $n_X$ and $n_Y$ and an executions from node $\sucnode{X}$ to node $\sucnode{Y}$ that starts in states $\cfgstate_1$. 
    Let $\cfgstate_2$ be a state such that $(2) ~\cfgstate_1 \equalupto{X} \cfgstate_2$ and
    $$ (3)~\sconf{\cfgstate_1}{\sucnode{X}} \nstep{as}{}{N_Y}{i} \sconf{\cfgstate_1'}{\sucnode{Y}}$$ 
    identity the previous execution with path $\as$.

    From (3) we know from correctness of slicing (\Cref{proof:correctness_slicing_sets}) that
    $$ (5)~ \exists as'. \sconf{\cfgstate_1}{\sucnode{X}} \nmstep{as'}{\bsi(N_Y)}{N_Y}{j}{m} \sconf{\cfgstate_1^*}{\sucnode{Y}},$$%
    $$ (6)~ \eqstatevar{\cfgstate_1'}{\cfgstate_1^*}{Y} $$
    $$ \text{and}~(7)~ \projectpath{as}{N_Y} ~=~ as'$$
    Equations (5) and (6) say that there exists a corresponding path in the backward slice and it has the correct value for $Y$.
    Equation (7) implies that all nodes in the backward slice of $N_Y$ are visited in the same order in $as$ and $as'$ \cite{wasserrab2009pdg}. 
    Since all nodes where $Y$ is defined are by defintion of $N_Y$ in the backward slice, we know from (7) that $i = j$. \\
    We now show that path $\as'$ exists in the sliced-graph when starting from state $\cfgstate_2$:
    $$
    (9)~ \sconf{\cfgstate_2}{n^{+1}_X} \nmstep{as'}{\bsi(n_Y)}{N_Y}{j}{m} \sconf{\cfgstate_2^*}{n^{+1}_Y}
    \quad \text{ and } \quad
    \eqstatevar{\cfgstate_1^*}{\cfgstate_2^*}{Y}
    $$
    With \Cref{lemma:no_x_use_in_bs} we get from (1) and (5) for all $k \leq m$ that
    $$
    (8)~ \sconf{\cfgstate_1}{\sucnode{X}} \xrightarrow{ }^k_{\bsi(N_Y)} \sconf{\cfgstate_1^+}{n^+} 
    \implies \notuseset{X}{n^+} 
    $$
    We can use (8) to show (9) since at all nodes $n^+$, that would have effected the trace, $X$ is not used.
    From (2) we know by the slicing framework's well-formedness properties that $X$ is only source of different effects.
    Therefore, we know that (9) holds. \\
    From (9) we can use \Cref{lemma:VMDP_big} to conclude that 
    $$ \sconf{\cfgstate_2}{n^{+1}_X} \nstep{as''}{}{N_Y}{i} \sconf{\cfgstate_2'}{n^{+1}_Y} ~\wedge~ \forall V \in \textit{rv}~N_Y~ n_Y.~ \eqstatevar{\cfgstate_2'}{\cfgstate_2^*}{V}$$ 
    and in particular 
    $$\eqstatevar{\cfgstate_2'}{\cfgstate_2^*}{Y}$$
    since $Y \in \textit{rv}~N_Y~ n_Y$.
\end{proof}

\paragraph{Instruction Independence}
\begin{lemma}[Instruction Dependency Predicate]
    \label{lem:inst_pred_slicing_app}
    Let $\precontract$ be a contract, $n$ be a node and $X$ be variable thereof. Then it holds that 
    \begin{align*}
        \pred{InstMayDepOn}(n,X) \not \in 	\lfp{\rules{\precontract}} &\Rightarrow \forall \nif, n_X.~ \nif \xrightarrow{~}_{\cd} n \Rightarrow \notbs{n_X}{\nif}
    \end{align*}
\end{lemma}

\begin{definition}[Instruction Independence]
    \label{def:inst_ind}
    A node $n$ is independent of a variable $X$ if and only if for all states $\cfgstate_1$ and $\cfgstate_1$ it holds that 
    \begin{align*}
        &\forall n_X, n', i. ~
        \cfgstate_1 \equalupto{X} \cfgstate_2 
        \wedge 
        \conf{n^{+1}_X}{\cfgstate_1} \nstep{\as}{}{n'}{i} \conf{n'}{\_} \wedge \conf{\sucnode{X}}{\cfgstate_2} \nstep{\as'}{}{n'}{i} \conf{n'}{\_}\\
        &\implies \projectpatheqlength{\as}{\as'}{n}
    \end{align*}
    We will refer to this definition by $\iio(n,X)$.
\end{definition}

\begin{theorem}[Soundness of Instruction Independence]
    \label{thm:inst-soundness-independence}
    Let $n$ be a node, $X$ be a variable and $\notbs{n_X}{\nif}$ hold for all nodes $n_X$ and $\nif$ where $\nif \xrightarrow{~}_{cd} n$. 
    Then node $n$ is independent of variable $X$. 
\end{theorem}

\begin{proof}
    Let $C$ be a contract and $\precontract$ be a contract with sound preprocessing information that is consistent with $C$.
    We assume that $(1) ~\notbs{n_X}{\nif}$ holds for all nodes $n_X$ and $\nif$ such that $\nif \xrightarrow{~}_{cd} n$ and two executions from node $\sucnode{X}$ to node $n'$ that starts in states $\cfgstate_1$ and $\cfgstate_2$ with $(4) ~\cfgstate_1 \equalupto{X} \cfgstate_2$ and visit $n'$ on their paths $i$ times.
    Let 
    $$(2) ~ \sconf{\cfgstate_1}{\sucnode{X}} \nstep{\as}{}{n'}{i} \conf{n'}{\_}
    \quad \text{and} \quad
    (3) ~ \sconf{\cfgstate_2}{\sucnode{X}} \nstep{\as'}{}{n'}{i}\conf{n'}{\_}$$ 
    be these two executions with paths $\as$ and $\as'$.
    
    The interesting case is node $n' \not = n$ since the conclusion for node $n' = n$ follows with assumption (3) and (4) by definition of $\nstep{}{}{n'}{i}\!\!$. 
    Therefore, we show the statement for $n' \not = n$.
    
    We show that paths $\as$ and $\as'$ visit node $n$ the same number of times by contradiction.
    W.l.o.g we assume that path $\as$ visits node $n$ at least once more than path $\as'$:
    $$ (*)~~ \mid \projectpath{\as}{n} \mid~ >~ \mid \projectpath{\as'}{n} \mid$$
    
    From $(*)$ we know that there exists a prefix $\overline{\as}$ of path $\as$ such that it can visit node $n$ once more than path $\as'$:
    \begin{align*}
        (5.1)& ~~ \sconf{\cfgstate_1}{\sucnode{X}} \nstep{\overline{\as}}{}{n}{k} \conf{n}{\_}  \nstep{}{}{n'}{0}\hspace{-1.4cm} &\conf{n'}{\_}&~\wedge~ \overline{\as} \subseteq_{od} \as& \\
        \intertext{
            with $k = \mid \projectpath{\as'}{n} \mid$ such that $\overline{\as}$ ends at the 	$k+1$ occurrence of $n$.
            Node $n'$ is reachable from $n$ because the prefix $\overline{\as}$ can allows be completed to the full path from (3).
        }
        \intertext{
            Next, we define a corresponding prefix $\overline{\as'}$ of $\as'$ that vists $n$ the maximal number of $k$ times and ends at the same appearence of $n'$.
        } 
        (5.2)& ~~ \sconf{\cfgstate_2}{\sucnode{X}} \; \; \nstep{\quad\qquad\overline{\as'}\quad\qquad}{}{n'}{j}\hspace{-1.4cm} &\conf{n'}{\_}&~\wedge~\overline{\as'} \subseteq_{od} \as'&
    \end{align*}
    with $j = \mid \projectpath{\overline{\as}}{n'} \mid$. 
    That path exists based on (4) and $(*)$.
    We know that path $\overline{\as'}$ visits $n$ less often than $\overline{\as}$; and $\overline{\as'}$ consists of all visits of $n$ in $\as'$.
    
    We know by construction of (5.1) and (5.2) that $\overline{\as}$ and $\overline{\as'}$ split up after the $k$-th visit at $n$ because otherwise both would visit node $n$ at least $k+1$ number of times. We call this split node $\nif$ and make the split explicit:
    \begin{align*}
        (6)~~ \exists \nif, w, z.~~ \conf{\sucnode{X}}{\cfgstate_1} \xrightarrow{ \overline{\as_1}}{}^* &\countnodes{\nif}{w}~ \conf{\nif}{\cfgstate^{if}_1} \xrightarrow{ \overline{\as_2}}{}^*~ \conf{n}{\_} \xrightarrow{}{}^*~ \conf{n'}{\_} \\
        \land~ 
        \conf{\sucnode{X}}{\cfgstate_2} \xrightarrow{ \overline{\as'_1}}{}^*&\countnodes{\nif}{z}~ \conf{\nif}{\cfgstate^{if}_2} \xrightarrow{\;\; \qquad \overline{\as'_2} \; \; \qquad}{}^*~ \conf{n'}{\_}
    \end{align*}
    where $\overline{\as} = \overline{\as_1}~ @ ~\overline{\as_2}$ such that $\overline{\as_2}$ and $\overline{\as'_2} $ do not share a single node, i.e.\ $\nif$ is the last node where the paths could split up.
    From (3) -- (6), we can conclude with standard control dependence that
    $$(7) ~~ \transscd{\nif}{n}$$
    By instantiating (1) with (7) we know
    $$ (8)~~ \forall n_X^0.~ \notbs{n_X^0}{\nif}$$ 
    From (7) we know that states $\cfgstate^{if}_1$ and $\cfgstate^{if}_2$ have at least one different value for some input variable at $\nif$.
    Only a different value in the usage set can lead to splitting control flow.
    Therefore, we get:
    $$ (9)~~ \neg \forall~ \useset{V}{\nif}. ~\eqstatevar{\cfgstate^{if}_1}{\cfgstate^{if}_2}{V} $$
    
    We now map paths $\overline{\as_1}$ and $ \overline{\as'_1}$ from (5) into the sliced graph with \Cref{proof:correctness_slicing_sets} of correctness of slicing.
    Those paths compute correct input values at $\nif$. 
    We get
    $$ (10)~~ \exists~ \overline{\overline{\as_1}}.~ \sconf{\cfgstate_1}{\sucnode{X}} \nmstep{\overline{\overline{\as_1}}}{\bsi(\nif)}{\nif}{w}{g} \sconf{\cfgstate^*_1}{\nif} ~
    \wedge~ \forall~ \useset{V}{\nif}.~ \eqstatevar{\cfgstate^{if}_1}{\cfgstate^*_1}{V} $$
    $$ (11)~~ \exists~ \overline{\overline{\as'_1}}.~ \sconf{\cfgstate_2}{\sucnode{X}} \nmstep{\overline{\overline{\as'_1}}}{\bsi(\nif)}{\nif}{z}{h} \sconf{\cfgstate^*_2}{\nif} ~
    \wedge~ \forall~ \useset{V}{\nif}. ~\eqstatevar{\cfgstate^{if}_2}{\cfgstate^*_2}{V}$$
    
    We apply paths $\overline{\overline{\as_1}}$ and $\overline{\overline{\as'_1}}$ to \Cref{lemma:no_x_use_in_bs} with (8) and get
    $$ (12)~ \forall~ q' \le q.~  \sconf{\cfgstate_p}{\sucnode{X}} \xrightarrow{ }^{q'}_{\bsi(L_{\textit{if}})} \sconf{\cfgstate^*_p}{n^+} \implies \notuseset{X}{n^+} $$ 
    with $(p, q) \in \{(1, g),~(2, h)\}$. \\
    Assumption (2) states that only $X$ can propagate changes, but we know from (12) that no node in the backward slice on either path uses $X$. 
    Therefore, executions $\sconf{\cfgstate_1}{\sucnode{X}} \xrightarrow{~}_{\bsi(\nif)}^* \_$  and $\sconf{\cfgstate_2}{\sucnode{X}} \xrightarrow{~}_{\bsi(\nif)}^* \_$ have the same states (up to $X$) after the same number of steps:
    \begin{align*}
        (13) ~~ 
        &\sconf{\cfgstate_1}{\sucnode{X}} \xrightarrow{xs}_{\bsi(\nif)}^t \sconf{\hat{\cfgstate_1}}{~\_} \land~
        \sconf{\cfgstate_2}{\sucnode{X}} \xrightarrow{ys}_{\bsi(\nif)}^t \sconf{\hat{\cfgstate_2}}{~\_} \\ 
        & \implies \hat{\cfgstate_1} \equalupto{X} \hat{\cfgstate_2}
    \end{align*}
    Since $\bs{\nif}{\nif}$, we know that paths in (10) and (11) reach the split node $\nif$ after the same number of steps in the backward slice:
    $$(14) ~~ w = z \qquad
    \text{ and therefore choose} \quad
    t =~ \mid \overline{\overline{as_1}} \mid~=~\mid \overline{\overline{as_1}} \mid$$
    Finally, by (13) and (14) on paths (10) and (11) we get $\cfgstate^*_1 \equalupto{X} \cfgstate^*_2 $ and conclude by (12):
    $$ (15)~~ \forall~ \useset{V}{\nif}.~ \eqstatevar{\cfgstate^*_1}{\cfgstate^*_2}{V} $$ 
    The state equivalence reasoning is summarized in the following for all $\useset{V}{\nif}$:
    \begin{center}
        \begin{tikzpicture}
            \node (A') at (0,-3) {$\cfgstate^{if}_1(V)$};
            \node (B') at (0,-5) {$\cfgstate_1^*(V)$};
            \node (C') at (4,-5) {$\cfgstate_2^*(V)$};
            \node (D') at (4,-3) {$\cfgstate^{if}_2(V)$};
            \node (AB') at (0,-4) {$\stackrel{(10)}{=}$};
            \node (BC') at (2,-5) {$\stackrel{(15)}{=}$};
            \node (CD') at (4,-4) {$\stackrel{(11)}{=}$};
        \end{tikzpicture}
    \end{center}
    Thereby, we conclude $\forall~ \useset{V}{\nif}. ~\eqstatevar{\cfgstate^{if}_1}{\cfgstate^{if}_2}{V}$ which contradicts (9).
\end{proof}

\begin{definition}[Environmental Instruction Independence]
    \label{def:iio_lifted}
    A node $n$ is independent of a constant environmental variable $X$ if and only if for all states $\cfgstate_1$ and $\cfgstate_2$ it holds that 
    \begin{align*}
        &\forall n^0, n', i. ~
        \cfgstate_1 \equalupto{X} \cfgstate_2
        \wedge 
        \conf{n^0}{\cfgstate_1} \nstep{\as}{}{n'}{i}  \conf{n'}{\_} \wedge \conf{n^0}{\cfgstate_2} \nstep{\as'}{}{n'}{i} \conf{n'}{\_}\\
        &\implies \projectpatheqlength{\as}{\as'}{n}
    \end{align*}
\end{definition}

\begin{lemma}[Soundness of Environmental Instruction Independence]
    \label{thm:evn-inst-soundness-independence}
    Let $n$ be a node, $X$ be a constant environmental variable and $\notbs{n_X}{\nif}$ for all nodes $n_X$ and $\nif$ where $\nif \xrightarrow{~}_{cd} n$. 
    Then node $n$ is independent of variable $X$.  
\end{lemma}

\begin{proof}
    Definite assignment was ensured by $\sucnode{X}$ in \Cref{thm:inst-soundness-independence}.
    Since constant environmental variables are definitely assigned and final by design, we can drop this requirement in the soundness claim and reuse the proof of \Cref{thm:inst-soundness-independence}.
\end{proof}

\subsection{Correctness of Trace Noninterference Pattern}
\label{sec:tni}
\newcommand{\fnodeset}{N_\outputproj(\evmcontract)}
\newcommand{\fargs}[1]{\textit{Args}(\evmcontract, #1)}
\newcommand{\fvarset}{\textit{Var}_\outputproj(\evmcontract)}
\newcommand{\fnode}{\node_\outputproj}
\newcommand{\fvar}{x_f}
\newcommand{\varset}{V}
\newcommand{\interval}[2]{[#1, #2]}
\newcommand{\conc}{\cdot}
\newcommand{\steps}[1]{\xrightarrow{#1}^*~}
\newcommand{\step}[2]{\xrightarrow{#1}^{#2}~}
\newcommand{\combifstep}[4]{\xhookrightarrow{#2}^{#1} \mid_{#4}^{#3}~}
\newcommand{\combinstep}[2]{\xhookrightarrow{#2}^{#1}~}
\newcommand{\combisteps}[4]{#1 \vDash #2 \xhookrightarrow{#4}^* #3}
\newcommand{\cfgmednstep}[1]{\Rightarrow^{#1}}

To prove~\Cref{thm:soundness-tni}

\begin{lemma}
    \label{lem:soundness-help}
Let $\varset$ be a set of CFG state variables and let
\begin{itemize}
    \item $\fnodeset\define \{\node ~|~ \exists ~i.~ \node = (\lpc, i) ~land~ \exists~\instruction ~\vec{x}~ \nextpc ~\pre.~ \evmcontract(\lpc) = (\instruction(\vec{x}), \nextpc, \pre) ~\land~ \outputproj(\instruction) \}$
    \item $\fargs{\node} \define  \{ x ~|~ \exists~i.~\land \node = (\lpc, i) ~\land~  \exists ~ \instruction ~\vec{x} ~\nextpc ~\pre.~ \evmcontract(\lpc) = (\instruction(\vec{x}), \nextpc, \pre) ~\land~ x \in \vec{x} \}$
    \item $\fvarset \define \{ x ~|~ \exists~ \fnode.~ \fnode \in \fnodeset ~\land~ x \in \fargs{\fnode} \}$
\end{itemize}
Further, let $m \in \NN$, $\evmcontract$, $\node, \fnode^1, \nodesucc{\fnode^1}, \cdots, \fnode^{m}, \nodesucc{\fnode^{m}}$, $\as_1, \cdots, \as_m$, $a_1, \cdots, a_m$, and $\cfgstate, \cfgstate', \cfgstate_1, \cfgstate_1^+, \cdots, \cfgstate_m, \cfgstate_m^+$ be arbitrary and assume that 
\begin{enumerate}
    \item $\forall \cfgstate \equalupto{V} \cfgstate'$
    \item $\forall \fvar \in \fvarset.~ \forall v \in \varset.~ \vio(\fvar, v)$ \label{asm:varind}
    \item $\forall \fnode \in \fnodeset.~ \forall v \in \varset.~ \iio(\fnode, v)$ \label{asm:instind}
    \item $\cfgconfig{\node}{\cfgstate} \left(\nstep{\as_i}{}{\fnodeset}{0} \cfgconfig{\fnode^i}{\cfgstate_i} \step{a_i}{} \cfgconfig{\nodesucc{\fnode^i}}{\cfgstate^+_i} \right)^m_{i=1}$
\end{enumerate}
Then there exist $\as_1', \cdots, \as_m'$, $a_1', \cdots, a_m'$, and $\cfgstate_1', \cfgstate_1'^+ \cdots, \cfgstate_m', \cfgstate_m'^+$ such that
\begin{align*}
    &\cfgconfig{\node}{\cfgstate} \left(\nstep{\as_i'}{}{\fnodeset}{0} \cfgconfig{\fnode^i}{\cfgstate_i'} \step{a_i'}{} \cfgconfig{\nodesucc{\fnode^i}}{\cfgstate_i'^+}\right)^m_{i=1} \\
    &~\land~ \forall i \in \interval{1}{m}.~ \forall x \in \fargs{\fnode^i}.~ \cfgstate_i(x) = \cfgstate_i'(x)
\end{align*}
\end{lemma}

\begin{proof}
By induction on $m \in \NN$.
\begin{enumerate}
    \item Let $m = 0$. The claim trivially holds. 
    \item Let $m > 0$. 
    Then $\left(\nstep{\as_i}{}{\fnodeset}{0} \cfgconfig{\fnode^i}{\cfgstate_i} \step{a_i}{} \cfgconfig{\nodesucc{\fnode^i}}{\cfgstate^+_i}\right)^{m-1}_{i=1} 
    \nstep{\as_m}{}{\fnodeset}{0} \cfgconfig{\fnode^m}{\cfgstate_m} \step{a_m}{} \cfgconfig{\nodesucc{\fnode^m}}{\cfgstate^+_m}$ 
    and by the inductive hypothesis also 
    $\cfgconfig{\node}{\cfgstate} \left(\nstep{\as_i'}{}{\fnodeset}{0} \cfgconfig{\fnode^i}{\cfgstate_i'} \step{a_i'}{} \cfgconfig{\nodesucc{\fnode^i}}{\cfgstate_i'^+}\right)^{m-1}_{i=1}$
    for some $\as_1', \cdots, \as_{m-1}'$, $a_1', \cdots, a_{m-1}$, and $\cfgstate_1', \cfgstate_1'^+, \cdots, \cfgstate_{m-1}', \cfgstate_{m-1}'^+$.
    such that $\forall i \in \interval{1}{m-1}.~ \forall x \in \fargs{\fnode^i}.~ \cfgstate_i(x) = \cfgstate_i'(x)$. 
    We are hence left to show that there exists some $\as_m'$, $a_m'$, $\cfgstate_m'$, $\cfgstate_m'^+$ such that
    $\cfgconfig{\nodesucc{\fnode^{m-1}}}{\cfgstate_{m-1}'^+} \nstep{\as_m'}{}{\fnodeset}{0} \cfgconfig{\fnode^m}{\cfgstate_m'} \step{a_m'}{} \cfgconfig{\nodesucc{\fnode^m}}{\cfgstate_m'^+}$, and 
    $\forall x \in \fargs{\fnode^m}.~ \cfgstate_m(x) = \cfgstate_m'(x)$.
    Assume towards contradiction that there is no $\as_m'$, $a_m'$ and $\cfgstate'_m, \cfgstate_m'^+$ such that $\cfgconfig{\nodesucc{\fnode^{m-1}}}{\cfgstate_{m-1}'^+} \nstep{\as_m'}{}{\fnodeset}{0} \cfgconfig{\fnode^m}{\cfgstate_m'} \step{a_m'}{} \cfgconfig{\nodesucc{\fnode^m}}{\cfgstate_m'^+}$, meaning that $\fnode^m$ is not reachable from $\cfgconfig{\fnode^{m-1}}{\cfgstate_{m-1}'}$ without stepping through another node in $\fnodeset$. 
    We consider two cases
    \begin{enumerate}
        \item $\fnode^m$ is not reachable from $\cfgconfig{\nodesucc{\fnode^{m-1}}}{\cfgstate_{m-1}'^+}$. Since every execution ends in the node $\exitnode$, we know that $\cfgconfig{\nodesucc{\fnode^{m-1}}}{\cfgstate_{m-1}'^+} \nstep{\as_{m}'}{}{\{ \fnode^m\}}{0} \cfgconfig{\exitnode}{\cfgstate'_\exitnode}$ and 
        $\cfgconfig{\fnode^{m}}{\cfgstate_{m}} \steps{\as_{m+1}} \cfgconfig{\exitnode}{\cfgstate_\exitnode}$
        for some $\cfgstate_\exitnode$ and $\cfgstate'_\exitnode$, and $\as_{m+1}$.
        Consequently, we know that must be executions 
        $\cfgconfig{\node}{\cfgstate} \steps{\as} \cfgconfig{\exitnode}{\cfgstate_\exitnode}$ and 
        $\cfgconfig{\node}{\cfgstate'} \steps{\as'} \cfgconfig{\exitnode}{\cfgstate'_\exitnode}$ such that 
        $\size{\project{\as}{\fnode^m}} > \size{\project{\as'}{\fnode^m}}$ 
        (with $\as = \sum_{i \in \interval{1}{m+1}} \as_i \cdot a_i$ and $\as' = \sum_{i \in \interval{0}{m}} \as_i' \cdot a_i'$).
        This immediately contradicts $\iio(\fnode^m, v)$ (assumption~\ref{asm:instind}).
        \item A node $\fnode^*$ from $\fnodeset$ different from $\fnode^m$, which is reached from $\cfgconfig{\fnode^{m-1}}{\cfgstate_{m-1}'}$ before reaching $\fnode^m$.
         So, there are $\as_m'$, $a_m'^*$, $\as_m'^*$, $\cfgstate'^*$, $\cfgstate'^{*+}$ such that 
         $\cfgconfig{\nodesucc{\fnode^{m-1}}}{\cfgstate_{m-1}'^+} \nstep{\as_{m}'}{}{\fnodeset}{0} \cfgconfig{\fnode^*}{\cfgstate'^*} \step{a_m'^*}{}
         \cfgconfig{\nodesucc{\fnode^*}}{\cfgstate'^{*+}} \step{\as_m'^*}{*} \cfgconfig{\fnode^m}{\cfgstate_m'}$.
         But then there must be executions $\cfgconfig{\node}{\cfgstate} \steps{\as} \cfgconfig{\fnode^m}{\cfgstate_m}$ and 
         $\cfgconfig{\node}{\cfgstate'} \steps{\as'} \cfgconfig{\fnode^m}{\cfgstate'_m}$ such that
         $\size{\project{\as}{\fnode^*}} < \size{\project{\as'}{\fnode^*}}$
         (with $\as = \sum_{i \in \interval{1}{m+1}} \as_i \cdot a_i$ and $\as' = (\sum_{i \in \interval{0}{m-1}} \as_i' \cdot a_i') \cdot \as_m' \cdot a_m'* \cdot \as_m'^*$). 
         This again immediately contradicts $\iio(\fnode^*, v)$ (assumption~\ref{asm:instind}).
    \end{enumerate}
    Hence, we know that $\cfgstate'_m, \cfgstate_m'^+$ such that $\cfgconfig{\nodesucc{\fnode^{m-1}}}{\cfgstate_{m-1}'^+} \nstep{\as_m'}{}{\fnodeset}{0} \cfgconfig{\fnode^m}{\cfgstate_m'} \step{a_m'}{} \cfgconfig{\nodesucc{\fnode^m}}{\cfgstate_m'^+}$. 
    Finally, we need to show that $\forall x \in \fargs{\fnode^m}.~ \cfgstate_m(x) = \cfgstate_m'(x)$. 
    This immediately follows from assumption~\ref{asm:varind} and concludes the proof.
\end{enumerate} 
\end{proof}

\newcommand{\evmcomponents}{\textsf{Comp}_{\textit{EVM}}}
\newcommand{\evmcomponentspropex}{\textsf{Comp}_{\textit{EVM}}^\exstate}
\newcommand{\evmcomponentsproptx}{\textsf{Comp}_{\textit{EVM}}^\transenv}
\newcommand{\evmcomponentssymb}{\textsf{Comp}_{\textit{EVM}}^S}
\newcommand{\this}{\textsf{this}}
\newcommand{\other}{\textsf{other}}
\newcommand{\valuc}{\textsf{value}}
\newcommand{\originc}{\textsf{origin}}
\newcommand{\gasprizec}{\textsf{gasprice}}
\newcommand{\storc}{\textsf{stor}}
\newcommand{\memoc}{\textsf{m}}
\newcommand{\gasc}{\textsf{gas}}
\newcommand{\msizec}{\textsf{msize}}
\newcommand{\activeaccountc}{\textsf{actor}}
\newcommand{\inputdatac}{\textsf{input}}
\newcommand{\senderc}{\textsf{sender}}
\newcommand{\activecodec}{\textsf{code}}
\newcommand{\accountcodec}{\textsf{code}}
\newcommand{\balancec}{\textsf{bal}}
\newcommand{\noncec}{\textsf{nonce}}
\newcommand{\stackc}{\textsf{stack}}
\newcommand{\headerc}{\textsf{H}}

We define the set of components of EVM configurations.

\begin{definition}[EVM configuration components]
The set of EVM configuration components $\evmcomponents$ is defined as follows
\begin{align*}
    \evmcomponents \define 
    \evmcomponentspropex ~\cup~ \evmcomponentsproptx ~\cup~ \evmcomponentssymb 
\end{align*}
with 
\begin{align*}
    \evmcomponentspropex & \define 
    \{ \mstate.\gasc, \mstate.\msizec, 
    \exenv.\inputdatac, \exenv.\senderc, \exenv.\valuc
    \} \\
    & \quad ~\cup~ 
    \{ \mstate.\stackc(x) ~|~ x \in \integer{8} \}
    ~\cup~
    \{ \mstate.\memoc(x) ~|~ x \in \integer{256} \}
\end{align*}
and
\begin{align*}
    \evmcomponentsproptx & \define 
    \{  \headerc.\parentc, \headerc.\beneficiaryc, \headerc.\difficultyc, \headerc.\blocknumberc, \headerc.\gaslimitc, \headerc.\timestampc, \originc, \gasprizec \}
\end{align*}
and 
\begin{align*}
    \evmcomponentssymb & \define 
    \{ \gstate(\other).\storc, \gstate(\other).\balancec, \gstate(\other).\accountcode, 
    \gstate(\this).\balancec, \gstate(\this).\noncec
    \} \\
    & \quad ~\cup~ \{ \gstate(\this).\storc(x) ~|~ x \in \integer{256} \}
\end{align*}
\end{definition}

Note that we explicitly exclude the program counter, the active account and the active code. 
The reason for this is that we will only compare executions of the same contract starting from the same instruction (program counter).

\begin{definition}
    Let $\inputcomponents \subseteq \evmcomponents$ be a set of EVM components. 
    We define two EVM configurations $(\transenv, \exstate)$, $(\transenv', \exstate')$ equal up to $\inputcomponents$ (written $(\transenv, \exstate) \equalupto{\inputcomponents} (\transenv', \exstate')$) if the following holds
\begin{align*}
    \forall \inputcomponent. \inputcomponent \not \in \inputcomponents
    &\Rightarrow (\inputcomponent \in \evmcomponentspropex \Rightarrow \exstate.\inputcomponent = \exstate'.\inputcomponent) \\ 
        & \quad ~\land~ 
        (\inputcomponent \in \evmcomponentsproptx \Rightarrow \transenv.\inputcomponent = \transenv'.\inputcomponent) \\ 
        & \quad ~\land~ 
        (\inputcomponent = \gstate(\other).\stor \Rightarrow \forall x \in \integer{256}~a \in \integer{160}.\ a \neq \exstate.\exenv.\activeaccountc \Rightarrow \exstate.\gstate(a).\storc(x) = \exstate'.\gstate(a).\storc(x)) \\
        & \quad ~\land~ 
        (\inputcomponent = \gstate(\other).\balancec \Rightarrow \forall a \in \integer{160}.\ a \neq \exstate.\exenv.\activeaccountc \Rightarrow \exstate.\gstate(a).\balancec = \exstate'.\gstate(a).\balancec) \\
        & \quad ~\land~ 
        (\inputcomponent = \gstate(\other).\noncec \Rightarrow \forall a \in \integer{160}.\ a \neq \exstate.\exenv.\activeaccountc \Rightarrow \exstate.\gstate(a).\noncec = \exstate'.\gstate(a).\noncec) \\
        & \quad ~\land~ 
        (\inputcomponent = \gstate(\other).\accountcodec \Rightarrow \forall a \in \integer{160}.\ a \neq \exstate.\exenv.\activeaccountc \Rightarrow \exstate.\gstate(a).\accountcodec = \exstate'.\gstate(a).\accountcodec) \\
        & \quad ~\land~ 
        (\forall x.\ \inputcomponent = \gstate(\this).\stor(x) \Rightarrow \forall a \in \integer{160}.\ a = \exstate.\exenv.\activeaccountc \Rightarrow \exstate.\gstate(a).\storc(x) = \exstate'.\gstate(a).\storc(x)) \\
        & \quad ~\land~ 
        (\inputcomponent = \gstate(\this).\balancec \Rightarrow \forall a \in \integer{160}.\ a = \exstate.\exenv.\activeaccountc \Rightarrow \exstate.\gstate(a).\balancec = \exstate'.\gstate(a).\balancec) \\
        & \quad ~\land~ 
        (\inputcomponent = \gstate(\this).\noncec \Rightarrow \forall a \in \integer{160}.\ a = \exstate.\exenv.\activeaccountc \Rightarrow \exstate.\gstate(a).\noncec = \exstate'.\gstate(a).\noncec)
\end{align*}
\end{definition}

We formally define the function $\componentToVar$, which maps components of EVM configurations to CFG variables. 
Note, that as opposed to the slightly simplified version in the main body of the paper, $\componentToVar$ maps to a set of variables.

\begin{definition}
    The function $\componentToVar \in \evmcomponents \to \setof{\varset}$ is defined as follows:
\begin{align*}
\componentToVar(\inputcomponent) \define
\begin{cases}
\{ \stackvar{x} \} & \inputcomponent = \mstate.\stackc(x) \\
\{ \memvarabs{x}{}, \memvarconc{x}{} \} & \inputcomponent = \mstate.\memoc(x) \\
\{ \locenvvar{\lgas} \} & \inputcomponent = \mstate.\gasc \\
\{ \locenvvar{\msize} \} & \inputcomponent = \mstate.\msizec \\
\{ \locenvvar{\inputdata} \} & \inputcomponent = \exenv.\inputdatac \\
\{ \locenvvar{\sender} \} & \inputcomponent = \exenv.\senderc \\
\{ \locenvvar{\valu} \} & \inputcomponent = \exenv.\valuc \\
\{ \globenvvar{\parent} \} & \inputcomponent = \headerc.\parentc \\
\{ \globenvvar{\beneficiary} \} & \inputcomponent = \headerc.\beneficiaryc \\
\{ \globenvvar{\difficulty} \} & \inputcomponent = \headerc.\difficultyc \\
\{ \globenvvar{\blocknumber} \} & \inputcomponent = \headerc.\blocknumberc \\
\{ \globenvvar{\gaslimit} \} & \inputcomponent = headerc.\gaslimitc \\
\{ \globenvvar{\timestamp} \} & \inputcomponent = \headerc.\timestampc \\
\{ \globenvvar{\origin} \} & \inputcomponent = \originc \\
\{ \globenvvar{\gasprize} \} & \inputcomponent = \gasprizec \\
\{ \globenvvar{\extenv} \} & \inputcomponent \in \{ \gstate(\other).\storc, \gstate(\other).\balancec, \gstate(\this).\balancec, \gstate(\other).\noncec, \gstate(\this).\noncec, \gstate(\other).\accountcode \} \\
\{ \storvarabs{x}{}, \storvarconc{x}{} \} & \inputcomponent = \gstate(\this).\storc(x) \\
\end{cases}
\end{align*}
\end{definition}

Next, we establish the relation between the notions of equivalence up to components: 
\begin{lemma}
    \label{lem:eq-upto-toCFG}
    Let $\transenv, \transenv'$ be transaction environments, $\exstate, \exstate'$ execution states and $\inputcomponents$ a set if EVM state components such that $(\transenv, \exstate) \equalupto{\inputcomponents} (\transenv', \exstate')$
    and $(\cfgstate, \evmcontract, \lpc) = \tocfgstate(\transenv, \exstate)$ and $(\cfgstate', \evmcontract, \lpc) = \tocfgstate(\transenv, \exstate)$ for some $\cfgstate$, $\cfgstate'$, $\evmcontract$. 
    Further, let $\inputvars = \{x ~|~ \exists \inputcomponent \in \inputcomponents.\ x \in \componentToVar(\inputcomponent) \}$.
    Then $\cfgstate \equalupto{\inputvars} \cfgstate'$.
\end{lemma}
\begin{proof}
Trivially follows from the definition of $\tocfgstate$ that maps components into variables in the same way as $\componentToVar$.

\end{proof}
Finally, by combining the previous lemmas, we can prove~\Cref{thm:soundness-tni}.
We restate the theorem for completeness with all assumptions here.
To this end, we restate the definition of \tracenoninterference to explicitly state the consistency assumption the involved executions: 

\begin{definition}[\Tracenoninterference]
	Let $\precontract$ be an EVM contract, $\inputcomponents \in \evmcomponents$ be a set of components of EVM configurations and $\outputproj \in \instructions \to \BB$ be a predicate on instructions. 
	Then \tracenoninterference{} of contract $\precontract$ w.r.t. $\inputcomponents$ and $\outputproj$ (written $\tracenoninter{\evmcontract}{\inputcomponents}{\outputproj}$) is defined as follows: 
	{\small
		\begin{align*}
			\tracenoninter{\evmcontract}{\inputcomponents}{\outputproj}
			&\define
			\forall ~\transenv ~\transenv' ~\exstate ~\exstate'~ t~ t' ~\pi. ~\pi'~ 
			\\
            & \strongconsistent{\exstate}{\precontract} 
            \\
            & \Rightarrow \strongconsistent{\exstate'}{\precontract}
            \\
			& \Rightarrow (\transenv, \exstate) \equalupto{\inputvars} (\transenv', \exstate') 
			\\
			&\Rightarrow \sstepstrace{\transenv}{\cons{\annotate{\exstate}{\evmcontract}}{\callstack}}{\cons{\annotate{t}{\evmcontract}}{\callstack}}{\pi} ~\land~ \finalstate{t}
			\\
			&\Rightarrow \sstepstrace{\transenv}{\cons{\annotate{\exstate'}{\evmcontract}}{\callstack}}{\cons{\annotate{t'}{\evmcontract}}{\callstack}}{\pi'} ~\land~ \finalstate{t'}
			\\
			&\Rightarrow \project{\pi}{\outputproj} = \project{\pi'}{\outputproj}
		\end{align*}
	}
		where $\project{\pi}{\outputproj}$ denotes the trace filtered by $\outputproj$, so containing only the instructions satisfying $\outputproj$.
		and $\exstate$ and $\exstate'$ are strongly consistent with contract $\evmcontract$.
	\end{definition}

Next, we formally define $\patnoninter{\inputcomponents}{\outputproj}{\precontract}$ considering the proper definition of $\componentToVar$:
\begin{definition}[Trace non-interference pattern]
    \label{def:patnoninter}
\begin{align*}
    \patnoninter{\inputcomponents}{\outputproj}{\precontract}\define
    &\{  \pred{InstMayDepOn}(\lpc, \inputvar) ~|~ 
    \exists ~\inputcomponent ~\inputvar~ \instruction~ \vec{x}~ \nextpc~ \pre.~ \inputcomponent \in \inputcomponents ~\land~ \inputvar \in \componentToVar(\inputcomponent)
        ~\land~ \precontract(\lpc) = \instruction(\vec{x}, \nextpc, \pre)
        ~\land~ \outputproj(\instruction)
        \} \\
    &\cup 
    \{  \pred{VarMayDepOn}(x_i, \inputvar) ~|~ 
     \exists ~\inputcomponent ~\inputvar~ \instruction~ \vec{x}~ \nextpc~ \pre.~ \inputcomponent \in \inputcomponents ~\land~ \inputvar \in \componentToVar(\inputcomponent)
    ~\land~ \precontract (\lpc) = (\instruction(\vec{x}, \nextpc, \pre))
    ~\land~ \outputproj(\instruction)
    ~\land~ x_i \in \vec{x} \}.
\end{align*}
\end{definition}

\begin{theorem}[Soundness of \tracenoninterference]
    \label{thm:soundness-tni-full}
    Let $\inputcomponents \subseteq \evmcomponents$ be a set of EVM components, and $\outputproj \in \instructions \to \BB$ an instruction-of-interest predicate.
    Further, let $\precontract$ be a store unreachable contract that does not exhibit local out-of-gas exceptions. 
    Then it holds that
    %
    \begin{align*}
        (\forall p\in \patnoninter{\inputcomponents}{\outputproj}{\precontract}.~ p \not \in \lfp{\rules{\precontract}}) 
        \Rightarrow \tracenoninter{\precontract}{\inputcomponents}{\outputproj}.
    \end{align*}
\end{theorem}

Note that due to the definition of EVM components, we can rely on the fact that 
\begin{align}
    (\transenv, \exstate) \equalupto{\inputvars} (\transenv', \exstate') \Rightarrow \access{\access{\exstate}{\mstate}}{\pc} = \access{\access{\exstate'}{\mstate}}{\pc} \label{asm:pc}    
\end{align}

\begin{proof}
Assume that $\forall p\in \patnoninter{\inputcomponents}{\outputproj}{\precontract}.~ p \not \in \lfp{\rules{\precontract}}$.
From the definition of $\patnoninter{\inputcomponents}{\outputproj}{\precontract}$ (~\Cref{def:patnoninter}), we know that hence
\begin{align*}
    \forall~ \inputcomponent \in \inputcomponents. 
    ~\forall~ \inputvar \in  \componentToVar(\inputcomponent). 
    ~\forall~ \lpc~ i~ \nextpc~ \pre~ \instruction ~\vec{x}. 
    ~\precontract(\lpc) = \instruction(\vec{x}, \nextpc, \pre)
    \Rightarrow \outputproj(\instruction)
    \Rightarrow \pred{InstMayDepOn}((\lpc, i), \inputvar) \not \in \lfp{\rules{\precontract}}
\end{align*}
and 
\begin{align*}
    \forall~ \inputcomponent \in \inputcomponents. 
    ~\forall~ \inputvar \in \componentToVar(\inputcomponent). 
    ~\forall~ x~\lpc~ \nextpc~ \pre~ \instruction ~\vec{x}. 
    ~\precontract(\lpc) = \instruction(\vec{x}, \nextpc, \pre)
    \Rightarrow \outputproj(\instruction)
    \Rightarrow x \in \vec{x}
    \Rightarrow \pred{VarMayDepOn}(x_i, \inputcomponent) \not \in \lfp{\rules{\precontract}}
\end{align*}

From \Cref{app:thm:soundness-independence},\Cref{lem:pred_slicing_app}, \Cref{thm:inst-soundness-independence},\Cref{lem:inst_pred_slicing_app},
we get that 
    \begin{align*}
        \forall~ \inputcomponent \in \inputcomponents. 
        ~\forall~ \inputvar \in \componentToVar(\inputcomponent). 
        ~\forall~ \lpc~ i~\nextpc~ \pre~ \instruction ~\vec{x}. 
        ~\precontract(\lpc) = \instruction(\vec{x}, \nextpc, \pre)
        \Rightarrow \outputproj(\instruction)
        \Rightarrow \iio((\lpc, i),, \inputvar)
    \end{align*}
    and 
    \begin{align*}
        \forall~ \inputcomponent \in \inputcomponents.
        ~\forall~ \inputvar \in \componentToVar(\inputcomponent). 
        ~\forall~ x~\lpc~ \nextpc~ \pre~ \instruction ~\vec{x}. 
        ~\precontract(\lpc) = \instruction(\vec{x}, \nextpc, \pre)
        \Rightarrow \outputproj(\instruction)
        \Rightarrow x \in \vec{x}
        \Rightarrow \vio(x_i, \inputvar)
    \end{align*}
Which is equivalent to 
\begin{align}
    \forall~v \in \{v ~|~ \exists \inputcomponent \in \inputcomponents.~v \in \componentToVar(\inputcomponent) \}. 
    ~\forall \fnode \in \fnodeset.
    ~ \iio(\fnode, v)
\end{align} \label{asm:iio}
and 
\begin{align}
    \forall~v \in \{v ~|~ \exists \inputcomponent \in \inputcomponents.~v \in  \componentToVar(\inputcomponent) \}. 
    ~\forall x \in \fvarset.
    ~\vio(x, v)
\end{align}. \label{asm:vio}

Now let $\transenv$, $\transenv'$ be transaction environments and  $\exstate$, $\exstate'$, $t$, $t'$ be execution states and  $\pi$, and  $\pi'$ be traces. 
Assume the following: 
\begin{enumerate}
    \item $(\transenv, \exstate) \equalupto{\inputvars} (\transenv', \exstate')$ \label{asm:p1}
    \item $\sstepstrace{\transenv}{\cons{\annotate{\exstate}{\evmcontract}}{\callstack}}{\cons{\annotate{t}{\evmcontract}}{\callstack}}{\pi}$\label{asm:p2}
    \item $\finalstate{t}$ \label{asm:p3}
    \item $\sstepstrace{\transenv}{\cons{\annotate{\exstate'}{\evmcontract}}{\callstack}}{\cons{\annotate{t'}{\evmcontract}}{\callstack}}{\pi'}$ \label{asm:p4}
    \item $\finalstate{t'}$ \label{asm:p5}
\end{enumerate}
We need to show that $\project{\pi}{\outputproj} = \project{\pi'}{\outputproj}$. 

We define the following relation to reason about steps that do not produce any actions satisfying $\outputproj$: 
\begin{align*}
    \transenv \vDash \cons{\exstate}{\callstack }\combifstep{*}{}{f}{} \cons{\exstate'}{\callstack }\ \define 
    \combisteps{\transenv}{\cons{\exstate}{\callstack}}{\cons{\exstate'}{\callstack}}{\pi} ~\land~ \project{\pi}{\outputproj} = \nil
\end{align*}

Using this definition, we can decompose every run 
$\sstepstrace{\transenv}{\cons{\annotate{\exstate}{\evmcontract}}{\callstack}}{\cons{\annotate{t}{\evmcontract}}{\callstack}}{\pi}$
into a run 
$\transenv \vDash \cons{\annotate{\exstate}{\evmcontract}}{\callstack} \left(\combifstep{*}{}{\outputproj}{} \cons{\annotate{\exstate^{f, i}}{\evmcontract}}{\callstack} \combinstep{}{\instruction_i(\vec{v}^{i})} \cons{\annotate{\exstate^{f^+,i}}{\evmcontract}}{\callstack} \right)_{i=1}^m \combifstep{*}{}{\outputproj}{} \cons{\annotate{t}{\evmcontract}}{\callstack}$. 
Such that $\project{\pi}{\outputproj} = \sum_{i =1}^m \instruction_i(\vec{v}^i)$.
Similarly, we can deconstruct every run $\sstepstrace{\transenv}{\cons{\annotate{\exstate'}{\evmcontract}}{\callstack}}{\cons{\annotate{t'}{\evmcontract}}{\callstack}}{\pi'}$ 
into a run 
$\transenv \vDash \cons{\annotate{\exstate'}{\evmcontract}}{\callstack} \left(\combifstep{*}{}{\outputproj}{} \cons{\annotate{\exstate'^{f, i}}{\evmcontract}}{\callstack} \combinstep{}{\instruction_i'(\vec{v}'^{i})} \cons{\annotate{\exstate'^{f^+,i}}{\evmcontract}}{\callstack} \right)_{i=1}^n \combifstep{*}{}{\outputproj}{} \cons{\annotate{t'}{\evmcontract}}{\callstack}$. 
Such that $\project{\pi'}{\outputproj} = \sum_{i =1}^n \instruction_i'(\vec{v}'^i)$.
Towards contradiction, we assume that $\project{\pi}{\outputproj} \neq \project{\pi'}{\outputproj}$. 
So, either there is a position $i$ such that $\instruction_i(\vec{v}^i) \neq \instruction_i'(\vec{v}'^i)$ and for all $j < i$ it holds that $\instruction_j(\vec{v}^j) = \instruction_j'(\vec{v}'^j)$ 
or $\size{\project{\pi}{\outputproj}} < \size{\project{\pi'}{\outputproj}}$ and $\project{\pi}{\outputproj}$ is a prefix of $\project{\pi'}{\outputproj}$
or $\size{\project{\pi'}{\outputproj}} < \size{\project{\pi}{\outputproj}}$ and $\project{\pi'}{\outputproj}$ is a prefix of $\project{\pi}{\outputproj}$.
We proceed by case distinction based on these cases:
\begin{enumerate}
    \item Assume that there is a position $i$ such that $\instruction_i(\vec{v}^i) \neq \instruction_i'(\vec{v}'^i)$ and for all $j < i$ it holds that $\instruction_j(\vec{v}^j) = \instruction_j'(\vec{v}'^j)$. 
    Then we know that $\transenv \vDash \cons{\annotate{\exstate}{\evmcontract}}{\callstack} \left(\combifstep{*}{}{\outputproj}{} \cons{\annotate{\exstate^{f, j}}{\evmcontract}}{\callstack} \combinstep{}{\instruction_j(\vec{v}^{j})} \cons{\annotate{\exstate^{f^+,j}}{\evmcontract}}{\callstack} \right)_{j=1}^{i-1} \combifstep{*}{}{\outputproj}{} \cons{\annotate{\exstate^{f, i}}{\evmcontract}}{\callstack} \combinstep{}{\instruction_i(\vec{v}^{i})} \cons{\annotate{\exstate^{f^+,i}}{\evmcontract}}{\callstack}$ and hence
    for $\mstate_i = \access{\exstate^{f, i}}{\mstate}$ we have that 
    $\evmcontract(\access{\mstate_i}{\pc}) = (\instruction_i(\vec{x}^i), \nextpc^i, \pre_i)$ and 
    for all $v^i_k \in \vec{v}^i$ it holds that $v^i_k = \access{\mstate_i}{\stack}(x^i_k)$.
    Similarly, we know that $\transenv \vDash \cons{\annotate{\exstate'}{\evmcontract}}{\callstack} \left(\combifstep{*}{}{\outputproj}{} \cons{\annotate{\exstate'^{f, j}}{\evmcontract}}{\callstack} \combinstep{}{\instruction_j'(\vec{v}'^{j})} \cons{\annotate{\exstate'^{f^+,j}}{\evmcontract}}{\callstack} \right)_{j=1}^{i-1} \combifstep{*}{}{\outputproj}{} \cons{\annotate{\exstate'^{f, i}}{\evmcontract}}{\callstack} \combinstep{}{\instruction_i'(\vec{v}'^{i})} \cons{\annotate{\exstate'^{f^+,i}}{\evmcontract}}{\callstack}$ and hence
    for $\mstate_i' = \access{\exstate'^{f, i}}{\mstate}$ we have that 
    $\evmcontract(\access{\mstate_i'}{\pc}) = (\instruction_i'(\vec{x}'^i), \nextpc'^i, \pre_i')$ and 
    for all $v_k'^i \in \vec{v}'^i$ it holds that $v_k'^i = \access{\mstate_i'}{\stack}(x_k'^i)$.
    From \Cref{thm:equivalence-evm-cfg} we get that 
    \begin{enumerate}
        \item  $\precontract, \size{\callstack} \vDash \cfgconfig{\node_\exstate}{\tocfgstate(\transenv, \exstate) \uplus \cfgstatecopybot} 
        \left( \nstep{}{}{\fnodeset}{0}\cfgconfig{\node_{\exstate^{f,j}}}{\tocfgstate(\transenv, \exstate^{f, j}) \uplus \cfgstatecopybot} \cfgmednstep{} \cfgconfig{\node_{\exstate^{f^+,j}}}{\tocfgstate(\transenv, \exstate^{f^+, j}) \uplus \cfgstatecopybot} \right)_{j=1}^{i}$
        such that all $\node_{\exstate^{f,j}} \in \fnodeset$ 
        \item  $\precontract, \size{\callstack} \vDash \cfgconfig{\node_{\exstate'}}{\tocfgstate(\transenv, \exstate') \uplus \cfgstatecopybot} 
        \left( \nstep{}{}{\fnodeset}{0}\cfgconfig{\node_{\exstate'^{f,j}}}{\tocfgstate(\transenv, \exstate'^{f, j}) \uplus \cfgstatecopybot} \cfgmednstep{} \cfgconfig{\node_{\exstate'^{f^+,j}}}{\tocfgstate(\transenv, \exstate'^{f^+, j}) \uplus \cfgstatecopybot} \right)_{j=1}^{i}$
        such that all $\node_{\exstate'^{f,j}} \in \fnodeset$.
    \end{enumerate}
    Note that from assumption~\ref{asm:pc} we know that $\node_\exstate = (\access{\access{\exstate}{\mstate}}{\pc}, 0) = (\access{\access{\exstate'}{\mstate}}{\pc}, 0) = \node_{\exstate'}$.
    Further, we know that
    \[
        \precontract, \size{\callstack} \vDash \cfgconfig{\node_\exstate}{\tocfgstate(\transenv, \exstate) \uplus \cfgstatecopybot} 
    \left( \nstep{}{}{\fnodeset}{0}\cfgconfig{\node_j^*}{\cfgstate_{j}} \step{}{} \cfgconfig{\node_j^*}{\cfgstate_{j+}} \right)_{j=1}^l
    \]
    for some $l \geq i$ 
    such that there is some $g\in \NN \to \NN$ 
    for which it holds that 
    1) $\forall\ n\ m.\ n < m \Rightarrow g(n) < g(m)$ and
    2) $\forall\ j \in \interval{1}{i}.\ \node_{\exstate^{f,j}} = \node_{g(j)}^*$ and
    3) $\forall\ j \in \interval{1}{i}.\ \tocfgstate(\transenv, \exstate^{f, j}) \uplus \cfgstatecopybot =  \cfgstate_{g(j)}$ and
    4) $\forall\ j \in \interval{1}{i}.\ \forall k.\ k > g(i) \Rightarrow k < g(i+1) \Rightarrow \forall \lpc.\ \node_{g(j)}^* = (\lpc, 0) \Rightarrow \exists q.\ \node_{k}^* = (\lpc, q)$.
    This is as all $\cfgmednstep{}$-steps from nodes in $\fnodeset$ can be expanded into the individual steps between the subnodes of the same $\pc$.

    Consequently, we can apply~\Cref{lem:soundness-help} (using~\Cref{lem:eq-upto-toCFG}) to obtain
    \[ 
        \precontract, \size{\callstack} \vDash \cfgconfig{\node_{\exstate'}}{\tocfgstate(\transenv, \exstate') \uplus \cfgstatecopybot} 
    \left( \nstep{}{}{\fnodeset}{0}\cfgconfig{\node_j^*}{\cfgstate_{j}'} \step{}{} \cfgconfig{\node_j^*}{\cfgstate_{j+}'} \right)_{j=1}^l
    \]
    and $\forall p \in \interval{1}{l}.~ \forall x \in \fargs{\node_p^*}.~ \cfgstate_p(x) = \cfgstate_p'(x)$. 
    Consequently, we can also conclude that 
    \[
        \precontract, \size{\callstack} \vDash \cfgconfig{\node_{\exstate'}}{\tocfgstate(\transenv, \exstate') \uplus \cfgstatecopybot} 
        \left( \nstep{}{}{\fnodeset}{0}\cfgconfig{\node_{\exstate^{f,j}}}{\cfgstate_j^\dagger} \cfgmednstep{} \cfgconfig{\node_{\exstate^{f^+,j}}}{\cfgstate_j^\dagger} \right)_{j=1}^{i}
    \]
    such that $\forall\ j \in \interval{1}{i}. \cfgstate_{j}^\dagger = \cfgstate_{g(j)}'$.
    In particular, this means that 
    $\forall\ j \in \interval{1}{i-1}.\ \forall x \in \fargs{\node_{\exstate^{f,j}}}.~ \cfgstate_{j}^\dagger(x) = \tocfgstate(\transenv, \exstate^{f, j}) \uplus \cfgstatecopybot(x)$.

    Since $\instruction_i(\vec{v}^i) \neq \instruction_i'(\vec{v}'^i)$, it must either holds that $\instruction_i(\vec{x}^i) \neq \instruction_i'(\vec{x}'^i)$ or there is some position $k$ such that $v_k^i \neq v_k'^i$. We do another case distinction:
    \begin{enumerate}
        \item Assume that $\instruction_i(\vec{x}^i) \neq \instruction_i'(\vec{x}'^i)$. Then it must hold that $\access{\mstate_i}{\pc} \neq \access{\mstate_i'}{\pc}$ (since $\evmcontract$ deterministically maps program counters triples $(\instruction(\vec{x}), \nextpc, \pre)$ ) and consequently $\node_{\exstate'^{f,i}} \neq \node_{\exstate^{f,i}}$.
        However, since execution is deterministic, and we know that we can reach $\node_{\exstate^{f,i}}$ as the $ith$ node from $\fnodeset$ (with a different $\pc$) when starting the execution in $\cfgconfig{\node_{\exstate'}}{\tocfgstate(\transenv, \exstate') \uplus \cfgstatecopybot}$, this leads to a contradiction.
        \item Assume that there is some position $k$ such that $v_k^i \neq v_k'^i$. This means that there exists $x_k^i \in \vec{x}^i$ (with $\precontract(\access{\mstate_i}{\pc}) = (\instruction(\vec{x}^i), \nextpc^i, \pre_i)$) such that $\access{\mstate_i'}{\stack}(x_k^i) \neq \access{\mstate_i}{\stack}(x_k^i)$.
        However, since execution is deterministic, and we know that we can reach the configuration $\cfgconfig{\node_{\exstate^{f,i}}}{\cfgstate_i^\dagger}$ as the $ith$ node from $\fnodeset$ when starting the execution in $\cfgconfig{\node_{\exstate'}}{\tocfgstate(\transenv, \exstate') \uplus \cfgstatecopybot}$, 
        we know that $\cfgstate_i^\dagger = \tocfgstate(\transenv, \exstate'^{f, j}) \uplus \cfgstatecopybot$. 
        Further, we know that $\cfgstate_{i}^\dagger(x_k^i) = \tocfgstate(\transenv, \exstate^{f, i}) \uplus \cfgstatecopybot(x_k^i)$ and consequently also 
        $\tocfgstate(\transenv, \exstate'^{f, j}) \uplus \cfgstatecopybot(x_k^i) = \tocfgstate(\transenv, \exstate^{f, i}) \uplus \cfgstatecopybot(x_k^i)$. 
        This contradicts $\access{\mstate_i'}{\stack}(x_k^i) \neq \access{\mstate_i}{\stack}(x_k^i)$. 
    \end{enumerate}
    \item Assume that $\size{\project{\pi}{\outputproj}} < \size{\project{\pi'}{\outputproj}}$ and $\project{\pi}{\outputproj}$ is a prefix of $\project{\pi'}{\outputproj}$. 
    So, $m < n$, which means that 
    \[
        \transenv \vDash \cons{\annotate{\exstate'}{\evmcontract}}{\callstack} 
        \left(\combifstep{*}{}{\outputproj}{} \cons{\annotate{\exstate'^{f, i}}{\evmcontract}}{\callstack} \combinstep{}{\instruction_i'(\vec{v}'^{i})} \cons{\annotate{\exstate'^{f^+,i}}{\evmcontract}}{\callstack} \right)_{i=1}^m 
        \left(\combifstep{*}{}{\outputproj}{} \cons{\annotate{\exstate'^{f, i}}{\evmcontract}}{\callstack} \combinstep{}{\instruction_i'(\vec{v}'^{i})} \cons{\annotate{\exstate'^{f^+,i}}{\evmcontract}}{\callstack} \right)_{i=m}^{m+k} 
    \combifstep{*}{}{\outputproj}{} \cons{\annotate{t'}{\evmcontract}}{\callstack} 
    \]
    for some $k > 1$. 
    From \Cref{thm:equivalence-evm-cfg} we get that
    \begin{align*}
        \precontract, \size{\callstack} \vDash \cfgconfig{\node_{\exstate'}}{\tocfgstate(\transenv, \exstate') \uplus \cfgstatecopybot} 
        \left( \nstep{}{}{\fnodeset}{0}\cfgconfig{\node_{\exstate'^{f,i}}}{\tocfgstate(\transenv, \exstate'^{f, i}) \uplus \cfgstatecopybot} \cfgmednstep{} \cfgconfig{\node_{\exstate'^{f^+,i}}}{\tocfgstate(\transenv, \exstate'^{f^+, i}) \uplus \cfgstatecopybot} \right)_{i=1}^{n}
    \end{align*}
    such that all $\node_{\exstate'^{f,j}} \in \fnodeset$.
    From this, we can conclude that
    \[
        \precontract, \size{\callstack} \vDash \cfgconfig{\node_{\exstate'}}{\tocfgstate(\transenv, \exstate') \uplus \cfgstatecopybot} 
    \left( \nstep{}{}{\fnodeset}{0}\cfgconfig{\node_j^*}{\cfgstate_{j}} \step{}{} \cfgconfig{\node_j^*}{\cfgstate_{j+}} \right)_{j=1}^l
    \]
    for some $l \geq n$ 
    such that there is some $g\in \NN \to \NN$ 
    for which it holds that 
    1) $\forall\ n\ m.\ n < m \Rightarrow g(n) < g(m)$ and
    2) $\forall\ j \in \interval{1}{n}.\ \node_{\exstate'^{f,j}} = \node_{g(j)}^*$ and
    3) $\forall\ j \in \interval{1}{n}. \tocfgstate(\transenv, \exstate'^{f, j}) \uplus \cfgstatecopybot =  \cfgstate_{g(j)}$ and
    4) $\forall\ j \in \interval{1}{i}.\ \forall k.\ k > g(i) \Rightarrow k < g(i+1) \Rightarrow \forall \lpc.\ \node_{g(j)}^* = (\lpc, 0) \Rightarrow \exists q.\ \node_{k}^* = (\lpc, q)$.
    This is as all $\cfgmednstep{}$-steps from nodes in $\fnodeset$ can be expanded into the individual steps between the subnodes of the same $\pc$.

    Note that from assumption~\ref{asm:pc} we know that $\node_\exstate = (\access{\access{\exstate}{\mstate}}{\pc}, 0) = (\access{\access{\exstate'}{\mstate}}{\pc}, 0) = \node_{\exstate'}$.
    Consequently, we can apply~\Cref{lem:soundness-help} (using~\Cref{lem:eq-upto-toCFG}) to obtain
    \[ 
        \precontract, \size{\callstack} \vDash \cfgconfig{\node_\exstate}{\tocfgstate(\transenv, \exstate') \uplus \cfgstatecopybot} 
    \left( \nstep{}{}{\fnodeset}{0}\cfgconfig{\node_j^*}{\cfgstate_{j}'} \step{}{} \cfgconfig{\node_j^*}{\cfgstate_{j+}'} \right)_{j=1}^l
    \]
    and $\forall p \in \interval{1}{l}.~ \forall x \in \fargs{\node_p^*}.~ \cfgstate_p(x) = \cfgstate_p'(x)$. 
    Consequently, we can also conclude that 
    \[
        \precontract, \size{\callstack} \vDash \cfgconfig{\node_{\exstate}}{\tocfgstate(\transenv, \exstate) \uplus \cfgstatecopybot} 
        \left( \nstep{}{}{\fnodeset}{0}\cfgconfig{\node_{\exstate'^{f,j}}}{\cfgstate_j^\dagger} \cfgmednstep{} \cfgconfig{\node_{\exstate'^{f^+,j}}}{\cfgstate_j^\dagger} \right)_{j=1}^{n}
    \]
    such that $\forall\ j \in \interval{1}{n}. \cfgstate_{j}^\dagger = \cfgstate_{g(j)}'$.
    Consequently, we have that
    \begin{align*}
        \precontract, \size{\callstack} \vDash \cfgconfig{\node_{\exstate}}{\tocfgstate(\transenv, \exstate) \uplus \cfgstatecopybot} 
        \left( \nstep{}{}{\fnodeset}{0}\cfgconfig{\node_{\exstate'^{f,j}}}{\cfgstate_j^\dagger} \cfgmednstep{} \cfgconfig{\node_{\exstate'^{f^+,j}}}{\cfgstate_j^\dagger} \right)_{j=1}^{m} \\
        \left( \nstep{}{}{\fnodeset}{0}\cfgconfig{\node_{\exstate'^{f,j}}}{\cfgstate_j^\dagger} \cfgmednstep{} \cfgconfig{\node_{\exstate'^{f^+,j}}}{\cfgstate_j^\dagger} \right)_{j=m}^{m+k} 
    \end{align*}

    However, from~\Cref{thm:equivalence-evm-cfg} we know that 
    \begin{align*}
        \precontract, \size{\callstack} \vDash \cfgconfig{\node_\exstate}{\tocfgstate(\transenv, \exstate) \uplus \cfgstatecopybot} 
        \left( \nstep{}{}{\fnodeset}{0}\cfgconfig{\node_{\exstate^{f,j}}}{\tocfgstate(\transenv, \exstate^{f, j}) \uplus \cfgstatecopybot} \cfgmednstep{} \cfgconfig{\node_{\exstate^{f^+,j}}}{\tocfgstate(\transenv, \exstate^{f^+, j}) \uplus \cfgstatecopybot} \right)_{j=1}^{m} \\
        \nstep{}{}{\fnodeset}{0}\cfgconfig{\node_{t}}{\tocfgstate(\transenv, t) \uplus \cfgstatecopybot}
    \end{align*}
    such that all $\node_{\exstate^{f,j}} \in \fnodeset$.
    This leads to a contradiction, since execution is deterministic and like this we obtain two executions starting in $\cfgconfig{\node_\exstate}{\tocfgstate(\transenv, \exstate) \uplus \cfgstatecopybot}$, which step through a different number of nodes (with different $\pc$s) in $\fnodeset$.
    \item Assume that $\size{\project{\pi'}{\outputproj}} < \size{\project{\pi}{\outputproj}}$ and $\project{\pi'}{\outputproj}$ is a prefix of $\project{\pi}{\outputproj}$. The proof is fully analogous to the previous case.
\end{enumerate}

\end{proof}

\section{Securify}
\label{sec:securify-spec}
\subsection{Analysis specification}
The following rules are extracted from the Securify fixed-point calculation \cite{securify2020repo}\cite{tsankov2018securify}.
We split the appendix in input facts in Section \Cref{app:a:1}, $\may$-semantic rules in \Cref{app:a:2}, and $\must$-semantic rules in \Cref{app:a:3}.
$\must$-analysis rules that are identical to the $\may$-analysis semantic rules are left out.
Identical rules are denoted with $\Leftarrow$ instead of $\mayrule$ in Appendix \ref{app:a:2}.
Additionally, analogous rules for storage are omitted.

\subsubsection{Input facts}
\label{app:a:1}
\begin{align}
    \pred{Source}(L,Y_0,\inst{inst}) & \leftarrow \inst{inst}(L, Y_0,\ldots)  
\end{align}
\begin{align}
    \pred{AssignVar}(L,Y_i,X_j) & \leftarrow \inst{inst}(L,~ \underbrace{\ldots,~ Y_i,~ \ldots~}_{Outputs},~ \underbrace{\ldots,~ X_j,~ \ldots~}_{Inputs})   \\
    \intertext{\centering \inst{inst} $\not \in \{\inst{mload}, \inst{sload},  \inst{sha3}\}$ (no propagation for known accesses)} \nonumber
\end{align}
\begin{align}
    \inst{mstore}(L, M_O, X) & \leftarrow \inst{mstore}(L, O, X),  hasConstantValue(O)   \\
    \inst{mstore}(L, \top, X) & \leftarrow \inst{mstore}(L, O, X), \neg hasConstantValue(O)  \\
    \inst{mload}(L, Y, M_O) & \leftarrow \inst{mload}(L, Y, O), hasConstantValue(O)  \\
    \inst{mload}(L, Y, \top) & \leftarrow \inst{mload}(L, Y, O), \neg hasConstantValue(O) 
\end{align}
\begin{align}
    \pred{AssignVar}(L,Y,O) & \leftarrow_{\may} \inst{mload}(L, Y, O), \neg hasConstantValue(O)  
\end{align}
\paragraph{Optional \pred{Source} Rules}
\begin{align}
    \pred{Source}(L,Y,M_O) & \leftarrow \inst{mload}(L, Y, O), hasConstantValue(O)   \\
    \pred{Source}(L,Y,M_{\top}) & \leftarrow \inst{mload}(L, Y, O), \neg hasConstantValue(O)  
\end{align}
\begin{align}
    \pred{Source}(L,Y,Y) & \leftarrow \inst{call}(L,Y,\ldots) ~ \textit{or} ~ \inst{staticcall}(L,Y, \ldots) 
\end{align} 

\paragraph{$\may$ Control Flow and Dependency Propagation}
\begin{align}
    \pred{Follow}(L_1, L_2) & \leftarrow \inst{inst}_{pc}(L_1, \ldots), \inst{inst}_{pc+1}(L_2, \ldots), \textit{hasLinearSuc}(L_1) \\
    \pred{Follow}(L_1, L_3) & \leftarrow \inst{jumpI}(L_1, \_,L_3)\\
    \pred{Follow}(L_1, L_2) & \leftarrow \inst{jump}(L_1,L_2) 
\end{align}
\begin{align}
    \pred{Taint}(L_1,L_3,X) &  \leftarrow_{\may} \inst{jumpI}(L_1, X, L_3), L_3 \not = \textit{MergeInstr}(L_1)   \\
    \pred{Taint}(L_1,L_2,X) &  \leftarrow_{\may} "\inst{jumpI}(L_1, X, L_3); \inst{inst}(L_2,\ldots)", L_2 \not = \textit{MergeInstr}(L_1)  \\
    \pred{Join}(L_1, L_2) &  \leftarrow_{\may} \inst{jumpI}(L_1, X, L_3), L_2 = \textit{MergeInstr}(L_1)  
\end{align}

\paragraph{$\must$ Control Flow and Dependency Propagation}
\begin{align}
    OneBranchTag(L) & \leftarrow_{\must} \inst{jumpDest}(L), \text{one incoming branch}, \text{no prev inst in BB} \\
    Tag(L) & \leftarrow_{\must} \inst{jumpDest}(L)  
\end{align}
\begin{align}
    Jump(L_1, L_2, L_4), Follows(L_1, L_2) & \leftarrow_{\must} "\inst{jumpI}(L_1,\ldots); \inst{inst}(L_2,\ldots)", L_4 = \text{MergeInstr} \\
    Jump(L_1, L_3, L_4) & \leftarrow_{\must} \inst{jumpI}(L_1, \_,L_3), L_4 = \text{MergeInstr} \\
    Jump(L_1, L_2, L_2) & \leftarrow_{\must} \inst{jump}(L_1, L_2)  
\end{align}
\begin{align}
    Follow(L_1, L_2) & \leftarrow_{\must} "\inst{inst}(L_1,\ldots); \inst{inst}(L_2,\ldots)" \; \text{(same basic block)}  
\end{align}
\pred{JoinIncBranches} (\pred{JIB}) connects all incoming branches such that the $\must$-analysis can check if a predicate holds on all preceeding nodes:
\begin{align}
    \pred{JIB}(L_1, L_2, L^{\textit{inc}}_2), \ldots ,\pred{JIB}(L^{\textit{inc}}_n, L_n, L') \leftarrow_{\must} \forall i,j,k.
    \begin{cases}
        \vspace{0.1cm}
        "\inst{inst}_{\textit{pc}}(L_i,\ldots); \inst{jumpDest}_{\textit{pc}+1}(L')"  \\ \vspace{0.1cm}
        \inst{jump}(L_j,L') \\ \vspace{0.1cm}
        \inst{jumpI}(L_k,\_,L') 
    \end{cases} 
\end{align}

\subsubsection{May-Semantic Rules}
\label{app:a:2}
\begin{align}
    \pred{VarMayDepOn}(Y,X) & \mayrule \pred{AssignVar}(\_, Y, Y'), \pred{VarMayDepOn}(Y',X)   \\
    \pred{VarMayDepOn}(Y,X) & \mayrule \pred{AssignVar}(L, Y, \_), \pred{Taint}(\_, L, Y'), \pred{VarMayDepOn}(Y', X)    \\
    \pred{VarMayDepOn}(Y,X) & \mayrule \pred{Source}(L, Y, \_), \pred{Taint}(\_, L, Y'), \pred{VarMayDepOn}(Y', X)   
\end{align}
\begin{align}
    \pred{VarMayDepOn}(Y,X) & \mayrule \pred{Source}(\_, Y, X)   
\end{align}
\begin{align}
    \pred{InstMayDepOn}(L,X) & \mayrule \pred{Taint}(\_, L, X)   \\
    \pred{InstMayDepOn}(L,X) & \mayrule \pred{Taint}(\_, L, Y), \pred{VarMayDepOn}(Y, X)  
\end{align}

\paragraph{Memory Dependency Propagation}

\begin{align}
    \pred{MemMayDepOn}(L,O,T) & \mayrule \inst{mstore}(L, O, X), \pred{VarMayDepOn}(X, T)  \\
    \pred{MemMayDepOn}(L,O,T) & \mayrule \pred{Follows}(L_1, L), \pred{MemMayDepOn}(L_1, O, T), \\ & \neg \pred{ReassignMem}(L, O)  \nonumber
\end{align}
\begin{align}
    \pred{ReassignMem}(L,O) \Leftarrow \inst{mstore}(L, O, \_), \pred{isConst}(O)  
\end{align}
\begin{align}
    \pred{Source}(L, Y, T) & \Leftarrow \inst{mload}(L, Y, O), \pred{MemMayDepOn}(L,O,T), \pred{isConst}(O)  \\
    \pred{Source}(L, Y, T) & \mayrule \inst{mload}(L, Y, O), \pred{MemMayDepOn}(L,\_ ,T), \neg \pred{isConst}(O)  
\end{align}
\paragraph{Control Dependence Propagation} 
\begin{align}
    \pred{Taint}(L_1, L_2, X) & \mayrule \pred{Follow}(L_3,L_2), \pred{Taint}(L_1, L_3, X), \neg \pred{Join}(L_1, L_2) 
\end{align}

\subsubsection{Must-Semantic Rules}
\label{app:a:3} 

\begin{align}
    \pred{MustFollow}(L_1, L_2)  & \mustrule \pred{MustPrecedeStep}(L_1,L_2)   \\
    \pred{MustFollow}(L_1, L_3)  & \mustrule \pred{MustFollow}(L_1,L_2), \pred{MustFollow}(L_2,L_3)  \\
    \nonumber \\
    \pred{MustPrecedeStep}(L_1, L_2)  & \mustrule \pred{Follow}(L_1,L_2), \neg \pred{Tag}(L_2) \\
    \pred{MustPrecedeStep}(L_1, L_3)  & \mustrule \pred{Jump}(L_1, L_2, \_ ), \pred{oneBranchTag}(L_2)  \\
    \pred{MustPrecedeStep}(L_1, L_2)  & \mustrule \pred{Jump}(L_1, \_ , L_2 ) 
\end{align}
\begin{align}
    \pred{DetBy}(L,Y,X) & \mustrule \pred{Source}(L, Y, X)  \\
    \pred{DetBy}(L,Y,X) & \mustrule \pred{AssignVar}(L, Y, Y'), \pred{DetBy}(L,Y',X)  \\
    \nonumber\\
    \pred{DetBy}(L_2,Y,X) & \mustrule \pred{MustFollow}(L_1, L_2), \pred{DetBy}(L_1, Y, X)  \\
    \pred{DetBy}(L,Y,X) & \mustrule \pred{JoinIncBr}(L_1, L_2, L), \pred{DetBy}(L_1, Y, X), \pred{DetBy}(L_2, Y, X) 
\end{align}
\begin{align}
    \pred{MemDetBy}(L,O,T) & \mustrule \inst{mstore}(L, O, X), \pred{DetBy}(L,X,T), \pred{isConst}(O)   \\
    \pred{MemDetBy}(L_2,O,T) & \mustrule \pred{MustPrecedeStep}(L_1,L_2), \pred{MemDetBy}(L_1,O,T), \\ & ~~ \neg \pred{ReassignMem}(L_2, O) \nonumber \\
    \pred{MemDetBy}(L,O,T) & \mustrule \pred{JoinIncBr}(L_1, L_2, L), \pred{MemDetBy}(L_1, O, T), \\ & \; \; 
    \pred{MemDetBy}(L_2, O, T), (\pred{isConst}(O))   \nonumber
\end{align}


\end{appendices}
\fi

\end{document}